\documentclass[12pt]{report}   
\usepackage{titlesec}
\usepackage[draft=false]{microtype}
\usepackage{mathtools, adjustbox, float, bm, graphicx, epsfig, amsmath, amssymb, amsthm, url}
\usepackage[margin= 0.5in,footnotesize, sf, bf]{caption}
\usepackage{subfig, array, tikz, xspace, enumerate, amsfonts, ragged2e}
\usetikzlibrary{positioning}
\usepackage{grffile} 
\usepackage[ruled,linesnumbered]{algorithm2e}
\usepackage{wrapfig}
\usepackage{multirow}
\usepackage{msc}
\usepackage[numbers,sort&compress]{natbib}

\usetikzlibrary{
  decorations.pathreplacing,
  calligraphy,
  positioning,
  shapes,
  automata,
  arrows, 
  calc,
  decorations.markings,
  hobby,
  arrows.meta,
  shadows,
  backgrounds}
\tikzset{
        >=stealth, 
        node distance=3.5cm,
        every state/.style={rectangle, thick, fill=gray!10}, 
        initial text=$ $,
        looped/.style={
        decoration={markings,mark=at position 0.999 with {\arrow[scale=2]{>}}},
        postaction={decorate},
        >=stealth
    },
    straight/.style={
        decoration={markings,mark=at position 1 with {\arrow[scale=2]{>}}},
        postaction={decorate},
        >=stealth
    },
    loopedSF/.style={
        decoration={
            markings,
            mark=at position 0.999 with {\arrow[scale=2]{>}},
            mark=at position 0.5 with {\arrow[scale=2]{>}}},
        postaction={decorate},
        >=stealth
    },
    straightSF/.style={
        decoration={
            markings,
            mark=at position 0.999 with {\arrow[scale=2]{>}},
            mark=at position 0.5 with {\arrow[scale=2]{>}}},
        postaction={decorate},
        >=stealth
    },
    triangle/.style = {fill=white, draw=black, regular polygon, regular polygon sides=3 },
    node rotated/.style = {rotate=180},
    border rotated/.style = {shape border rotate=180}
    }
\usepackage{times, centernot}

\newtheoremstyle{normaltext}
  {3pt} 
  {3pt} 
  {\normalfont} 
  {} 
  {\bfseries} 
  {.} 
  {.5em} 
  {} 

\theoremstyle{normaltext}
\newtheorem{observation}{Observation}
\newtheorem{definition}{Definition}
\newtheorem{problem}{Problem}
\newtheorem{theorem}{Theorem}

\usepackage{algorithmic}

\newcommand{\Eref}[1] {Equation (\ref{#1})} 

\newcommand{\sem}[2]{#1[\hspace{-.2em}|#2\vert\hspace{-.2em}]}

\newcommand{\snd}{\mathit{snd}}
\newcommand{\dlv}{\mathit{rcv}}
\newcommand{\sndr}{\snd_r}
\newcommand{\dlvrr}{\dlv_s} 
\newcommand{\snds}{\snd_s}
\newcommand{\dlvrs}{\dlv_r}

\newcommand{\id}{\mathit{id}}

\newcommand{\sample}{\mbox{\sf S}\xspace}
\newcommand{\srtt}{\mbox{\sf srtt}\xspace}
\newcommand{\rtt}{\mbox{\sf rtt}\xspace}
\newcommand{\rttvar}{\mbox{\sf rttvar}\xspace}
\newcommand{\rto}{\mbox{\sf rto}\xspace}

\newcommand{\acls}{\textsc{ACL2s}\xspace}
\newcommand{\acl}{\textsc{ACL2}\xspace}
\newcommand{\aclr}{\textsc{ACL2(r)}\xspace}

\newcommand{\nums}{\textit{numT}}
\newcommand{\tims}{\textit{time}}
\newcommand{\band}{~\wedge~}
\newcommand{\Secl}[1]{\label{Sec#1}}
\newcommand{\Secr}[1]{Sec.~\ref{Sec#1}}
\newcommand{\Obl}[1]{\label{Obs#1}}
\newcommand{\Obr}[1]{Ob.~\ref{Obs#1}}
\newcommand{\Thl}[1]{\label{Thm#1}}
\newcommand{\Thr}[1]{Thm.~\ref{Thm#1}}
\newcommand{\Algl}[1]{\label{Alg#1}}
\newcommand{\Algr}[1]{Alg.~\ref{Alg#1}}
\newcommand{\Figl}[1]{\label{Fig#1}}
\newcommand{\Figr}[1]{Fig.~\ref{Fig#1}}
\newcommand{\Tabl}[1]{\label{Tab#1}}
\newcommand{\Tabr}[1]{Table~\ref{Tab#1}}
\newcommand{\Subsecl}[1]{\label{Subsec#1}}
\newcommand{\Subsecr}[1]{Subsec.~\ref{Subsec#1}}

\newcommand{\ack}{\textsc{Ack}\xspace}
\newcommand{\acks}{\textsc{Ack}s\xspace}
\newcommand{\nack}{N\ack}
\newcommand{\nacks}{N\textsc{Ack}s\xspace}
\newcommand{\sack}{S\ack}

\newcommand{\act}{\mathit{Act}}

\newcommand{\F}[0]{\mathbf{F}}
\newcommand{\U}[0]{\mathbf{U}}
\newcommand{\G}[0]{\mathbf{G}}
\newcommand{\X}[0]{\mathbf{X}}

\newcommand{\SYN}[0]{\texttt{SYN}\xspace}
\newcommand{\FIN}[0]{\texttt{FIN}\xspace}
\newcommand{\ACK}[0]{\texttt{ACK}\xspace}

\newcommand{\CLOSED}[0]{Closed\xspace}
\newcommand{\SYNSENT}[0]{Syn\_Sent\xspace}

\newcommand{\SYNREC}[0]{Syn\_Received\xspace}
\newcommand{\ESTABLISHED}[0]{Established\xspace}

\newcommand{\FINWAITTWO}[0]{Fin\_Wait\_2\xspace}
\newcommand{\CLOSEWAIT}[0]{Close\_Wait\xspace}

\newcommand{\LASTACK}[0]{Last\_Ack\xspace}
\newcommand{\TIMEWAIT}[0]{Time\_Wait\xspace}
\newcommand{\CLOSEREQ}[0]{Close\_Req\xspace}
\newcommand{\OPEN}[0]{Open\xspace}

\newcommand{\Closed}[0]{Closed\xspace}
\newcommand{\CookieWait}[0]{Cookie\_Wait\xspace}
\newcommand{\CookieEchoed}[0]{Cookie\_Echoed\xspace}
\newcommand{\Established}[0]{Established\xspace}
\newcommand{\ShutdownAckSent}[0]{Shutdown\_Ack\_Sent\xspace}
\newcommand{\ShutdownReceived}[0]{Shutdown\_Received\xspace}
\newcommand{\ShutdownSent}[0]{Shutdown\_Sent\xspace}
\newcommand{\ShutdownPending}[0]{Shutdown\_Pending\xspace}
\newcommand{\UserAssoc}[0]{User\_Assoc\xspace}
\newcommand{\UserShutdown}[0]{User\_Shutdown\xspace}
\newcommand{\UserAbort}[0]{User\_Abort\xspace}
\newcommand{\Init}[0]{\texttt{INIT}\xspace}
\newcommand{\InitAck}[0]{\texttt{INIT\_ACK}\xspace}
\newcommand{\CookieError}[0]{\texttt{COOKIE\_ERROR}\xspace}
\newcommand{\CookieEcho}[0]{\texttt{COOKIE\_ECHO}\xspace}
\newcommand{\Data}[0]{\texttt{DATA}\xspace}
\newcommand{\DataAck}[0]{\texttt{DATA\_ACK}\xspace}
\newcommand{\CookieAck}[0]{\texttt{COOKIE\_ACK}\xspace}
\newcommand{\Shutdown}[0]{\texttt{SHUTDOWN}\xspace}
\newcommand{\ShutdownAck}[0]{\texttt{SHUTDOWN\_ACK}\xspace}
\newcommand{\ShutdownComplete}[0]{\texttt{SHUTDOWN\_COMPLETE}\xspace}

\newcommand{\Abort}[0]{\texttt{ABORT}\xspace}
\newcommand{\spin}[0]{\textsc{SPIN}\xspace}
\newcommand{\amodel}[0]{attacker model\xspace}
\newcommand{\amodels}[0]{attacker models\xspace}
\newcommand{\Amodel}[0]{Attacker Model\xspace}

\newcommand{\cur}[0]{\texttt{curPkt}\xspace}
\newcommand{\ha}[0]{\texttt{hiAck}\xspace}
\newcommand{\hp}[0]{\texttt{hiPkt}\xspace}

\newcommand{\continue}[0]{\continue\xspace}

\newcommand{\rcvd}[0]{\textbf{r}\xspace}

\newcommand{\bcap}[0]{\texttt{bcap}\xspace}
\newcommand{\dcap}[0]{\texttt{dcap}\xspace}
\newcommand{\bucket}[0]{\texttt{bkt}\xspace}
\newcommand{\rat}[0]{\texttt{rt}}
\newcommand{\delay}[0]{\texttt{del}}
\newcommand{\data}[0]{\texttt{dgs}}
\newcommand{\ttl}[0]{\delay}
\newcommand{\promela}[0]{\textsc{Promela}\xspace}
\newcommand{\korg}[0]{\textsc{Korg}\xspace}

\newenvironment{sketch}{%
  \proof
}{\endproof}

\newcommand*{\funcfont}{\fontfamily{lmss}\selectfont}
\newcommand*{\codefont}{\ttfamily\small}

\newcommand{\Chapl}[1]{\label{Chapter#1}}
\newcommand{\Chapr}[1]{Chapter~\ref{Chapter#1}}

\newcommand{\Eql}[1]{\label{Equation#1}}
\newcommand{\Eqr}[1]{\Eref{Equation#1}}

\newcommand{\ffirst}{\mathit{high}}

\newcommand{\sender}[0]{\textbf{s}}

\newcommand{\dg}[0]{\textbf{d}}
\newcommand{\trsender}[0]{\textsf{senderR}}

\newcommand{\evt}[0]{\textbf{e}}
\newcommand{\trreceiver}[0]{\textsf{receiverR}}
\newcommand{\tbf}[0]{\textbf{tbf}}

\newcommand{\trtbf}[0]{\textsf{tbfR}}
\newcommand{\trsys}[0]{\textsf{sysR}}
\newcommand{\sys}[0]{\textbf{sys}}

\newcounter{invc}
\newcommand\inv{\emph{I\arabic{invc}}\addtocounter{invc}{1}}

\DeclarePairedDelimiter{\ceil}{\lceil}{\rceil}

\usepackage[sc]{mathpazo} 
\usepackage[T1]{fontenc}

\tikzstyle{thmbox} = [
  draw=blue,
  rectangle,
  align=center,
  text width=3cm,
  inner sep=0.3cm
]

\DeclareTextFontCommand{\funcfontify}{\funcfont}
\DeclareTextFontCommand{\codefontify}{\codefont}
\newcommand{\codify}[1]{\ensuremath{\mbox{\codefontify{#1}}}}

\usepackage[defaultsans, scale=0.86]{opensans}
\usepackage[T1]{fontenc}
\usepackage[scaled=1.03]{inconsolata} 
\usepackage[T1]{fontenc}

\usepackage{xcolor}

\definecolor{mycolor}{rgb}{0.5,0.05,0.0}
\usepackage{hyperref} 
\hypersetup{
    colorlinks=true,
    citecolor=black,
    linkcolor=black,
    filecolor=black,      
    urlcolor=mycolor,
}

\usepackage{sectsty}
\allsectionsfont{\bfseries\sffamily\color{black}}
\subparagraphfont{\hspace*{2em}\bfseries\sffamily\color{black}}
\definecolor{Goldenrod}{rgb}{0.85, 0.65, 0.13}

\titleformat{\chapter}[display]
  {\normalfont\Large\bfseries\color{black}}
  {\chaptertitlename\ \thechapter}{20pt}{\Large}

\usepackage[strict]{changepage}  
\usepackage{lstautogobble}
\lstdefinelanguage{acl2s}{
  keywords={defdata,definecd,definec,defun,property,create-map*,create-reduce*},
  comment=[l]{;;},
}

\author{Max von Hippel}
\title{Verification and Attack Synthesis for Network Protocols}
\begin{document}
\begin{titlepage}
\centering
\begin{minipage}[c]{0.9\textwidth}
        \centering
        {\quad}\vspace{2.2in} \\
        \textsf{\Large Ph.D. Dissertation}
        \rule{\linewidth}{1pt}\\[6pt]{\huge Verification and Attack Synthesis for Network Protocols}\\ \rule{\linewidth}{2pt}\\[0.1in]
\end{minipage}
\vfill 
{\Large{Max von Hippel}\\[0.5em]}
Khoury College of Computer Science\\[1em]
Northeastern University\\[0.5in]
{ \bfseries\sffamily Ph.D. Committee \\[1.5em]}
\begin{table}[H]
\centering
\begin{tabular}{rcl}
        \bfseries\sffamily{Advisor} & Cristina Nita-Rotaru & \\[1.2em]       
        & Pete Manolios  & \\[1.2em]         
        & Guevara Noubir  &\\[1.2em]     
        \bfseries\sffamily{Ext. member} & Joseph Kiniry & \bfseries\sffamily{Galois} \\[2.5em]
\end{tabular} 
\end{table}
\today
\end{titlepage}
\thispagestyle{empty}
\newpage

\quad\vspace{2in}\quad 
\abstract{
Network protocols are programs with inputs and outputs that follow predefined communication patterns to synchronize and exchange information.  There are many protocols and each serves a different purpose, e.g., routing, transport, secure communication, etc.
The functional and performance requirements for a protocol can be expressed using a formal specification,
such as, a set of logical predicates over its traces.
A protocol could be prevented from achieving its requirements due to a bug in its design or implementation,
a component failure (e.g., a crash),
or an attack.
This dissertation shows that formal methods
can feasibly characterize the functionality and performance of network protocols 
under normal conditions as well as when subjected to attacks.

\medskip

We study the formal verification of protocol correctness and performance in the absence of an attack through the lens of three case studies: Karn’s Algorithm, the retransmission timeout (RTO), and Go-Back-$N$.  Karn’s Algorithm has been widely used to sample round-trip times (RTTs) on the Internet since 1987, particularly for congestion control, but until now, it was never formally analyzed.  We formalize it in Ivy and prove novel correctness properties, e.g. that it measures a real and pessimistic RTT.  The RTO is defined in RFC~6298 and computes, as a function of the outputs from Karn’s Algorithm, the time the sender will wait for a new \ack before timing out and retransmitting unacknowledged packets.  If the RTO is too small then the sender will timeout unnecessarily, leading to congestion, but if it is too large then the sender will take too long to respond when congestion does occur.  We model the RTO calculation using \acls and verify bounds on its internal variables, concretely and asymptotically.  Then we illustrate an edge-case where infinitely many timeouts could occur despite stable network conditions.  Finally, also in \acls, we model Go-Back-$N$, which is the basis for TCP’s sliding window mechanism.  Using our model, we formally analyze the performance of Go-Back-$N$ in the presence of losses -- in particular those caused by the queuing mechanism, which we model as a generalized token bucket filter (TBF).  Using bisimulation arguments, we prove that Go-Back-$N$ can theoretically achieve perfect efficiency, and we derive a formula for its efficiency when the sender constantly over-transmits.

\medskip

Then, we turn our attention to the automated discovery of attacks which,
under a given \amodel,
can cause a protocol to malfunction.
Many prior works automatically found attacks using heuristic or randomized techniques, however, our approach is novel and rooted in formal methods.
Specifically, we
explore the under-studied approach of attacker synthesis, which is challenging and different from program synthesis because it takes into account the existing protocol as well as the attacker model.  In contrast to heuristic attack discovery techniques, attacker synthesis is rooted in formal methods and involves automatically generating attacks in a way that is sound and, in the setting we study, complete.  We propose a novel formalization for a general attacker synthesis problem, taking into account the protocol, placement and capabilities of the attacker, requirement that the attack terminates, and correctness definition for the system.  To the best of our knowledge no prior works proposed such a general framework.  The correctness specification is the negation of the attacker goal, formally capturing the intuition that the goals of the system builder and hacker are at odds.  We propose a solution to our problem, based on model-checking, and implement it in an open-source tool called \korg.  We apply \korg to TCP, DCCP, and SCTP, reporting attacks against each.  In SCTP we find two specification ambiguities, each of which, we show, can open the protocol to attack, as confirmed by the chair of the SCTP RFC committee, and we suggest edits to clarify both.  Finally, we prove that \korg is sound and complete, and can thus be used to prove that a patch resolves a vulnerability, which we demonstrate with SCTP.
}
\newpage

\tableofcontents
\newpage
\pagenumbering{arabic}

\chapter{Introduction}\Chapl{intro}
In this chapter, we begin the dissertation by providing an overview of network protocols,
including all the case studies we analyze, as well as the formal methods we use for our analysis.
First we explain why the correctness of these protocols matters (\Secr{intro:motivation}), and
what role each case study plays in the proper functioning of the Internet (\Secr{intro:background:protocols}).
Then in \Secr{intro:formalmethods} we describe the formal methods we use to analyze these protocols,
including theorem proving, model checking, and synthesis.
We describe our contributions in \Secr{intro:contributions} and outline the rest of the dissertation in
\Secr{intro:outline}.

\section{Motivation}\Secl{intro:motivation}

The Internet consists of \emph{protocols}, which
    allow computers to connect and communicate --
    for example, the Transmission Control Protocol (TCP)~\cite{rfc9293_tcp_new}, 
    Datagram Congestion Control Protocol (DCCP)~\cite{rfc4340_dccp}, 
    Stream Control Transmission Protocol (SCTP)~\cite{rfc9260}, and so on.
Each protocol is designed to give slightly different guarantees,
such as, reliable communication, secure communication, or eventual consensus.
Unfortunately, Internet protocols are not typically designed from the ground up in a mathematically
rigorous way that could assure they actually deliver on those promises.
For example, none of the widely used Internet protocols were generated using program synthesis
techniques to provably satisfy a logical specification.
For many protocols, there does not even exist a mathematical specification of what 
it would mean for the protocol to be correct or incorrect, i.e., of the protocol goals,
let alone a proof thereof.
The performance requirements protocols must meet in order to be practically useful
are likewise often left unstated.
This presents a serious problem because the Internet is the backbone of the modern world economy~\cite{manyika2011great} 
and undergirds essential infrastructure such as emergency services~\cite{voip911} 
and power grids~\cite{powerGrid}.
It is therefore essential that Internet protocols work correctly, since malfunction could mean not only 
serious monetary loss but potentially even loss of life.
The situation is made more grave by the fact that the Internet is rife with hackers, 
namely, attackers who try to maliciously trick protocols and other programs into malfunctioning
for economic or sociopolitical gain.

Note that most of the time, the Internet works correctly -- emails are sent and received,
webpages load in fractions of a second, etc.
But, this status quo is sometimes interrupted by malfunctions or attacks.
For example:
\begin{itemize}
    \item In October 1986, the National Science Foundation Network, a predecessor to the World Wide Web, 
    dropped in throughput from 32 Kbps to 40 bps.
    The drop was caused by (random) congestion on the network, which the protocol in use
    was not equipped to deal with (congestion control had not yet been invented).
    The incident inspired the first (and seminal) work on congestion control~\cite{jacobson1988congestion}.
    \item In October 2013, Hurricane Sandy physically damaged network infrastructure leading to a double in the number of Internet outages across the United States over a four-day period~\cite{heidemann2012preliminary}.
    \item In late 2016, the Mirai botnet infected over 600k Internet-of-Things devices, such as routers, DVRs, and cameras, particularly in Brazil, Columbia, and Vietnam.
    The botnet performed denial-of-service attacks on multiple targets, including the popular blog Krebs on Security as well as the telecommunications company Deutsche Telekom~\cite{antonakakis2017understanding}.
    The latter attack caused an Internet outage for around 900k customers~\cite{dutchTelekom}. 
    \item In June 2019, a BGP routing leak in a fiber-optic services provider used by Verizon lead to roughly day-long outages at Reddit, Discord, Google, Amazon, Verizon, and Spectrum~\cite{routingOutage}.
    \item In July 2024, a bug in the Crowdstrike Falcon software caused a global internet outage grounding United, American, Delta, and Allegiant airlines, delaying US/Mexico border crossings, disrupting courts in Massachusetts and New York, and even forcing some hospitals to suspend visitation~\cite{crowdstrikeOutage}.
\end{itemize}
Thus, although the Internet generally functions correctly, it sometimes malfunctions, leading to outages
or decreased performance.
These malfunctions can be caused by flaws or limitations in the protocols in use, physical damage to networking equipment, bugs, or even attacks.
For a detailed analysis of Internet outages and their causes, the reader is referred to~\cite{bogutz2019identifying}.

However rare, malfunctions or attacks like these have clear real-world impacts.
Motivated by these impacts, in this dissertation, we show that
formal methods can feasibly characterize the functionality and performance 
of network protocols under normal conditions as well as when subjected to attacks.

\section{Network Protocols}\Secl{intro:background:protocols}

The Internet was first conceived by J.C.R. Licklider in 1962,
and the first computer network, consisting of just two nodes,
was established between Massachusetts and California in 1965 over a telephone line~\cite{leiner2009brief}.
Today, the ``Internet'' refers to the World Wide Web, which operates according to dozens of protocols
defined in academic papers or
by the Internet Engineering Task Force (IETF) in so-called Request For Comments documents, or RFCs.
Internet applications communicate by implementing the logic outlined in these papers or documents,
allowing them to send and receive messages according to a common set of shared rules.

From its conception, the Internet was built to tolerate unreliability in a layered, best-effort fashion
known as the \emph{end-to-end argument}~\cite{saltzer1984end}.
The idea is that certain functions of a modular, multi-layered system (such as the Internet)
can only be reliably provided at the application layer, that is, from the perspective of a service
which controls all ``end points'' of the system.
This application should assume that faults may have been introduced at any point between those ends,
and do error detection (and potentially, correction) on a best-effort basis.
As an example, \emph{transport protocols} are protocols that provide communication services to applications
running on different hosts~\cite{ross2012computer}.
In a transport protocol, the receiver of a sequence of messages cannot assume that they are un-corrupted, nor can it assume that they were delivered in the same order they were sent.
Rather, it must use application-level mechanisms such as checksums or sequence numbers to gain these assurances.
Such mechanisms allow transport protocols to achieve a number of useful goals, such as reliability (where messages are delivered to the application by the receiver in the same order that they were transmitted to the receiver from the sender) or latency guarantees.

One common feature of transport protocols is the need to deal with \emph{congestion}, where the sender transmits packets more quickly than the network is able to deliver them, leading to losses.
The problem is tricky because messages between networked computers experience at least speed-of-light delay between transmission
and delivery, and the exact delay depends on physical conditions and the network state.  Worse still, data can be reordered or lost in transit.
Hence, no computer in a network can ever know the current, instantaneous state of all
the other computers in the network~\cite{lamport2019time}, which in the context of controlling congestion,
means that the sender cannot directly determine the instantaneous congestive state of the network.
Protocols deal with this epistemic dilemma using various kinds of feedback and measurements.
For example, in many protocols, the receiver of a message provides feedback in the form of a special acknowledgment message, called an \ack.
\acks are essential for building reliable protocols since they let a sender determine when
some data has been successfully delivered, so the sender can send the next chunk of data
in its queue.
A measurement which is used in many protocols is the round-trip-time, or RTT.
This is the time elapsed between when a sender transmits a message and when it first receives
an \ack indicating the message was delivered.
Intuitively, the RTT measures the speed of the network, and is useful for detecting congestion,
where the network becomes overwhelmed and starts dropping messages.

Measuring the RTT is straightforward if every message has a unique identification number
(commonly called a sequence number), and each \ack includes information indicating which specific
sequence numbers are being acknowledged, as is the case in the protocol QUIC~\cite{rfc9000_quic}.\footnote{QUIC initially stood for ``Quick UDP Internet Connections''~\cite{quicDraft}, but today, the IEEE does not consider it to be an acronym~\cite{rfc9000_quic}.}
However, in many protocols, such as TCP, when a message is deemed to be lost and is therefore retransmitted, the retransmission carries the same sequence number as the original.
Then, when a corresponding \ack arrives, it is impossible to tell if the \ack is for the retransmission or the original.
The most popular solution to this dilemma is called Karn's Algorithm~\cite{karn1987},
and the idea is simple: only measure RTTs using unambiguous \acks.

Protocols use measurements, such as the RTT measurements output by Karn's Algorithm,
to make inferences about the likely current state of the network and,
as a consequence of those inferences, concrete decisions about what action to take next.
For example, many protocols utilize the RTT samples output by Karn's Algorithm
to compute a Retransmission TimeOut (RTO) value, which is the amount of time the
sender will wait for any \ack to arrive acknowledging previously unacknowledged data,
before it assumes that the data in-transit must have been lost, and retransmits.
This computation is most commonly done using the RTO formula defined in RFC 6298~\cite{RFC6298}, or some variant thereof.
And yet, despite the widespread use of both Karn's Algorithm to measure RTT samples,
and the RTO computation based on those samples defined in RFC 6298, neither of these
critical protocol components were ever previously studied using formal methods.

Although the RTT and RTO are used in many types of protocols, 
perhaps their most fundamental role is in the implementation of reliable transport protocols such as TCP.
Transport protocols form the \emph{transport} layer of the Internet, facilitating end-to-end
communication between computers.
Reliable transport protocols are ones where packets are delivered to the application by the receiver
 in the same order
that they are transmitted by the sender, without omissions.
These protocols typically use the RTO to detect when messages were lost, and retransmit accordingly.
Such protocols face an inherent trade-off between how quickly they can progress
in the best case (when there are no timeouts) and worst 
case (when the sender is forced to retransmit).

As an example of this trade-off, consider the difference between the toy protocol Stop-and-Wait,
and the protocol Go-Back-$N$.
In Stop-and-Wait, the sender transmits one message at a time,
and will not send the next message in its queue until it has received an \ack
for the prior one.
So, in the best case, when there are no timeouts, the Stop-and-Wait sender progresses
slowly, requiring a new \ack after each transmission and before the next.
But in the worst case, it only needs to retransmit one message, since all the previous
ones were already acknowledged.
In contrast, in Go-Back-$N$, for some positive fixed integer $N$,
the sender may transmit as many as $N$ messages before requiring that any of them
be acknowledged.  In the best case, this means the sender can progress more quickly than
it could in Stop-and-Wait, since it can send the next message in the queue while still awaiting
the \ack for the prior one.  But in the worst case, it could be forced to retransmit all $N$
messages.
Note, Stop-and-Wait is simply Go-Back-1.
Although some prior works analyzed the average performance of Go-Back-$N$,
no prior works formally analyzed its best and worst-case performance,
nor how this trade-off scales as a function of $N$.

In transport protocols, communication does not just happen out of the blue.
Rather, the sender and receiver establish a connection using a communication pattern
known as an \emph{establishment routine}.
Once a connection is established, the sender begins transmitting its internal
message queue to the receiver, who responds with corresponding \acks.
Then at some point, either the sender or the receiver initiates a \emph{tear-down routine},
which is similar to the establishment routine but serves to de-associate,
deleting the connection.
The conjunction of the two routines is commonly referred to as the protocol \emph{handshake}.

There are many transport protocols, and in general each provides a slightly different trade-off between
features (such as reliability, in-order delivery, congestion control features, etc.)
and performance.
TCP is the most fundamental and oldest reliable transport protocol on the Internet,
and guarantees reliable, in-order packet delivery.
It has many variants, 
e.g., TCP Vegas~\cite{brakmo1995tcp}, TCP New Reno~\cite{rfc6582}, etc., 
but all them use the same handshake,
defined in RFC 9293~\cite{rfc9293_tcp_new}.
DCCP is similar to TCP, but does not guarantee in-order message delivery~\cite{rfc4340_dccp}.
SCTP is a comparatively newer transport 
protocol proposed as an alternative to TCP,
offering enhanced performance, security features, 
and greater flexibility.
It is specified in several RFCs, each introducing significant modifications. 
RFC 9260~\cite{rfc9260}, which obsolesced RFC 4960~\cite{rfc4960}, 
made numerous small clarifications and improvements, 
including a critical patch for CVE-2021-3772~\cite{cve},
a denial-of-service attack made possible by an ambiguity in RFC 4860
which the Linux implementation misinterpreted~\cite{linux}.
On the other hand, RFC 4960, which obsolesced the original
specification in RFC 2960~\cite{rfc2960}, 
introduced major structural
changes to the protocol as described 
in the errata RFC 4460~\cite{rfc4460}. 
Although each RFC ostensibly represents an improvement over the prior,
it is not obvious that these improvements do not introduce new bugs or vulnerabilities --
to confirm this, we need some kind of formal verification.
Each of these protocols are crucial to the proper functioning of the Internet, 
and each one uses a different and unique handshake.

The classical way to verify a protocol handshake is to encode its goals as logical properties,
encode the handshake as a state machine,
and then use a model checker to verify that the state machine satisfies those properties.
Unfortunately, the state machine descriptions given in RFC documents are informal
and may have omissions, mistakes, or simplifications.
Moreover, the correctness properties these machines are expected to satisfy
are rarely made explicit.
To assure that commonly used transport protocols like TCP, DCCP, and SCTP
operate correctly, what we need are corresponding mathematical state machine models
and the logical properties those models are expected to satisfy, based on a close
reading of the RFCs (and not just a literal interpretation of the ASCII diagrams they contain).

Finally, once we have rigorously determined that a protocol 
works correctly at all levels,
we still need to show that it is robust against attacks.
This requires formalizing a notion of \emph{attacker model}, taking into consideration
the placement and capabilities of the attacker, and then showing that even under that
attacker model, the protocol still satisfies all of its correctness properties.
If a protocol property can be violated under a realistic attacker model, this implies that
the protocol is not secure against the modeled attack, and therefore must either be 
patched to provide an adequate defense, or restricted in its use to only scenarios where such an attack is impossible.

\section{Formal Methods}\Secl{intro:formalmethods}

In this dissertation we study network protocols using \emph{formal methods}.
These are techniques for analyzing or generating systems, particularly software systems,
using formal mathematics in a computer-aided environment.
At high level, the primary techniques in formal methods include theorem proving,
model checking, synthesis, and lightweight formal methods such as property-based testing and 
grammar-based fuzzing.
We use the first three techniques in this dissertation -- each of which we describe below.
Our thesis is that these methods make it feasible to rigorously analyze the correctness
and performance of network protocols both in isolation and when subjected to attacks.

\subsection{Theorem Proving}

An \emph{interactive theorem prover} is a software system 
in which a computer and a human can collaborate to write a mathematical proof.
In other words, a theorem prover is like an integrated developer environment (IDE)
for mathematical reasoning.
The least powerful kind of theorem prover is one that checks a human-written
proof and confirms that it is devoid of mistakes, that is, that each step in the
proof syntactically follows from the previous steps.
This style of reasoning -- evocative of the ultra-formalism of the Bourbaki group~\cite{dieudonne1970work} --
can be quite onerous, but has the benefit of producing bulletproof arguments.
On the other hand, the most powerful kind of theorem prover is one that 
automates a significant portion of the proof-writing process (in addition to checking
that each proof step follows from the prior ones).
In practice, most provers fall somewhere between those two extremes --
    at times automating proof steps, saving a considerable amount of proof effort,
    but at other times obligating the human to justify intuitively obvious proof steps.
For a nice history and survey of interactive theorem proving, the reader is referred to~\cite{maric2015survey} or~\cite{mackenzie1995automation}.

Unfortunately, it is not possible to outline a single set of mathematical principals 
which suffice to understand all of the interactive theorem provers.  This is because 
different theorem provers accommodate different \emph{logics}, which can 
differ in terms of both their foundations and logical order.
The \emph{foundations} of a logic are the axioms it assumes, while the \emph{order} of a logic
refers to its level of abstraction.
A first-order logic allows predicates over atomic propositions,
while a second-order logic allows predicates over sets of propositions,
a third-order logic allows predicates over sets of sets of propositions, etc.
More philosophically, a first-order logic allows one to reason about all objects
in a universe; a second-order one about all properties of
objects in a universe; a third-order logic about properties
of properties of objects in a universe; and so on.

Although this can all seem quite abstract, these distinctions have a very real impact
on the types of theorems one can prove.
For example, most mathematicians today work within Zermelo–Fraenkel set theory
with the Axiom of Choice (aka ZFC), which is a first-order logic highly 
amenable to set-theoretic reasoning.  However, this logic allows a proof which says that
a single unit sphere can be split into an infinite number of slices, which can be re-assembled (without collision)
into two unit spheres each equal in volume to the original~\cite{stdecomposition,wilson2005continuous}.\footnote{See also~\cite{bapanapally2022complete} for a formal verification of the result in question.}  This proof contradicts our natural
intuition about surface area and volume, drawing into question the 
closeness of ZFC to our lived experience of the universe we reside in.
On the other hand, Homotopy Type Theory (HoTT) is a newer, alternative type-theoretic foundation for mathematics 
in which, loosely speaking, isomorphism and strict
equality are defined to mean the same thing~\cite{program2013homotopy}.
HoTT, in contrast to ZFC, does not include the Axiom of Choice.
Note that some provers can support multiple logics, e.g., it is possible to use either ZFC or HoTT in Rocq\footnote{Formerly known as Coq.}~\cite{werner1996encoding,bauer2017hott}.

In this dissertation, we use two provers.
The first, Ivy~\cite{padon2016ivy}, is a tool for proving inductive invariants of protocols.
It is highly automated, and attempts to split theorems into individual proof obligations
in logics for which it has decision procedures.
Ivy is very flexible and allows the user to design and specify any logical foundations they please.
However, in practice, the tool becomes highly unstable as soon as sufficient axioms are introduced
to leave the decidable fragment, at which point even a very small model change can cause the tool
to be unable to generate a proof or disproof.
For example, the Peano Arithmetic axioms, which are the most commonly used axioms for arithmetic,
suffice to exit the decidable fragment.
In this dissertation we use Ivy with its default theory, which provides useful axioms for reasoning
about lists and list manipulations.

The second prover we use is a Boyer-Moore theorem prover~\cite{boyer1975proving,boyer1979computational} called A Computational Logic for Applicative Common Lisp (\acl)~\cite{kaufmann1997industrial}.
\acl uses an extensible foundation built on top of traditional propositional calculus with equality.
Its exact foundations are fairly elaborate because it accommodates all of Common Lisp,
but informally: it allows the user to express and prove formulas over recursive functions on variables and constants~\cite{kaufmann1997precise}.
These formulas are quantifier-free, meaning, they are implicitly universally quantified.
We also use two variants of \acl.
The first, the \acl Sedan (\acls)~\cite{dillinger2007acl2s}, extends \acl with a data definition framework (DefData)~\cite{chamarthi2014data}, ordinals~\cite{manolios2004integrating}, termination analysis based on context-calling graphs~\cite{manolios2006termination}, and counterexample generation via the cgen library~\cite{chamarthi2011integrating}.
The second, \aclr, uses a slightly different foundation in order to support nonstandard analysis with real numbers~\cite{gamboa1999mechanically,gamboa2001nonstandard}.
However, we do not perform any nonstandard analysis in this dissertation; we only use \aclr to prove a theorem involving irrationals which could not be proven in \acl or \acls (neither of which supports a theory of irrational numbers).

\subsection{Model Checking}\Secl{modelchecking}

In contrast to theorem proving, which is inherently interactive and highly flexible,
model checking is totally automatic but restricted to only problems over very small domains.
In this dissertation, we use the \spin model checker~\cite{holzmann1996nested} to verify Linear Temporal Logic (LTL) properties of finite Kripke Structures, which are finite state transition systems where the states are labeled with atomic propositions.  When a finite Kripke structure $K$ takes a sequence of transitions through its states, we refer to the sequence as a \emph{run}, and to the corresponding sequence of labels on those states as an \emph{execution}.  LTL allows us to write statements about the temporal occurrence of different labels in an execution, using the operators ``until'' and ``next''.  For instance, if the second state in the execution~$\sigma$ of a run~$r$ has the label \texttt{crit}, then the corresponding execution~$\sigma$ satisfies ``next \texttt{crit}'', written $\X \texttt{crit}$,
and we write $\sigma \models \X \texttt{crit}$.
On the other hand, if the trace does not satisfy $\X \texttt{crit}$, then we would write $\sigma \centernot{\models} \X \texttt{crit}$.
Likewise, if the trace induced by the run $r$ satisfies a property $\phi$ then we write $r \models \phi$, else we write $r \centernot{\models} \phi$.
We naturally lift this notation to finite Kripke structures, in the sense that if every execution of $K$ satisfies $\phi$ then we write $K \models \phi$, else if any execution violates $\phi$ then we write $K \centernot{\models} \phi$.

An LTL model checker takes as input a finite Kripke Structure $K$ and an LTL property $\phi$ and return \emph{true} iff $K \models \phi$, else some $r \in \text{runs}(K)$ such that $r \centernot{\models} \phi$.
The decision procedure for LTL model checking was discovered by Vardi and Wolper~\cite{IEEECS86}
and implemented, with some optimizations (e.g.~\cite{holzmann1998analysis,holzmann1995improvement,wolper1993reliable,holzmann1996nested,holzmann2007design,holzmann2008swarm}) 
in the model checker \spin~\cite{Holzmann03}.
The basic premise is as follows.
First, the LTL property $\phi$ is translated to a so-called B{\"u}chi automaton $B(\phi)$,
according to the procedure outlined in~\cite{gerth1995simple}.
The B{\"u}chi automaton can be viewed as a finite Kripke structure with the atomic propositions 
$\text{props}(\phi) \uplus \{ \text{accepting} \}$,
where 
$\uplus$ denotes disjoint union,
$\text{props}(\phi)$ are the atomic propositions which appear in $\phi$,
and the language of the automaton,
denoted $\mathcal{L}(B(\phi))$, is the subset of its traces in which it passes through an accepting state
infinitely many times.\footnote{Since the automaton is finite-state, if it passes through the set of accepting
states infinitely often, then it must also pass through some particular accepting state infinitely often.}
The interesting thing about the B{\"u}chi automaton is that its language is precisely the complement of the language of the property from which it was generated.
That is to say, if $\mathcal{L}(\phi)$ is the set of all possible infinite sequences $\sigma$
 of sets of atomic propositions such that for each $\sigma \in \mathcal{L}(\phi)$, $\sigma \models \phi$, then $\overline{\mathcal{L}(\phi)} = \mathcal{L}(B(\phi))$ (and vice versa).
Thus, the model-checking problem reduces to checking language emptiness on $\mathcal{L}(K) \cap \mathcal{L}(B(\phi))$.
For a tutorial on the topic, the reader is referred to~\cite{buchiNotes},
or for a more comprehensive treatment,~\cite{baier2008principles}.

\subsection{Program Synthesis}

Program synthesis is the task of, given some logical specification, automatically generating a program that meets it.  The concept was first introduced by Church in an unpublished talk at the Institute for Defense Analysis in 1957~\cite{alonzo57}, and has since grown into an expansive field with myriad approaches and sub-problems, e.g., where the specification is written in LTL~\cite{pnueli1989synthesis} or in Computational Tree Logic~\cite{clarke1981design}.

Generally speaking, the synthesizer performs a search over a program-space, which it constrains
(often the constraint is iterative) in order to find a satisfying example.
Unfortunately, the general program synthesis problem is undecidable, since the search involves
checking non-trivial features of Turing Machines.
However, like many problems in formal methods, it can be made tractable
for real-world problems by limiting the specification language
and augmenting the search algorithm with clever tricks, heuristics, and optimizations~\cite{gulwani2017program}.
For example, Flash Fill is a feature of Microsoft Excel that digests 
some example input cells and an output cell --
    for instance, as inputs, ``November'', ``3'', and ``2012'', and as output, ``11/3/12'' --
    and fills in the remainder of the corresponding column according to a pattern it derives
    which maps the example inputs to the example output, in fractions of a second~\cite{gantenbein2013flash}.
Part of what makes the algorithm fast is that it is limited to disallow the Kleene star or the disjunction operator, which allows much of its decision procedure to be reduced to regular expressions.
In addition, it is designed to ask the user for more examples, when necessary~\cite{gulwani2011automating}.
Because program synthesis inherently involves searching a space of possible programs, most techniques involve reducing the search to a common search technique such as integer linear programming or satisfiability modulo theories (for a nice survey of such techniques the reader is referred to~\cite{gulwani2017program}).
However, with the advent of language models, there are now a new class of neurosymbolic techniques which leverage machine learning algorithms trained on vast quantities of human-written computer code to synthesize programs.  For a survey of these (rapidly emerging) techniques, the reader is referred to~\cite{chaudhuri2021neurosymbolic}.
In this dissertation, we define a constrained type of synthesis, where the program being generated
only needs to have at least a single run in which it can induce a particular system to misbehave,
and we reduce the program-search to an LTL model-checking problem.

\section{Thesis Contribution}\Secl{intro:contributions}

In this dissertation we study protocols from the ground up using formal methods.
Our contributions are as follows.

\begin{itemize}
\item \textbf{Models.}
We develop formal models of Karn's Algorithm, the RTO computation, Go-Back-$N$, TCP, DCCP, and SCTP.
To the best of our knowledge, 
    neither Karn's Algorithm nor the RTO computation was ever previously formally modeled,
    and we are the first to model Go-Back-$N$ non-probabilistically in the context of a non-trivial
    rate-limiting channel.
Our channel is, we argue,
more realistic than those used in comparable prior works,
while still being compositional in the sense that the serial composition of two channels
can be simulated by just one single one.
Finally, our TCP, DCCP, and SCTP handshake models are more complete than comparable 
models introduced in prior works.
We provide detailed comparisons to prior works in each chapter.

\item \textbf{Properties.}
In addition to new models, we also introduce formal properties which, we claim, the modeled protocols should satisfy.  We justify our properties based on a close reading of the corresponding academic literature and RFC documents.  In the cases of Karn's Algorithm, the RTO computation, and Go-Back-$N$, the properties we formulate and prove are totally novel.  Some of the properties we prove about the three protocol handshakes are novel, while others serve to replicate prior results, in the context of our more detailed models.

\item \textbf{Proofs.}
We use a blend of formal methods and many proof strategies, including 
    inductive invariants,
    real analysis ($\epsilon$/$\delta$ proofs),
    bisimulation arguments,
    and LTL model checking.
Our multifaceted approach provides a useful case study in the benefits and drawbacks of multiple formal methods.

\item \textbf{Attacker Synthesis.}
To the best of our knowledge, we are the first to introduce a fully formal problem definition and solution for the automated synthesis of attacks against network protocols.  
We create an open-source tool called \korg, in which we implement our approach, and which we apply to TCP, DCCP, and SCTP as case studies.  
For TCP and DCCP we automatically find known attack strategies.
SCTP was recently patched to resolve a security vulnerability caused by an ambiguity in its RFC, and we use \korg to show the highlighted vulnerability could be automatically found before the patch was applied, and the patch resolved the vulnerability.  Then we identify two ambiguities in the RFC and show that either, if misinterpreted, could lead to a vulnerability.  Our analysis resulted in an erratum to the RFC.
\end{itemize}

\medskip

All our models and code are open-source and freely available with the dissertation artifacts.
We also provide open-source scripts with which to automatically reproduce all of our results.

\section{Thesis Outline}\Secl{intro:outline}

The rest of the dissertation is organized as follows.
\begin{enumerate}[{Chapter }1:]
    \setcounter{enumi}{1}
    \item \textbf{Verification of RTT Estimates and Asymptotic Analysis of Timeouts.} We analyze the RTT measurements produced by Karn's Algorithm,
    and the RTO computation based on them defined in RFC 6298~\cite{RFC6298}.
    We use a blend of formal methods to prove hitherto unformalized 
    invariants of Karn's Algorithm, and long-term bounds on the
    variables of the RTO computation.
    
    \medskip
    
    This chapter includes work originally presented in the following publications:

    \medskip

    Max von Hippel, Kenneth L. McMillan, Cristina Nita-Rotaru, and Lenore D. Zuck. \emph{A Formal Analysis of Karn’s Algorithm.} International Conference on Networked Systems, 2023.

    \medskip

    Max von Hippel, Panagiotis Manolios, Kenneth L. McMillan, Cristina Nita-Rotaru, and Lenore Zuck. \emph{A Case Study in Analytic Protocol Analysis in ACL2.} ACL2, 2023.

    \medskip

    All of our code is open-source and available at \url{github.com/rto-karn}.
    The \acls proofs are also made available with the \acl books in \texttt{workshops/2023/vonhippel-etal}.
    \item \textbf{Formal Performance Analysis of Go-Back-$N$.} We formally define best and worst-case scenarios for Go-Back-$N$
    and then prove bounds on the performance
    of the protocol in each, parameterized by~$N$,
    in the context of a realistic channel model which we prove to be compositional.

    \medskip
    
    Our models and proofs are open-source and freely available at \url{https://github.com/maxvonhippel/go-back-n-fm}.

    \medskip

    \item \textbf{Protocol Correctness for Handshakes.}  We formally model the handshakes of TCP, DCCP, and SCTP,
    all of which are important and widely-used transport protocols.
    We define logical properties each handshake should satisfy, based on
    a close reading of the corresponding RFC, which we verify 
    using the LTL model checker \spin.

    \medskip

    This chapter and the next include work originally presented in the following publications:

    \medskip

    Max von Hippel, Cole Vick, Stavros Tripakis, and Cristina Nita-Rotaru. \emph{Automated attacker synthesis for distributed protocols.} Computer Safety, Reliability, and Security, 2020.

    \medskip

    Maria Leonor Pacheco, Max von Hippel, Ben Weintraub, Dan Goldwasser, and Cristina Nita-Rotaru. \emph{Automated attack synthesis by extracting finite state machines from protocol specification documents.} IEEE Symposium on Security and Privacy, 2022.

    \medskip

    Jacob Ginesin, Max von Hippel, Evan Defloor, Cristina Nita-Rotaru, and Michael T{\"u}xen. \emph{A Formal Analysis of SCTP: Attack Synthesis and Patch Verification.} USENIX, 2024.

    \medskip

    All our models and properties are open-source and freely available at \url{https://github.com/maxvonhippel/attackerSynthesis} and \url{https://github.com/sctpfm}.
    \item \textbf{Automated Attacker Synthesis.}  We introduce and formally define the \emph{attacker synthesis} problem
    for network protocols,
    where the goal is, given a protocol which satisfies its LTL specification
    in the absence of an attacker, 
    to generate a non-trivial attacker which can cause the protocol to 
    violate its specification.
    We propose an automated solution based on LTL model-checking, which we prove
    to be both sound and, for the restricted class of attack programs it is 
    designed to generate, complete.
    Then we create an open-source attacker synthesis tool called \korg in which we implement
    our solution.
    We apply \korg to our TCP, DCCP, and SCTP models in the context of several 
    representative attacker models.
    \korg is open-source and freely available at \url{https://github.com/maxvonhippel/attackerSynthesis}.

    \medskip

    \item \textbf{Conclusion.}  We summarize our work and discuss limitations therein and future research directions.
\end{enumerate}

\setcounter{observation}{0}
\setcounter{definition}{0}
\setcounter{problem}{0}
\setcounter{theorem}{0}

\chapter{Verification of RTT Estimates and Asymptotic Analysis of Timeouts}\Chapl{karn-rto}
\textbf{Summary.}
The stability of the Internet relies on timeouts. The timeout value, known as the Retransmission TimeOut (RTO), is constantly updated, based on sampling the Round Trip Time (RTT) of each packet as measured by its sender -- that is, the time between when the sender transmits a packet and receives a corresponding acknowledgement. Many of the Internet protocols compute those samples via the same sampling mechanism, known as Karn's Algorithm. 

\medskip
We present a formal description of the algorithm, and study its properties. We prove the computed samples reflect the RTT of some packets, but it is not always possible to determine which. We then study some of the properties of RTO computations as described in the commonly used RFC~6298, using real analysis in \acls.  We present this as a case study in analytic protocol verification using a theorem prover.  All properties are mechanically verified using \textsc{Ivy} or \acls.

\medskip

This chapter includes work originally presented in the following publications:

\medskip

\noindent~Max von Hippel, Kenneth L. McMillan, Cristina Nita-Rotaru, and Lenore D. Zuck. \emph{A Formal Analysis of Karn’s Algorithm.} In International Conference on Networked Systems, 2023.  
\begin{description}
\item \underline{Contribution:} MvH co-authored the Karn's Algorithm model and proofs and helped write the corresponding text.  MvH solely authored the RTO model and proofs, but followed a proof sketch from KLM for the limit.
\end{description}

\medskip

\noindent~Max von Hippel, Panagiotis Manolios, Kenneth L. McMillan, Cristina Nita-Rotaru, and Lenore Zuck. \emph{A Case Study in Analytic Protocol Analysis in ACL2.} ACL2, 2023.  
\begin{description}
\item \underline{Contribution:} MvH wrote the majority of the proof code and all of the paper.
\end{description}

\section{Karn's Algorithm and the RTO Computation}\Secl{karn:intro}
Protocols leverage RTT information for many purposes, e.g.,
one-way delay estimation~\cite{abdou2015accurate} or
network topology optimization~\cite{tang2005gocast,rfc8305}, but the most common use is for the RTO computation, defined in RFC~6298~\cite{RFC6298}, which states:
\begin{quote}
The Internet, to a considerable degree, relies on the correct
   implementation of the RTO algorithm [\ldots] in order to preserve network stability and avoid congestion
   collapse.
\end{quote}
An RTO that is too low may cause false timeouts by hastily triggering a timeout mechanism that delays the proper functioning of the protocol, and thus, may expose the protocol to denial-of-service attacks.
On the other hand, an RTO that is too high causes overuse of resources~\cite{PM2006} by unnecessarily delaying the invocation of timeout mechanisms when congestion occurs.
A poorly chosen RTO can have disastrous consequences, including \emph{congestion collapse},
wherein the demands put on the network far exceed its capacity,
leading to excessive message dropping and thus excessive retransmission.
Congestion collapse was first observed in October 1986, during which time total Internet traffic dropped by over 1000x~\cite{jacobson1988congestion}.  At the time this kind of network failure was an engineering curiosity, but today it would spell global economic disaster, loss of life, infrastructural damage, etc.

Both Karn's algorithm and the RTO computation are widely used across the Internet, as we detail in \Subsecr{karn:usage}.
Hence, the correctness of these two mechanisms is fundamental for the correctness of the Internet as a whole.
Yet, some theoretical papers analyzing congestion control -- the original motivation for computing the RTO -- 
explicitly ignore the topic of timeouts, and hence implicitly ignore the RTO computation 
(e.g., \cite{mathis1997macroscopic,ccac,DMSS19}).

Computing a good RTO requires a good estimate of the  RTT. The RTO computation depends solely on the estimated RTT and some parameters that are fixed. Thus, understanding the mechanism which estimates RTT is fundamental to understanding any quantitative property of the Internet. The RTT of a packet (or message, datagram, frame, segment, etc.) is precisely the time that elapsed between its transmission and some confirmation of its delivery.  Both events (transmission and receipt of confirmation of delivery) occur at the same endpoint, namely, the one that transmits the packet, which we call the \emph{sender}.  In essence, if the sender transmits a packet at its local time $t$, and first learns of its delivery at time $t + \delta$, it estimates the RTT for this packet as $\delta$.

TCP uses a \emph{cumulative} acknowledgement mechanism where every packet received generates an \ack with the sequence number of the first un-received packet.\footnote{Some implementations of TCP use additional types of \acks, yet, the cumulative ones are common to TCP implementations.} Thus, if packets with sequence numbers $1,\dots,x$ are received and the packet with sequence number $x+1$ is not, the receiver will \ack\ with $x+1$, indicating the first un-received packet in the sequence, even if packets whose sequence numbers exceed $x+1$ were received. 

If the Internet's delivery mechanism were perfect, then packets would be received and acknowledged in order, and the sender would always be able to compute the RTT of each packet.  Unfortunately, the Internet is imperfect. TCP operates on top of IP, whose only guarantee is that every message received was sent. Thus, messages are neither invented nor corrupted, but at least theoretically, may be duplicated, reordered, or lost. In practice duplication is sufficiently rare that it is ignored, and re-ordering is sometimes ignored and sometimes restricted. But losses are never ignored, and are the main focus of all congestion control algorithms. When a loss is suspected, a packet is retransmitted.  If it is later acknowledged, one cannot determine whether the \ack\ is for the initial transmission or for the retransmission.  Karn's algorithm~\cite{karn1987} addresses this ambiguity by only using unambiguous \acks to compute RTT estimates.
RFC~6298~\cite{RFC6298} then computes an estimated RTT as a weighted (decaying) average of the samples output by Karn's algorithm,
and computes an RTO based on this estimate and a measure of the RTT variance.  The RTO is then used to gauge whether a packet is lost, and then, usually, to transition a state where transmission rate is reduced. Thus, the RTT sampling in Karn's algorithm is what ultimately informs the transmission rate of protocols.   And while RFC 6298 pertains to TCP, numerous non-TCP protocols also refer to RFC~6298 for the RTO computation, as we outline in \Subsecr{karn:usage}.

\begin{figure}
\centering  
\includegraphics[width=0.7\textwidth]{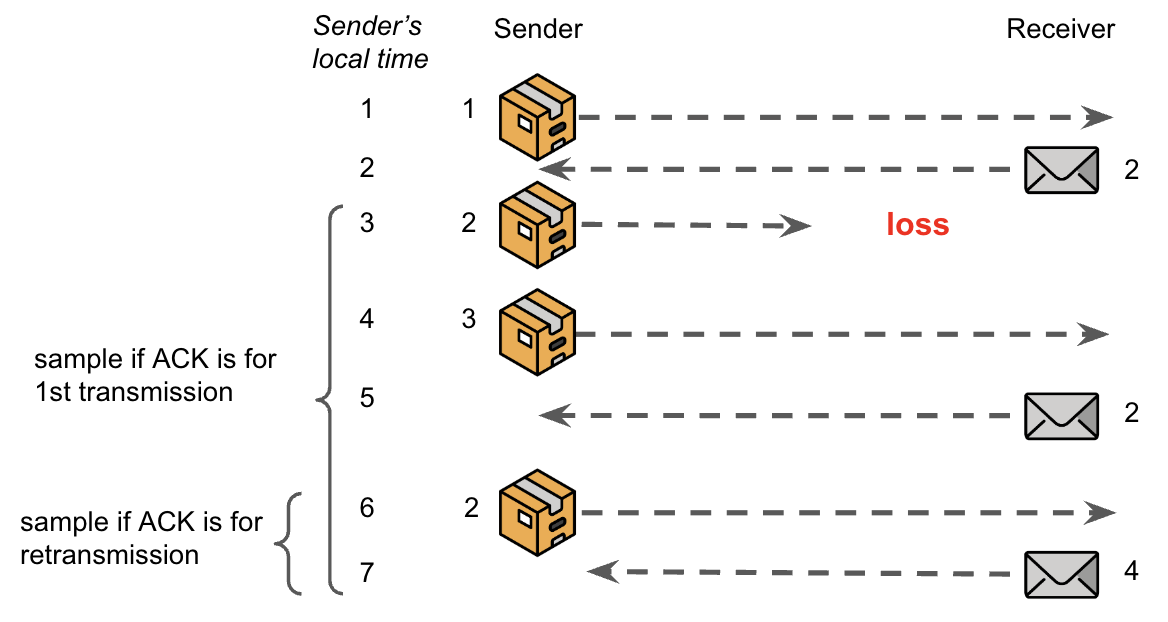}
\caption{Illustration of an ambiguous \ack, with the sender's local clock shown on the left. Sender's packets are illustrated as packets, while receiver's \acks are shown as envelopes.  The first time the sender transmits 2 the packet is lost in-transit. Later, upon receiving a cumulative \ack of 2, the sender determines the receiver had not yet received the 2 packet and thus the packet might be lost in transit. It thus retransmits 2. Ultimately the receiver receives the retransmission and responds with a cumulative \ack of 4. When the sender receives this \ack it cannot determine which 2 packet delivery triggered the ACK transmission and thus, it does not know whether to measure an RTT of 7-3=4 or 7-6=1. Hence, the \ack is ambiguous, so any sampled RTT would be as well.}
\Figl{karn:ambiguity}
\end{figure}

\subsection{Contribution}
Here, we first formalize Karn's algorithm~\cite{RFC6298}, and prove some high-level properties about the relationship between \acks and packets.  In particular, we show that Karn's algorithm computes the ``real" RTT of some packet, but the identity of this packet may be impossible to determine, unless one assumes (as many do) that \acks are delivered in a FIFO ordering. 
Next, we examine the RTO computation defined in RFC~6298~\cite{RFC6298} and its relationship to Karn's algorithm.
For example, we show that when the samples fluctuate within a known interval, the estimated RTT eventually converges to the same interval. This confirms and generalizes prior results.

All our results are automatically checked.  
For the first part, where we study Karn's algorithm, we use Ivy~\cite{ivy}.
Ivy is an interactive prover for inductive invariants, and provides convenient, built-in facilities
for specifying and proving properties about protocols, which makes it ideal for this part of the chapter.
For the second part, we study the RTO computation (and other computations it relies on), defined in RFC~6298.  
These are purely numerical computations and, in isolation, do not involve reasoning about the interleaving of processes
or their communication.  Each computation has rational inputs and outputs, and the theorems we prove bound these computations using exponents and rational multiplication.  We also prove the asymptotic limits of these bounds in steady-state conditions, which we define.  Since Ivy lacks a theory of rational numbers or exponentiation, we turn to \acls~\cite{dillinger2007acl2s,chamarthi2011acl2} for the remainder of the chapter.
We believe this is the first work that formalizes properties of the RTT sampling via Karn's algorithm, as well as properties of the quantities RFC~6298 computes, including the RTO.  Our work provides a useful example of how multiple formal methods approaches can be used to study different angles of a single system.
Finally, the \acls component provides a case study in real analysis using a theorem-prover.

\subsection{Usage of Karn's Algorithm and RFC~6298\Subsecl{karn:usage}}
Many protocols use Karn's Algorithm to sample RTT, e.g.,~\cite{RFC6298,rfc3124,rfc3481,rfc3539,rfc4015,rfc8085}.
Unfortunately, the samples output by Karn's Algorithm could be noisy or outdated.
RFC~6298 addresses this problem 
by using a rolling average 
called the \emph{smoothed RTT}, or \srtt.
Protocols that use the \srtt in conjunction with Karn's Algorithm
(at least optionally) 
include~\cite{rfc7765,rfc4960,rfc6538,rfc6817,rfc8085,rfc8305,rfc8445,rfc8489,rfc8855,rfc8931}. 
RFC~6298 then proposes an RTO computation based on the \srtt and another value called the \rttvar, 
which is intended to capture the variance in the samples.
Note, when referring specifically to the RTO output by RFC~6298, we use the convention \rto.
This is a subtle distinction as the RTO can be implemented in other ways as well (see e.g.,~\cite{kesselman2005optimizing,balandina2013computing}).
These three computations (\srtt, \rttvar, and \rto) are used in TCP
and in many other protocols, e.g.~\cite{rfc6538,rfc6817,rfc6940,rfc8489,rfc8931},
although some such protocols omit explicit mention of RFC~6298 (see \cite{PM2006}).

Not all protocols use retransmission. For example, in QUIC~\cite{rfc9000_quic} every packet has a unique identifier, hence retransmitting a packet assigns it a new unique identifier and the matching \ack\ indicates whether it is for the old or new transmission. Consequently, Karn's algorithm is only used when a real retransmission occurs, which covers most of the protocols designed when one had to be mindful of the length of the transmitted packets and could not afford unique identifiers.   On the other hand, even protocols that do not use Karn's algorithm nevertheless utilize a retransmission timeout that is at least adapted from RFC~6298 -- and in fact, QUIC is one such protocol.

\section{Formal Model of Sender, Channel, and Receiver} \Secl{karn:setup}
We partition messages, or datagrams, 
into \emph{packets} $P$ and \emph{acknowledgments} $A$.
Each packet $p\in P$ is uniquely identified by its id $p.\id \in \mathbb{N}$.
Each \ack $a \in A$ is also uniquely identified by its id $a.\id$. 
Whenever possible, we identify packets and acknowledgments by their $\id$s.

Messages (packets and acknowledgments) typically include additional information such as destination port or sequence number, however, we abstract away such  information in our model.  Also, some protocols distinguish between packets and \emph{segments}, but we abstract away this distinction as well.

The model consists of two endpoints (\emph{sender} and \emph{receiver}) connected over a bi-directional \emph{channel}, shown in \Figr{karn:sender_receiver_channel}. The sender sends packets through the channel to the receiver, and the receiver sends acknowledgements through the channel to the sender.  

\begin{figure}[htb]
\centering
\begin{tikzpicture}
\node[draw, rectangle, rounded corners, minimum width=1.7cm, minimum height=1.5cm, fill=white] (channel) at (3,0) 
    {channel};
\draw[->,>=latex] 
    ([yshift=-0.4cm,xshift=-1.6cm]channel.north west) 
    to node[above] {$\snds$} 
    ([yshift=-0.4cm]channel.north west);
\draw[->,>=latex] 
    ([yshift=0.4cm]channel.south west) 
    to node[below] {$\dlvrr$} 
    ([yshift=0.4cm,xshift=-1.48cm]channel.south west) ;
\draw[->,>=latex] 
    ([yshift=-0.4cm]channel.north east) 
    to node[above] {$\dlvrs$} 
    ([yshift=-0.4cm,xshift=1.48cm]channel.north east) ;
\draw[->,>=latex] 
    ([yshift=0.4cm,xshift=1.5cm]channel.south east) 
    to node[below] {$\sndr$} 
    ([yshift=0.4cm]channel.south east);
\node[draw, circle, minimum height=1.5cm,fill=white] (sender) at (0,0) 
    {sender};
\node[draw, circle, minimum height=1.5cm,fill=white] (receiver) at (6,0) 
    {receiver};
\end{tikzpicture}
\caption{The sender, channel, and receiver. The sender sends packets by $\snds$ actions which are received by $\dlvrs$ actions at the receiver's endpoint, and similarly, the receiver sends \acks by $\sndr$ actions which are received by $\dlvrr$ actions at the receiver's endpoint.}
\Figl{karn:sender_receiver_channel}
\end{figure}
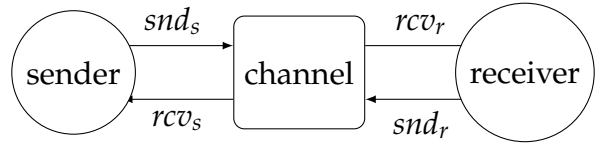

\paragraph{Actions.}
The set of actions, $\act$, is partitioned into four action types:
\begin{enumerate}
    \item $\snds$ that consists of the set of the sender's transmit actions, i.e.: 
    \(\snds = \cup_{p\in P} \{\snds(\id) :\id = p.\id\}\).
    These actions encode when the sender sends a packet.

    \item $\dlvrs$ that consists of the set of the sender's delivery actions, i.e.: 
    \(\dlvrs = \cup_{p\in P} \{\dlvrs(\id) :\id = p.\id\}\).
    These actions encode when the receiver receives a packet.

    \item $\snd_r$ that consists of the set of the receiver's transmit actions, i.e.: 
    \(\sndr = \cup_{a\in A} \{\sndr(\id) :\id = a.\id\}\).
    These actions encode when the receiver sends an \ack.
    
    \item $\dlvrr$ that consists of the set of the receiver's delivery actions, i.e.:
    \(\dlvrr = \cup_{a\in A} \{\dlvrr(\id) :\id = a.\id\}\).
    These actions encode when the sender receives an \ack.
\end{enumerate}

For a finite sequence $\sigma$ over $\act$, we denote the length of $\sigma$ by $\vert\sigma\vert$ and refer to an occurrence of an action in $\sigma$ as an \emph{event}. That is, an event in $\sigma$ consists of an action and its position in $\sigma$. 

The sender's input actions are $\dlvrr$, and its output actions are $\snds$. The receiver's input actions are $\dlvrs$ and its output actions are $\snd_r$. The channel's input actions are $\snds\cup\sndr$ and its output actions are $\dlvrs \cup \dlvrr$.

We assume that the channel is synchronously composed with its two endpoints, the sender and the receiver. 
That is, a $\sndr$ action occurs simultaneously at both the receiver and the channel, 
a $\dlvrr$ action occurs simultaneously at both the sender and the channel, and so on.
The sender and the receiver can be asynchronous. The sender, receiver, and channel are input-enabled in the I/O-automata sense, i.e., each can always receive inputs (messages). In real implementations, the inputs to each component are restricted by buffers, but since the channel is allowed to drop messages (as we see later), restrictions on the input buffer sizes can be modeled using loss. Hence the assumption of input-enabledness does not restrict the model. 

\paragraph{Model Executions.}
Let $\sigma$ be a sequence of actions. We say that $\sigma$ is an \emph{execution} if every delivery event in $\sigma$ is preceded by a matching transmission event, that is, both events carry the same message. (This does not rule out duplication, reordering, or loss -- more on that below.) 
Formally, if $e_i=\dlvrr(x) \in \sigma$, then for some $j<i$, $e_j=\sndr(x) \in \sigma$;
and likewise in the opposite direction.
This requirement rules out corruption and insertion of messages.  
In addition, for TCP-like executions, we may impose additional requirements on the ordering of $\snd$-events of the endpoints.  An example execution is illustrated in the rightmost column of \Figr{karn:execution}.

\paragraph{The Sender.} We adopt the convention that it only transmits a packet after it had transmitted all the preceding ones. Formally, for every $x\in\mathbb{N}$, if $e_i = \snds(x+1)\in \sigma$, then for some $j < i$, $e_j = \snds(x)\in \sigma$.

\paragraph{The Receiver.} We assume here the model of \emph{cumulative \acks}.  That is, the receiver executes a $\snd_r(\id)$ action only if it has been delivered all packets $p$ such that $p.\id < \id$ and it had not been delivered packet $p$ such that $p.\id = \id$. Thus, for example, the receiver can execute $\snd_r(17)$ only after it had been delivered all packets whose id is $<17$ and had not been delivered the packet whose id is 17.  In particular, it may have been delivered packets whose id is $>17$, just not the packet with id 17. 

Many TCP models mandate the receiver transmits exactly one \ack in response to each packet delivered (e.g., \cite{ccac,arun2022starvation,rfc6582,brakmo1995tcp,rfc8985,gerla2001tcp}).  The assumption is common in congestion control algorithms where the sender uses the number of copies of the same acknowledgement it is delivered to estimate how many packets were delivered after a packet was dropped, and thus the number of lost packets. There are however some TCP variants, such as Data Center TCP and TCP Westwood, that allow a \emph{delayed \ack} option wherein the receiver transmits an \ack after every $n^{th}$ packet delivery~\cite{rfc8257,mascolo2001tcp}\footnote{We discuss such \ack strategies further in \Chapr{gbn} as well as \Secr{subsec:delay-ack} in the Appendix.}, or Compound TCP that allows \emph{proactive acknowledgments} where the receiver transmits before having receiving all the acknowledged packets, albeit at a pace that is proportional to the pace of packet deliveries~\cite{ctcp}.  Another mechanism that is sometimes allowed is \nack (for Negative \ack) where the receiver sends, in addition to the cumulative acknowledgement, a list of gaps of missing packets~\cite{rfc3941}.  Since TCP datagrams are restricted in size, the \nacks are partial. Newer protocols (such as QUIC) allow for full (unrestricted) \nacks~\cite{rfc9000_quic}.

Our Ivy model assumes the receiver transmits one \ack per packet delivered.
That is, we assume that in the projection of  $\sigma$ onto the receiver's actions, $\sndr$ and $\dlvrs$ events are alternating.
In fact, the results listed in this paper would still hold even under the slightly weaker assumption
that the receiver transmits an \ack whenever it is delivered a packet that it had not previously
been delivered, but for which it had previously been delivered all lesser ones.
However, the stronger assumption is easier to reason about, and is more commonly used in the literature
(for example it is the default assumption for congestion control algorithms where the pace of delivered
acknowledgments is used to infer the pace of delivered packets).
Consequently, our results apply to traditional congestion control algorithms like TCP Vegas
and TCP New Reno where the receiver transmits one acknowledgement per packet delivered,
however, our results might not apply to atypical 
protocols like Data Center TCP, TCP Westwood, or Compound TCP,
that use alternative \ack schemes.

\paragraph{The Channel.}  So far,  we only required that the channel never deliver messages to one endpoint that were not previously transmitted by the other.  This does not rule out loss, reordering, nor duplication of messages. In the literature, message duplication is assumed to be so uncommon that it can be disregarded. The traditional congestion control protocols (\cite{rfc5681,rfc6582,liu2006tcp,rfc9002,ctcp}) assume bounded reordering, namely, that once a message is delivered, an older one can be delivered only if transmitted no more than $k$ transmissions ago (usually, $k=4$).  Packet losses are always assumed to occur, but the possibility of losing acknowledgements is often ignored. 

It is possible to formalize further constraints on the channel, e.g., by restricting the receiver-to-sender path to be loss- and reordering-free.  For instance, the work in~\cite{afek1994reliable} formalizes a constrained channel by assuming a mapping from delivery to transmission events, and using properties of this mapping to express restrictions. Reordering is ruled out by having this mapping be monotonic, duplication is ruled out by having it be one-to-one, and loss is ruled out by having it be onto. 

Most prior works assume no loss or reordering of \acks~\cite{baccelli2000tcp,hu2013modeling,DMSS19,ccac,arun2022starvation}, or did not model loss or reordering at all~\cite{lockefeer2016formal,von2020automated,cluzel2021layered}. Some prior works assume both loss and reordering but do not study the computation of RTO or other aspects of
congestion control~\cite{afek1994reliable,smith1997formal}.

Since, as we describe in \Secr{karn:related}, some works on RTO assume the channel delivers \acks in perfect order, and since this assumption has implications on the RTT computation (see \Obr{fifo_rtt}), we define executions where the receiver's messages are delivered, without losses, in the order they are transmitted as follows. An execution $\sigma$ is a \emph{FIFO-acknowledgement execution} if 
\( \sigma|_{\dlvrr} ~\preceq \sigma|_{\sndr} \)
is an invariant of sigma,
where $\sigma|_a$ is the projection of $\sigma$ onto the $a$ actions, and $\preceq$ is the prefix relation. That is, in a FIFO-acknowledgement execution, the sequences of \acks delivered to the sender is always a prefix of the sequence of \acks transmitted by the receiver.

The following observation establishes that the sequence of acknowledgements the receiver transmits is monotonically increasing. Its proof follows directly from the fact that the receiver is generating cumulative \acks. (All Observations in this section and the next are established in Ivy.)
\begin{observation}\Obl{r_sends_monotone}
Let $\sigma$ be an execution, and assume $i<j$ such that $e_i = \sndr(a_i),e_j = \sndr(a_j)$ are in $\sigma$. Then $a_i \leq a_j$.
\end{observation}

\paragraph{Sender's Computations.}
So far, we abstracted away from the internals of the sender, receiver, and channel, and focused on the executions their composition allows.  As we pointed out at the beginning of this section, real datagrams can contain information far beyond ids, and there are many mechanisms for their generation, depending on the protocol being implemented and the implementation choices made. Such real implementations have \emph{states}.  All we care about here, however, is the set of observable behaviors they allow, in terms of packet and acknowledgement ids. We thus choose to ignore implementation details, including states, and focus on executions, namely abstract observable behaviors. 

In the next section we study a mechanism that is imposed over executions. In particular, we describe an algorithm for sampling the RTT of packets, namely, Karn's Algorithm. This algorithm, $P$, is (synchronously) composed with the sender's algorithm (on which we only make a single assumption, that is, that a packet is transmitted only after all prior ones were transmitted).  We can view the algorithm as a \emph{non-interfering monitor}, that is, $P$ observes the sender's actions ($\snds$ and $\dlvrr$) and performs some bookkeeping when each occurs. In fact, after initialization of variables, it consists of two parts, one that describes the update to its variables upon a $\snds$ action, and one that describes the updates to its variables after a $\dlvrr$ action. 

Let $V$ be the set of variables $P$ uses. To be non-interfering, $V$ has to be disjoint from the set of variables that the sender uses to determine when to generate $\snds$s and process $\dlvrr$s.   We ignore this latter set of variables since it is of no relevance to our purposes. Let a \emph{sender's state} be a type-consistent assignment of values $V$. For a sender's state $s$ and a variable $v\in V$, let $\sem{s}{v}$ be the value of $v$ at state $s$.  For simplicity's sake (and consistent with the pseudocode we present in the next section) assume that $P$ is deterministic, that is, given a state $s$ and a sender's action $\alpha$, there is a unique sender state $s'$ such that $s'$ is the successor of $s$ given $\alpha$.

Let $\sigma$ be an execution.   Let $\sigma|_s$ be the projection of $\sigma$ onto the sender's events (the $\snds$ and $\dlvrr$ events).  Since $P$ is deterministic, the sequence $\sigma|_s$ uniquely defines a sequence of sender's states $\kappa_\sigma: s_0, \ldots$ such that $s_0$ is the initial state, and every $s_{i+1}$ is a successor of $s_i$ under $P$ according to $\sigma|_s$.  We refer to $\kappa_\sigma$ as the \emph{sender's computation under $P$ and $\sigma$}.

\section{Formal Model of Karn's Algorithm}\Secl{karn:karn}

As discussed in \Secr{karn:intro}, having a good estimate of RTT, the round-trip time of a packet, is essential for determining the value of RTO, which is crucial for many of the Internet's protocols (see \Subsecr{karn:usage} for a listing thereof).  The value of RTT varies over the lifetime of a protocol, and is therefore often sampled. Since the sender knows the time it transmits a packet, and is also the recipient of acknowledgements, it is the sender whose role it is to sample the RTT.  If the channel over which packets and acknowledgements are communicated were a perfect FIFO channel, then RTT would be easy to compute, since then each packet would generate a single acknowledgement, and the time between the transmission of the packets and the delivery of its acknowledgement would be the RTT. However, channels are not perfect. Senders retransmit packets they believe to be lost, and when those are acknowledged the sender cannot disambiguate which of the transmissions to associate with the acknowledgements.  Moreover, transmitted acknowledgments can be lost, or delivered out of order. In \cite{karn1987}, an idea, referred to as Karn's Algorithm, was introduced to address the first issue.  There, sampling of RTT is only performed when the sender receives a new acknowledgement, say $h$, greater than the previously highest received acknowledgement, say $\ell$, where all the packets whose id is in the range $[\ell,h)$ were transmitted only once. It then outputs a new sample whose value is the time that elapsed between the transmission of the packet whose id is $\ell$ and delivery of the acknowledgement $h$.  The reason $\ell$ (as opposed to $h$) is used for the base of calculations is the possibility that the id of the packet whose delivery triggers the new acknowledgement is $\ell$, and the RTT computation has to be cautious in the sense of over-approximating RTT.

\begin{figure}[H]
\begin{minipage}{\linewidth}
\begin{algorithm}[H]
\SetAlgoLined
\DontPrintSemicolon
\SetKwInOut{Input}{input}
\SetKwInOut{Output}{output}
\Input{$\snds(i),\dlvrr(j),~i,j\in \mathbb{N}^+$}
\Output{$\sample \in \mathbb{N}^+$}
$\nums,\tims\colon \mathbb{N}^+ \to \mathbb{N}$  {\bf init} all 0\; 
$\ffirst\colon \mathbb{N}$ {\bf init} $0$ \;
$\tau\colon \mathbb{N}$ {\bf init} $1$ \;
\If{$\snds(i)$ is received}{
    $\nums[i] := \nums[i] +1$\;
    \If {$\tims[i] = 0$}{
        $\tims[i] := \tau$\;   
    }\label{ss_update}
    $\tau := \tau + 1$\;}
\If{$\dlvrr(j)$ is received}{
\If{$j > \ffirst$}{\label{bgn_block}
    \If{\textit{ok-to-sample}($\nums,\ffirst$)\label{ok_to}}{
    $\sample := \tau - \tims[\ffirst]$\label{sample}\;
    {\bf Ouput} $\sample$}\label{end_sample}
    $\ffirst := j$\label{retransmission}}\label{end_block}
    $\tau := \tau + 1$\;}
   \caption{Karn's Algorithm}\Algl{Karn}
\end{algorithm}
\end{minipage}
\end{figure}

The real RTT of a packet may be tricky to define. The only case where it is clear is when packet $i$ is transmitted once, and an \ack $i+1$ is delivered before any other \ack $\ge i+1$ is delivered. We can then define the RTT of packet $i$, $\rtt(i)$, to be the time, on the sender's clock, that elapses between the (first and only) $\snds(i)$ action and the $\dlvrr(i+1)$ action. Since the channel is not FIFO, it's possible that $h > \ell+1$, and then the sample, that is, the time that elapses between $\snds[\ell]$ and $\dlvrr(h)$ is the RTT for some packet $j\in[\ell,h)$, denoted by, $\rtt(j)$, but we may not be able to identify $j$.  Moreover, the sample over-approximates the RTT of all packets in the range. Note that $\rtt$ is a partial function. We show that when the channel delivers the receiver's messages in FIFO ordering, then the computed sample is exactly $\rtt(\ell)$.

We model the sender's sampling of RTT according to Karn's Algorithm (\Algr{Karn}). The sampling is a non-interfering monitor of the sender. Its inputs are the sender's actions, the  $\snds(i)$'s and $\dlvrr(j)$'s. Its output is a (possibly empty) sequence of samples denoted by $\sample$.  To model time,  we use an integer counter ($\tau$) that is initialized to 1 (we reserve 0 for undefined) and is incremented with each step. Upon a $\snds(i)$ input, the algorithm stores, in $\nums[i]$, the number of times packet $i$ is transmitted, and in $\tims[i]$ the time of the first time it is transmitted. The second step is for $\dlvrr$ events, where the sender determines whether a new sample can be computed, and if so, computes it.  An example execution, concluding with the computation of a sample via Karn's Algorithm, is given in \Figr{karn:execution}.

\begin{figure}[h]
\centering
\begin{tikzpicture}[align=center]
\node[draw,rectangle] (sender) {Sender};
\node[draw,rectangle,right=0.7cm of sender] (channel) {Channel};
\node[draw,rectangle,right=0.7cm of channel] (receiver) {Receiver};
\node[right=0.3cm of receiver] {Execution~$\sigma$};
\node[left=0.3cm of sender] {Karn's Algorithm};
\node[below=0.25cm of sender] (s1) {};
\node[below=0.25cm of channel] (c1) {};
\node[below=0.25cm of receiver] (r1) {};
\node[right=0.8cm of r1] (e1) {$e_1=\snds(1)$};
\draw[->,>=latex] (s1) to node[above] {1} (c1);
\node[left=0.6cm of s1] {$\nums[1]=1; \tims[1]=1; \tau=2$};
\node[below=0.25cm of s1] (s2) {};
\node[below=0.25cm of c1] (c2) {};
\node[below=0.25cm of r1] (r2) {};
\node[right=0.8cm of r2] (e2) {$e_2=\snds(2)$};
\draw[->,>=latex] (s2) to node[above] {2} (c2);
\node[left=0.6cm of s2] {$\nums[2]=1; \tims[2]=2; \tau=3$};
\node[below=0.25cm of s2] (s3) {};
\node[below=0.25cm of c2] (c3) {};
\node[below=0.25cm of r2] (r3) {};
\node[right=0.8cm of r3] (e3) {$e_3=\dlvrs(2)$};
\draw[->,>=latex] (c3) to node[above] {2} (r3);
\node[left=0.6cm of s3] {};
\node[below=0.25cm of s3] (s4) {};
\node[below=0.25cm of c3] (c4) {};
\node[below=0.25cm of r3] (r4) {};
\node[right=0.8cm of r4] (e4) {$e_4=\sndr(1)$};
\draw[->,>=latex] (r4) to node[above] {1} (c4);
\node[left=0.6cm of s4] {};
\node[below=0.25cm of s4] (s5) {};
\node[below=0.25cm of c4] (c5) {};
\node[below=0.25cm of r4] (r5) {};
\node[right=0.8cm of r5] (e5) {$e_5=\dlvrr(1)$};
\draw[->,>=latex] (c5) to node[above] {1} (s5);
\node[left=0.6cm of s5] {$\neg\textit{ok-to-sample};\ffirst=1;\tau=4$};
\node[below=0.25cm of s5] (s6) {};
\node[below=0.25cm of c5] (c6) {};
\node[below=0.25cm of r5] (r6) {};
\node[right=0.8cm of r6] (e6) {$e_6=\dlvrs(1)$};
\draw[->,>=latex] (c6) to node[above] {1} (r6);
\node[left=0.6cm of s6] {};
\node[below=0.25cm of s6] (s7) {};
\node[below=0.25cm of c6] (c7) {};
\node[below=0.25cm of r6] (r7) {};
\node[right=0.8cm of r7] (e7) {$e_7=\snds(3)$};
\draw[->,>=latex] (s7) to node[above] {3} (c7);
\node[left=0.6cm of s7] {$\nums[5]=1; \tims[5]=4; \tau=5$};
\node[below=0.25cm of s7] (s8) {};
\node[below=0.25cm of c7] (c8) {};
\node[below=0.25cm of r7] (r8) {};
\node[right=0.8cm of r8] (e8) {$e_8=\sndr(3)$};
\draw[->,>=latex] (r8) to node[above] {3} (c8);
\node[left=0.6cm of s8] {};
\node[below=0.25cm of s8] (s9) {};
\node[below=0.25cm of c8] (c9) {};
\node[below=0.25cm of r8] (r9) {};
\node[right=0.8cm of r9] (e9) {$e_9=\dlvrs(3)$};
\draw[->,>=latex] (c9) to node[above] {3} (r9);
\node[left=0.6cm of s9] {};
\node[below=0.25cm of s9] (s10) {};
\node[below=0.25cm of c9] (c10) {};
\node[below=0.25cm of r9] (r10) {};
\node[right=0.8cm of r10] (e10) {$e_{10}=\sndr(4)$};
\draw[->,>=latex] (r10) to node[above] {4} (c10);
\node[left=0.6cm of s10] {};
\node[below=0.25cm of s10] (s11) {};
\node[below=0.25cm of c10] (c11) {};
\node[below=0.25cm of r10] (r11) {};
\node[right=0.8cm of r11] (e11) {$e_{11}=\dlvrr(4)$};
\draw[->,>=latex] (c11) to node[above] {4} (s11);
\node[left=0.6cm of s11] {$\textit{ok-to-sample};\ffirst=4;\tau=6$};
\end{tikzpicture}
\caption{Message sequence chart illustrating an example execution.  Time progresses from top down.  Instructions executed by \Algr{Karn} are shown on the left, and the sender's execution is on the right.  $\snds$ events are indicated with arrows from sender to channel, $\dlvrs$ events with arrows from channel to receiver, etc.  After the final $\dlvrr$ event, sender executes Line \ref{sample} and outputs the computation $\sample = 6 - 2 = 4$.}
\Figl{karn:execution}
\end{figure}
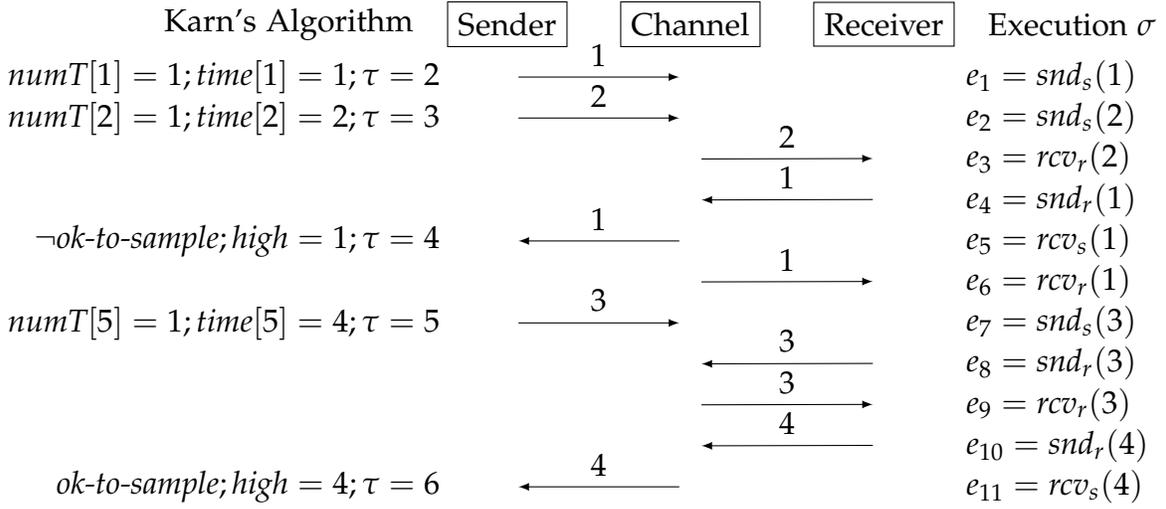

In \Algr{Karn}, $\nums[i]$ stores the number of times a packet whose id is $i$ is transmitted, $\tims[i]$ stores the sender's time where  packet whose id is $i$ is first transmitted, $\ffirst$ records the highest delivered acknowledgement, and when a new sample is computed (in $\sample$) it is recorded as an output. The condition \emph{ok-to-sample}($\nums,\ffirst$) in Line \ref{ok_to} checks whether sampling should occur.  When $\ffirst > 0$, that is, when this is not the first \ack received, then the condition is that all the packets in the range $[\ffirst,j)$ were transmitted once. If, however, $\ffirst = 0$, since ids are positive, the condition is  that all the packets in the range $[1,j)$ were transmitted once. Hence, \emph{ok-to-sample}($\nums,\ffirst$) is:
$$ (\forall k. \ffirst < k < j \to \nums[k] = 1)\band (\ffirst > 0 \rightarrow \nums[\ffirst]=1)$$
If \emph{ok-to-sample}($\nums,\ffirst$), Line \ref{sample} computes a new sample $\sample$ as the time that elapsed since packet $\ffirst$ was transmitted until acknowledgement $j$ is delivered, and outputs it in the next line. Thus, a new sample is \emph{not computed} when a new \ack, that is greater than $\ffirst$, is delivered but some packets whose id is less than the new \ack, yet $\ge \ffirst$ were retransmitted. Whether or not a new sample is computed, when such an \ack is delivered, $\ffirst$ is updated to its value to reflect the currently highest delivered \ack.

\section{Properties of Karn's Algorithm}

We show, through a sequence of observations, that \Algr{Karn} computes the true RTT of \emph{some} packet, whose identity cannot also be uniquely determined.  While much was written about the algorithm, we failed to find a clear statement of what exactly it computes.  In \cite{karn1987}, it is shown that if a small number of consecutive samples are equal then the computed RTT (which is a weighted average of the sampled RTTs) is close to the value of those samples. See the next section for further discussion on this issue. Our focus in this section is what exactly is computed by the algorithm. 

The set of variables in \Algr{Karn} is $V=\{ \tau, \nums,\tims, \ffirst, \sample\}$.  Let $\sigma$ be an execution, and let $\kappa_\sigma$ be the sender's computation under \Algr{Karn} and $\sigma$.
The following observation establishes two invariants over $\kappa_\sigma$. Both follow from the assumption we made on the sender's execution, namely that the sender does not transmit $p$ without first transmitting $1, \ldots, p-1$. The first establishes that if a packet is transmitted (as viewed by $\nums$), all preceding ones were transmitted, and the second 
that the first time a packet is transmitted must be later than the first time every preceding packet was transmitted. 

\begin{samepage}
\begin{observation}\Obl{karn_invs}
The following are invariants over sender's computations: 
\[ \begin{array}{ll}
0 < i < j \band \nums[j] > 0 \longrightarrow \nums[i] > 0 & (\inv) \\
0 < i < j \band \nums[j] > 0 \longrightarrow \tims[i] < \tims[j] & (\inv)
\end{array}\]
\end{observation}
\end{samepage}

Assume $\kappa_\sigma: s_0, s_1, \ldots$. We say that a state $s_i\in \kappa_\sigma$ is a \emph{fresh sample} state if the transition leading into it contains an execution of Lines \ref{ok_to}--\ref{end_sample} of \Algr{Karn}. The following observation establishes that in a fresh sample state, the new sample is an upper bound for the RTT of a particular range of packets (whose ids range from the previous $\ffirst$ up to, but excluding, the new $\ffirst$), and is the real RTT of one of them.

\begin{observation}\Obl{sample_upper_bounds}
Let $\sigma$ and $\kappa_\sigma$ be as above and assume that $s_i\in \kappa_\sigma$ is a fresh sample state. Then the following all hold:
\begin{enumerate}
    \item For every packet with id $\ell$, $\sem{s_{i-1}}{\ffirst} \le \ell < \sem{s_i}{\ffirst}$ implies that $\rtt(\ell) \le \sem{s_i}{\sample}$. That is, the fresh sample is an upper bound of the RTT for all packets between the old and the new $\ffirst$.
    \item There exists a packet with id $\ell$, $\sem{s_{i-1}}{\ffirst} \le \ell < \sem{s_i}{\ffirst}$ such that $\rtt(\ell) = \sem{s_i}{\sample}$. That is, the fresh sample is the RTT of some packet 
    between the old and 
    new $\ffirst$.
    \end{enumerate}
\end{observation}

We next show under the (somewhat unrealistic, yet often made) assumption of FIFO-acknowledgement executions,  the packet whose RTT is computed in the second clause of \Obr{sample_upper_bounds} is exactly the packet whose id equals to the prior $\ffirst$. In particular, that if $s_i$ is a fresh sample state, then the packet whose RTT is computed is $p$ such that $p.\id$ equals to the value of $\ffirst$ just before the new fresh state is reached.

\begin{observation}\Obl{fifo_rtt}
Let $\sigma$ be a FIFO-acknowledgement execution $\sigma$, and assume $\kappa_\sigma$ contains a fresh sample state  $s_\ell$.  Then  $\sem{s_\ell}{\sample}= \rtt(\sem{s_{\ell-1}}{\ffirst})$.
\end{observation}

Let $\sigma$ be a (not necessarily FIFO) execution and let $\kappa_\sigma$ be the sender's computation under \Algr{Karn} and $\sigma$ that outputs some samples.  We denote by $\sample_1, \ldots$ the sequence of samples that is the output of $\kappa_\sigma$. That is, $\sample_k$ is the $k^{th}$ sample obtained by \Algr{Karn}  given the execution $\sigma$.

\section{Formal Model of the RTO Computation}\Secl{RFC6298}

 We next analyze the computation of RTOs as described in RFC~6298. 
Each new sample triggers a new RTO computation, that depends on sequences of two other variables ($\srtt$ and $\rttvar$) and three constants ($\alpha$, $\beta$, and $G$).
In this section, we consider the scenario in which the samples produced by Karn's algorithm are consecutively bounded.
We show that in this context, we can compute corresponding bounds on \srtt, as well as an upper bound on \rttvar;
and that these bounds converge to the bounds on the samples and the distance between those bounds, respectively,
as the number of bounded samples grows.
These observations allow us to characterize the asymptotic conditions under which the RTO will generally exceed the RTT values, and by how much.  In other words, these observations allow us to reason about whether timeouts will occur in the long run.
 
 Let $\{ \srtt, \rttvar, \rto, \alpha, \beta, G \}\in\mathbb{Q}^+$ be fresh variables. 
 As mentioned before, $\alpha < 1$, $\beta < 1$, and $G$ are constants. Let $\sigma$ be an execution and $\kappa_\sigma$ be the sender’s computation under \Algr{Karn} and $\sigma$. Assume that $\kappa_\sigma$ outputs some samples $\sample_1,\ldots,\sample_N$.

RFC~6298 defines the RTO and the computations it depends upon as follows:
\[\begin{aligned}
\rto_i & = \srtt_i + \max(G, 4\cdot \rttvar_i) \\
\srtt_i & = \begin{cases} 
    \sample_i & \text{ if } i = 1 \\
    (1-\alpha)\srtt_{i-1} + \alpha \sample_i & \text{ if } i > 1 
    \end{cases} \\
\rttvar_i & = \begin{cases}
\sample_i / 2 & \text{ if } i = 1 \\
(1-\beta) \rttvar_{i-1} + \beta \lvert \srtt_{i-1} - \sample_i \rvert & \text{ if } i > 1
\end{cases}
\end{aligned}
\]
where $G$ is the clock granularity (of $\tau$),  \srtt is referred to in RFC~6298 as the \emph{smoothed RTT}, 
and \rttvar as the \emph{RTT variance}.
The \srtt is a rolling weighted average of the sample values and is meant to give
an RTT estimate that is resilient to noisy samples.
The \rttvar is described as a measure of variance in the sample values, although as we show below,
it is not the usual statistical variance.
The \rto is computed from  \srtt and \rttvar and is the amount of time the sender will wait without receiving an \ack before it determines that congestion has occurred and takes some action such as decreasing its output and retransmitting unacknowledged messages.
We manually compute these variables, and mechanically verify the computations thereof, using \acls. The choice of \acls stems from Ivy's lack of support of the theory of the Rationals, which is necessary for this analysis.

\section{Properties of the RTO Computation}

Intuitively, the \srtt is meant to give an estimate of the (recent) samples, while the \rttvar is meant to provide a measure of the degree to which these samples vary.  However, 
the \rttvar is not actually a variance in the statistical sense.
For example,
 if $\sample_1 = 1$, $\sample_2 = 44$, $\sample_3 = 13$, $\alpha=1/8$, and $ \beta=1/4$, then the statistical variance of the samples is $1477/3$ but $\rttvar_3 = 4977361/65536 \neq 1477/3$.

If the \rttvar does not compute the statistical variance, then what does it compute?  And what does the \srtt compute?
We answer these questions under the (realistic) restriction that the samples fall within some bounds,
which we formalize as follows.
Let $c$ and $r$ be positive rationals and let $i$ and $n$ be positive naturals.
Suppose that $\sample_i, \ldots, \sample_{i+n}$ all fall within the bounded interval $[c-r, c+r]$
with center $c$ and radius $r$.  Then we refer to $\sample_i, \ldots, \sample_{i+n}$ as \emph{$c/r$ steady-state samples}.
In the remainder of this section, 
we study $c/r$ steady-state samples and prove both instantaneous and asymptotic
bounds on the \rttvar and \srtt values they produce.
\Figr{karn:scenarios} illustrates two scenarios with $c/r$ steady-state samples.
In the first, the samples
are randomly drawn from a uniform distribution, 
while in the second, they are pathologically crafted to cause infinitely many
timeouts.
The figure shows for each scenario
the lower and upper bounds on the \srtt which we report below in \Obr{srtt_approaches_s}, 
as well as the upper bound on the \rttvar which we report below in \Obr{rttvar_upper_bound}.
The asymptotic behavior of the reported bounds is also clearly visible.

In~\cite{karn1987}, Karn and Partridge argue that, given $\alpha=1/8$ and $\beta=1/4$, after six consecutive identical samples $\sample$, assuming the initial $\srtt\ge \beta \sample$, the final \srtt approximates $\sample$ within some tolerable~$\epsilon$. We generalize this result 
in the following observation.

\begin{observation}\Obl{srtt_approaches_s}
Suppose $\alpha, c,$ and $r$ are reals, $c$ is positive, $r$ is non-negative, and $\alpha \in (0, 1]$.
Further suppose $i$ and $n$ are positive naturals, and $\sample_i, \ldots, \sample_{i+n}$ are $c/r$ steady-state samples.
Define $L$ and $H$ as follows.
\[\begin{aligned}
    L & = (1-\alpha)^{n+1} \srtt_{i-1} + (1 - (1-\alpha)^{n+1})(c - r) \\
    H & = (1-\alpha)^{n+1} \srtt_{i-1} + (1 - (1-\alpha)^{n+1})(c + r)
\end{aligned}\]
Then $L \leq \srtt_{i+n} \leq H$.  Moreover, $\lim_{n\to\infty} L = c - r$, and $\lim_{n \to \infty} H = c + r$.
\end{observation}
As an example, suppose that
$n=5$, $\alpha=1/8$, $\beta=1/4$, $r=0$, and $\srtt_{i-1} = 3 \beta c$.
Then $L = H \approx 0.89 c$, 
hence $\srtt_{i+4}$ differs from $\sample_i, \ldots, \sample_{i+4} = c$
by about 10\% or less.
\Obr{srtt_approaches_s} also generalizes in the sense that as $n$ grows to infinity, $[L, H]$ converges to $[c-r, c+r]$,
meaning the bounds on the $\srtt$ converge to the bounds on the samples, or if $r=0$, to just the (repeated) sample value
$\sample_i = c$.

Next, we turn our attention to bounding the \rttvar. 
The following observation establishes that when the difference between each sample and the previous \srtt is bounded above
by some constant $\Delta$, then each \rttvar is bounded above by a function of this $\Delta$.  Moreover, as the number of consecutive samples grows for which this bound holds, the upper bound on the \rttvar converges to precisely $\Delta$.
Note, in this observation we use the convention $f^{(m)}$ to denote $m$-repeated compositions of $f$, for any function $f$, e.g., $f^{(3)}(x) = f(f(f(x)))$.

\begin{observation}\Obl{rttvar_upper_bound}
Suppose $1 < i$, and
$0 < \Delta \in \mathbb{Q}$ is such that
\(\lvert \sample_j - \srtt_{j-1} \rvert \leq \Delta\)
for all $j \in [i, i+n]$.
Define $B_\Delta(x) = (1-\beta)x + \beta \Delta$.
Then all the following hold.
\begin{itemize}
    \item Each $\rttvar_j$ is bounded above by the function $B_\Delta(\rttvar_{j-1})$.
    \item We can rewrite the (recursive) upper bound on $\rttvar_{i+n}$ as follows:
    \[B_\Delta^{(n+1)}(\rttvar_{i-1}) = (1-\beta)^{n+1} \rttvar_{i-1} + (1 - (1 - \beta)^{n+1}) \Delta\]
    \item Moreover, this bound converges to $\Delta$, i.e.,
    \(\lim_{n \to \infty} B_\Delta^{(n+1)}(\rttvar_{i-1}) = \Delta\).
\end{itemize}
\end{observation}
Note that if $\sample_i, \ldots, \sample_{i+n}$ are $c/r$ steady-state samples then by \Obr{srtt_approaches_s}:
\[\lvert \sample_n - \srtt_{n-1} \rvert \leq \Delta = (1-\alpha)^{n+1} \srtt_{i-1} + 2r - (1-\alpha)^{n+1}(c + r)\]
Since
\(\lim_{n \to \infty} \Delta = 2r\),
in $c/r$ steady-state conditions, it follows that
the $\rttvar$ asymptotically measures the diameter $2r$ of the sample interval $[c-r, c+r]$.

\paragraph{Implications for the \rto Computation.}
Assume $n$ are $c/r$ consecutive steady-state samples.
As $n \to \infty$, the bounds on $\srtt_n$ approach $[c-r, c+r]$, and the upper bound $\Delta$ on $\rttvar_n$ approaches $2r$.
Thus, 
as $n$ increases, 
assuming $G < 4  \rttvar_n$, 
\(c - r + 4\rttvar_n \leq \rto_n \leq c + 3r\).
With these bounds, if $\rttvar_n$ is always bounded from below by  $r$, then the \rto 
exceeds the (steady) RTT,
hence no timeout will occur.
On the other hand, we can construct a pathological case where the samples are $c/r$ steady-state
but the \rttvar dips below $r$, allowing the \rto to drop below the RTT.
One such case is illustrated in the bottom of \Figr{karn:scenarios}.
In that case, every $100^{\text{th}}$ sample is equal to $c+r = 75$, and the rest are equal to $c-r=60$.  At the spikes (where $\sample_i = 75$) the sampled RTT exceeds the \rto, and so a timeout would occur.
This suffices to show that steady-state conditions alone do not guarantee a steady-state in terms of avoiding timeouts.
Characterizing the minimal, sufficient conditions for avoiding timeouts during a $c/r$ steady-state
is a problem left for future work.

\begin{figure}[h]
\centering
\includegraphics[width=\textwidth]{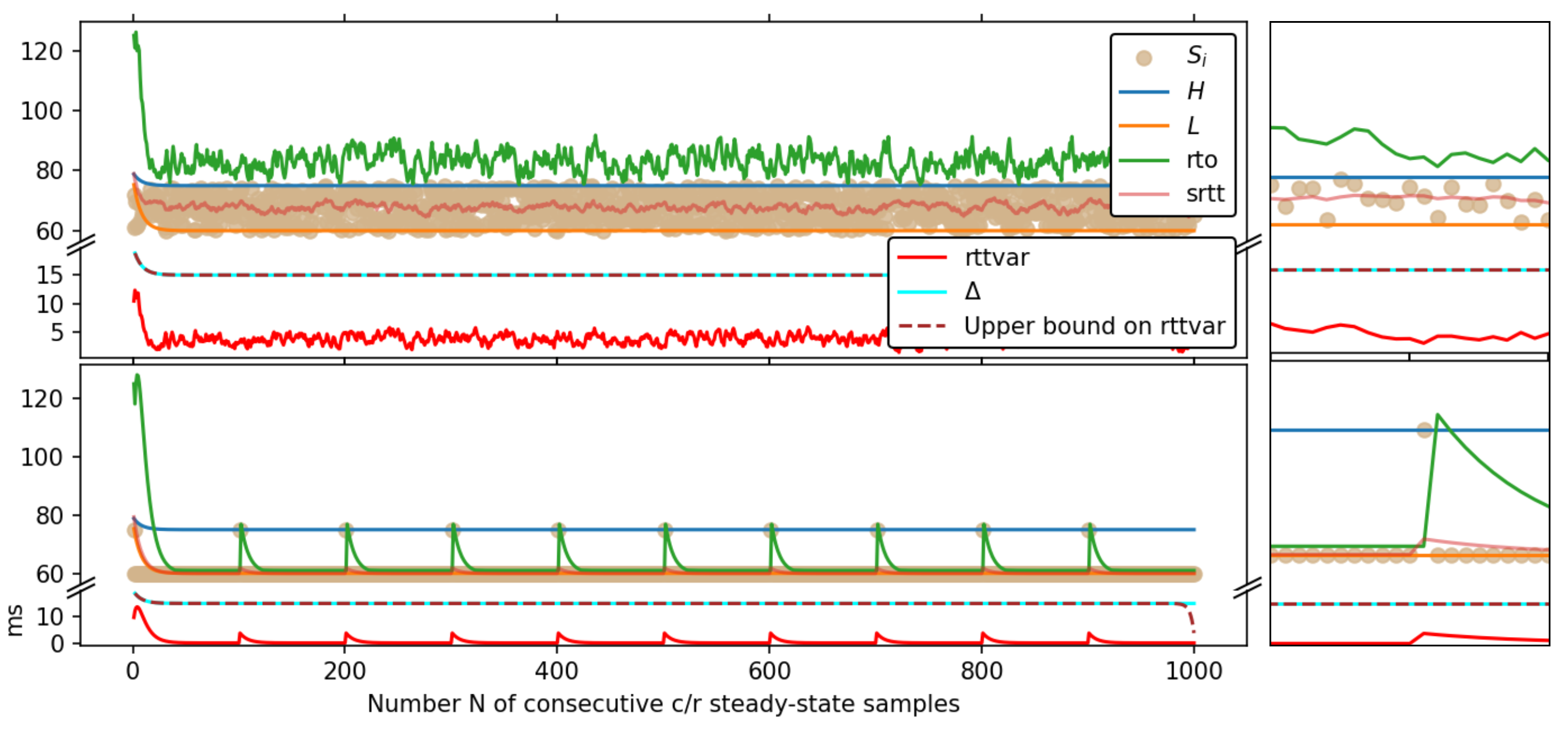}
\caption{On the left are two $67.5/7.5$ steady-state scenarios.  On top the samples are drawn from the uniform distribution over the bounds, and timeouts rarely, if ever, occur.  In the bottom (pathological) scenario, every 100$^{\text{th}}$ sample equals $c+r = 75$ while the rest equal $c-r = 60$, and at each ``spike'', a timeout occurs.  There are infinitely many spikes, and one is shown on the right ($n=[350,450]$).}
\Figl{karn:scenarios}
\end{figure}

\subsection{Real Analysis in \acl, \acls, and \aclr}
In order to prove \Obr{srtt_approaches_s} and \Obr{rttvar_upper_bound} in \acls, we first had to show that 
\(\forall \alpha \in [0, 1) \, :: \, \lim_{n \to \infty} \alpha^n = 0\),
which turned out to be surprisingly challenging.
The most obvious pen-and-paper proof strategy is the following.
\begin{proof}
Let $\epsilon > 0$ and $0 \leq \alpha < 1$ arbitrarily. 
Set $\delta = \log_\alpha(\epsilon)$.
Then \(n > \delta \iff n > \log_\alpha(\epsilon) \iff \alpha^n < \epsilon\).
\end{proof}
However, \acl and \acls do not support irrational numbers, and the logarithm of a rational may be irrational.  Therefore, this proof strategy is not possible in either.  
To address this problem we tried three approaches: 
(1) using \aclr, a variant of \acl designed for non-standard analysis; (2) a direct $\epsilon/\delta$ proof leveraging properties of the ceiling function; and (3) an alternative $\epsilon/\delta$ proof leveraging the bionomial theorem.  We discuss each strategy briefly below.

\subsubsection{\aclr.}~The first strategy was to change tools again and use \aclr, a variant of \acl designed for non-standard analysis.  We formalized the theorem statement using skolemization, like so.
\begin{lstlisting}
(defun-sk lim-0 (a e n)
  (exists (d)
    (=> (^ (realp e) (< 0 e) (< d n)) (< (raise a n) e))))

(defthm lim-a^n->0
  (=> (^ (realp a) (< 0 a) (< a 1) (realp e) (< 0 e) (natp n))
      (lim-0 a e n)) :instructions ...) ;; proof will go here
\end{lstlisting}
Then we defined $\delta$.
\begin{lstlisting}
(defun d (eps a) (/ (acl2-ln eps) (acl2-ln a)))
\end{lstlisting}
After that, we proved some straightforward arithmetic properties, as well as the lemma that $e^{n \ln(\alpha)} = \alpha^n$.  With these challenges surpassed, the remainder of the proof immediately followed:
\begin{sketch}
Let $\epsilon > 0$ and $0 \leq \alpha < 1$ arbitrarily. 
If $\alpha=0$ the result is immediate; suppose $\alpha>0$.
Suppose $\epsilon < 1$, noting that if the theorem
holds for $\epsilon < 1$ then it holds for $\epsilon \geq 1$.
Let $\delta = \ln(\epsilon)/\ln(\alpha)$.
Note that $\ln(\epsilon)$ and $\ln(\alpha)$ are negative.
Let $n$ be some natural number and observe that $n \ln(\alpha) = \ln(\alpha^n)$.
Thus:
\[
\begin{aligned}
n > \delta & \iff n > \ln(\epsilon)/\ln(\alpha) & \text{ by definition of }\delta\\
           & \iff n \ln(\alpha) < \ln(\epsilon) & \text{ multiplying each side by }\ln(a)\\
           & \iff e^{n \ln(\alpha)} < e^{\ln(\epsilon)} & \text{ raising each side above }e\\
           & \iff e^{\ln(\alpha^n)} < \epsilon & \text{ because }e^{\ln(x)}=x\text{ for all }x\text{, and }n \ln(\alpha) = \ln(\alpha^n)\\
           & \iff \alpha^n < \epsilon & \text{ because }e^{\ln(x)}=x\text{ for all }x
\end{aligned}
\] 
\end{sketch}

\subsubsection{Ceiling Proof.}~In contrast to the \aclr proof, this one only uses rationals and therefore could be formalized in \acls.  The skolemized theorem statement with types goes as follows.
\begin{lstlisting}
(defun-sk lim-0 (a e n)
  (declare (xargs :guard (and (posratp a) (< a 1) (posratp e) (natp n))
          :verify-guards t))
  (exists (d) (and (natp d) (implies (< d n) (< (expt a n) e)))))

(property lim-a^n->0 (a e :pos-rational n :nat)
  :hyps (< a 1)
  (lim-0 a e n) :instructions ...) ;; proof will go here
\end{lstlisting}
The proof goes as follows.
\begin{sketch}
Let $0 \leq \alpha < 1$ and $\epsilon > 0$, arbitrarily.
Let $k = \ceil{a/(1-a)}$ and observe that $a \leq k/(k+1)$.
Let $f(n) = k\alpha^k/n$.  
As an intermediary lemma, we claim that for all $n \geq k$, $\alpha^n \leq f(n)$.
\begin{description}
    \item \emph{Base Case}: $n = k$ thus $f(n) = \alpha^k \geq \alpha^n$ and we are done.  
    \item \emph{Inductive Step}: By inductive hypothesis, we have
    \begin{equation}
    a^n \leq k \alpha^k / n
    \label{eqn:ind}
    \end{equation}
    and $k \leq n$.
    This gives us $k/(k+1) \leq n/(n+1)$ and thus:
    \begin{equation}
    \alpha \leq n/(n+1)
    \label{eqn:a-ineq}
    \end{equation}
    Multiplying Eqn.~\ref{eqn:ind} through by $\alpha$, we get
    \(
    \alpha^{n+1} \leq k \alpha^{k+1} / n
    \).
    Combining this with Eqn.~\ref{eqn:a-ineq}:
    \begin{equation}
    \alpha^{n+1} \leq ( k \alpha^k / n ) \frac{n}{n+1} = k \alpha^k / (n+1)
    \end{equation}
    and we are done.
\end{description}
Hence induction: $\forall n \geq k$, $\alpha^n \leq f(n)$.  Now, 
let $\delta = \ceil{k\alpha^k/\epsilon}$.  It follows that 
$\forall n \geq \delta$, $f(n) \leq \epsilon$, and thus by the above result,
$\alpha^n \leq \epsilon$.  We get $\alpha^n < \epsilon$ by repeating this process
for $\epsilon / 2$, and we are done.
\end{sketch}
Although the proof is relatively straightforward on paper,
we found that it required a large number of arithmetic lemmas to pass in \acls, making it cumbersome from a proof-engineering standpoint.

\subsubsection{Binomial Proof}
Finally, we found a direct proof using the binomial theorem.
The proof goes as follows.
\begin{sketch}
Let $\epsilon=x/y > 0$, $\alpha=p/q$, and $b = p/(p+1)$.
First observe that $\alpha \leq b$.
Second, observe that $b^p = p^p/(p+1)^p$.
By the binomial theorem, $(p+1)^p > 2p^p$.
Finally observe that $1/2^y < \epsilon$.
Combining these results, if $\delta = py$ then $n > \delta$
implies $\alpha^n < \alpha^{py} \leq b^{py} \leq 1/2^y < \epsilon$, and we are done.
\end{sketch}
This proof was much simpler than the ceiling proof to implement in \acls, and compared to the proof in \aclr, had the advantage of working in the prover we were already using for our analysis.
We implemented two variants of the Binomial Proof: one which was completely manual,
and another where we leveraged the termination analysis in \acls to find a $\delta$ semi-automatically.
The latter was more elegant as it took greater advantage of the features built into \acls.

\subsubsection{Comparison}
Comparing these proofs leads us to four conclusions.
First, we implemented the ceiling proof in both \acl and \acls, and found it was considerably easier to execute in the latter due to automated termination analysis and contracts checking.
Additionally, the inclusion of types as first-class citizens in \acls made the proof much easier to follow.
Second, \acl (and \acls in particular) 
could benefit better-documented
and more easily searchable library of purely mathematical theorems, relating to the ceiling, floor, exponent, and logarithm, as well as metric spaces and limits.  
Searching for proofs is difficult enough, and \acl does not come with any kind of semantic proof search tool.
And often, even when the desired theorems exist in the \acl books, they are unmentioned in the documentation.
For example, the documentation on ``arithmetic'' does not mention the RTL books, 
and neither does the documentation on ``math''.  Moreover, the rewrite rules from different libraries may conflict,
so even if you find the desired theorems, importing them into a singular environment
may be non-trivial.
Third, \aclr could benefit from the addition of the generic exponent and logarithm.  
This could be done using the lemma outlined in our proof.
Fourth, though \aclr and \acl have incompatible theories, it is nevertheless true that certain kinds of theorems over the reals should hold over the rationals, because the rationals are dense in the reals.  It would useful to have a kind of ``bridge'' between \aclr and \acl, by which the user could justify that a given theorem, if true over the reals, must also hold over the rationals; prove the theorem in \aclr; and then import the theorem, using its ``justification'', into \acl.

\begin{table}[]
\centering
\begin{tabular}{lllllll}
Proof         & LoC & Chars  & Props/Thms & Functions & Books & Cert Time (s) \\\hline
Real          & 161 & 4,224  & 17         & 1         & 5     & 0.58         \\
Ceiling       & 408 & 16,103 & 20         & 3         & 0     & 64.17         \\
Binomial (M)  & 154 & 5,652  & 22         & 1         & 2     & 2.54          \\
Binomial (SA) & 122 & 5,402  & 22         & 2         & 1     & 3.84         
\end{tabular}
\caption{Proof comparison.  (M) refers to ``manual'' while (SA) refers to ``semi-automatic''.  Lines of code and character count are computed without comments or empty lines, however, the proofs are not styled identically.  Props/Thms counts instances of \codify{property} and \codify{defthm}, while Functions counts \codify{definec}s, \codify{definecd}s, and \codify{defun}s.  Certification time is measured on a 16GB M1 Macbook Air.}
\label{tab:comparison}
\end{table}

\section{Related Work}\Secl{karn:related}

To the best of our knowledge, ours is the first work to formally verify properties of Karn's algorithm or the RTO defined in RFC~6298.  However, formal methods have previously been applied to proving protocol correctness~\cite{smith1997formal,lockefeer2016formal,okumura2020formal,cluzel2021layered}, and lightweight formal methods have been used for protocol testing~\cite{bishop2005rigorous,mcmillan2019formal}. 
One such lightweight approach, called \textsc{PacketDrill}, was used to test a new implementation of the RTO computation from RFC~6298~\cite{cardwell2013packetdrill}. The \textsc{PacketDrill} authors performed fourteen successful tests on the new RTO implementation.  After publication, their tool was used externally to find a bug in the tested RTO implementation~\cite{packetDrillRTO}.  In contrast to such lightweight FM, in which an implementation is strategically tested, we took a proof-based approach to the verification of fundamental properties of the protocol design.

Some prior works applied formal methods to congestion control algorithms~\cite{lomuscio2010model,hespanha2001hybrid,konur2011formal,malik2016formal,arun2022starvation,ccac}.  A common theme of these works is that they make strong assumptions about the network model, e.g., assuming the channel never duplicates messages or reorders or loses acknowledgments.  In this vein, we study the case in which acknowledgments are communicated FIFO in \Obr{fifo_rtt}.
Congestion control algorithms were also classically studied using manual mathematics (as opposed to formal methods)~\cite{mathis1997macroscopic,DMSS19,srikant2004mathematics}.  One such approach is called \emph{network calculus}~\cite{le2001network} and has been used to simulate congestion control algorithms~\cite{kim2004network}.  Network calculus has the advantage that it can be used to study realistic network dynamics, in contrast to our Ivy-based approach, which is catered to logical properties of the system.  For example, Kim and Hou~\cite{kim2004network} are able to determine the minimum and maximum throughput of traditional TCP congestion control, but do not prove any properties about what precisely Karn's algorithm measures, or about bounds on the variables used to compute the RTO.

\section{Conclusion}\Secl{karn:conclusion}

In this chapter we applied formal methods to Karn's algorithm, 
as well as the \rto computation 
described by RFC~6298 and used in many of the Internet's protocols.
These two algorithms were previously only studied with manual mathematics or experimentation.
We presented open-source formal models of each, with which we formally verified the following 
important properties.
\begin{enumerate}[{Obs. }1:]
\item Acknowledgements are transmitted in non-decreasing order.
\item Two inductive invariants regarding the internal variables of Karn's algorithm.
\item  Karn's algorithm samples a real RTT, but a pessimistic one.
\item In the case where acknowledgments are neither dropped, duplicated, nor reordered, Karn's algorithm samples the highest \ack received by the sender before the sampled one.
\item For the \rto computation, when the samples are bounded, so is the \srtt.  As the number of bounded samples increases, the bounds on the \srtt converge to the bounds on the samples.
\item For the \rto computation, when the samples are bounded, so is the \rttvar.  As the number of bounded samples increases, the upper bound on the \rttvar converges to the difference between the lower and upper bounds on the samples.
\end{enumerate}
We concluded by discussing the implications of these bounds for the \rto.

In addition to rigorously examining some fundamental building blocks of the Internet, 
we also provide an example
of how multiple provers can be used in harmony to prove more than either could handle alone.
First, we used Ivy to model the underlying system and Karn's algorithm. Ivy offers an easy treatment for concurrency, which was vital for the behavior of the under-specfied models we used for the sender, receiver, and channel. The under-specification renders our results their generality. We guided Ivy by providing supplemental invariants as hints, e.g., \emph{if $\dlvrr(a)$ occurs in an execution, then for all $p < a$, $\dlvrs(p)$ occurred previously}.  Then, since Ivy lacks a theory of rationals, we turned to \acls.
We began by proving two lemmas.  
\begin{itemize}
    \item The $\alpha$-summation ``unfolds'': \((1-\alpha)\sum_{i=0}^N (1-\alpha)^i \alpha + \alpha = \sum_{i=0}^{N+1}(1-\alpha)^i \alpha\).
    \item The \srtt is ``linear'': if $\srtt_{i-1} \leq \srtt_{i-1}'$ and, for all $i \leq j \leq i+n$, $\sample_j \leq \sample_j'$, then $\srtt_{i+n} \leq \srtt_{i+n}'$.
\end{itemize}
Then we steered \acls to prove \Obr{srtt_approaches_s} and \Obr{rttvar_upper_bound} with these lemmas as hints.

Proving the limits of the bounds on \srtt and \rttvar was much trickier, and required manually writing $\epsilon/\delta$ proofs directly in the \acls proof-builder.  We experimented with doing this three different ways, using \acl, \acls, and \aclr, and found that the easiest approach in the context of our pre-existing model was a semi-automated proof in \acls.  These proofs would have been impossible to do in Ivy.  On the other hand, since \acls does not come with built-in facilities for reasoning about interleaved network semantics, we opted to leave the RTT computation proofs in Ivy.  These choices were easier, and yielded cleaner proofs, compared to doing everything using just one of the two tools.

\setcounter{observation}{0}
\setcounter{definition}{0}
\setcounter{problem}{0}
\setcounter{theorem}{0}

\chapter{Formal Performance Analysis of Go-Back-$N$}\Chapl{gbn}
\textbf{Summary.}
In this chapter, we study Go-Back-N, a.k.a. GB$(N)$, a classical automatic repeat request protocol
which was historically used in telecommunications networks and today serves as
the basis for more complex sliding window mechanisms such as the ones found
in TCP Tahoe and New Reno.
We formally model a GB$(N)$ system consisting of a sender and a receiver,
each connected to the other by a token bucket filter (TBF).
The TBF model is meant to capture the behaviors of a real router, or series of routers,
including rate-limiting, reordering, nondeterministic loss, and bounded and unbounded delay --
and we formally verify that, indeed, a single TBF can simulate a series of TBFs in serial composition.
We prove a variety of correctness invariants for our model.
Then, we study the efficiency of GB$(N)$, namely, the fraction of packets received by the
receiver that the receiver delivers to the application.
GB$(N)$ provides reliable FIFO communication, which means that 
the receiver delivers a packet to the application
only once, and only after all packets with lesser sequence numbers were delivered.
Under the simplifying assumption that every packet is the same size, 
we show that 
in the absence of reordering, delay, or nondeterministic loss,
GB$(N)$ can achieve perfect efficiency (efficiency=1).
Citing measurement studies, we argue that a common cause of losses is over-transmission,
where the sender transmits packets faster than the sender-to-receiver TBF can forward them.
We describe a set of constraints under which the GB$(N)$ sender over-transmits, and 
formally characterize the impact the resulting losses have on the efficiency of the protocol
(again, in the absence of other kinds of faults, and assuming packets are equally sized).
Our results are parameterized by the window size~$N$ of the protocol,
transmission rate of the sender, and parameters of the two TBFs; and we formally verify
all our theorems in \acls.

\begin{description}
\item \underline{Contribution:} MvH created the model and proofs, and wrote the chapter.
\end{description}

\section{Overview of Go-Back-$N$}\Secl{gbn-intro}

Automatic repeat request (ARQ) protocols provide reliable FIFO communication over an unreliable bidirectional channel
connecting a sender and a receiver.
In every ARQ protocol, the sender transmits a sequence of packets to a receiver, who provides feedback in the form of \acks.  The sender uses this feedback to decide what packets to send next.  There are multiple ARQ protocols, such as Stop-and-Wait, a.k.a. the Alternating Bit Protocol (ABP); Go-Back-$N$, abbreviated GB$(N)$; Selective Repeat; Hybrid ARQ; etc.  The simplest is ABP, where the sender does not transmit the next packet until it has received confirmation that the prior one was delivered.  GB$(N)$ extends this idea by using a \emph{window} of $N$-many packets the sender can transmit at a time, for some fixed positive integer $N$.

At a high level, GB$(N)$ works as follows.
The sender has a list of datagrams referred to as \emph{packets}
which it intends to transmit.
Each packet has a sequence number.
The sender begins transmitting starting with the packet with the lowest sequence number, which is one.
It transmits the first $N$ packets, ordered by sequence number, then starts a timer\footnote{(namely, the RTO timer studied in \Chapr{karn-rto})}.
If a cumulative \ack for any of the $N$ packets it just sent arrives before the timer goes off, then the sender
cancels the timer and ``slides the window'', beginning the transmission of a new window starting with the new \ack value.  (Note, if the \ack does not acknowledge the entirety of the prior window, then the new window and the prior one will overlap.)
On the other hand, if the timer expires without any new \ack arriving, the sender ``goes back~$N$'', retransmitting the window from its start.

Unfortunately, this high-level description leaves many details unstated,
and to make matters worse, GB$(N)$ does not have a single, canonical definition.
The protocol is 
defined in several networking textbooks (e.g.,~\cite{ross2012computer,peterson2007computer,tanenbaum,bertsekas2021data}),
without citation to any original definition.
Each defines GB$(N)$ slightly differently, or omits key details making it 
unclear what precisely they believe GB$(N)$ does.
Points of contention include:
\begin{enumerate}[(1)]

\item Whether the receiver is expected to \ack every message (as in~\cite{ross2012computer}), or just some (\cite{peterson2007computer} describes both options; while \cite{bertsekas2021data} describes a receiver who sends \acks within some bounded time of each packet receival).
In the latter case, the receiver might send \acks on some temporal schedule (as in \cite{bertsekas2021data}),
or it might \ack every $k^{\text{th}}$ message received, or delivered, for some positive integer $k$.
On one hand, if the receiver \acks every message received, the sender can quickly determine if 
a packet was lost (e.g., using a duplicate \ack heuristic)\footnote{Thank you to Lenore Zuck for pointing this out.}.
On the other hand, unless the \acks are piggy-backed on existing messages in the opposite direction,
acknowledging every message could considerably increase the burden on the network in the
receiver-to-sender direction.
It is also worth noting that a receiver which does not send an \ack until a certain number of packets
were \emph{delivered} -- that is, received in-order -- may not transmit any \acks for a long time if some packets near the bottom of the window are reordered in transit.

\item Whether the receiver ignores out-of-order packets when computing the cumulative \ack (as in~\cite{ross2012computer,tanenbaum,bertsekas2021data}, but not~\cite{peterson2007computer}).
A receiver who ignores out-of-order packets only needs to keep track of the most recent \ack it sent,
whereas one who processes all packets needs to constantly maintain a bit-vector of size $N$ in order to compute its next \ack transmission.
On the other hand, if two packets in a window are reordered, a receiver who ignores out-of-order packets will force the sender to retransmit the window portion beginning with the lesser of the two reordered packets, whereas a receiver who processes all packets will not force a retransmission.
Forcing the retransmission is inefficient, but not necessarily ``wrong''.

\item Whether it is possible for the sender to process an \ack part-way through transmitting its window,
as in~\cite{peterson2007computer,ross2012computer,bertsekas2021data}, 
or, if the sender will wait until all $N$ packets have been transmitted.
This is unspecified in~\cite{tanenbaum}.
To see why this matters, suppose the sender has just begun retransmitting a window when an \ack for the entire window arrives.  Ideally, the sender would process the \ack and forgo unnecessary retransmissions.
Yet, it is unclear whether a sender who ignores the \ack until it has transmitted the entire window is ``wrong'', per se.

\end{enumerate}

Despite it not being well-defined, GB$(N)$ is referenced in many RFCs (\cite{rfc7414,rfc6675,rfc6538,rfc5236,rfc4821,rfc4015,rfc3782,rfc3522,rfc3517,rfc3366,rfc1458,rfc1254,rfc1152}).
Thus, Bertsekas and Gallager refer to it as ``the basis for most error recovery procedures in the transport layer''~\cite{bertsekas2021data}.
The protocols described in these RFCs also differ in the points raised above,
so, we cannot simply define GB$(N)$ to be ``whatever it is in practice''.
For example, in TCP New Reno~\cite{rfc3782} the receiver must consider out-of-order packets
in order for the duplicate-\ack recovery mechanism to work.
Yet, RFC 3366~\cite{rfc3366} says that GB$(N)$ is alternatively known as ``Reject'',
implying it involves a receiver who rejects out-of-order packets.

\subsection{Prior Models}
Several prior works formally model GB$(N)$, yet, do not arrive at any consensus 
regarding the three points of contention outlined above,
nor agree on what assumptions to make about the network.
Two prior works (with intersecting authorship) use a Markov chain model to derive a probability generating function for the delay between when a packet
is transmitted and when a corresponding \ack is first delivered, 
i.e., for the average RTT,
in GB$(N)$~\cite{de2005delay} 
and protocols that extend it~\cite{de2002queueing}.
In their models, the receiver (1) \acks every $N^{\text{th}}$ packet delivered and
(2) ignores out-of-order packets, and (3) the sender does not process
an \ack until the entire window was transmitted.
They assume \acks are forwarded from the receiver to the sender at a fixed temporal schedule, 
and that the sender's transmission rate is constant.
In a related work~\cite{hasan2008performance}, 
Hasan and Tahar formalize ABP, GB$(N)$, and Selective-Repeat using Higher Order Logic,
and compute (and verify) the average RTT.
In their model, the receiver (1) \acks every packet and
(2) ignores out-of-order packets, and (3) the sender buffers \acks as soon as they arrive.
Hasan and Tahar assume the RTO is not greater than the sum of the average time between when the sender
sends a packet and when the receiver receives it, and the average time between when the 
receiver sends an \ack and the sender receives it.
They refer to this sum as the RTT, although as we explain in \Chapr{karn-rto},
it is not the same as the ``RTT'' sampled by Karn's Algorithm.\footnote{In contrast, whenever we refer to the RTT, we are referring to the value that Karn's Algorithm measures.}
All three of these works abstract packet reordering and corruption using randomized errors; ignore
errors in the receiver-to-sender direction; treat lost messages identically to corrupted ones; and assume that the time it takes the channel to transport a message from one endpoint to the other is constant.
It is also worth noting that all three works define ``average'' to mean the expected value,
using probabilities.

There are also several prior works which use formal methods to study the correctness (as opposed to performance) of Go-Back-$N$, meaning, they attempt to verify that the protocol provides reliable, in-order message delivery.
Most of these assume an idealized network without packet or \ack reordering, 
and/or use a specific and unrealistically small window 
size~\cite{bezem1994correctness,kaivola1997using,stahl1999divide}.
However, Chkliaev et.~al. model an improved sliding window protocol based on GB$(N)$,
which they prove provides reliable in-order delivery
    under liveness assumptions and
    restrictions relating the window size and
    maximum sequence number, 
    using a network model with loss, reordering, delay, and 
    even duplication~\cite{chkliaev2002formal,chkliaev2003verification}.
In their model, (1) the receiver's acknowledgment strategy is left nondeterministic,
yet it (2) buffers out-of-order packets, and (3) the sender buffers \acks as soon as they arrive.
El Minouni and Bouhdadi took a different approach, using a refinement argument~\cite{el2015applying}.
In their model, the receiver (1) \acks every delivered packet and (2) does not buffer out-of-order packets,
and (3) the sender buffers an \ack as soon as it arrives.
Much like the probabilistic models described in the previous paragraph,
 Minouni and Bouhdadi's abstracts all network faults identically.
All the mentioned prior GB$(N)$ models are summarized in \Tabr{prior-gbn}.

\begin{table}[h]
\centering
\small
\begin{tabular}{p{1cm}p{5cm}p{1.8cm}p{1.8cm}p{1cm}p{1cm}p{1cm}p{1cm}p{1cm}}
Model & Receiver strategy & OOO packets? & ACKs mid-window? & Reorder & Loss & Dupl & Delay & Prop \\\hline
\cite{de2005delay,de2002queueing}                  & Every $N^{\text{th}}$ packet delivered & No                  & No                       & Abs & Abs & No          & Const & Avg. RTT    \\
\cite{hasan2008performance}                        & Every packet delivered                 & No                  & Yes                      & Abs & Abs & No          & Const & Avg. RTT    \\
\cite{bezem1994correctness}                        & Every packet received                  & N/A                 & Yes                      & No         & Yes        & No          & No       & Corr \\
\cite{kaivola1997using}                            & Every packet received                  & Yes                 & Yes                      & No         & Yes        & No          & No       & Corr \\
\cite{stahl1999divide}                             & ND                       & Yes                 & Yes                      & No         & Yes        & No          & No       & Corr \\
\cite{chkliaev2002formal,chkliaev2003verification} & ND                       & Yes                 & Yes                      & Yes        & Yes        & Yes         & Bnd  & Corr \\
\cite{el2015applying}                              & Every packet delivered                 & No                  & Yes                      & Abs & Abs & Abs  & Bnd  & Corr
\end{tabular}
\caption{Summary of prior models of GB$(N)$ or extensions thereof.  For each model, we summarize the receiver strategy, whether or not the receiver buffers out-of-order (OOO) packets when computing its next \ack transmission, whether or not the sender processes \acks as soon as they arrive (even if it has not yet completed its current window transmission), whether the network captures reordering, loss, duplication, and/or delay, and what property the model was used to study.  We use ND to mean nondeterministic, Abs to mean abstracted, Const to mean constant,  Bnd to mean bounded, and Corr to mean correctness.}
\Tabl{prior-gbn}
\end{table}

There is an emerging body of work which attempts to study congestion control algorithms
using formal methods~\cite{ccac,agarwaltowards,arun2022starvation,DMSS19}, and in that context, it is important for the network model to be 
somewhat realistic so that the algorithm under study is not scrutinized using implausible
traffic flows.
One feature which real networks tend to implement is \emph{rate limiting}, and 
the most common model for a single-direction
rate limiting network is called a Token Bucket Filter, or TBF~\cite{tang1999network}.
The basic idea of the TBF is that it has a counter, called a \emph{bucket}, which it increments 
at a constant rate up to a fixed capacity, and an internal, fixed-byte-capacity queue which holds the 
messages it intends to forward.
The bucket is commonly described as holding ``tokens'', e.g., if the bucket is set to 17, then we say 
the TBF has 17 tokens.  Intuitively, tokens are the currency the TBF needs to forward messages to the receiving endpoint.
The TBF drops any messages sent to it for which there is not sufficient space remaining in the queue, 
and whenever it forwards a message, it reduces the number of tokens in the bucket by the size, in bytes, of the message.
The combination of the fixed-byte-capacity queue and capped bucket suffice to implement rate-limiting.
Although variations on the TBF have been employed for several congestion control
verification tasks~\cite{ccac,agarwaltowards}, to the best of our knowledge, 
no prior works verified aspects of GB$(N)$ in the context of a TBF-based network model.
It is therefore an open question how precisely GB$(N)$ behaves when configured over a rate-limiting network.

\subsection{Our Model and Contribution}
In this chapter our goal is twofold.
First, we want to build an \emph{executable} model that is flexible to many non-probabilistic 
verification tasks and captures relevant details of GB$(N)$ which were ignored in prior works,
most notably, the behavior of GB$(N)$ when the sender and receiver are connected by a rate-limiting
network consisting of a TBF in each direction.
The fact that the model is executable is important because it means that in the future, it can be used for not just verification
tasks (which is how we use it) but also for simulations and attack discovery (as we did in~\cite{kumar2024formal}).
Second, to illustrate the utility of our model, we aim to answer a question which was not directly
studied in prior works, namely, what kind of performance we can expect from GB$(N)$ in the best and worst
case scenarios.
This question relates directly to the network definition because the worst-case scenario that we study for GB$(N)$
arises from the interaction between the GB$(N)$ sender and the sender-to-receiver TBF.

Our definition of the ``best case'' is the scenario where nondeterministic network
failures (such as nondeterministic loss, reordering, and delay) do not occur, i.e.,
the network behaves in an idealized fashion, and the sender and receiver both transmit
at rates $\leq$ the rates at which the buckets of the corresponding TBFs refill,
meaning, neither component over-transmits.
And, we define the ``worst case'' scenario as one where the sender overwhelms the network with packet 
transmissions, leading to congestion.
We formalize both in this chapter, and argue, citing measurement studies, that
the latter of the two is a realistic ``worst case''.
To the best of our knowledge, we are the first to formally define and study this over-transmission
scenario for GB$(N)$.
However, since we focus on best and worst case behaviors, 
we do not need to reason about the most \emph{likely} behavior of 
the protocol, and thus, our model does not capture
the probability of events such as nondeterministic loss or reordering.
For this reason, our model (in its current form) is not appropriate for reasoning about the average case, as was done in~\cite{de2005delay,de2002queueing,hasan2008performance}.

Our model resolves the three points of contention outlined above in the following ways.
\begin{enumerate}[(1)]

\item Rather than explicitly encoding how often the receiver should transmit an \ack, 
we leave it nondeterministic, allowing us to model many different receiver strategies.
This is the same approach which was taken in~\cite{chkliaev2002formal,chkliaev2003verification},
but not~\cite{bertsekas2021data}, which assumed the receiver transmits an \ack within some temporal window
of receiving a packet.
An advantage of our approach is that our model is not restricted to just the particular 
problems we study, yet, we are still able to prove theorems about specific receiver
strategies, by making the receiver's strategy an assumption of the theorem.
For instance, in our best and worst case theorems, we assume the receiver
waits to receive $N$ packets before sending another \ack.\footnote{Note -- $N$ is the number of packets received, not necessarily the number which were delivered.  In other words, all $N$ could potentially be out-of-order.} 
We view this as the worst possible realistic receiver strategy.
If the receiver waited to receive $> N$ packets before sending an \ack, it would cause unnecessary
retransmissions after every window, which seems unrealistic;
and if it counted delivered rather than received packets, it would not be able to provide
any feedback if the first packet was reordered.
But, the longer it waits, the more the system suffers from losses or retransmissions;
thus of the realistic options, waiting for $N$ packets is the worst.

\item We assume the receiver does not buffer out-of-order packets,
even when using the strategy described above.
In other words, the receiver might count the number of packets received,
including out-of-order packets, and send an \ack after, for example, every $N^{\text{th}}$
receive event, but it will not acknowledge any packets which were received out of
order with its next \ack transmission.
The reason we make this assumption is because it was common to most of the GB$(N)$ works we 
surveyed, some of which claimed that buffering out-of-order packets is a feature which 
distinguishes Selective-Repeat ARQ from GB$(N)$ (see e.g.,~\cite{hasan2008performance}).

\item We assume the sender will process any \ack it receives 
immediately, before sending more packets, even if it is not yet
done sending the current window.
This assumption is also common to most of the works we surveyed,
obviously improves the performance of the protocol, 
and in contrast to buffering out-of-order packets,
is not explicitly ruled out for GB$(N)$ by any prior work or textbook we found.

\end{enumerate}
We define our system in the context of a network model which uses a TBF
in each direction.
Our TBF does not just capture deterministic rate limiting, but also 
nondeterministic loss, delay, and reordering, allowing us to simulate numerous
possible network failure conditions.

Despite being in many ways more realistic than prior models,
ours still makes several simplifying assumptions or abstractions.
First, we model the retransmission timer nondeterministically -- it is allowed to fire at any time after the sender has transmitted the last packet in a window and before it has received any \ack which acknowledges any portion of that window.
We make the liveness assumption that the retransmission timer is not enabled forever without firing, so,
it cannot block the system from progressing.
For a more detailed treatment of the retransmission timer, the reader is referred to our work in \Chapr{karn-rto}.
Additionally, we assume that sequence numbers are unbounded.
This assumption drastically simplifies our proofs, 
but it can only be safely assumed if the network satisfies some formal criteria which were
previously verified in~\cite{chkliaev2003verification}.
Third, we assume that \acks count packet sequence numbers rather than bytes.
That is to say, an \ack of 7 acknowledges the packets with sequence numbers 1, 2, 3, 4, 5, and 6, as opposed to the first 6 bytes of the byte-stream encoded by the payloads of the in-order packets.
We discuss the latter two assumptions further in \Secr{gbn:model:setup}.

\begin{figure}[h]
    \centering
    \begin{tikzpicture}
\node[draw, rectangle] (sender) {\small Sender};
\node[draw, rectangle, right=0.9cm of sender] (receiver) {\small Receiver};
\node[below=4cm of sender] (endS) {};
\node[below=4cm of receiver] (endR) {};
\draw[-] (sender) -- (endS);
\draw[-] (receiver) -- (endR);
\draw[fill=gray!20] (0.2,-0.4) rectangle (2.3,-4.2);
\node[] (Channel) at (1.3,-0.6) {\small Channel};
\node[] (GB1) at (1.3,0.8) {\small GB$(1)$=ABP=Stop-and-Wait};
\draw[->] (0,-1) to node[fill=gray!20]{\small $1$} (2.5,-1.5);
\draw[->] (2.5,-1.6) to node[fill=gray!20]{\small $2$} (0,-2.1);
\node[draw,fill=red!20] (loss0) at (2.5,-2.7) {loss!};
\draw[->] (0,-2.2) to node[fill=gray!20]{\small $2$} (loss0);
\draw[->] (0,-3) to node[fill=gray!20]{\small $2$} (2.5,-3.5);
\node[fill=blue!20,draw] at (-0.8,-3) {\small \emph{timeout}};
\draw[->] (2.5,-3.6) to node[fill=gray!20] {\small $3$} (0,-4.1);
\node[draw,rectangle,right=1.3cm of receiver] (sender2) {\small Sender};
\node[draw,rectangle,right=0.9cm of sender2] (receiver2) {\small Receiver};
\node[below=6cm of sender2] (endS2) {};
\node[below=6cm of receiver2] (endR2) {};
\draw[-] (sender2) -- (endS2);
\draw[-] (receiver2) -- (endR2);
\draw[fill=gray!20] (5.65,-0.4) rectangle (7.75,-6.2);
\node[] (Channel) at (6.75,-0.6) {\small Channel};
\node[] (GB1) at (6.75,0.8) {\small GB$(2)$};
\draw[->] (5.45,-1.4) to node[fill=gray!20]{\small $2$} (7.95,-1.9);
\draw[->] (5.45,-1) to node[fill=gray!20]{\small $1$} (7.95,-1.5); 
\draw[->] (7.95,-2) to node[fill=gray!20]{\small $3$} (5.45,-2.5);
\node[draw,fill=red!20] (loss0) at (7.95,-3.2) {\small loss!};
\draw[->] (5.45,-2.6) to node[fill=gray!20]{\small $3$} (loss0);
\draw[->] (5.45,-3.4) to node[fill=gray!20]{\small $4$} (7.95,-3.9);
\node[fill=blue!20,draw] at (4.6,-4.5) {\small \emph{timeout}};
\draw[->] (5.45,-5) to node[fill=gray!20,xshift=-0.5cm,yshift=0.05cm] 
    {\small $4$} (7.95,-5.5); 
\draw[->] (7.95,-5.2) to node[fill=gray!20,yshift=-0.1cm] {\small $4$} (5.45,-5.7);
\draw[->] (5.45,-4.6) to node[fill=gray!20] 
    {\small $3$} (7.95,-5.1);
\end{tikzpicture}
    \caption{Example message sequence charts for GB$(1)$ and GB$(2)$.  In both, the receiver waits to receive $N$ packets (regardless of whether or not these packets are in order) before transmitting an \ack~$a$ cumulatively acknowledging all packets $p < a$.  The sender in both charts successfully transmits an entire window, but then loses the first transmission of the subsequent window.  In GB$(1)$, the entire second window is lost, resulting in a retransmission.  On the other hand, in GB$(2)$, the second window consists of two packets, the second of which (carrying sequence number $4$) is successfully received, but not delivered since it is out of order.  Since the packet with sequence number $3$ did not make it through, the sender is forced to retransmit.}
    \Figl{gbn:seq}
\end{figure}
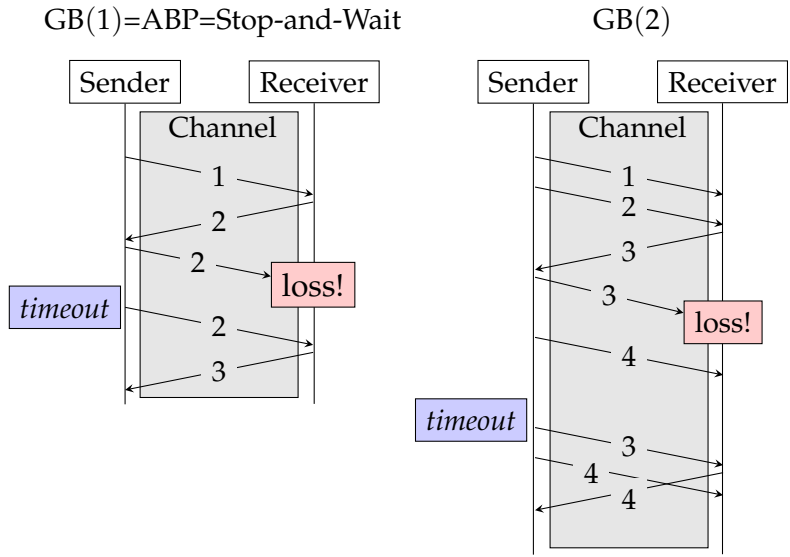

In order to characterize system performance, we study the \emph{efficiency} of GB$(N)$, namely, 
the fraction of packets received by the receiver which the receiver delivers to the application.
Put differently, this is the fraction of received packets which are useful.
Thus, the worst possible efficiency is zero, and perfect efficiency is one.
In ARQ protocols that use a cumulative acknowledgment scheme, 
whenever the receiver transmits a cumulative \ack,
the \ack is equal to one plus the number of useful packets received so far.
For example, if the receiver receives packets with sequence numbers 1, 2, 2, 1, and 3,
then its efficiency is 3/5, since two of the packets were duplicates and therefore not useful.
If it subsequently transmits an \ack, that \ack will equal 4.
Thus, we can measure the long-term efficiency of an ARQ protocol by looking at the average increase
in subsequent \acks transmitted by the receiver, divided by the number of packets the receiver receives
between \ack transmissions.
Using our model, we show that GB$(N)$ can, in the absence of loss, reordering, or 
delay, achieve perfect efficiency.
Then, we compute the efficiency of the system when the sender transmits packets faster than the TBF
can forward them to the receiver, leading to losses.
To the best of our knowledge, we are the first to formally model this problem, which arises from 
the interaction between GB$(N)$ and the sender-to-receiver TBF, and therefore could not have been
studied using the previously mentioned models which did not include a TBF in either direction.

The rest of the chapter is organized as follows.
Using traditional, pen-and-paper mathematics, we
describe the setup of our model in \Secr{gbn:model:setup},
and then describe how we model the sender, receiver, and each TBF,
and the invariants we prove about each component, in \Secr{gbn:model:sender},
\Secr{gbn:model:receiver}, and \Secr{gbn:model:tbf}, respectively.
We tie it all together by explicitly encoding the entire system transition
relation in \Secr{gbn:model:tran}.
Then we analyze the efficiency of GB$(N)$ in the best case,
and in a scenario where the sender over-transmits, in \Secr{gbn:efficiency}.
In \Secr{acl2s}, we explain how we formalize these pen-and-paper models and properties
in \acls, and what our mechanized proofs look like.
The section assumes familiarity with \acls, and therefore, readers unfamiliar with the prover
may find the preceeding five sections more useful for understanding our model and results.
Conversely, readers familiar with \acls may find it easier to skim the pen-and-paper mathematics
 and focus more on how the models and proofs were formalized in the theorem prover.
Finally, we survey related works in \Secr{gbn:related} --
other than those already discussed above --
and conclude in \Secr{gbn:conclusion}.

\section{Setup for Formal Model of Go-Back-$N$}\Secl{gbn:model:setup}

In our model, a \emph{datagram} is a tuple 
\(\dg = (i, x)\)
consisting of a positive integer~$i$,
which we refer to as the id of the datagram,
and a string payload~$x$.
For convenience, we use Dg to denote the set of all datagrams.
In a real network, the payload is a byte array, but we model it as a string so that
the traces which get printed when the model is executed are easier to read (in 
the sense that they have interpretable messages like \texttt{HELLO} or \texttt{ACK}).
We use the length of the payload~$x$, denoted $\emph{sz}(\dg)$, as a proxy for the byte-size of the datagram.
In this convention (and others) we drop redundant parentheses, for example,
writing $\emph{sz}((1, \texttt{EAT})) = \emph{sz}(1, \texttt{EAT}) = 3$.
We assume the existence of a fixed, maximum payload size, but we do not assume this maximum size
is any one particular value.

Datagrams are separated into (data) \emph{packets} (which the sender sends) 
and \emph{acknowledgments}, or \acks (which the receiver sends).
The goal of the sender is to communicate a stream of bytes to the receiver, e.g., 
\texttt{welcome to Alaska}, in order and without omissions.  
The byte-stream is broken into packets,
ordered by consecutive id, starting with 1, e.g., 
$[(1, \texttt{welcome}), (2, \texttt{to}), (3, \texttt{Alaska})]$.\footnote{In practice the byte-stream would probably not be broken up by spaces -- we just present it this way for illustration.}
We refer to the id of a packet as its \emph{sequence number}.
Meanwhile, in our model, every \ack $(j, y)$ has the payload $y=\texttt{ACK}$
and is said to be \emph{cumulative} in the sense that it acknowledges 
all transmitted packets with sequence numbers $< j$, but no packet with sequence number $j$.
So, in our example, 
$(3, \texttt{ACK})$ acknowledges $(1, \texttt{welcome})$ and $(2, \texttt{to})$ 
but not $(3, \texttt{Alaska})$.
This nomenclature is consistent with our definitions in \Chapr{karn-rto},
except that now the datagram type is enriched with a payload.

In our model, we make some simplifying assumptions about both sequence numbers and cumulative acknowledgments.

\emph{Sequence numbers.}
In our model we simplify how sequence numbers are treated in two important ways.
First, in the real world, sequence numbers are bounded, typically by $2^{32}$, and once a sender
has transmitted that many packets, the sequence number of the next packet ``wraps around'' back
to one.  
For example, in a Gbps network, the sequence number can wrap in $\leq$ 34s~\cite{microsoftTCPfeatures}.
This can cause ambiguities where the receiver of a packet is not certain whether it was sent
after the sequence number wrapped or before (in which case it must have been delayed in transit).
Such ambiguities can cause problems, such as
    stale RTT estimates via Karn's Algorithm 
    or the inability to detect false reactions to losses
    in loss-based congestion control algorithms~\cite{pawsLinux}.
The classical solution, called Protection Against Wrapped Sequences (PAWS), is to include a 12 byte timestamp in each datagram, and use it to disambiguate datagram order~\cite{rfc1323_timestamp}.
Although the timestamp also wraps, just like the sequence number, 
it does so at most once every 24 days~\cite{wright1995tcp}.
In practice, ambiguities caused by wrapped sequence numbers are considered sufficiently rare
that many TCP applications do not use PAWS by default~\cite{pawsMicrosoft}
and for those where such ambiguities do occur (and matter)
PAWS is generally considered an adequate solution.
Consequently, all of the related works we survey in \Secr{gbn:related}
except for~\cite{chkliaev2003verification} 
make the simplifying assumption that sequence numbers are unbounded
or that the bound is much larger than the byte-length of the data stream the sender aims to transmit.
(The work in~\cite{chkliaev2003verification} 
explicitly defines and proves the conditions under which
bounded sequence numbers are unambiguous and thus our unbounded simplification is acceptable.)
In this chapter, we assume sequence numbers are unbounded, but when a theorem statement would change
under a model with bounded sequence numbers, we say so and explain how.

\emph{Cumulative ACKs.}
In TCP and similar protocols, an \ack cumulatively acknowledges the bytes delivered so far, but in our model (in both this chapter and the previous), \acks are cumulative over packet sequence numbers, not bytes.
This simplification makes our proofs easier but, in contrast to the bounded/unbounded simplification we just described,
it does not lose any model fidelity,
because we explicitly encode the payload of each packet in the model, and therefore,
we can always convert an \ack from sequence-number form to byte-form.

We model a Go-Back-$N$ system consisting of four components:
    a sender who sends packets and receives acknowledgments,
    a receiver who sends acknowledgments and receives packets,
    and two Token Bucket Filters (one in each direction) which (attempt to)
    move datagrams from endpoint to endpoint (e.g. from sender to receiver).
Our model is illustrated in \Figr{gbn:system}.

\begin{figure}[H]
    \centering
    \begin{tikzpicture}
\node[fill=white, draw,minimum height=1cm,align=center] (sender) {sender};
%
%
\node[fill=white, draw, above right=0.1cm and 1.5cm of sender] (ts) {$\mathcal{F}_s$};
\node[fill=white, draw, below right=0.1cm and 1.5cm of sender] (tr) {$\mathcal{F}_r$};
\node[fill=white, draw, below right=0.1cm and 1.5cm of ts,minimum height=1cm,align=center] (receiver) {receiver};
\draw[->] (sender) to[out=north east,in=west] node[above] {packets} (ts);
\draw[->] (ts) to[out=east,in=north west] (receiver);
\draw[->] (receiver) to[out=south west,in=east] node[below] {\acks} (tr);
\draw[->] (tr) to[out=west,in=south east] (sender);
\end{tikzpicture}
    \caption{The system, consisting of the sender, receiver, and two TBFs ($\mathcal{F}_s$ and $\mathcal{F}_r$).  Each time the sender transmits a packet it flows through $\mathcal{F}_s$ before reaching the receiver (or gets dropped in-transit), and likewise, when the receiver eventually responds with an ACK, it flows through $\mathcal{F}_r$ before reaching the sender (or gets dropped).}
    \Figl{gbn:system}
\end{figure}
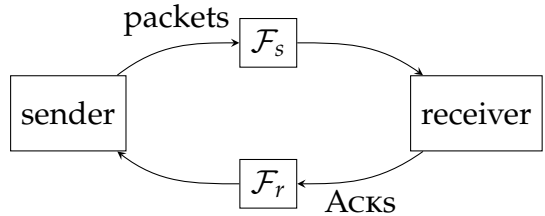

Components in our Go-Back-$N$ system synchronize on the following events:
\begin{itemize}
    \item $\snds(\dg)$ in which the sender sends the packet $\dg$ into $\mathcal{F}_s$.  This is an output event for the sender and an input event for $\mathcal{F}_s$.
    \item $\sndr(\dg)$ in which the receiver sends the \ack $\dg$ into $\mathcal{F}_r$.  This is an output event for the receiver and an input event for $\mathcal{F}_r$.
    \item $\dlvrs(\dg)$ in which the receiver receives the packet $\dg$ from $\mathcal{F}_s$.  This is an output event for $\mathcal{F}_s$ and an input event for the receiver.
    \item Finally, $\dlvrr(\dg)$ in which the sender receives the \ack $\dg$ from $\mathcal{F}_r$.  This is an output event for $\mathcal{F}_r$ and an input event for the sender.
\end{itemize}
When convenient we drop redundant parentheses, e.g.,
writing $\snds(i, x)$ rather than $\snds((i, x))$.

Like in \Chapr{karn-rto}, the sender, receiver, and TBFs are non-blocking 
in the I/O automata sense~\cite{lynch1988introduction},
that is to say, no component of the system can block one of its input events from occurring.
Mathematically, this works as follows.
Each component in the model
evolves according to a set of \emph{state update functions}.
Each state update function $f$ takes as input a component state $s$, 
and a (potentially empty) argument list $\alpha$, 
and outputs an updated component state $s'$, and either an output event $o$
or the special symbol $\perp$, which denotes null.
There are two types of state updates: \emph{internal} updates of the form 
$(s', o) = f(s, x_1, x_2, \ldots)$, where $\alpha = x_1, x_2, \ldots$
is a list of typed variables nondeterministically selected by the acting
component, and \emph{external} updates of the form $(s', \perp) = f(s, \evt)$,
 where $\alpha = \evt$ is an input event of the
component, and $o = \perp$ (i.e., the update does not output an event).
An internal update is allowed to have a \emph{precondition}, namely, some predicate over
the state and arguments which must hold in order for the update to occur.
We do not allow preconditions on external updates as this would enable blocking.
Lastly, an event $\evt$ is an input event to a component $c$ if and only if $c$ has just one (and not more than one)
external state update function which takes $\evt$ as its argument.

The rules of concurrency naturally follow:
each component can execute at most one state update function at a time; and
two (or more) components can update concurrently, provided that if one of the components
outputs an event which is an input to another component, the two synchronize on the given 
event (updating in lockstep).
Naturally, this means a single event cannot be both an input to and an output of the same component.
However, we do not encode these concurrency rules in our \acls code,
because they are not relevent for the theorems we prove.
Rather, the \acls code simply describes each state update function individually,
and the places where components synchronize on events.
We explain the encoding of the system transition relation in the \acls code 
in more detail in \Secr{gbn:model:tran} and \Secr{acl2s}.

We model our system in steps.
First, we define the state update functions for each component.
We define the sender's transition relation in \Secr{gbn:model:sender},
the receiver's in \Secr{gbn:model:receiver},
and the transition relations for the two TBFs in \Secr{gbn:model:tbf}.
Then we use those functions to define the component transition relations.
Finally, we build the composite transition relation for the entire system
out of the individual transition relations of the components, in \Secr{gbn:model:tran}.
The composite relation captures the semantics described above, albeit, 
with the caveat that if two components in the real system update at once,
the corresponding model trace consists of two updates in a row (which commute).
Throughout, we use the 
following conventions: 
$\mathbb{N}$ denotes the naturals (including zero); 
$\setminus$ denotes set subtraction;
$\mathbb{N}_+ = \mathbb{N} \setminus \{ 0 \}$ denotes the positive
naturals;  
for any lists $A$ and $B$, $A;B$ denotes their concatenation;
and Str denotes the set of all strings whose length does not exceed the maximum
payload size.

\section{Formal Model and Correctness of the Go-Back-$N$ Sender}\Secl{gbn:model:sender}

Our GB$(N)$ sender model is quite general, capturing the behaviors of a number of possible
implementations at once.  In this section we first formalize our model, and then explain how
it captures numerous potential implementation choices.

We assume $N$ is a fixed positive constant integer, and
model the sender's state as a tuple of positive integers~$\sender= (\ha, \hp, \cur)$, where
    $\ha$ is the highest \ack id received so far (or one if none were received so far);
    $\hp$ is the highest packet id sent so far; and
    $\cur$ is the id of the next packet the sender plans to transmit (initially one).
We take the convention that $\hp$ is unitialized (i.e. null) until the sender has sent at least one packet.
At any given time, the current \emph{window} is the integer interval $[\ha, \ha + N]$.
The sender updates its state according to the following three functions.

\paragraph{\(\textsf{rcvAck}(\sender, \evt)\):} An external update triggered by $\evt = \dlvrr(a, \texttt{ACK})$.  If $\ha < a \leq \hp + 1$, it ``slides the window'' by setting $\ha := a$ and $\cur := \max(\cur, a)$.  Else, it does nothing.

\paragraph{\(\textsf{advCur}(\sender, x)\):} An internal update with the precondition that $\cur < \ha + N$ and $x$ is a string.
Emits $\snds(\cur, x)$, and
sets $\hp := \max(\hp, \cur)$, and, subsequently, $\cur := \cur + 1$.

\paragraph{\(\textsf{timeout}(\sender)\):} An internal update with the precondition that $\cur = \ha + N$.
Sets $\cur := \ha$ and emits nothing.

\medskip

We encode the sender's behavior using a transition relation $\trsender(\sender, \evt, \sender')$
which describes how a sender in state $\sender = (\ha, \hp, \cur)$
transitions to a state $\sender' = (\ha', \hp', \cur')$
after receiving as input, or outputting, the event $\evt$ (or neither if $\evt = \perp$).
There are three possible cases.
\begin{enumerate}
    \item $\evt = \dlvrr(a, \texttt{ACK})$ and the sender updates using \(\textsf{rcvAck}(\sender, a)\).
    \item $\evt = \snds(\cur, x)$ and is emitted by the sender as it updates using \(\textsf{advCur}(\sender, x)\).
    \item $\evt = \perp$, because the sender updates using \(\textsf{timeout}(\sender)\), which is an internal
    update with no output event.
\end{enumerate}
The transition relation captures all three.

\begin{equation}
\begin{aligned}
\trsender(\sender, \evt, \sender') & := 
    (\exists \, a \in \mathbb{N}_+ \,::\, 
    \evt = \dlvrr(a, \texttt{ACK}) \land (\sender', \perp) = \textsf{rcvAck}(\sender, \evt)) \\
    & \lor 
    (\exists \, x \in \text{Str} \,::\, 
        \evt = \snds(\cur, x) \land 
        \cur < \ha + N \land 
        (\sender', \evt) = \textsf{advCur}(\sender, x)) \\
    & \lor
    (\evt = \perp \land 
        \,\cur = \ha + N \land 
        (\sender', \evt) = \textsf{timeout}(\sender))
\end{aligned}
\Eql{trsender}
\end{equation}

Real GB$(N)$ implementations may differ on how they the prioritize these three functions.
For instance, if the sender's timer expires at the same time that it receives a new \ack,
should it first process the \ack, or process the timeout?
By defining the transition relation the way we do, we are able to capture all possible choices
for which functions to prioritize.

\begin{samepage}
\begin{theorem}
All the following are invariants of the sender's transition relation 
$\trsender$.
\begin{enumerate}[{Inv }1:]
    \item The sender's high ACK (\ha) only acknowledges packets it transmitted: \(\ha \leq \hp + 1\).
    \item The sender's next transmission (\cur) is always either within the current window, or in the first position of the next window: \(\ha \leq \cur \leq \ha + N\).
    \item The sender's high ACK (\ha) and highest transmission (\hp) are both non-decreasing
    with time, according to the sender's local clock.
    That is, if \(\trsender(\sender, \evt, \sender')\),
    then \(\ha \leq \ha'\) and
    \(\hp \leq \hp'\).
\end{enumerate}
\Thl{sender:gbn}
\end{theorem}
\end{samepage}

Note, if we adjusted our model to have bounded ids with wrap-around, we would need to modify the invariants in \Thr{sender:gbn} to take the id bound into account.  This could be done either by: 
(1) assuming a connection never lasts more than 34s, or 
(2) adding timestamps to datagrams, implementing PAWS, 
modifying \(\trsender\)
to transmit $\cur \text{ mod } 2^{32}$ and to infer the unbounded value of an \ack
id based on the ordering that PAWS infers, 
and assuming connections never last >24 days.
For a detailed formal treatment of bounded sequence numbers, 
    the reader is referred to~\cite{chkliaev2003verification}.

\section{Formal Model and Correctness of the Go-Back-$N$ Receiver}\Secl{gbn:model:receiver}

We model the GB$(N)$ receiver as having an unbounded internal set of naturals \rcvd.
Whenever it receives the event $\dlvrs(i, x)$, it checks if $i$ is a cumulative \ack
for \rcvd, that is, if $i=\min(\mathbb{N}_+\setminus\rcvd)$, in which case 
the receiver adds $i$ to \rcvd.  Else it does nothing.

Note, we allow $\rcvd$ to be unbounded in order to make our model more general in the sense
that it could be more easily adapted to describe a receiver who buffers out-of-order packets.
In that case, the unbounded nature of $\rcvd$ is still acceptable because it 
abstracts a bit-vector of size $N$, which is bounded.
However, we prove that the receiver we describe, which ignores out-of-order packets,
is equivalent to one where the $\rcvd$ set is replaced with a single integer $p$ which tracks the
next packet sequence number the receiver expects to receive.
Therefore, we are not concerned that the unbounded nature of the set is unrealistic; it is simply
a useful modeling abstraction.

The receiver has two state update functions.

\paragraph{\(\textsf{rcvPkt}(\rcvd, i)\):} An external update triggered by $\dlvrs(i, x)$.  
If $i = \min(\mathbb{N}_+\setminus\rcvd)$, sets $\rcvd := \rcvd \cup \{ i \}$, else does nothing.

\paragraph{\(\textsf{sndAck}(\rcvd)\):} An internal update with no precondition.
Outputs $\sndr(\min(\mathbb{N}_+\setminus\rcvd), \texttt{ACK})$ and leaves $\rcvd$ unchanged.

\medskip

By leaving the receiver's acknowledgment strategy nondeterministic, our model is able
to capture all possible \ack strategies, such as: \ack every packet; \ack every other 
packet; \ack every $N^{\text{th}}$ packet; send an \ack on a temporal schedule; etc.
We can reason about a particular \ack strategy by phrasing it as a predicate over the order of events,
for example, ``precisely $N$ $\dlvrs$ events must occur after each $\sndr$ and before the next''.
We encode the receiver's transition relation~$\trreceiver$ as follows.

\begin{equation}
\begin{aligned}
\trreceiver(\rcvd, \evt, \rcvd') & :=
    (\exists i \in \mathbb{N}_+ \,::\, \evt = \sndr(i, \texttt{ACK}) \land 
        (\rcvd', \evt) = \textsf{sndAck}(\rcvd)) \lor \\
    & (\exists x \in \text{Str}, i \in \mathbb{N}_+ \,::\, \evt = \dlvrs(i, x) \land 
     (\rcvd', \perp) = \textsf{rcvPkt}(\rcvd, \evt))
\end{aligned}
\end{equation}

Researchers interested in ARQ protocols where the receiver does buffer out-of-order packets
can modify our \textsf{rcvPkt} definition to set $\rcvd := \rcvd \cup \{ i \}$ regardless of whether
or not $i = \min(\mathbb{N}_+\setminus\rcvd)$.
We actually prove that the following invariant holds for both versions of the receiver.

\begin{theorem}
Suppose $\trreceiver(\rcvd, \evt, \rcvd')$.
Then $\rcvd \subseteq \rcvd'$.
\Thl{rcvr:invs}
\end{theorem}

\section{Formal Model and Correctness of the Token Bucket Filter}\Secl{gbn:model:tbf}

Recall that our model has two TBFs: $\mathcal{F}_s$ (which connects the sender to the receiver)
and $\mathcal{F}_r$ (which connects the receiver to the sender).
Since both TBFs work the same way, in this section, we describe the single
TBF definition we use in both places.

At a high level, the TBF works as follows.
The TBF has an internal list of datagrams, called~$\data$, which has a fixed byte-capacity $\dcap$.
When the sending endpoint sends a datagram to the TBF, if $\data$ is full (i.e., if the cumulative
size in bytes of the payloads of the datagrams in $\data$ equals $\dcap$) then the datagram is dropped.
Otherwise, it is inserted into the first position in the list.
Note, this means the sending endpoint can never successfully transmit any datagram whose payload
is longer than $\min(\dcap, \bcap)$ -- so there is an effective cap on the size of payloads.
Naturally this means that, in order to avoid unecessary losses, the maximum payload size should be $\leq \min(\dcap, \bcap)$.
Since (for convenience) we assume the payload of every \ack is \texttt{ACK}, 
we therefore assume $\min(\dcap_r, \bcap_r) \geq 3 = \emph{sz}(\texttt{ACK})$.

In addition, the TBF has a counter $\bucket$, called a \emph{bucket}, which is initially zero, 
and capped above by~$\bcap$.
The TBF has an internal clock which ticks, and with each tick, the 
$\bucket$ increases by a fixed rate $\rat$, up to $\bcap$. 
The bucket is commonly described as holding ``tokens'', for example, if $\bucket=4$ then we say
the TBF has 4 tokens in its bucket,
and for the TBF to forward a datagram from its list~$\data$ to the receiving endpoint, it must
remove a number of tokens equal to the size of the datagram's payload in bytes from the bucket. 
Finally, we assume the TBF is configured with some maximum delay value $\ttl$, which is either a positive integer or infinity, and any datagram which persists in $\data$ for that many ticks is dropped. 
The way this is implemented in the model is by tracking for each datagram in the TBF 
how many ticks the datagram has survived, and after each tick,
dropping any datagram which has reached its expiration.

Now we formalize that high-level description.  Our formal model
also includes token decay,
reordering, and datagram loss.
Token decay is modeled as a nondeterministic function which, when executed,
decrements the number of tokens in the bucket.
The purpose of token decay is to capture \emph{wastage}.
Reordering is captured implicitly, by allowing the TBF to forward any datagram,
not just the oldest one. 
That is to say, whereas a real TBF would \texttt{pop()} a datagram and then forward it,
ours chooses some value of $i$ less than the length of $\data$,
removes the $i^{\text{th}}$ datagram, and delivers that one.
Nondeterministic loss is modeled the same way as forwarding, except that any datagram
in $\data$ can be nondeterministically lost at any time.

Consider a TBF $\mathcal{F}$ which connects endpoint $a$ to endpoint $b$,
configured with positive integer caps $\bcap$ and $\dcap$, 
positive integer bucket rate $\rat \leq \bcap$, and ordinal maximum delay $\ttl$. 
Note, an ordinal is a type that includes
the naturals 0, 1, 2, etc.,
    the value $\omega$ which is greater than all naturals,
    as well as the addition of any pair of ordinals.
    In other words, the ordinals include both the natural numbers
    and (infinitely many, increasing flavors of) infinity~\cite{manolios2004integrating}.
    We use natural \(\ttl\) values to model bounded maximum delay,
    and infinite ones to model unbounded maximum delay 
    (where a datagram could theoretically stay in \(\data\) forever).
We use Ord to denote the set of all ordinals.
Next, we formalize the state update functions of $\mathcal{F}$,
which we then use to build its transition relation.

We model the state of $\mathcal{F}$ using the tuple 
\(\tbf = (\bucket, \data)\)
where 
\(\bucket\) is a mutable natural (initially set to zero), and
\(\data\) is a mutable list of tuples $(t, \dg)$, where 
\(\dg\) is a datagram and
the ordinal $t$ 
is the number of clock ticks that $\dg$ 
can remain in $\data$ before it must be dropped (initially, $\ttl$).
We use $\emph{len}(\data)$ to denote the number of datagrams in $\data$, 
e.g., $\emph{len}([(4,(2,\texttt{MANGO}))])=1$,
and $\emph{sz}(\data)$ to denote the cumulative size of the payloads of its entries,
e.g., $\emph{sz}([(\omega, (1, \texttt{EAT})), (\omega, (2, \texttt{BANANA}))]) = \emph{sz}(1, \texttt{EAT}) + \emph{sz}(2, \texttt{BANANA}) = 3 + 6 = 9$. 
We summarize all the variables and parameters of $\mathcal{F}$ in \Tabr{tbf-vars-params}.

\begin{table}[h]
\centering
\begin{tabular}{lp{3cm}p{2.5cm}p{2cm}p{6cm}}
Name   & Type & (V)ariable or (P)arameter & Initial Value & Description \\\hline
$\bucket$ & $\mathbb{N}$ & V & 0 & Number of tokens \\
$\bcap$ & $\mathbb{N}_+$ & P & N/A & Maximum value of $\bucket$ \\  
$\rat$ & $\mathbb{N}_+$ & P & N/A & Rate at which the bucket refills, up to $\bcap$ \\   
$\ttl$ & Ord & P & $\omega$ & Maximum datagram delay \\
$\data$ & List of (Ord, Dg) & V &  $[]$ & Datagrams to be forwarded \\
$\dcap$ & $\mathbb{N}_+$ & P & N/A & Maximum value of $\emph{sz}(\data)$ \\    
\end{tabular}
\caption{Variables and parameters of the TBF.}
\Tabl{tbf-vars-params}
\end{table}

The TBF updates its state using the following state update functions. 
The first function, \textsf{tick}, encodes a cycle of the TBF's internal clock, which increases its bucket until the bucket reaches its cap.

\paragraph{$\textsf{tick}(\tbf)$:} Internal update which 
decrements the remaining delay value for each datagram in $\data$,
removing any which has persisted for $\ttl$ \textsf{tick}s,
and sets $\bucket := \min(\bucket + \rat, \bcap)$. Outputs nothing.

\medskip
The second function, \textsf{decay}, captures wastage behaviors where a TBF loses tokens.
Such behaviors are included in some, but not all, TBF definitions in the literature.

\paragraph{$\textsf{decay}(\tbf)$:} Internal update which
sets $\bucket := \max(\bucket - 1, 0)$ and outputs nothing.

\medskip
The third function, \textsf{process}, captures the event where the sending endpoint sends a datagram into the TBF (which may be lost because the TBF does not have enough space for the datagram, or enqueued in $\data$).

\paragraph{$\textsf{process}(\tbf, \evt)$:} External update triggered by $\evt = \snd_a(\dg)$.  
If $\emph{sz}(\data) + \emph{sz}(\dg) \leq \dcap$, pushes $\dg$ into $\data$.
Otherwise the function does nothing, meaning, the datagram is dropped. 

\medskip
The fourth function, \textsf{drop}, describes nondeterministic loss.

\paragraph{$\textsf{drop}(\tbf, i)$:} Internal update with the precondition $i < \emph{len}(\data)$. Removes the $(i+1)^{\text{th}}$ element of $\data$ and outputs nothing.

\medskip
Finally, the fifth function, \textsf{forward}, captures the event where the TBF forwards a datagram to the receiving endpoint.

\paragraph{$\textsf{forward}(\tbf, i)$:} Internal update with the precondition $i < \emph{len}(\data)$ and $\emph{sz}(\data[i]) \leq \bucket$.  Sets $\bucket := \bucket - \emph{sz}(\data[i])$, removes the $(i+1)^{\text{th}}$ element of $\data$ from $\data$, and outputs $\dlv_b(\data[i])$.

\medskip

Using the functions outlined above, a single TBF $\mathcal{F}$ may progress through a long series
of consecutive states $\tbf_0, \tbf_1, \tbf_2, \ldots, $ etc.
It does so according to the transition relation $\trtbf$, defined in \Eqr{trtbf}.
We define $\trtbf$ piece-wise through three sub-relations.
The first, \textsf{tbfIntR}, describes the internal events of the TBF, namely, 
\textsf{tick}, \textsf{decay}, and \textsf{drop}.

\begin{equation}
\begin{aligned}
\textsf{tbfIntR}(\tbf, \evt, \tbf') & := \evt = \perp \land \big(
    (\tbf', \evt) \in \{ \textsf{tick}(\tbf), \textsf{decay}(\tbf) \} \\
    & \lor 
    \exists i \in \mathbb{N} \,::\, i < \emph{len}(\data) \land
        (\tbf', \evt) = \textsf{drop}(\tbf, i)
    \big)
\end{aligned}
\end{equation}
Next, we define \textsf{tbfProcR}, the sub-relation which captures how the TBF responds when the 
sending endpoint (endpoint a) transmits a datagram into it.
\begin{equation}
\textsf{tbfProcR}(\tbf, \evt, \tbf') := \exists \dg \in \text{Dg} \,::\, \evt = \snd_a(\dg) \land \tbf' = \textsf{process}(\tbf, \evt)
\end{equation}
Finally, we define \textsf{tbfFwdR}, which captures the step where the TBF forwards a datagram from $\data$ into the receiving endpoint~(b).
\begin{equation}
\begin{aligned}
\textsf{tbfFwdR}(\tbf, \evt, \tbf') := \exists i \in \mathbb{N} \,::\, &
i < \emph{len}(\data) \land \emph{sz}(\data[i]) \leq \bucket \\ & \land
\evt = \dlv_b(\data[i]) \land \tbf' = \textsf{forward}(\tbf, i)
\end{aligned}
\end{equation}
Taking the disjunction of these three relations yields the transition relation for the TBF.
\begin{equation}
\trtbf(\tbf, \evt, \tbf') := 
    \textsf{tbfIntR}(\tbf, \evt, \tbf') \lor 
    \textsf{tbfProcR}(\tbf, \evt, \tbf') \lor 
    \textsf{tbfFwdR}(\tbf, \evt, \tbf')
\Eql{trtbf}
\end{equation}

The most obvious property of the TBF, which we verify, is that the cumulative size of the payloads it
forwards between ticks is bounded by the number of tokens in its bucket.
\begin{theorem}
Suppose the TBF forwards datagrams $\dg_1, \dg_2, \ldots, \dg_j$ between two \textsf{tick}s.
Let $\bucket$ be the number of tokens the TBF has after the first tick 
and before it begins forwarding datagrams.
Then $\emph{sz}(\dg_1) + \emph{sz}(\dg_2) + \ldots + \emph{sz}(\dg_j) \leq \bucket$.
\end{theorem}

Next, we identify an important property of the TBF, namely, that it is compositional.
The purpose of this property is to show that we can reason about connections over multiple sequential
TBF links by reasoning about just a single TBF.
In order to formalize the property, we first need to introduce some useful definitions.

First, we define the serial composition of two TBFs as follows.
\begin{definition}[Serial TBF Composition]
Let $\mathcal{F}_i$ be the TBF with input event $\snd_i$ and output event $\dlv_{i+1}$
for each $i = 1, 2$, and let $\trtbf_i$ denote the transition relation of $\mathcal{F}_i$.
Then the \emph{serial composition of $\mathcal{F}_1$ with $\mathcal{F}_2$}, denoted $\mathcal{F}_1 \triangleright \mathcal{F}_2$, is the system in which the output event 
$\dlv_{i+1}$ of $\mathcal{F}_1$ 
is considered equal to the input event $\snd_{i+1}$ of $\mathcal{F}_2$.
That is to say, when $\mathcal{F}_1$ forwards a datagram to $\mathcal{F}_2$, the datagram gets
processed immediately by $\mathcal{F}_2$.
\end{definition}

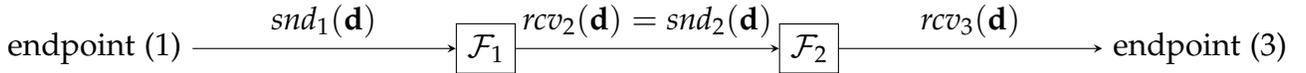
\begin{figure}[H]
    \centering
    \begin{tikzpicture}
\node[] (one) {endpoint (1)};
\node[draw, rectangle, right=of one] (dgs1) {$\mathcal{F}_1$};
\draw[->] (one) to node[above] {$\snd_1(\dg)$} (dgs1);
\node[draw, rectangle, right=of dgs1] (dgs2) {$\mathcal{F}_2$};
\draw[->] (dgs1) to node[above] {$\dlv_2(\dg)=\snd_2(\dg)$} (dgs2);
\node[right=of dgs2] (three) {endpoint (3)};
\draw[->] (dgs2) to node[above] {$\dlv_3(\dg)$} (three);
\end{tikzpicture}
    \caption{The serial composition of TBF $\mathcal{F}_1$ with TBF $\mathcal{F}_2$, denoted $\mathcal{F}_1 \triangleright \mathcal{F}_2$.}
    \Figl{gbn:system}
\end{figure}

Next, we introduce a composition operator $\oplus$ for TBFs and their states.
The idea here is, given two TBFs, to generate a third which can simulate their serial composition.
\begin{definition}[Abstract TBF Composition]
Let $\tbf_i$ be TBF states for $i = 1, 2$, of the TBFs $\mathcal{F}_i$.
Then the \emph{abstract composition} of $\tbf_1$ and $\tbf_2$, 
denoted $\tbf_1 \oplus \tbf_2$, is 
the state $(\bucket_2, \data_1';\data_2)$ where 
$\data_1' = [(t + \ttl_2, \dg) \text{ for } (t, \dg) \text{ in } \data_1]$.
The \emph{abstract composition of $\mathcal{F}_1$ and $\mathcal{F}_2$}, 
denoted $\mathcal{F}_1 \oplus \mathcal{F}_2$,
is the TBF with parameters $\bcap_2, \dcap_1 + \dcap_2, \rat_2$, and $\ttl_1 + \ttl_2$;
and $\tbf_1 \oplus \tbf_2$ is a state of $\mathcal{F}_1 \oplus \mathcal{F}_2$.
\end{definition}

Intuitively, the abstract TBF composition is meant to produce a single TBF which can simulate
the serial composition of two TBFs.
This intuition drives our choices of parameters and variables.
First, it is important to note that we do not assume the two TBFs are synchronized, i.e., 
we do not assume they \textsf{tick} at the same time.
For this reason, the number of tokens in the second TBF is the limiting factor for whether
or not a datagram can be forwarded, and so, we set the bucket in $\tbf_1 \oplus \tbf_2$
to $\bucket_2$, its cap to $\bcap_2$, and its refill rate to $\rat_2$.
Next, the single TBF must contain every datagram from the two individual TBFs, but critically,
we need to simulate the fact that the datagrams in the first TBF might go through some sequence of 
\textsf{tick}s before arriving at the second.  This is why we set $\data$ in $\tbf_1 \oplus \tbf_2$
to $\data_1';\data_2$.
It naturally follows that the cap on $\data$ should be $\dcap_1 + \dcap_2$ and the 
maximum delay should be $\ttl_1 + \ttl_2$.

Finally, we need a notion of reachability (which applies to not just TBFs but also any other system component).
\begin{definition}[Reachability]
Let $C$ be a component with transition relation 
$\textsf{trancR}$.
Let $E = \evt_0, \evt_1, \ldots, \evt_k$ 
be a finite sequence such that each $\evt_i$ is either null $(\perp)$
or an input or output event of $C$.
Let $c_0, c_1, \ldots, c_k$ be a sequence of states of $C$ 
such that $\bigwedge_{0 \leq i < k} \textsf{trancR}(c_i, \evt_i, c_{i+1})$.
Then we say $C$ can \emph{reach} $c_k$ from $c_0$ by following the
event sequence $E$.
\end{definition}

Naturally, we can extend this vocabulary to also reason about
composite systems, for example, the sequential composition of two
TBFs.  That is, suppose $E$ is a sequence of input or output events of $\mathcal{F}_1$ or $\mathcal{F}_2$,
and let 
\[(\tbf_0^1, \tbf_0^2), (\tbf_1^1, \tbf_1^2), \ldots, (\tbf_k^1, \tbf_k^2)\]
be a sequence of 
states of $\mathcal{F}_1 \triangleright \mathcal{F}_2$.
For each $i = 1, 2$, let $E|_{i}$ denote the projection of $E$ onto the input and output events of 
$\mathcal{F}_i$.
Suppose that $\mathcal{F}_1$ can reach $\tbf_k^1$ from $\tbf_0^1$ by following $E|_1$,
$\mathcal{F}_2$ can reach $\tbf_k^2$ from $\tbf_0^2$ by following $E|_2$, 
and for all $0 \leq i < k$,
if $E[i] = \dlv_{i+1}(\dg)$, then $E[i+1] = \snd_{i+1}(\dg)$.
Then we say $\mathcal{F}_1 \triangleright \mathcal{F}_2$ can reach $(\tbf_k^1, \tbf_k^2)$ from $(\tbf_0^1, \tbf_0^2)$ by following $E$.

Lastly, we need a notion of equivalence of TBFs and their states.  
This equivalence notion is what we will use to argue that the abstract composition of two TBFs can simulate
the serial composition thereof.
\begin{definition}[TBF Equivalence]
We will say the state $\tbf$ of $\mathcal{F}$ 
is equivalent to the state $\tbf'$ of $\mathcal{F}'$,
and write $\tbf \approx \tbf'$, if $\bcap_1 = \bcap_2, \dcap_1 = \dcap_2, \rat_1 = \rat_2, \ttl_1 = \ttl_2,$
and the ids in $\data_1$ form a permutation of the ids in $\data_2$.
\end{definition}
The intuition behind our equivalence notion is to capture the closest thing to strict equality possible.
The only reason we do not use equality is because the datagrams might get permuted depending on the order
in which they are forwarded from $\data_1$ to $\data_2$.
There is one other subtlety to note here, which is that we ignore the payloads in the permutation condition.
This is because, in an association, we assume an endpoint never sends two 
datagrams with the same id but different payloads.
However, this would no longer hold if we were going to model bounded ids with wrap-around, as discussed previously, in which case we would need to require that (all of) $\data_1$ is a permutation of (all of) $\data_2$.

With these definitions in mind, we can now formalize our property.

\begin{theorem}
Let $\mathcal{F}_i$ be TBFs for $i = 1, 2$ and $E$ a sequence of events, each of which is either
null, or an input or output event of at least one of the two TBFs.
Suppose $\mathcal{F}_1 \triangleright \mathcal{F}_2$ can reach $(\tbf_1', \tbf_2')$
from $(\tbf_1, \tbf_2)$ by following $E$.
Then the TBF $\mathcal{F}$ with initial state $\tbf_1 \oplus \tbf_2$ can reach a state 
$\tbf_3$ by following $E$, such that, $\tbf_3 \approx \tbf_1' \oplus \tbf_2'$.
\end{theorem}

\begin{sketch}
We break the problem down into cases, following the transition relation of $\mathcal{F}_1 \triangleright \mathcal{F}_2$.
\begin{description}
    \item $\textsf{tick}(\mathcal{F}_1)$: Equivalent to a noop in $\mathcal{F}_1 \oplus \mathcal{F}_2$, provided no datagrams expire in $\mathcal{F}_1$.
    If something does age out, we can simulate its erasure using a \textsf{drop}.

    \item $\textsf{tick}(\mathcal{F}_2)$: Equivalent to a \textsf{tick} in 
    $\mathcal{F}_1 \oplus \mathcal{F}_2$.

    \item $\textsf{decay}(\mathcal{F}_1)$: Equivalent to a noop in $\mathcal{F}_1 \oplus \mathcal{F}_2$.
    \item $\textsf{decay}(\mathcal{F}_2)$: Equivalent to $\textsf{decay}(\mathcal{F}_1 \oplus \mathcal{F}_2)$.
    \item $\textsf{process}(\mathcal{F}_1, \dlv_1(\dg))$: Equivalent to $\textsf{process}(\mathcal{F}_1 \oplus \mathcal{F}_2, \dlv_{1 \oplus 2}(\dg))$.
    \item $\textsf{process}(\mathcal{F}_2, \dlv_2(\dg))$: Can only happen in conjunction with $\textsf{forward}(\mathcal{F}_1, i)$ where $\data_1[i] = \dg$ (which emits $\dlv_2(\dg)$).  Equivalent to a noop in $\mathcal{F}_1 \oplus \mathcal{F}_2$.
    \item $\textsf{drop}(\mathcal{F}_1, i)$ or $\textsf{drop}(\mathcal{F}_2, i)$: Equivalent to $\textsf{drop}(\mathcal{F}_1 \oplus \mathcal{F}_2, j)$, for some value of $j$.
    \item $\textsf{forward}(\mathcal{F}_1, i)$: Can only happen in conjunction with $\textsf{process}(\mathcal{F}_2, \dlv_2(\dg))$ where $\dg = \data_1[i]$; case covered above.
    \item $\textsf{forward}(\mathcal{F}_2, i)$: Equivalent to $\textsf{forward}(\mathcal{F}_1 \oplus \mathcal{F}_2, j)$ for some value of $j$.
\end{description}
\end{sketch}

A limitation of this result is that we use \textsf{drop}, i.e., nondeterministic loss, to emulate the case where a datagram expires in the first TBF before it can be forwarded to the second.
Although nondeterministic loss is assumed in some models,
such as the model we used in \Chapr{karn-rto}, it may be too expressive in others.
For instance, later, in \Thr{gbn:worst-case}, we examine exclusively losses caused by a sender who transmits
datagrams into a TBF faster than the TBF can deliver them.  In this case, we do not want to include nondeterministic losses, since our goal is to measure just the losses caused by over-transmission.
If we removed nondeterministic loss from our model, we could still prove composition, 
but we would need to assume the two TBFs synchronize in the sense that they \textsf{tick} and \textsf{decay}
at the same time, 
and also, that 
datagrams are not lost due to throttling at the interface between the first and second TBF, 
i.e., \(\rat_1 \leq \rat_2 \land \bcap_1 \leq \bcap_2\).
It is not yet known whether the result can be proven with weaker assumptions.
This problem was first identified by Arun et. al.~\cite{ccac}, in the context of their ``path model'', a.k.a, CCAC.

\section{Formal Definition of the Composite Transition Relation of Go-Back-$N$}\Secl{gbn:model:tran}

Having defined the sender, receiver, and TBF, and their respective transition relations,
we now define the composite transition relation~$\trsys$ for the entire system, and then briefly
discuss its semantics.  We use $\trtbf_s$ to denote the transition relation of $\mathcal{F}_s$
and $\trtbf_r$ to denote the transition relation of $\mathcal{F}_r$.
We encode the state of the entire system using the tuple $\sys = (\sender, \tbf_s, \tbf_r, \rcvd)$,
and as before, we take the convention $\sys' = (\sender', \tbf_s', \tbf_r', \rcvd')$.
Note that $\evt$ could be any event in the model, or $\perp$, 
and we use the convention $\max(\emptyset)=0$.
We build the transition relation piece-by-piece.
Our transition relation explicitly encodes the intuition that two components synchronize on an event which is an input to one and an output of another, but internal events occur asynchronously.

First we define the steps where the sender transmits 
the next packet in its window (with id=$\cur$), or the receiver transmits a cumulative \ack.
The intuition here is that the transmitting component takes a step on its output event,
and the TBF it transmits to reacts synchronously, but the rest of the system stays still.
\begin{equation}
\begin{aligned}
\textsf{senderSnd}(\sys, \evt, \sys') := & \, \exists x \in \text{Str} \,::\, \evt = \snds(\cur, x)
    \land \trsender(\sender, \evt, \sender') \\
    & \land \trtbf_s(\tbf_s, \evt, \tbf_s')
    \land \rcvd = \rcvd' 
    \land \tbf_r = \tbf_r' \\
\textsf{receiverSnd}(\sys, \evt, \sys') := & \, \evt = \sndr(\min(\mathbb{N}_+ \setminus \rcvd), \texttt{ACK})
    \land \sender = \sender'
    \land \tbf_s = \tbf_s'\\
    & \land \trreceiver(\rcvd, \evt, \rcvd')
    \land \trtbf_r(\tbf_r, \evt, \tbf_r') \\
\end{aligned}
\end{equation}
Notice how so long as the receiver ignores out-of-order packets, $\min(\mathbb{N}_+ \setminus \rcvd) = \max(\rcvd) + 1$, under the convention that $\max(\emptyset) = 0$.
Next, we define the step where the sender performs an internal update.  The only internal update of the sender is the timeout, so, this is when the sender ``goes back~$N$''.
\begin{equation}
\begin{aligned}
\textsf{senderInt}(\sys, \evt, \sys') := & \, \evt = \perp \land\, \trsender(\sender, \evt, \sender') \land \tbf_s = \tbf_s' \land \rcvd = \rcvd' \land \tbf_r = \tbf_r' \\
\end{aligned}
\end{equation}
Likewise, we define the internal steps for the two TBFs (where they \textsf{tick}, \textsf{decay}, or \textsf{drop}).
\begin{equation}
\begin{aligned}
\textsf{tbfSint}(\sys, \evt, \sys') := & \, \evt = \perp \land\, \sender = \sender' \land \tbf_s = \tbf_s' \land \rcvd = \rcvd' \land \trtbf_r(\tbf_r, \evt, \tbf_r') \\
\textsf{tbfRint}(\sys, \evt, \sys') := & \, \evt = \perp \land\, \sender = \sender' \land \trtbf_s(\tbf_s, \evt, \tbf_s') \land \rcvd = \rcvd' \land \tbf_r = \tbf_r' \\
\end{aligned}
\end{equation}
Finally, we define the steps where the sender receives an \ack or the receiver receives a packet.
In these, the receiving component and the forwarding TBF both transition, while everything else stays still.
\begin{equation}
\begin{aligned}
\textsf{senderRcv}(\sys, \evt, \sys') := & \, \exists i \in \mathbb{N}_+ \,::\, 
    \evt = \dlvrr(i, \texttt{ACK}) 
    \land \trsender(\sender, \evt, \sender') \\
    & \land \tbf_s = \tbf_s'
    \land \rcvd = \rcvd' 
    \land \trtbf_r(\tbf_r, \evt, \tbf_r') \\
\textsf{receiverRcv}(\sys, \evt, \sys') := & \, \exists i \in \mathbb{N}_+, x \in \text{Str} \,::\, \evt = \dlvrs(i, x)
    \land \sender = \sender' \\
    & \land \trtbf_s(\tbf_s, \evt, \tbf_s')
    \land \trreceiver(\rcvd, \evt, \rcvd') 
    \land  \tbf_r = \tbf_r'\\
\end{aligned}
\end{equation}

Combining all these steps, we get the entire transition relation for the composite system.
\begin{equation}
\begin{aligned}
\trsys(\sys, \evt, \sys') := &
    \, \textsf{senderSnd}(\sys, \evt, \sys')
    \lor 
    \textsf{receiverSnd}(\sys, \evt, \sys') \\
    & \lor
    \textsf{senderInt}(\sys, \evt, \sys')
    \lor
    \textsf{tbfSint}(\sys, \evt, \sys')
    \lor
    \textsf{tbfRint}(\sys, \evt, \sys') \\
    & \lor 
    \textsf{senderRcv}(\sys, \evt, \sys')
    \lor 
    \textsf{receiverRcv}(\sys, \evt, \sys') 
\end{aligned}
\Eql{gbn:composite:tranr}
\end{equation}

The transition system in \Eqr{gbn:composite:tranr} relates a system state $\sys$
to the resulting $\sys'$ after a single component has taken an internal step ($\evt = \perp$) or two components have synchronized on an event (e.g., $\evt = \sndr(i, x)$).
Of course, in a real GB$(N)$ system such events may occur concurrently, e.g., 
if the sender transmits one packet ($\sndr(i, x)$) at the same time that the receiver receives another ($\dlvrs(j, y)$).
Although we do not explicitly model concurrency, 
the semantics of concurrency for the
system we define are clear.
As previously described in \Secr{gbn:model:setup}:
each component (sender, receiver, $\mathcal{F}_s$, and $\mathcal{F}_r$) can execute
at most one state update function at a time; and
two or more components can update at once provided that, if one of the updates outputs an event $e$,
which is an input to another component, the latter must execute its corresponding (external) update
at the same time.

The way we would capture this in our model is with a skipping refinement~\cite{jain2015skipping}.
Essentially, the refinement would map the abstract sequence from the prior example
\begin{equation}
\textsf{senderSnd}(\sys_1, \snds(i, x), \sys_2) \land
\textsf{receiverRcv}(\sys_2, \dlvrs(j, y), \sys_3)
\end{equation}
to the concrete transition $(\sys_1, \{\sndr(i, x),\dlvrs(j, y)\}, \sys_3)$, ``skipping'' 
the intermediate state $\sys_2$.  Note that a necessary but insufficient condition
for these events to be potentially concurrent is that they commute, that is, that the following holds
for some $\sys_2'$.
\begin{equation}
\textsf{receiverRcv}(\sys_1, \dlvrs(j, y), \sys_2')
\land
\textsf{senderSnd}(\sys_2', \snds(i, x), \sys_3)
\end{equation}
However, the properties we prove in this chapter do not relate to the nuances of concurrency, 
so, we leave this refinement to future work.

\section{Formal Efficiency Analysis of Go-Back-$N$}\Secl{gbn:efficiency}

Next, we formally analyze the performance of GB$(N)$.
In this case, what we mean by performance is the efficiency of the system,
that is, the fraction of packets received by the receiver which are
considered useful.
In the context of GB$(N)$, a packet is considered useful if it is (a) not a duplicate
and (b) cumulatively acknowledged.
Thus, in the long run, the efficiency of the system is precisely
\(\max(\rcvd)\)
divided by the number of packets received by the receiver.\footnote{Where, as before, $\rcvd$ is the set of packets delivered to the receiver.  Note, if the receiver is redefined to also buffer out-of-order packets, as discussed in \Secr{gbn:model:receiver}, then the efficiency is \(\min(\mathbb{N} \setminus \rcvd)\) divided by the number of packets received.}
Under the simplifying assumption that every packet is the same size, we prove two results.
First, it is possible for GB$(N)$ to achieve perfect efficiency.
And second, we compute the efficiency of GB$(N)$ in the absence of nondeterministic losses,
token decay, or reordering, under the assumption that the sender transmits at a constant
rate which exceeds the rate at which the sender's TBF $\mathcal{F}_s$ can deliver (leading to losses).
We argue that this second scenario is realistic and explain how it can be avoided by 
carefully configuring the sender relative to the parameters of the TBF.

\begin{theorem}
GB$(N)$ can achieve perfect efficiency of one.
\end{theorem}
\begin{sketch}
Suppose $\ttl_s, \ttl_r > 1$, 
$(x_i)_{i=1}^{N}$ is a sequence of strings,
such that for all $1 \leq i \leq N$, $\text{sz}(x_i) = 1$,
and let $E$ be the following event sequence.
\[
\begin{aligned}
E = \, & \snds(1, x_1), \perp, \dlvrs(1, x_1), \\
    & \snds(2, x_2), \perp, \dlvrs(2, x_2), \\
    & \ldots, \\
    & \snds(N, x_N), \perp, \dlvrs(N, x_N), \\
    & \sndr(N+1, \texttt{ACK}), \perp, \dlvrr(N+1, \texttt{ACK})
\end{aligned}
\]
Let $\sys_0$ be the initial state where 
$\sender = (1,1,1)$, 
$\rcvd = []$, and
$\tbf_a = (0,[])$ for each $a \in \{ s, r \}$.
Let $\sys_N$ be the state which is identical to $\sys_0$ except that
$\sender = (N+1,N,N+1)$, and
$\rcvd = [1,N]$.
Then $\sys_N$ is reachable from $\sys_0$ by following the event sequence $E$.
Moreover, the efficiency of the system between $\sys_0$ and $\sys_N$ is 1,
since $N$ packets were received by the receiver,
and in the end, $\max(\rcvd) = N$.
The general case follows by an induction on this argument.
\end{sketch}

In the real world, datagrams can be lost for a number of reasons.
Datagrams on wireless networks get corrupted due to radio interference (collision)
or weak signals, and are thus automatically dropped~\cite{rayanchu2008diagnosing}.
Another possibility is buggy code, e.g.,
Hoque~et.~al.~\cite{hoque2013adversarial} found a bug in 
    AODV~\cite{ahoddvr} where the first packet in each window transmitted along a previously
    untraveled route was lost by the router due to a mis-ordering of notification events.
More exotically, a compromised router could deliberately drop packets 
    in a targeted fashion to stealthily sabotage communication between some victim 
    computers~\cite{mzrak2008detecting,rayanchu2008diagnosing}.
But in traditional wired networks, 
according to measurement studies, 
the most common kind of loss can be attributed to the
queuing mechanism on the router (e.g. the TBF), 
which drops datagrams as part of its effort to rate-limit~\cite{tang1999network,borella1998internet,borella2000measurement}.\footnote{
In the context of the TBF, the resulting losses are typically geometrically distributed~\cite{yajnik1999measurement}.  Consequently, a geometric loss pattern is assumed in some works that study GB$(N)$ probabilistically, e.g.,~\cite{hasan2008performance}.}
This motivates us to analyze the scenario in which the sender-to-receiver 
TBF ($\mathcal{F}_s$) is overwhelmed with packets, leading to deterministic losses.
In order to understand just the impact of over-transmission on performance,
we assume everything else about the system is ideal, i.e.,
packets are never reordered, \acks are received immediately after being sent, 
$\mathcal{F}_s$ has unbounded delay, etc.

For the over-transmission scenario, 
suppose the sender transmits at some positive integer 
rate $R$, such that
$\rat_s < R < \dcap_s < N$.
Intuitively, this means the sender transmits $R$ packets for every one \textsf{tick}
of $\mathcal{F}_s$ (the sender-to-receiver TBF).
Since the bucket of $\mathcal{F}_s$ refills slower than the sender sends,
we get over-transmission, where the sender is sending into a full TBF and the extra
packets are deterministically lost.
We assume $\rat_s = \bucket_s = \bcap_s$ and $\ttl_s = \omega$,
meaning the bucket refills as quickly as possible and packets do not expire.
We further assume that
while $\cur < \ha + N$, the system progresses through the following pattern:
the sender transmits $R$ packets, all of equal (constant) size, then
$\mathcal{F}_s$ ticks and forwards $\rat_s$ packets to the receiver, then the cycle repeats.
Clearly $\data_s$ fills at a net rate $R - \rat_s$, until it reaches $\dcap_s$.
However, the story is more complicated once the channel fills.
Let $w = (\dcap_s - R)/(R - \rat_s)$,
so, after $w - 1$ bursts of $R$ packets each, $\dcap_s - \emph{sz}(\data_s) = R$.
Then in the next step, $\data_s$ becomes full,
i.e., $\emph{sz}(\data_s) = \dcap_s$, and then $\rat_s$ packets are delivered.
And in the step after that, the first $\rat_s$ packets enter $\data_s$ before
losses begin to occur, after which any subsequent packets that enter $\data_s$ are
out-of-order and therefore ignored by the receiver upon being received.
We assume that before the sender times out, $\mathcal{F}_s$ is able to deliver every packet in $\data_s$.

From this analysis we can draw two conclusions.
First, over-transmission will occur if 
$(w + 1) R + \rat_s < N$, where $w$ is defined as in the previous paragraph.
And second, if over-transmission occurs, the number of packets delivered to the 
receiver will be 
\[R w + R + \rat_s = R(\dcap_s - R)/(R - \rat_s) + R + \rat_s\]
which means the efficiency of the entire system is:
\[(R(\dcap_s - R)/(R - \rat_s) + R + \rat_s)/N\]
Note, strictly speaking we formally verify the theorem with $\emph{sz}(\dg)=1$ for all packets,
but the result clearly scales for any constant datagram size less than the maximum,
by just multiplying $\dcap_s, \rat_s, \bucket_s,$ and $\bcap_s$ by the constant packet size.

\begin{theorem}
Suppose \(\rat_s < R < \dcap_s < N\) such that $R - \rat_s$ divides $\dcap_s - R$.
Further suppose
        \(\ttl_s \notin \mathbb{N}\), 
        \(\ttl_r > 1\), and 
        \(R(\dcap_s - R)/(R - \rat_s) + R + \rat_s < N\).
Let $\sys_0$ be the initial state as before and suppose $E$ is an event sequence of length $k$ such that,
 when the system follows $E$ from $\sys_0$ to $\sys_k$, it does so 
according to the following pattern, repeated an arbitrary number of times.
\begin{enumerate}
    \item The sender transmits $R$ one-byte packets.
    Then $\mathcal{F}_s$ ticks, refilling its bucket, and forwards $\bucket_s$ packets to the receiver, FIFO.
    This repeats until the sender has transmitted its entire window.
    \item $\mathcal{F}_s$ ticks and forwards $\bucket_s$ packets to the receiver, FIFO.  
    This repeats until $\emph{sz}(\data_s) = 0$.
    \item If the receiver has received $N$ packets (FIFO or otherwise, including duplicates)
    since it last transmitted an \ack, it transmits an \ack, $\mathcal{F}_r$ ticks, then $\mathcal{F}_r$ forwards the \ack to the sender.  Otherwise, the sender has a timeout and ``goes back~$N$'', and the process repeats from (1).
\end{enumerate}
The efficiency of the system between $\sys_0$ and $\sys_k$ is
\((R(\dcap_s - R)/(R - \rat_s) + R + \rat_s)/N\).
\Thl{gbn:worst-case}
\end{theorem}

\begin{sketch}
Suppose $\sys$ and $R$ are as described in \Thl{gbn:worst-case}.
First, we prove that after each repetition of step (1), $\text{len}(\data_s)$
increases by $R - \rat_s$ packets.
We thus derive that $\dcap_s / (R - \rat_s)$ repetitions of step (1)
suffice to fill $\data_s$ to $R$ less than its capacity, after which, 
the next burst brings $\emph{sz}(\data_s)$ to $\dcap_s - \rat_s$.
In the next burst, the last $R - \rat_s$ transmissions are lost,
meaning all subsequent packet transmissions before the timeout are out-of-order
and therefore, even if they are received by the receiver, the receiver ignores them.
It follows that the total number of delivered packets before the timeout is 
$R(\dcap_s - R)/(R - \rat_s) + R + \rat_s$.
After the next timeout, the process repeats from the start, deterministically, 
over and over, until 
the receiver has received $N$ packets,
at which point it sends an \ack.
Thus, the actual efficiency is \((R(\dcap_s - R)/(R - \rat_s) + R + \rat_s)/N\).
\end{sketch}


To get a sense of how bad performance can be in an over-transmitting scenario,
suppose $\rat_s = R/10$, $\dcap_s = N/10$, and $R = N/20$.  Then the over-transmitting system would
have an efficiency of 199/1800 $\approx 0.11$.  

As explained earlier, this problem can be avoided entirely by configuring
the sender such that \(R(\dcap_s - R)/(R - \rat_s) + R + \rat_s \geq N\) or \(R \leq \rat_s \leq \dcap_s\), in which case,
the over-transmission scenario we describe is impossible.
However, this could be difficult in protocols where the window size or transmission rate
evolves with time, or where the TBF is allowed to change the rate at which it refills its bucket.
In such cases, the system may require a tight coupling of the evolution of the window size with
feedback about the state of the TBF in order to avoid over-transmitting.

If the receiver is modified to also buffer out-of-order packets, then the equality
in \Thr{gbn:worst-case} becomes an inequality, that is, the system achieves an efficiency
$\geq (R(\dcap_s - R)/(R - \rat_s) + R + \rat_s)/N$.
The reason it might be greater is that some out-of-order packets received in a prior window
might fill the gaps in the current one, allowing the cumulative \ack to increase by more
than just the number of in-order packets received in the current window.
However, in our \acls model, we do not formalize the over-transmission scenario 
for such a receiver who buffers out-of-order packets.

\section{Formalization in \acls}\Secl{acl2s}

Our model consists of four components: the sender, receiver, and two TBFs.
In this section, we describe how we model each component in \acls,
the theorems we prove about each and about the overall system,
and the proof strategies we use.  We begin with the sender.

\subsection{Formalization of the Sender in \acls}

Our model relies heavily on the DefData framework for type definitions~\cite{defdata},
which allows us to easily define new types for both data and states.
For example, we define the record type \texttt{sstate} to 
encode the sender's variables and parameters.
\begin{lstlisting}
(defdata sstate ;; window size, high ack, high pkt, next transmission
  (record (N . pos) (hiA . pos) (hiP . pos) (cur . pos)))
\end{lstlisting}
When we enter a record type into the proof state, \acls generates accessor
functions allowing us to read the record's entries.
For example, \texttt{sstate-hiA} is a function which maps an \texttt{sstate}
to its \texttt{hiA} value.
Conversely, given an \texttt{sstate}, we can set one of its values using \texttt{mset}
or multiple values at once with \texttt{msets}.
All three concepts are illustrated in the code snippet below.
Note that \texttt{mset} requires the record as its final argument while \texttt{msets}
requires that the record comes first.
\begin{lstlisting}
(= (sstate-hiA (mset :hiA 3 ss)) 3)
(= (sstate-hiP (msets ss :hiA 3 :N 5)) (sstate-hiP ss))
\end{lstlisting}

The sender evolves according to three update functions: \texttt{rcvAck}
in which it receives an \ack, \texttt{advCur} in which it transmits a packet in the
window (and advances to the next), and \texttt{timeout} in which, after transmitting
an entire window and waiting for an \ack, it times out, and ``goes back~$N$''.
Each function is defined using a \texttt{definecd} block,
which takes the form 
\begin{lstlisting}
(definecd f (arg0 :argT0 arg1 :argT1 ...) :retT :ic (icond) :oc (ocond) (body))
\end{lstlisting}
denoting the function named \texttt{f}
takes as input arguments \texttt{arg0} of type \texttt{argT0},
\texttt{arg1} of type \texttt{argT1}, etc.,
satisfying the precondition \texttt{icond},
and then executes the (terminating) code in \texttt{body},
returning a result of type \texttt{retT} which satisfies the postcondition \texttt{ocond}.
If \acls is unable to prove the postcondition automatically, it can be prompted to the solution using
hints.  The three functions are defined as follows.
\begin{lstlisting}
;; The sender receives an ack, and potentially slides the window.
(definecd rcvAck (ss :sstate ack :pos) :sstate
  (if (<= ack (1+ (sstate-hiP ss))) 
      (b* ((hiA (max (sstate-hiA ss) ack))
           (cur (max (sstate ss) hiA)))
    (msets ss :hiA hiA :cur cur))
    ss))
;; The sender sends and then increments "cur", until the entire window is sent.
(definecd advCur (ss :sstate) :sstate
  :ic (< (sstate-cur ss) (+ (sstate-N ss) (sstate-hiA ss)))
  (let* ((cur (sstate-cur ss))
         (hiP (max (sstate-hiP ss) cur)))
    (msets ss :cur (1+ cur) :hiP hiP)))
;; The sender times out, and "goes back N".
(definecd timeout (ss :sstate) :sstate
  :ic (= (sstate-cur ss) (+ (sstate-N ss) (sstate-hiA ss)))
  (mset :cur (sstate-hiA ss) ss))
\end{lstlisting}

These update functions and, when applicable, their preconditions, naturally
give rise to the transition relation for the sender, \texttt{stranr}.
Note how we can safely assume the \texttt{ack} in \texttt{rcvAck} is \texttt{(sstate-hiA ss1)}
since the resulting \texttt{sstate} is the same regardless.
\begin{lstlisting}
(definecd stranr (ss0 ss1 :sstate) :bool
  (v (== (rcvAck ss0 (sstate-hiA ss1)) ss1)
     (^ (< (sstate-cur ss0) (+ (sstate-N ss0) (sstate-hiA ss0)))
        (== (advCur ss0) ss1))
     (^ (= (sstate-cur ss0) (+ (sstate-N ss0) (sstate-hiA ss0)))
        (== (timeout ss0) ss1))))
\end{lstlisting}

Defining the initial state for the sender is slightly tricky, since we want its variables
to be positive integers, but if \texttt{hiP}>0 then surely the sender has sent a packet.
So, we assume that the sender has already sent one packet, and define the initial state to be
the one where \texttt{hiA}=\texttt{hiP}=1 and \texttt{cur}=2.
\begin{lstlisting}
(defconst *initial-ss-10* (sstate 10 1 1 2))
(definecd initial-ss (N :pos) :sstate (mset :N N *initial-ss-10*))
\end{lstlisting}
When we prove an invariant about the sender, we first prove that the invariant
holds initially, and then show that if it holds in \texttt{ss0}, and \texttt{(stranr0 ss0 ss1)},
then it also holds in \texttt{ss1}.
We prove three non-obvious invariants: (1) \texttt{hiA} $\leq$ \texttt{hiP} + 1,
(2) \texttt{hiA} $\leq$ \texttt{cur} $\leq$ \texttt{hiA} + $N$, and
(3) \texttt{hiA} and \texttt{hiP} are non-decreasing with \texttt{stranr}.
All three go through automatically after the definitions for the update functions and 
\texttt{stranr} are enabled.
As an example, here is the statement of invariant (1).
\begin{lstlisting}
(property (N :pos) 
  (<= (sstate-hiA (initial-ss N)) (1+ (sstate-hiP (initial-ss N)))))

(property (ss0 ss1 :sstate)
    :h (^ (stranr ss0 ss1) (<= (sstate-hiA ss0) (1+ (sstate-hiP ss0))))
    (<= (sstate-hiA ss1) (1+ (sstate-hiP ss1))))
\end{lstlisting}

\subsection{Formalization of the Receiver in \acls}
Next, we formalize the receiver.  This is much simpler than the sender since the only variable the
receiver needs to keep track of is the set of packets delivered to far.  We model this set as a list of positive integers,
and define a function to recognize when an \texttt{ack} is cumulative with respect to the received set.
\begin{lstlisting}
(defdata poss (listof pos))
;; Does rcvd have everything in the range [1, p]?
(definecd has-all (p :pos rcvd :poss) :bool
  (^ (in p rcvd) (v (= 1 p) (has-all (1- p) rcvd))))
;; Is ack a cumulative acknowledgment for the received set ps?
(definecd cumackp (ack :pos rcvd :poss) :bool
  (^ (! (in ack rcvd)) (v (= 1 ack) (has-all (1- ack) rcvd))))
\end{lstlisting}
As a sanity check, we prove that the cumulative \ack is unique, in the sense that if \texttt{(cumackp ack0 rcvd)} and \texttt{(cumackp ack1 rcvd)} then \texttt{(= ack0 ack1)}.  This proof requires two hints: one saying that if \texttt{ack1} were cumulative, this would imply that \texttt{ack0} $\in$ \texttt{rcvd}; and a second saying that, based on the first hint, if both \acks are cumulative then therefore \texttt{ack0} $\centernot{<}$ \texttt{ack1}.  With these, the proof goes through automatically.  

The receiver has two state update functions: one where it sends an \ack and one where it receives a packet.
Only the latter updates the received set.
It is therefore unsurprising that \acls easily dispatches the proof that the received set is non-decreasing under the subset relation.

\subsection{Formalization of the TBF in \acls}
In order to formalize the TBF we first need to define two important data types.
The first, \texttt{nat-ord}, describes the ordinals, namely, the naturals 0, 1, 2, 3, ..., as well as 
infinitely many flavors of infinity~\cite{manolios2004integrating}.
\begin{lstlisting}
(defun nth-ord (n) (if (== n 0) (omega) (1+ n)))
(register-type nat-ord :predicate o-p :enumerator nth-ord)
\end{lstlisting}
The second type we define is the timed datagram, namely, a record containing the contents
of a datagram (a positive integer id and a string payload) as well as an ordinal denoting the 
maximum possible remaining delay before the datagram must be either dropped or
forwarded to its destination.
\begin{lstlisting}
(defdata tdg (record (id . pos) (del . nat-ord) (pld . string)))
(defdata tdgs (listof tdg)) ;; Convenient type for lists of timed datagrams
\end{lstlisting}

With these type definitions out of the way, we next define the state of the TBF.
Like with the sender, we include both constants and variables in the same record.
\begin{lstlisting}
(defdata tbf
  (record (b-cap . pos) ;; bucket capacity (how large can bkt be?)
    (d-cap . pos) ;; link capacity (how many bytes can be in data?)
    (bkt . nat) ;; bucket, which must always be <= b-cap
    (rat . pos) ;; rate at which the bucket refills
    (del . nat-ord) ;; maximum delay of a datagram in data
    (data . tdgs))) ;; data in-transit, must satisfy sz(D) <= d-cap
\end{lstlisting}
The TBF has five update functions: \texttt{tick} which decrements the \texttt{del} on each \texttt{tdg} in \texttt{data},
removing any with \texttt{del}=0, and sets \texttt{bkt} to $\min(\texttt{bkt} + \texttt{rat}, \texttt{b-cap})$;
\texttt{decay} which sets \texttt{bkt} to $\max(0, \texttt{bkt} - 1)$; \texttt{prc} which takes as input the contents of a datagram, and either does nothing if the size of the datagram exceeds the remaining space in \texttt{data}, else, enqueues it in \texttt{data} with \texttt{del} set to the default delay; \texttt{drop} which takes as input some $i <$ the length of \texttt{data}, and removes the corresponding element from \texttt{data}; and \texttt{fwd} which takes the same input, but requires as a precondition that the $i^{\text{th}}$ element of \texttt{data} is not greater in size than \texttt{bkt}, and decrements \texttt{bkt} by the size of the datagram upon removal.  As an example, here is the code for \texttt{fwd}.  Note how we prove using a postcondition that the TBF is limited in how much it can deliver by the value of its \texttt{bkt}, which decrements with the delivery.
\begin{lstlisting}
;; The sz of a (timed) datagram is the length of the payload.
(definecd sz (tdgs :tdgs) :nat
  (match tdgs (() 0) ((tdg . rst) (+ (length (tdg-pld tdg)) (sz rst)))))

(definecd fwd (tbf :tbf i :nat) :tbf
  :ic (^ (< i (len (tbf-data tbf)))
   (<= (length (tdg-pld (nth i (tbf-data tbf)))) (tbf-bkt tbf)))
  ;; Theorem: TBF can only fwd bkt many bytes, and after forwarding, its bkt
  ;; is decremented by the sz of the forwarded datagram.
  :oc (^ (<= (- (sz (tbf-data (fwd tbf i))) (sz (tbf-data tbf))) (tbf-bkt tbf))
         (= (- (sz (tbf-data tbf)) (sz (tbf-data (fwd tbf i))))
            (length (tdg-pld (nth i (tbf-data tbf))))))
  (msets tbf :bkt (- (tbf-bkt tbf)
         (length (tdg-pld (nth i (tbf-data tbf)))))
   :data (remove-ith (tbf-data tbf) i))
  :function-contract-hints (("Goal" :use (:instance remove-ith-decreases-sz
                (tdgs (tbf-data tbf))))))
\end{lstlisting}

In order to prove that the serial composition of two TBFs can be simulated by a single (third) TBF,
we need four ingredients: a function to compute the third TBF,
which we refer to as the \emph{abstract composition} of the original two; a function to determine if
two TBFs are ``equivalent''; and for each function of each TBF in the serial composition,
a theorem equating (under the equivalence definition)
the serial composition after the function is applied, to some operation on the abstract composition.
To begin, we define a type \texttt{(defdata two-tbf (list tbf tbf))} to encode the internal state
of two TBFs serially composed, and an operator \texttt{[+]} to compute the corresponding abstract composition.
The intuition behind the abstract composition definition is explained above, in \Secr{gbn:model:tbf}.
Here, \texttt{(incr-del tdgs del)} adds \texttt{del} to the maximum delay of each timed datagram in \texttt{tdgs}.
\begin{lstlisting}
(definecd [+] (ttbf :two-tbf) :tbf
  (tbf
   (tbf-b-cap (cadr ttbf)) ;; bkt capacity = bkt capacity of the second TBF
   (+ (tbf-d-cap (car ttbf)) (tbf-d-cap (cadr ttbf))) ;; link capacity = sum
   (tbf-bkt (cadr ttbf)) ;; bkt = bkt of the second TBF
   (tbf-rat (cadr ttbf)) ;; rate = rate of the second TBF
   (o+ (tbf-del (car ttbf)) (tbf-del (cadr ttbf))) ;; max delay = sum
   ;; incr the delays on the first data and prepend the result to the second
   (append (incr-del (tbf-data (car ttbf)) (tbf-del (cadr ttbf)))
           (tbf-data (cadr ttbf)))))
\end{lstlisting}

Then we define our equivalence notion, which is that two TBFs are equivalent if they have equal caps and variables,
except for the \texttt{data}s, for which we require that the ids in the former are a permutation of the ids in the latter.
Given two \texttt{poss}s, say, \texttt{ids0} and \texttt{ids1}, the way we show one is a permutation of the other is by
proving that for all \texttt{x} $\in$ \texttt{pos}, the \texttt{count} of \texttt{x} in \texttt{ids0} equals the \texttt{count} of \texttt{x} in \texttt{ids1}.  We refer to this kind of equivalence as \verb|~|\texttt{=}.
Using this notion of equivalence, we dispatch the theorems relating steps of the serial composition to steps of the abstract one with either hints to the automated prover, or a manual proof.  For example, here is the theorem which states that when the first TBF in the serial composition processes a datagram, the result is equivalent to when the abstract composition processes a datagram.
This theorem goes through with 18 proof instructions.
\begin{lstlisting}
(defthm transmission-rule
  (=> (^ (two-tbfp ttbf) (posp p) (stringp pld)
         (<= (+ (sz (tbf-data (car ttbf))) (length pld)) (tbf-d-cap (car ttbf)))
         (<= (sz (tbf-data (cadr ttbf))) (tbf-d-cap (cadr ttbf))))
      (~= ([+] (list (prc (car ttbf) p pld) (cadr ttbf))) 
          (prc ([+] ttbf) p pld))))
\end{lstlisting}

The most tricky is the theorem which says that when in the serial composition a datagram is forwarded from the first TBF
to the second, the result is equivalent to a noop in the second, provided that the datagram is not lost in the process.
The crux of this theorem is the following lemma, which says that when we move an item from one list to another,
the concatenation of the original two lists is a permutation of the concatenation of the latter two.
\acls proves this theorem automatically, after being provided seven hints (two instantiations each of three lemmas, plus a case-split).
\begin{lstlisting}
(property mv-is-a-permutation (ps0 ps1 :tl i :nat p :all)
  :h (< i (len ps0))
  (= (count p (append (remove-ith ps0 i) (cons (nth i ps0) ps1)))
     (count p (append ps0 ps1))))
\end{lstlisting}
After a number of additional (smaller) lemmas, we are able to lift this result to an equivalence theorem on the serial and abstract compositions.

\subsection{Formalization of Efficiency Analysis in \acls}

Originally we proved each efficiency result (best and worst case) separately,
but then when revising the proofs, we realized that the ``worst case'' proof strategy
could be modified to dispatch the best-case result as well.
The key idea is to define a simplified model which is easier to reason about, 
and prove that this simplified model adequately simulates the real system.
\begin{lstlisting}
;; The real system under study.
(defdata system
  (record (sender . sstate) (receiver . poss) (s2r . tbf) (r2s . tbf)))
;; The simplified model.  The channel contains only ids.
(defdata simplified-system
  (record (chan . poss) (d-cap . nat) (ack . pos) 
          (cur . pos) (hiA . pos) (N . pos)))
\end{lstlisting}

To show that we can reason about the \texttt{system} by reasoning about its
simplification, we first show that the map from the former to the latter is preserved
when the sender transmits a packet ...
\begin{lstlisting}
(== (simplify (prc-1 sys x)) (prc-1-simplified (simplify sys)))
\end{lstlisting}
 ... or when the TBF forwards a packet to the receiver ...
\begin{lstlisting}
(== (simplify (fwd-1 sys)) (fwd-1-simplified (simplify sys)))
\end{lstlisting}
... under the appropriate preconditions for each, and with the assumptions that 
every packet has size one (\texttt{(all-1 (tdgs->poss (tbf-data (system-s2r sys))))})
and an unbounded delay value (\texttt{(all-inf (tdgs->poss (tbf-data (system-s2r sys))))}),
and the \texttt{s2r} TBF has unbounded delay (\texttt{(! (natp (tbf-del (system-s2r sys))))}).

Next, we repeat this step for the repetition of each function.  That is, we define a function \texttt{prc-R} which repeats \texttt{prc-1} $R$ times, for some $R \leq \texttt{hiA} + \texttt{N} - \texttt{cur}$, sending a default packet ``p'' each time.
(The choice of char for the payload of the packet does not matter; we use ``p'' arbitrarily.)
We define another function \texttt{fwd-b} which repeats \texttt{fwd-1} $b$ times, for some $b \leq \texttt{b-cap}$; and we define a simplified version of each function.
Then we connect the simplifications to the originals in the same way we did for \texttt{prc-1} and \texttt{fwd-1},
under the assumption that $\bcap_s = \rat_s \leq \dcap_s$.
After this, we define a function \texttt{single-step} which applies \texttt{prc-R}, then makes \texttt{s2r} tick,
before finally applying \texttt{dlv-b}; and we show that so long as \texttt{s2r} has an infinite (non natural) delay,
and the packets in transit satisfy \texttt{all-1} and \texttt{all-inf}, then the simplification of \texttt{single-step} equals \texttt{single-step-simplified} applied to the simplification of the system.
We are then able to prove the best-case efficiency by analyzing \texttt{single-step-simplified}.
\begin{lstlisting}
(property best-case-efficiency (sm :simplified-model R :pos)
  :h (^ (endp (simplified-model-chan sm))
        (<= (+ (simplified-model-cur sm) R)
            (+ (simplified-model-hiA sm) (simplified-model-N sm)))
        (<= R (simplified-model-d-cap sm))
        (= (simplified-model-cur sm) (simplified-model-ack sm))
        (< 1 (simplified-model-cur sm)))
  (^ ;; Preserve input contracts
   (endp (simplified-model-chan (single-step-simplified sm R R)))
   (= (simplified-model-cur (single-step-simplified sm R R))
      (simplified-model-ack (single-step-simplified sm R R)))
   ;; Actual efficiency theorem
   (= (/ R (- (simplified-model-ack (single-step-simplified sm R R))
        (simplified-model-ack sm)))
      1))
\end{lstlisting}

For the worst-case result, we need to reason about multiple steps -- first a series of steps which
fill the channel in the sender-to-receiver direction, then one or more steps that occur in which the channel
overflows and losses occur.
To do this, we lift \texttt{single-step-simplified} to a function \texttt{multi-step-simplified},
which simply repeats \texttt{single-step-simplified} a given number of times.
\begin{lstlisting}
(definecd many-steps-simplified (sm :simplified-model R b steps :pos) :simplified-model
   :ic (^ (<= (+ (simplified-model-cur sm) (* R steps))
              (+ (simplified-model-hiA sm) (simplified-model-N sm)))
          (<= b (min R (simplified-model-d-cap sm)))
          (<= (len (simplified-model-chan sm)) (simplified-model-d-cap sm)))
   (if (= steps 1)
       (single-step-simplified sm R b)
     (many-steps-simplified (single-step-simplified sm R b) R b (1- steps))))
\end{lstlisting}
We define a function to compute the number of repetitions of \texttt{single-step-simplified}
that will be needed to fill the channel to $R$ less than its capacity.
\begin{lstlisting}
(definecd steps-to-fill (R b d-cap :pos) :pos
  :ic (^ (< b R) ;; overtransmission
         (< R d-cap)
         ;; simplifying assumption that R - b divides d-cap - R
         (natp (/ (- d-cap R) (- R b))))
  (/ (- d-cap R) (- R b)))
\end{lstlisting}
We prove that after $(\dcap_s - R)/(R - b)$ \texttt{single-step-simplified}s, 
all the following hold:
\begin{enumerate}[i.]
  \item The \texttt{channel} (which, recall, contains the ids of the packets in \texttt{s2r})
  equals the descending list
  \(
[\texttt{cur}_0 + R (\dcap_s - R)/(R - b) - 1, 
  \texttt{cur}_0 + R (\dcap_s - R)/(R - b) - 2, 
  \ldots]
  \)
  of length $\dcap_s - R$, where $\texttt{cur}_0$ was the \texttt{cur} value before the 
  $(\dcap_s - R)/(R - b)$ steps were taken.
  \item The \texttt{ack} value (i.e., the cumulative acknowledgment the receiver would send next,
  were it to send one) has increased by $\rat_s * (\dcap_s - R) / (R - b)$.
  \item The \texttt{cur} value has increased by $R * (\dcap_s - R) / (R - b)$.
\end{enumerate}
In \acls, this looks like the following.
\begin{lstlisting}
(let* ((warmup-period (steps-to-fill R b (simplified-model-d-cap sm)))
       (many-steps-later (many-steps-simplified sm R b warmup-period)))
       (^ (== (simplified-model-chan many-steps-later)
              (top-dn (+ (simplified-model-cur sm) (* R warmup-period) -1)
                      (* (- R b) warmup-period)))
          (= (simplified-model-ack many-steps-later)
             (+ (* b warmup-period) (simplified-model-ack sm)))
          (= (simplified-model-cur many-steps-later)
            (+ (* R warmup-period) (simplified-model-cur sm)))))
\end{lstlisting}

We then prove two additional theorems, characterizing what happens to the \texttt{channel}, \texttt{cur}, and \texttt{ack} after each of the next
two \texttt{single-step-simplified}s.
In the first, the channel becomes full, and then $\rat_s$ packets are delivered.
In the second, the first $\rat_s$ packets make it into the channel FIFO before losses occur.
Since the \texttt{cur} value increases until a timeout occurs, we are able to show that
no subsequent packet transmissions will be delivered by proving a gap between \texttt{cur}
and the most recently processed value in the channel.

Combining these facts, if the sender transmits $R$ packets into \texttt{s2r},
and $b$ are delivered, then we know the length of \texttt{s2R} increased by $R - b$ up to $\texttt{d-cap}$,
at which point, the invariant that the channel is of top-down form is no longer satisfied.
Thus, $R(\dcap_s - R)/(R - b) + R + b$ total packets make it from the sender to the receiver
FIFO, before losses begin occurring, after which the packets are not FIFO and therefore 
do not get delivered after being received by the receiver.
Since we assume the receiver does not send an \ack until it has received $N$ packets,
it follows that when $R > b$ the efficiency is $(R(\dcap_s - R)/(R - b) + R + b) / N$.
Plugging in $b = \bcap_s = \rat_s$ yields the worst case result.

\section{Related Work}\Secl{gbn:related}

Several prior works analyzed the performance of other ARQ protocols using pen-and-paper 
mathematics~\cite{fayolle1978analytic,easton1980batch,bruneel1994analytic}.
In that vein, Lockefeer et. al. used pen-and-paper mathematics to prove that the selective
acknowledgment (\sack) feature could improve the performance 
of the sliding window mechanism in TCP~\cite{IEEEACMTON02}.
They modeled \sack using the I/O automata formalism of Lynch and Tuttle~\cite{lynch1988introduction},
which is equivalent to our formalism.
Using a refinement argument, they showed that the traces of TCP with \sack
are equivalent to a subset of the traces of a generic specification for an end-to-end reliable message service.
Then, they extended their model to include a notion of time, and showed that in certain worst-case scenarios,
\sack can decrease packet latency by an amount proportional 
to the product of the RTT and the number of packet losses.
This second result had at least two major limitations.
(1) Because they made stronger assumptions than we did, they report that the true worst
case performance of the system could be much worse then what they computed,
if the RTO exceeds the RTT.
As we showed in \Chapr{karn-rto}, even when the RTTs are bounded in the infinite time horizon, 
the RTO may exceed the RTT by as much as the difference in the bounds, which could be considerable.
(2) They only showed that it is \emph{possible} for \sack to improve performance relative to a standard 
cumulative \ack scheme -- they did not show that the performance of \sack is always no worse than that of
the standard scheme.
It is also unclear how precisely they defined the RTT.
As we show in \Chapr{karn-rto},
one cannot simply assume that the ``true'' RTT is identical to the value sampled by Karn's Algorithm,
since in the presence of retransmissions, Karn's Algorithm cannot sample at all.
Moreover, the value sampled by Karn's Algorithm is not necessarily identical to the sum of the average time it takes
for a packet to travel from sender to receiver plus the average time it takes for an \ack to travel
from receiver to sender (a misconception common to several of the prior works we referenced in \Secr{gbn-intro}).
Unfortunately, the authors do not include their timed model for us to check.

The refinement map Lockefeer et. al. used
connected the \emph{send}, \emph{retransmission}, and \emph{receive} buffers to a single queue which abstracted reliable communication~\cite{IEEEACMTON02}.
Our over-transmission proof actually does something similar.
Since we know that in the scenario we analyze, all packet losses occur at transmission time,
given the event sequence we assume the worst-case system follows,
clearly every packet which enters $\mathcal{F}_s$ eventually reaches the receiver.
Therefore, we prove the worst-case performance bounds by defining an invariant which says that the \rcvd set contains $1, \ldots, \ha - 1$ and a postfix of the packets in transit are precisely $\ha, \ldots, \cur - 1$ (where $\ha \leq \cur - 1$),
and then proving that if there are initially zero packets in transit then the invariant holds for $\dcap_s / (R - \rat_s)$ bursts of $R$ packets each.
This proof strategy can be seen as connecting the packets in transit ($\data_s$) to the cumulatively received packets (\rcvd).
An interesting direction for future work is to see if our over-transmission analysis can be simplified using an explicit refinement argument.
However, doing this in \acls may be more challenging than making an analogous argument with pen-and-paper, as Lockefeer et. al. did, because \acls requires the argument to be fully formal.

Works which apply formal methods to congestion control algorithms are also closely related
because these algorithms, for the most part, build on GB$(N)$ by modifying the window size~$N$
(referred to as the congestion window, or \emph{cwnd}) on the fly.
In~\cite{zarchy2017axiomatic}, Zarchy et. al. defined ``axioms'' for congestion control algorithms
characterizing certain fundamental guarantees the algorithms might want to satisfy,
and then showed that some axioms were incompatible with others.
They did not use a formal methods software, but their approach was logically grounded
and fully formal in practice.
Since then, Venkat, Agarwal, and colleagues have published a number of works applying formal methods
to congestion control algorithms: 
proposing a unified formal framework for congestion control algorithm verification~\cite{ccac},
defining and proving the possibility of \emph{starvation} in certain algorithms~\cite{arun2022starvation},
and most recently, automatically synthesizing congestion control algorithms to meet 
formally specified performance guarantees~\cite{agarwaltowards}.
Their formal framework~\cite{ccac} included a TBF in the sender-to-receiver direction, 
albeit, with slightly different features from ours (e.g., no nondeterministic loss).
They proved a composition theorem for their TBF definition but reported that they were unable to handle
the case with unbounded delay (in our model, unbounded $\ttl$).
We were able to dispatch both the bounded and unbounded cases at once, by modeling the $\ttl$
as an ordinal.

To the best of our knowledge, ours is the first work to formally analyze the efficiency of GB$(N)$.
A limitation of our work is that we do not characterize how long sequences of events can take in the real world.
Probably the best way to solve this is by applying something similar to the \emph{symbolic latency} approach proposed by Zhang, Sharma, and Kapritsos~\cite{zhang2023performal}.
The basic idea is to define a happens-before relation on events in the system, which then yields a symbolic calculus for how long a trace could potentially take to execute depending on the distributions of durations of particular events when measured in the wild, and the different ways those events might overlap.
In a related work,
Arashloo, Beckett, and Agarwal suggest an approach to distributed systems testing
where the tests are concrete workloads generated by a synthesizer in response to abstract
queries about possible system performance~\cite{arashloo2023formal}.
We could do something similar by implementing our concurrency model as a happens-before relation and then
defining probability distributions for the durations of time required for different events in the model.
A benefit of doing this in \acls would be the ability to generate workloads ``for free'',
using enumerators~\cite{walter2022enumerative}.

\section{Conclusion}\Secl{gbn:conclusion}

In this chapter we formally modeled the GB$(N)$ protocol over a network with a 
Token Bucket Filter in each direction.  
Since there is no singular, canonical definition of GB$(N)$, we wrote our model in a way
that could capture many plausible variations of the protocol at once.
Using our model, we proved the following theorems.
\begin{enumerate}[{Thm. }1:]
\item Three inductive invariants confirming that the sender updates its internal variables correctly.
\item That the set of packets the receiver has cumulatively received, stored in the receiver's local variable \rcvd, is non-decreasing under the subset relation.
\item The TBF cannot forward more bytes of data than it has tokens to spend.
\item The serial composition of two TBFs can be simulated by a single (larger) TBF.
\item It is possible for GB$(N)$ to achieve perfect efficiency.  
\item A formula for the efficiency of GB$(N)$ when the sender constantly over-transmits,
leading to deterministic losses.
\end{enumerate}

These results provide a first step toward characterizing the performance of
more complex protocols including the sliding window logic in modern TCP implementations like New Reno,
where the window size evolves with time.
In particular, our over-transmission analysis provides insight into how a sender
should be configured, relative to the TBF it transmits into, in order to avoid
deterministic losses.

\setcounter{observation}{0}
\setcounter{definition}{0}
\setcounter{problem}{0}
\setcounter{theorem}{0}

\chapter{Protocol Correctness for Handshakes}\Chapl{handshakes}
\textbf{Summary.}
An important component of every transport protocol is its handshake, i.e., the mechanisms by which it forms and deletes associations.  We explain how handshakes work at a high level and give some examples.  Then, we describe a formal modeling language which allows us to describe protocol handshakes as finite Kripke structures.  We model and write LTL correctness properties for three protocol handshakes: TCP, DCCP, and SCTP.  Our models and properties are carefuly justified based on the corresponding protocol RFC documents.  Using the \spin model checker, we prove that all three handshakes satisfy the correctness properties we write for them, in the absence of an attacker.
These properties have to do with the interactions between the \emph{active} peer, 
who initiates an exchange, and the second peer, who either \emph{passively} responds,
or simultaneously initiates.

Our major results are as follows.
The TCP handshake avoids half-open connections and deadlocks, and its active/passive establishment routine
works as expected.
The DCCP handshake avoids infinite looping behaviors, and supports neither active/active nor passive/passive teardown.
And finally, the SCTP handshake avoids multiple unsafe states which are explicitly precluded in the RFC, responds appropriately to messages, uses its timers when needed, and satisfies numerous additional safety and liveness properties implied by its RFC.

\medskip

This chapter includes work originally presented in the following publications:

\medskip

\noindent~Max von Hippel, Cole Vick, Stavros Tripakis, and Cristina Nita-Rotaru. \emph{Automated attacker synthesis for distributed protocols.} Computer Safety, Reliability, and Security, 2020.
\begin{description}
\item \underline{Contribution:} MvH formalized the problem with help from ST, invented the solution, wrote the proofs, wrote most of the code for the implementation and TCP case study, and wrote most of the paper.
\end{description}

\medskip

\noindent~Maria Leonor Pacheco, Max von Hippel, Ben Weintraub, Dan Goldwasser, and Cristina Nita-Rotaru. \emph{Automated attack synthesis by extracting finite state machines from protocol specification documents.} IEEE Symposium on Security and Privacy, 2022.
\begin{description}
\item \underline{Contribution:} MvH wrote the models and properties, as well as the FSM extraction algorithm (not included in this dissertation).
\end{description}

\medskip

\noindent~Jacob Ginesin, Max von Hippel, Evan Defloor, Cristina Nita-Rotaru, and Michael T{\"u}xen. \emph{A Formal Analysis of SCTP: Attack Synthesis and Patch Verification.} USENIX, 2024.
\begin{description}
\item \underline{Contribution:} MvH co-authored the models and properties and wrote more than half of the paper.
\end{description}

\section{Transport Protocol Handshakes}\Secl{handshakes:intro}

Transport protocols represent the fundamental communication backbone for much of the Internet.  
In the prior two chapters, we showed how provers can be used to verify both inductive
invariants as well as performance bounds of protocols.
Now, we focus on a different aspect of correctness: modeling and 
proving temporal properties of transport protocol \emph{handshakes}.

Each transport protocol has a handshake mechanism, namely, some procedure
by which a sender and a receiver can establish an association before exchanging data,
and tear down the association upon concluding the exchange.
During establishment, a peer who attempts to initiate a handshake is called \emph{active}.
If both peers attempt to initiate the same handshake at once, then they are both called active;
otherwise the responding peer is referred to as \emph{passive}.
Likewise, during teardown, a peer who initiates teardown is referred to as active,
while a peer who responds to a request to tear down an existing association is called passive.
Thus, a handshake might have both active/active and active/passive establishment routines,
as well as potentially both active/active and active/passive teardown routines.
However, passive/passive routines are impossible by definition.
As an example, we illustrate active/passive establishment and teardown for SCTP in
\Figr{handshakes:sctp:activePassive}.

\setlength{\levelheight}{0.8cm}

\begin{figure}[h]
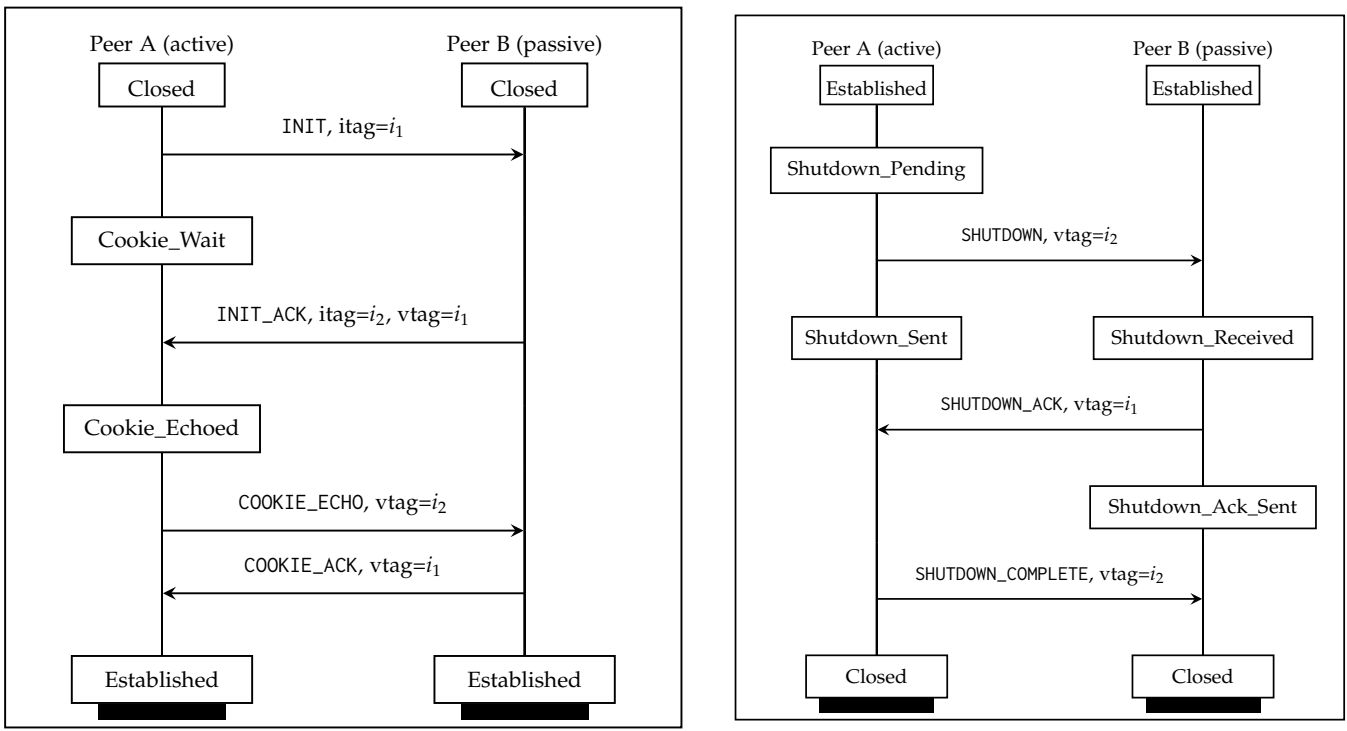

\centering
\begin{minipage}{0.48\textwidth}
\begin{adjustbox}{width=\textwidth,center}
\begin{msc}[head top distance=0.7cm,msc keyword=,left environment distance=2cm,right environment distance=2cm,foot distance=0.1cm, instance distance=3cm, action width=3cm]{}
    \declinst{A}{\scriptsize Peer A (active)}{\scriptsize \Closed}
    \declinst{B}{\scriptsize Peer B (passive)}{\scriptsize \Closed}
    
    \mess{\scriptsize \Init, itag=$i_1$}{A}{B}
    
    \nextlevel
    \action[action width=2.3cm]{\scriptsize \CookieWait}{A}
    \nextlevel
    \nextlevel

    \mess{\scriptsize \InitAck, itag=$i_2$, vtag=$i_1$}{B}{A}
    
    \nextlevel
    \action[action width=2.5cm]{\scriptsize \CookieEchoed}{A}
    \nextlevel
    \nextlevel

    \mess{\scriptsize \CookieEcho, vtag=$i_2$}{A}{B}
    \nextlevel
    \mess{\scriptsize \CookieAck, vtag=$i_1$}{B}{A}

    \nextlevel
    \action[action width=2.3cm]{\scriptsize \Established}{A}
    \action[action width=2.3cm]{\scriptsize \Established}{B}
\end{msc}
\end{adjustbox}
\end{minipage}
\begin{minipage}{0.48\textwidth}
\begin{adjustbox}{width=0.9\textwidth,center}
\begin{msc}[head top distance=0.7cm,msc keyword=,left environment distance=2cm,right environment distance=2cm,foot distance=0.1cm, instance distance=3cm, action width=3cm]{}
    \declinst{A}{\scriptsize Peer A (active)}{\scriptsize \Established}
    \declinst{B}{\scriptsize Peer B (passive)}{\scriptsize \Established}

    \action[action width=3cm]{\scriptsize \ShutdownPending}{A}

    \nextlevel
    \nextlevel

    \mess{\scriptsize \Shutdown, vtag=$i_2$}{A}{B}
    
    \nextlevel
    \action[action width=2.4cm]{\scriptsize \ShutdownSent}{A}
    \action[action width=3.1cm]{\scriptsize \ShutdownReceived}{B}
    \nextlevel
    \nextlevel

    \mess{\scriptsize \ShutdownAck, vtag=$i_1$}{B}{A}
    
    \nextlevel
    \action[action width=3.2cm]{\scriptsize \ShutdownAckSent}{B}
    \nextlevel
    \nextlevel

    \mess{\scriptsize \ShutdownComplete, vtag=$i_2$}{A}{B}
    \nextlevel

    \action[action width=2cm]{\scriptsize \Closed}{A}
    \action[action width=2cm]{\scriptsize \Closed}{B}
\end{msc}
\end{adjustbox}
\end{minipage}
\caption{Message sequence charts illustrating SCTP active/passive association establishment routine (left) and active/passive teardown (right).  Arrows indicate communication direction and time flows from the top down.  We discuss the message components further in \Secr{handshakes:sctp}, but briefly: each message consists of a control message (e.g., $\ShutdownAck$), and optionally a verification or initiate tag (vtag or itag).  The itag is a random integer, and sets the corresponding vtag for the rest of the handshake.}
\Figl{handshakes:sctp:activePassive}
\end{figure}

Typically, RFC documents describe handshakes using message sequence charts (such as \Figr{handshakes:sctp:activePassive}),
as well as finite state machine diagrams (like our \Figr{sctp:fsm}).  
But the way that these illustrations are provided in the RFC documents is often 
vague or imprecise.  Moreover, RFCs rarely explicitly state protocol goals
as logical properties, rather, the goals are left implicit in the high-level protocol
description and use-cases it was ostensibly developed for, or scattered in off-hand comments throughout
the document (which must be manually coalesced to form a cohesive specification).
This status quo creates a situation in which much of the web relies on handshake mechanisms
with vague or unclear requirements and no formal assurance that those requirements,
should they exist, are always met.

Transport protocol handshakes are finite-state in the sense that there are only two participants in a
handshake, each participant moves through a pre-defined finite set of states according
to a common procedure, and the messages the participants send and receive are drawn
from a finite set of control messages.
Because handshakes are finite-state, we can forego theorem proving and instead
analyze them automatically using a model checker.
In this chapter we do exactly that.
We formally model the handshakes of three commonly used transport protocols as finite-state
processes based off a close reading of the respective RFC documents.
Then, we logically formulate temporal properties those protocols should satisfy, again
reading between the lines of the RFCs.
Finally, we use a model checker to prove that the modeled handshakes satisfy the 
transcribed properties, for the system consisting of two protocol peers connecting over
a FIFO channel with a size-1 buffer in each direction (illustrated in \Figr{common:network:model}).
Note that we use the model checker in its exhaustive mode, which is only possible because our
models are relatively small.

The rest of this chapter is organized as follows.
We give an overview of TCP, DCCP, and SCTP in \Secr{intro-tcp-dccp-sctp}.
We formally define the semantics of LTL over finite Kripke structures in \Secr{ks-ltl}.
In \Secr{handshakes:math}, we provide formal definitions of processes and process composition,
allowing us to reduce a handshake involving two participants and a bidirectional channel
to a single finite Kripke structure (which can then be model checked).
Put differently, \Secr{ks-ltl} explains the basics of LTL model checking,
while \Secr{handshakes:math} shows how we can use this framework to analyze a handshake
involving two communicating protocol participants.
Next, we look at three important case studies: TCP (\Secr{handshakes:tcp}),
DCCP (\Secr{handshakes:dccp}), and SCTP (\Secr{handshakes:sctp}).
In each, we give a brief overview of the protocol handshake being studied,
provide a fully formal model and LTL properties the model should satisfy,
and justify our model and properties based off a close reading of the corresponding RFC.
We find that all three models are correct, in the sense that they satisfy all of the
correctness properties we found.
We conclude in \Secr{handshakes:conclusion}.

\section{Overview of TCP, DCCP, and SCTP}\Secl{intro-tcp-dccp-sctp}

TCP was first proposed by Cerf and Kahn in 1974~\cite{vint1974protocol},
    as the singular transport protocol for the Internet, 
    providing reliable, in-order packet delivery -- a contribution for which 
    they were awarded the ACM Turing Award in 2004~\cite{turing}.
Only after researchers began investigating voice-over-IP in the 1970s did it
become clear that this guarantee came with a performance trade-off~\cite{leiner1997past},
    ultimately leading to the split of TCP and Internet Protocol (IP) into separate protocols,
    and the development of the User Datagram protocol (UDP)~\cite{rfc768_udp},
    an unreliable transport protocol designed for time-sensitive applications.
Early applications of TCP included email~\cite{rfc4321}, file transfer~\cite{rfc913}, 
and remote login~\cite{rfc4253}, all of which are still used today.
There are many TCP variants, such as TCP Vegas~\cite{brakmo1995tcp} or Westwood~\cite{mascolo2001tcp},
but all of them use the common handshake described in RFC 9293~\cite{rfc9293_tcp_new}.
In this handshake at least one peer must take an active role during the establishment routine,
and likewise for the teardown routine;
however either peer could switch roles between routines so long as at least one is active.
This is an unusual characteristic not shared by DCCP or SCTP (which we discuss next).

DCCP is canonically specified in RFC 4340~\cite{rfc4340_dccp}.
It is similar to TCP, but does not guarantee in-order message delivery,
and does not support active/active establishment.
On the other hand, it is faster than TCP, and thus appropriate for applications like
telephony or media streaming where speed is more important than reliability.
In contrast to UDP, DCCP provides built-in congestion control features,
without needing to implement them in the application layer.
Note, we do not model congestion control algorithms in this dissertation.
Also in contrast to TCP, the active and passive peers have fixed roles for the lifetime of the
association.

SCTP is a transport protocol
offering features such as multi-homing, 
multi-streaming, and message-oriented delivery.
Among other use-cases, it is the data channel for WebRTC~\cite{webrtcSCTP}, 
which is used by such applications as Facebook Messenger~\cite{fbWebRTC},
Microsoft Teams~\cite{teamsWebRTC},
and Discord~\cite{discordWebRTC}.
The design of SCTP
is described in RFC 9260~\cite{rfc9260}, and implemented in Linux~\cite{linux} and FreeBSD~\cite{freebsd}.
Much like DCCP, SCTP only supports active/passive establishment\footnote{SCTP also supports an initialization routine where both peers are active, called ``initialization collision''.  However, this routine is described in the RFC as an edge-case, rather than an intended use-case.  
},
but unlike DCCP, the active peer during establishment does not need to be active during teardown.
For teardown there are two options: graceful or graceless.
During graceful tear-down, one peer can act actively and the other passively, or they can both take an active role.
Graceless teardown happens in a single step.

\section{Finite Kripke Structures and Linear Temporal Logic}\Secl{ks-ltl}

Next, we provide the semantics of LTL for finite Kripke structures.
Note, we use $2^X$ to denote the power-set of $X$, and
$\omega$-exponentiation to denote infinite repetition, e.g., $a^\omega=a a a \cdots$.

\begin{sloppypar}
\begin{definition}[Finite Kripke Structure]
A \emph{finite Kripke structure} is a tuple 
$K = \langle \text{AP}, S, s_0, T, L \rangle$ with set of \emph{atomic propositions} AP, 
set of \emph{states} $S$, 
\emph{initial state} $s_0 \in S$, \emph{transition relation} $T \subseteq S \times S$, 
and (total) \emph{labeling function} $L : S \to 2^{\text{AP}}$, such that
$\text{AP}$ and $S$ are finite.
\end{definition}
\end{sloppypar}

A \emph{run} of a finite Kripke structure $K$ is any sequence of 
transitions $t_0, t_1, \ldots \in T$ such that
states $s_0, s_1, \ldots$ such that 
$T(s_i, s_{i+1})$ for each $i$.
In other words, a run is a behavior of the structure.
A \emph{trace} of $K$ is the sequence $L(s_0), L(s_1), \ldots$ where $s_0, s_1, \ldots$ is a run.
A trace is an observable behavior of the system.
When reasoning about runs or traces, we use the following (Pythonic) indexing notation.
Given a (zero-indexed) sequence $\nu$, we let $\nu[i]$ denote the $i^{\text{th}}$ element of $\nu$; $\nu[i:j]$, where $i \leq j$, denote the finite infix $(\nu[t])_{t = i}^{j}$; and $\nu[i:]$ denote the infinite postfix $(\nu[t])_{t = i}^{\infty}$; we will use this notation for runs and computations.

\emph{LTL}~\cite{IEEEASFCS77}  is a temporal logic for reasoning about traces of finite Kripke Structures.
The syntax of LTL is defined by the following grammar, where $\U$ means ``until'' and $\X$ means ``next'':
\begin{equation}
\phi ::= \underbrace{p \mid q \mid ...}_{\in \text{AP}} \mid \phi_1 \land \phi_2 \mid \neg \phi_1 \mid \X \phi_1 \mid \phi_1 \U \phi_2
\end{equation}
... where $p, q, ... \in \text{AP}$ can be any atomic propositions, and $\phi_1, \phi_2$ can be any LTL formulae.  Let $\sigma$ be a computation of a finite Kripke structure $K$.  If an LTL formula $\phi$ is true about $\sigma$, we write $\sigma \models \phi$.  On the other hand, if $\neg ( \sigma \models \phi )$, then we write $\sigma \centernot{\models} \phi$.  The semantics of LTL with respect to $\sigma$ are as follows.
\begin{equation}
\begin{array}{lcl}
\sigma \models p & \text{ iff } & p \in \sigma[0] \\
\sigma \models \phi_1 \land \phi_2 & \text{ iff } & \sigma \models \phi_1 \text{ and } \sigma \models \phi_2 \\
\sigma \models \neg \phi_1 & \text{ iff } & \sigma \centernot{\models} \phi_1 \\
\sigma \models \X \phi_1 & \text{ iff } & \sigma[1:] \models \phi_1 \\
\sigma \models \phi_1 \U \phi_2 & \text{ iff } & \big( \exists \, \kappa \geq 0 \, : \, \sigma[\kappa:] \models \phi_2 \text{, and }  \\
                                &              & \, \, \, \forall \, 0 \leq j < \kappa \, : \, \sigma[j:] \models \phi_1 \big) \\
\end{array}
\end{equation}

Essentially, $p$ holds iff it holds at the first step of the computation; the conjunction of two formulae holds if both formulae hold; the negation of a formula holds if the formula does not hold; $\X \phi_1$ holds if $\phi_1$ holds in the next step of the computation; and $\phi_1 \U \phi_2$ holds if $\phi_2$ holds at some future step of the computation, and until then, $\phi_1$ holds.  Standard syntactic sugar include $\lor$, \textbf{true}, \textbf{false}, $\F$ (``eventually''), $\G$ (``globally''), and $\to$ (``implies'').  For all LTL formulae $\phi_1, \phi_2$ and atomic propositions~$p \in \text{AP}$: $\phi_1 \lor \phi_2 \equiv \neg (\neg \phi_1 \land \neg \phi_2)$; $\textbf{true} \equiv p \lor \neg p$; $\textbf{false} \equiv \neg \textbf{true}$; $\F \phi_1 \equiv \textbf{true} \U \phi_1$; $\G \phi_1\equiv \neg \F \neg \phi_1$; and $\phi_1 \to \phi_2 \equiv (\neg \phi_1) \lor (\phi_1 \land \phi_2)$.  We provide some example formulae in \Secr{ltl:examples} in the Appendix.

An LTL formula $\phi$ is called a \emph{safety property} iff it can be violated by a finite prefix of a computation, or a \emph{liveness property} iff it can only be violated by an infinite computation \cite{BaierKatoenBook}.  Every LTL formula is the intersection of a safety property and a liveness property~\cite{alpern1985defining}, and moreover, the decomposition can be done entirely within LTL~\cite{maretic2014ltl}.  For a finite Kripke structure $K$ and LTL formula $\phi$, we write $K \models \phi$ iff, for every computation $\sigma$ of $K$, $\sigma \models \phi$.  For convenience, we naturally elevate our notation for satisfaction on computations to satisfaction on runs, that is, for a run $r$ of a process $K$ inducing a computation $\sigma$, we write $r \models \phi$ and say ``$r$ satisfies $\phi$" iff $\sigma \models \phi$, or write $r \centernot{\models} \phi$ and say ``$r$ violates $\phi$" iff $\sigma \centernot{\models} \phi$.

\section{Formal Setup for Transport Protocol Handshake Models}\Secl{handshakes:math}

We model protocols as interacting \emph{processes}, in the spirit of~\cite{SIGACT17}.
A process is just a Kripke Structure with {\em inputs} and {\em outputs}.
The composition of these processes can be projected onto a finite Kripke structure
amenable to model checking, as we explain shortly.

\begin{sloppypar}
\begin{definition}[Process]
A \emph{process} is a tuple 
$P = \langle \text{AP}, I, O, S, s_0, T, L \rangle$ 
such that $\langle \text{AP}, S, s_0, \{ (s, s') \mid \exists x \in I \cup O \,::\, (s, x, s') \in T \}, L \rangle$
is a finite Kripke structure,
$T \subseteq S \times (I \cup O) \times S$,
and $I \cap O = \emptyset$.
\end{definition}
\end{sloppypar}

The state $s$ is called \emph{reachable} if either it is the initial state or there exists a sequence of transitions \[\big( (s_i, x_i, s_{i+1}) \big)_{i = 0}^m \subseteq T\]
starting at the initial state $s_0$ and ending at $s_{m + 1} = s$.  Otherwise, $s$ is called \emph{unreachable}.

The composition of two processes $P_1$ and $P_2$ is another process denoted $P_1 \parallel P_2$, capturing both the individual behaviors of $P_1$ and $P_2$ as well as their interactions with one another (e.g. \Figr{exampleComposition}).  We define the asynchronous parallel composition operator $\parallel$ with rendezvous communication as in \cite{SIGACT17}.
\begin{definition}[Process Composition]
Let \(P_i = \langle \text{AP}_i, I_i, O_i, S_i, s_0^i, T_i, L_i \rangle\) be processes,
for $i = 1, 2$.  For the composition of $P_1$ and $P_2$ 
(denoted $P_1 \parallel P_2$) to be well-defined,
the processes must have no common outputs,
    i.e., $O_1 \cap O_2 = \emptyset$, 
    and no common atomic propositions,
    i.e., $\text{AP}_1 \cap \text{AP}_2 = \emptyset$.
Then $P_1 \parallel P_2$ is defined below:
\begin{equation}
P_1 \parallel P_2 
= \langle 
\text{AP}_1 \cup \text{AP}_2, 
(I_1 \cup I_2) \setminus (O_1 \cup O_2), 
O_1 \cup O_2, 
S_1 \times S_2, 
(s_0^1, s_0^2), 
T, 
L \rangle
\end{equation}
... where the transition relation $T$ is precisely the set of transitions $(s_1, s_2) \xrightarrow[]{x} (s_1', s_2')$ such that, for $i = 1, 2$, if the label $x \in I_i \cup O_i$ is a label of $P_i$, then $s_i \xrightarrow[]{x} s_i' \in T_i$, else $s_i = s_i'$.  $L : S_1 \times S_2 \to 2^{\text{AP}_1 \cup \text{AP}_2}$ is the function defined as $L(s_1, s_2) = L_1(s_1) \cup L_2(s_2)$.
\end{definition}

Intuitively, we define process composition to capture two primary ideas: (1) \emph{rendezvous communication}, meaning that a message is sent at the same time that it is received, and (2) \emph{multi-casting}, meaning that a single message could be sent to multiple parties at once. We can use so-called \emph{channel} processes to build asynchronous communication out of rendezvous communication (as we do in the next three sections), and we can easily preclude multi-casting by manipulating process interfaces.  Our definition therefore allows for a variety of communication models, making it flexible for diverse research problems.  However, as we explain shortly, in the context of handshakes, we look at one model setup which is common to transport protocols.

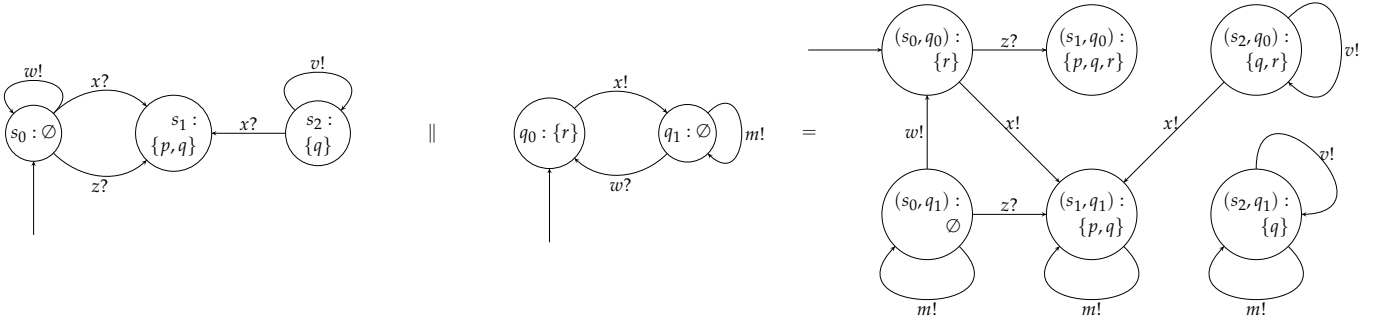
\begin{figure}[h]
\centering
\begin{adjustbox}{max totalsize={1.0\textwidth}{.5\textheight}}
\begin{tikzpicture}
\node[] (empty) {};
\node[draw,circle] (s0) [above=of empty] {\Huge $s_0 : \emptyset$};
\node[draw,circle] (s1) [right=of s0] {\Huge $\begin{aligned}s_1:\\\{ p, q \}\end{aligned}$};
\node[draw,circle] (s2) [right=of s1] {\Huge $\begin{aligned}s_2:\\\{ q \}\end{aligned}$};
\draw[straight] (empty) to (s0);
\draw[straight] (s2) to[above] node {\Huge $x?$} (s1);
\draw[looped] (s0) to[above,out=north east,in=north west,looseness=1] node {\Huge $x?$} (s1);
\draw[looped] (s0) to[below,out=south east,in=south west,looseness=1] node {\Huge $z?$} (s1);
\draw[looped] (s2) to[above,out=north west,in=north east,looseness=4] node {\Huge $v!$} (s2);
\draw[looped] (s0) to[above,out=north east,in=north west,looseness=4] node {\Huge $w!$} (s0);

\node[] (comp) [right=of s2] {\Huge $\parallel$};

\node[draw,circle] (q0) [right=of comp] {\Huge $q_0 : \{ r \}$};
\node[] (empty2) [below=of q0] {};
\node[draw,circle] (q1) [right=of q0] {\Huge $q_1 : \emptyset$};
\draw[straight] (empty2) to (q0);
\draw[looped] (q0) to[out=north east,in=north west,above] node {\Huge $x!$} (q1);
\draw[looped] (q1) to[out=south west,in=south east,below] node {\Huge $w?$} (q0);
\draw[looped] (q1) to[out=north east,in=south east,right,looseness=4] node {\Huge $m!$} (q1);

\node[] (equals) [right=of q1] {\Huge \, $=$};

\node[] (emptyFinal) [above=of equals] {};

\node[draw,circle] (s0q0) [right=of emptyFinal] 
    {\Huge $\begin{aligned}(s_0,q_0):\\\{ r \}\end{aligned}$};
\node[draw,circle] (s1q0) [right=of s0q0]
    {\Huge $\begin{aligned}(s_1,q_0):\\\{ p, q, r \}\end{aligned}$};
\node[draw,circle] (s2q0) [right=of s1q0]
    {\Huge $\begin{aligned}(s_2,q_0):\\\{ q, r \}\end{aligned}$};
\node[draw,circle] (s0q1) [below=of s0q0]
    {\Huge $\begin{aligned}(s_0,q_1):\\\emptyset\end{aligned}$};
\node[draw,circle] (s1q1) [right=of s0q1]
    {\Huge $\begin{aligned}(s_1,q_1):\\\{ p, q \}\end{aligned}$};
\node[draw,circle] (s2q1) [right=of s1q1]
    {\Huge $\begin{aligned}(s_2,q_1):\\\{ q \}\end{aligned}$};

\draw[straight] (emptyFinal) to (s0q0);

\draw[straight] (s0q0) to[above] node {\Huge $z?$} (s1q0);
\draw[straight] (s0q0) to[above] node {\Huge $\, \, x!$} (s1q1);
\draw[looped] (s2q0) to[out=north east,in=south east,looseness=4,right] node {\Huge $v!$} (s2q0);
\draw[straight] (s0q1) to[left] node {\Huge $w!$} (s0q0);
\draw[looped] (s2q1) to[out=north,in=east,looseness=4,right] node {\Huge $v!$} (s2q1);
\draw[straight] (s2q0) to[above] node {\Huge $x! \, \,$} (s1q1);
\draw[straight] (s0q1) to[above] node {\Huge $z?$} (s1q1);

\draw[looped] (s0q1) to[out=south east,in=south west,below,looseness=4] node {\Huge $m!$} (s0q1);
\draw[looped] (s1q1) to[out=south east,in=south west,below,looseness=4] node {\Huge $m!$} (s1q1);
\draw[looped] (s2q1) to[out=south east,in=south west,below,looseness=4] node {\Huge $m!$} (s2q1);
\end{tikzpicture}
\end{adjustbox}
\caption{Left is a process $P$ with atomic propositions $\text{AP} = \{ p, q \}$, inputs $I = \{ x, z \},$ outputs $O = \{ v, w \},$ states $S = \{ s_0, s_1, s_2 \},$ transition relation $T = \{ (s_0, w, s_0), (s_0, x, s_1), (s_0, z, s_1), (s_2, x, s_1), (s_2, v, s_2) \},$ and labeling function $L$ where $L(s_0) = \emptyset$, $L(s_1) = \{ p, q \},$ and $L(s_2) = \{ q \}$.  
Center is a process $Q = \langle \{ r \}, \{ w \}, \{ x, m \}, \{ q_0, q_1 \}, q_0, \{ (q_0, x, q_1), (q_1, m, q_1), (q_1, w, q_0) \}, L_Q \rangle$ where $L_Q(q_0) = \{ r \}$ and $L_Q(q_1) = \emptyset$.  Processes $P$ and $Q$ have neither common atomic propositions ($\{ p, q \} \cap \{ r \} = \emptyset$), nor common outputs ($\{ w, v \} \cap \{ x, m \} = \emptyset$), so the composition $P \parallel Q$ is well-defined.  Right is the process $P \parallel Q$.  Although $P \parallel Q$ is rather complicated, its only reachable states are $(s_0, q_0), (s_1, q_0),$ and $(s_1, q_1)$, and its only run is $r = \big( (s_0, q_0), x, (s_1, q_1) \big), \big( (s_1, q_1), m, (s_1, q_1) \big)^{\omega}$.  Non-obviously, the only computation of $P \parallel Q$ is $\sigma = \{ r \}, \{ p, q \}^{\omega}$.}
\Figl{exampleComposition}
\end{figure}

A state of the composite process $P_1 \parallel P_2$ is a pair $(s_1, s_2)$ consisting of a state $s_1 \in S_1$ of $P_1$ and a state $s_2 \in S_2$ of $P_2$.  The initial state of $P_1 \parallel P_2$ is a pair $(s_0^1, s_0^2)$ consisting of the initial state $s_0^1$ of $P_1$ and the initial state $s_0^2$ of $P_2$.  The inputs of the composite process are all the inputs of $P_1$ that are not outputs of $P_2$, and all the inputs of $P_2$ that are not outputs of $P_1$.  The outputs of the composite process are the outputs of the individual processes.
$P_1 \parallel P_2$ has three kinds of transitions $(s_1, s_2) \xrightarrow[]{z} (s_1', s_2')$.  In the first case, $P_1$ may issue an output $z$.  If this output $z$ is an input of $P_2$, then $P_1$ and $P_2$ move simultaneously and $P_1 \parallel P_2$ outputs $z$.  Otherwise, $P_1$ moves, outputting $z$, but $P_2$ stays still (so $s_2 = s_2'$).  The second case is symmetric to the first, except that $P_2$ issues the output.  In the third case, $z$ is neither an output for $P_1$ nor for $P_2$.  If $z$ is an input for both, then they synchronize.  Otherwise, whichever process has $z$ as an input moves, while the other stays still.

Note that sometimes rendezvous composition is defined to match $s_1 \xrightarrow[]{z?} s_1'$ with $s_2 \xrightarrow[]{z!} s_2'$ to form a {\em silent} transition $(s_1, s_2) \xrightarrow[]{} (s_1', s_2')$, but with our definition the output is preserved, so the composite transition would be $(s_1, s_2) \xrightarrow[]{z!} (s_1', s_2')$.  This allows for \emph{multi-casting}, where an output event of one process can synchronize with multiple input events from multiple other processes.  It also means there are no silent transitions.
A major benefit of multi-casting is that the composition operator can be commutative (up to isomorphism) and associative.

The labeling function $L$ is total as $L_1$ and $L_2$ are total. Since we required the processes $P_1, P_2$ to have disjoint sets of atomic propositions, $L$ does not change the logic of the two processes under composition.
Additionally, $\parallel$ is commutative and associative \cite{SIGACT17}.

Naturally, we can project a process onto a Kripke Structure by removing its inputs and outputs.
That is to say, the projection of a process $P = \langle \text{AP}, I, O, S, s_0, T, L \rangle$ 
is precisely the finite Kripke Structure $K_P = \langle \text{AP}, S, s_0, \{ (s, s') \mid \exists x \in I \cup O \,::\, (s, x, s') \in T \}, L \rangle$.  This is useful because it means we can model a system consisting of multiple interacting components using processes, then compute the composition thereof, project it onto a finite Kripke structure, and model check the result.  In this chapter, we do exactly that for the TCP, DCCP, and SCTP protocol handshakes.

We use a common model setup throughout, which we illustrate in \Figr{common:network:model}.
The setup consists of two protocol peers with isomorphic process logic, each connected to the other
by a unidirectional channel.
The channel has a size one buffer, meaning, it receives a message, and then waits to deliver it.
When we describe a protocol peer we do so generically, for example, saying that it ``sends SYN'' or ``receives ACK'', but
on paper, each output of each peer encodes the identity of the peer who sent it, e.g., peer A could send $\text{SYN}_A$.
This is how we are able to define two peers which apparently have the same inputs and outputs without 
violating our composition definition.

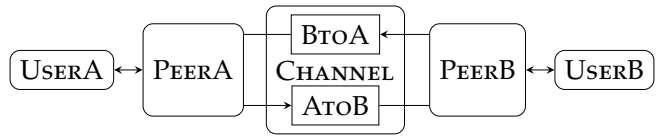
\begin{figure}[h]
\begin{adjustbox}{width=0.45\textwidth,center}
\begin{tikzpicture}
\node[draw, rectangle, rounded corners, minimum width=1.8cm, minimum height=1.8cm, fill=white] (channel) at (3,0) 
    {\small \textsc{Channel}};
\node[draw,rectangle] (AtoB) at (3,-0.5) {\small \textsc{AtoB}};
\node[draw,rectangle] (BtoA) at (3,0.5) {\small \textsc{BtoA}};
\draw[->] (0.5,-0.5) to (AtoB);
\draw[->] (AtoB) to (4.49,-0.5);
\draw[->] (5,0.5) to (BtoA);
\draw[->] (BtoA) to (1.49,0.5);
\node[draw, rectangle, rounded corners, minimum height=1.3cm,fill=white] (sender) at (1,0) 
    {\small \textsc{PeerA}};
\node[draw, rectangle, rounded corners, minimum height=1.3cm,fill=white] (receiver) at (5,0) 
    {\small \textsc{PeerB}};
\node[draw,rectangle,rounded corners,left=0.4cm of sender,align=center] (userA) {\small \textsc{UserA}};
\node[draw,rectangle,rounded corners,right=0.4cm of receiver,align=center] (userB) {\small \textsc{UserB}};
\draw[<->] (userA) to (sender);
\draw[<->] (userB) to (receiver);
\end{tikzpicture}
\end{adjustbox}
\caption{The system $\textsc{UserA} \parallel \textsc{PeerA} \parallel \textsc{Channel} \parallel \textsc{PeerB} \parallel \textsc{UserB}$. Processes are shown in rectangles, and arrows indicate communication direction, i.e., an arrow $A \xrightarrow[]{} B$ indicates that an output of $A$ is an input of $B$.  \textsc{Channel} contains a size-1 FIFO buffer in each direction (AtoB and BtoA, respectively).  The internal buffer is used to model delay.  The user processes are nondeterministic and simply transmit user commands to the peers.  Each of the peers runs the protocol handshake state machine, which takes as input user commands and messages from the other peer.}
\Figl{common:network:model}
\end{figure}

For our analyses of TCP, DCCP, and SCTP, we write properties which relate the current state of each peer
to its prior state (where record-keeping happens after each transition).
On paper, the way this is done is by transforming the process 
\(\langle \text{AP}, I, O, S, s_0, T, L \rangle\)
into 
\(\langle \text{AP} \uplus S^2, I, O, S^2, s_0, T', L' \rangle\)
where
\(T'((s_a, s_b), (s_c, s_d))\) holds iff $s_b = s_c$ and $T(s_b, s_d)$,
and
\(L'((s_a, s_b)) = L(s_b) \cup \{ (s_a, s_b) \}\).
But in \promela, the manipulation is much easier: we simply define the variables
\begin{lstlisting}
int state[2];
int before_state[2];
\end{lstlisting}
and then update them after each transition, e.g., for Peer A, upon transitioning into the state \texttt{DCCP\_Request}:
\begin{lstlisting}
REQUEST:
    before_state[0] = state[0];
    state[0] = RequestState;
\end{lstlisting}

In our DCCP model, we include a boolean state variable \texttt{active} encoding the role of the peer in the current association.  The formal state-space of one peer in the model is the Cartesian product of the list of DCCP state names and the possible values of \texttt{active} (true or false).

We also use \promela's \texttt{timeout} feature in our TCP and DCCP models.
This is a special transition type which allows a transition to occur only when,
if the transition did not exist, the system would deadlock.
The transition is implemented by adding an additional proposition to the global labeling function~$L$,
encoding whether or not the system can progress from its current state without the \texttt{timeout} transition,
and then predicating the transition on the negation of this proposition~\cite{timeouts}.

Another syntactic sugar we use in our diagrams is the notion of implicit states.
Essentially, if a protocol peer first sends message A, then receives message B, before transitioning to a new state,
the formal process needs to transition after A and before B to an implicit state (awaiting B).
When we show protocol models diagrammatically, we elide these states, instead just stating the sequence of send 
and receive operations that must occur in order for the peer to enter its ultimate destination.

Finally, we use so-called $\epsilon$-transitions in all of our models.  These are transitions without inputs or outputs.
On paper, an $\epsilon$-transition can be encoded as a transition which outputs a special symbol which is not
an input to any process in the system (say, $\epsilon$).  We typically leave $\epsilon$-transitions unlabeled when we portray processes diagrammatically.

With these mathematical details out of the way, we next define and
model check three concrete systems: TCP, DCCP, and SCTP.
We describe each model in detail as well as the properties we verify.

\section{Formal Model of the Transmission Control Protocol Handshake}\Secl{handshakes:tcp}

Recall that at least one peer must take an active role in the TCP establishment and teardown routines,
however, the peer which is active during establishment does not need to be the active one during teardown.
The active participant is the one who initiates the routine, by sending a \SYN in the case
of establishment, or a \FIN in the case of teardown.
The full TCP packet type grammar is 
\(\textit{msg} ::= \texttt{SYN} \mid \texttt{ACK} \mid \texttt{FIN}\).
Note, for simplicity, we model the message \texttt{SYN\_ACK} as the pair of messages $\texttt{SYN}, \texttt{ACK}$ and handle both possible orderings.
So, in our model, each message consists of just its type (and nothing else).

Our formal model is illustrated in \Figr{tcp}.
The user and user commands in this model are completely abstracted.
The model has eleven states, described below.
\begin{itemize}
  \item CLOSED -- This is described in the TCP RFC as a ``fictional state'' in which no association exists~\cite{rfc9293_tcp_new}.
  \item LISTEN -- The peer decided to take a passive role during establishment and is waiting to receive a \texttt{SYN} from the active participant.
  \item SYN\_SENT -- The peer decided to take an active role during establishment, sent a \texttt{SYN} to the other participant, and is waiting for either an \texttt{ACK} (indicating the other peer is taking a passive role) or a \texttt{SYN} (indicating the other peer also decided to be active).
  \item SYN\_RECEIVED -- The peer transitioned here from LISTEN or SYN\_SENT after receiving a \texttt{SYN}.
  It expects to receive an \texttt{ACK}, before it transitions to ESTABLISHED.
  \item ESTABLISHED -- The peer has established an association and can communicate.
  \item FIN\_WAIT\_1 -- The peer has begun the active role in the teardown routine.
  \item CLOSE\_WAIT -- The peer has begun the passive role in the teardown routine.
  \item CLOSING -- The peer is half-way through active/active teardown.
  \item FIN\_WAIT\_2 -- The peer is halfway through the active role in active/passive teardown.
  \item TIME\_WAIT -- The peer has completed active teardown and is giving the other participant time to complete its teardown.
  \item LAST\_ACK -- The peer is waiting to receive one last \texttt{ACK} in order to conclude passive teardown.
\end{itemize}

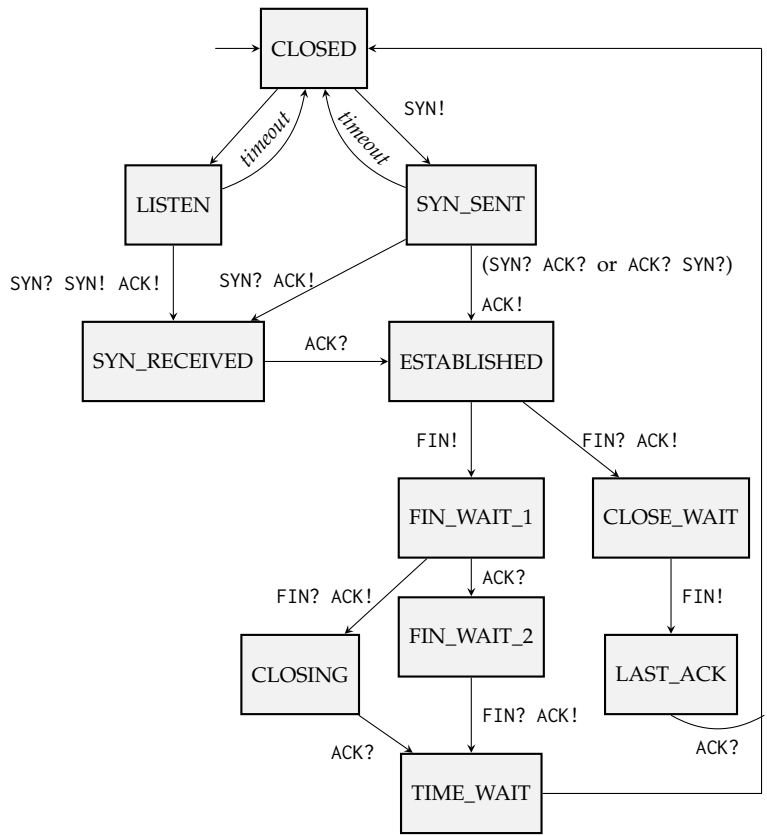
\begin{figure}[H]
\centering
    \begin{tikzpicture}
        \node[state, initial] (CLOSED) {\scriptsize CLOSED};
        \node[state, below right=1cm and 0.5cm of CLOSED] (SYN_SENT) {\scriptsize SYN\_SENT};
        \node[state, below left=1cm and 0.5cm of CLOSED] (LISTEN) {\scriptsize LISTEN};
        \node[state, below=1cm and 0.5cm of LISTEN] (SYN_RECEIVED) {\scriptsize SYN\_RECEIVED};
        \node[state, below=1cm and 0.5cm of SYN_SENT] (ESTABLISHED) {\scriptsize ESTABLISHED};
        \node[state, below=1cm and 0.5cm of ESTABLISHED] (FIN_WAIT_1) {\scriptsize FIN\_WAIT\_1};
        \node[state, below right=1cm and 0.5cm of ESTABLISHED] (CLOSE_WAIT) {\scriptsize CLOSE\_WAIT};
        \node[state, below left=1cm and 0.5cm of FIN_WAIT_1] (CLOSING) {\scriptsize CLOSING};
        \node[state, below=0.5cm of FIN_WAIT_1] (FIN_WAIT_2) {\scriptsize FIN\_WAIT\_2};
        \node[state, below=1cm of CLOSE_WAIT] (LAST_ACK) {\scriptsize LAST\_ACK};
        \node[state, below=1cm of FIN_WAIT_2] (TIME_WAIT) {\scriptsize TIME\_WAIT};

        \draw[->] (CLOSED) edge[left]  (LISTEN);
        \draw[->] (CLOSED) edge[above right, text width=2.2cm] 
                  node{\scriptsize \texttt{SYN!}} (SYN_SENT);
        \draw[->] (LISTEN) edge[left] node{\scriptsize
            \texttt{SYN? SYN! ACK!}} (SYN_RECEIVED);
        \draw[->] (LISTEN) edge[above, bend right] node[rotate=48]{\scriptsize \textit{timeout}} (CLOSED);
        \draw[->] (SYN_SENT) edge[right, text width=3.5cm] node{\scriptsize (\texttt{SYN? ACK?} or \texttt{ACK? SYN?}) \texttt{ACK!}} (ESTABLISHED);
        \draw[->] (SYN_SENT) edge[left] node{\scriptsize \texttt{SYN? ACK!}} (SYN_RECEIVED);
        \draw[->] (SYN_SENT) edge[above, bend left] node[rotate=-48]{\scriptsize \textit{timeout}} (CLOSED);
        \draw[->] (SYN_RECEIVED) edge[above] node{\scriptsize \texttt{ACK?}} (ESTABLISHED);
        \draw[->] (ESTABLISHED) edge[left] node{\scriptsize \texttt{FIN!}} (FIN_WAIT_1);
        \draw[->] (ESTABLISHED) edge[right] node{\scriptsize \texttt{FIN? ACK!}} (CLOSE_WAIT);
        \draw[->] (FIN_WAIT_1) edge[left] node{\scriptsize \texttt{FIN? ACK!}} (CLOSING);
        \draw[->] (FIN_WAIT_1) edge[right] node{\scriptsize \texttt{ACK?}} (FIN_WAIT_2);
        \draw[->] (CLOSE_WAIT) edge[right] node{\scriptsize \texttt{FIN!}} (LAST_ACK);
        \draw[->] (FIN_WAIT_2) edge[right] node{\scriptsize \texttt{FIN? ACK!}} (TIME_WAIT);
        \draw[->] (CLOSING) edge[below left] node{\scriptsize \texttt{ACK?}} (TIME_WAIT);
        \draw (LAST_ACK.south) edge[below,bend right] node[below]{\scriptsize \texttt{ACK?}} 
              ([xshift=1.24cm]LAST_ACK.south);
        \draw[->] (5.95,0) -- (CLOSED);
        \draw (TIME_WAIT) -| (5.95,0);
    \end{tikzpicture}
    \caption{\textbf{TCP Model}.  States are shown in boxes; the initial state is $\CLOSED$ and has an incoming arrow to indicate it is initial.  Transitions are shown in labeled edges between states.  Technically, any transition with more than one event on it actually amounts to multiple transitions in the process, with some implicit states in-between them.  For instance, the transition from \SYNSENT to \ESTABLISHED with label $\SYN? \ACK? \ACK!$ is actually encoded as the sequence of transitions $\SYNSENT \xrightarrow[]{\SYN?} s_a \xrightarrow[]{\ACK?} s_b \xrightarrow[]{\ACK!} \ESTABLISHED$ where $s_a$ and $s_b$ are implicit.}
    \Figl{tcp}
\end{figure}

\section{Properties of the Transmission Control Protocol Handshake}

We derived the following formal correctness properties from RFC 9293~\cite{rfc9293_tcp_new}.
Using \spin, we verified that the system consisting of two TCP participants satisfies all of these properties.

\begin{description}

\item \textbf{$\phi_1$: No half-open connections.}  According to $\mathsection$3.5.1. of the RFC, half-open connections, in which one peer is in \ESTABLISHED while the other is in \CLOSED, are considered anomalous and expected to only occur in the context of crashes (and crash recovery).  Since we do not model crashes, it follows that such scenarios should be impossible in our model.

\item \textbf{$\phi_2$: Passive/active establishment eventually succeeds.}  The RFC describes TCP as enabling peers to reliably exchange information.  But if the peers cannot establish a connection, then this is impossible.  Passive/active establishment is the default establishment mode, so in order for TCP to ``work'' property by default, it should eventually succeed.

\item \textbf{$\phi_3$: Peers don't get stuck.}  The RFC explicitly states that the TCP handshake was designed to avoid deadlocks, in $\mathsection$3.8.6.2.1, 3.9.1.2, and 3.9.1.3.  More generally, a deadlock in the handshake would constitute some kind of crash or DoS.

\item \textbf{$\phi_4$: \SYNREC is eventually followed by \ESTABLISHED.}  Follows from the establishment routines described in 3.5 as well as the Reset Processing logic outlined in 3.5.3.  Intuitively, the property says that the passive peer in active/passive establishment progresses through active/passive establishment.

\end{description}

\section{Formal Model of the Datagram Congestion Control Protocol Handshake}\Secl{handshakes:dccp}

Our formal model is illustrated in \Figr{dccp}.
Note, in DCCP, unexpected messages are automatically dropped.
We implemented this detail in our model but elide it in \Figr{dccp}
to avoid clutter.
The full DCCP packet type grammar is given in Eqn.~\ref{eqn:dccp-msgs}.

\begin{equation}
\begin{aligned}
\emph{msg} ::= & \texttt{DCCP\_REQUEST} \mid \texttt{DCCP\_RESPONSE} \mid 
\texttt{DCCP\_RESET} \mid \texttt{DCCP\_SYNC} \mid \texttt{DCCP\_ACK} \\
& \mid
\texttt{DCCP\_DATA} \mid \texttt{DCCP\_DATAACK} \mid \texttt{DCCP\_CLOSE} \mid
\texttt{DCCP\_CLOSEREQ}
\end{aligned}
\label{eqn:dccp-msgs}
\end{equation}

The model has nine states, described below.
\begin{itemize}
  \item CLOSED -- Much like in TCP, this is described as representing ``nonexistent connections''~\cite{rfc4340_dccp}.
  \item LISTEN -- The beginning of the passive establishment routine.
  \item REQUEST -- The beginning of the active establishment routine.
  \item RESPOND -- Step two of the passive establishment routine.
  \item PARTOPEN -- Step two of the active establishment routine.
  \item OPEN -- Equivalent to ESTABLISHED in TCP, represents the state where an association exists and the peer can communicate data.
  \item CLOSING -- The beginning  of the passive teardown routine.
  Also possible for an active peer if they request immediate teardown.
  \item CLOSEREQ -- The only state in the active teardown routine.
  \item TIMEWAIT -- Similar to the identically named state in the TCP machine.  Represents the final step in the passive teardown routine.
\end{itemize}

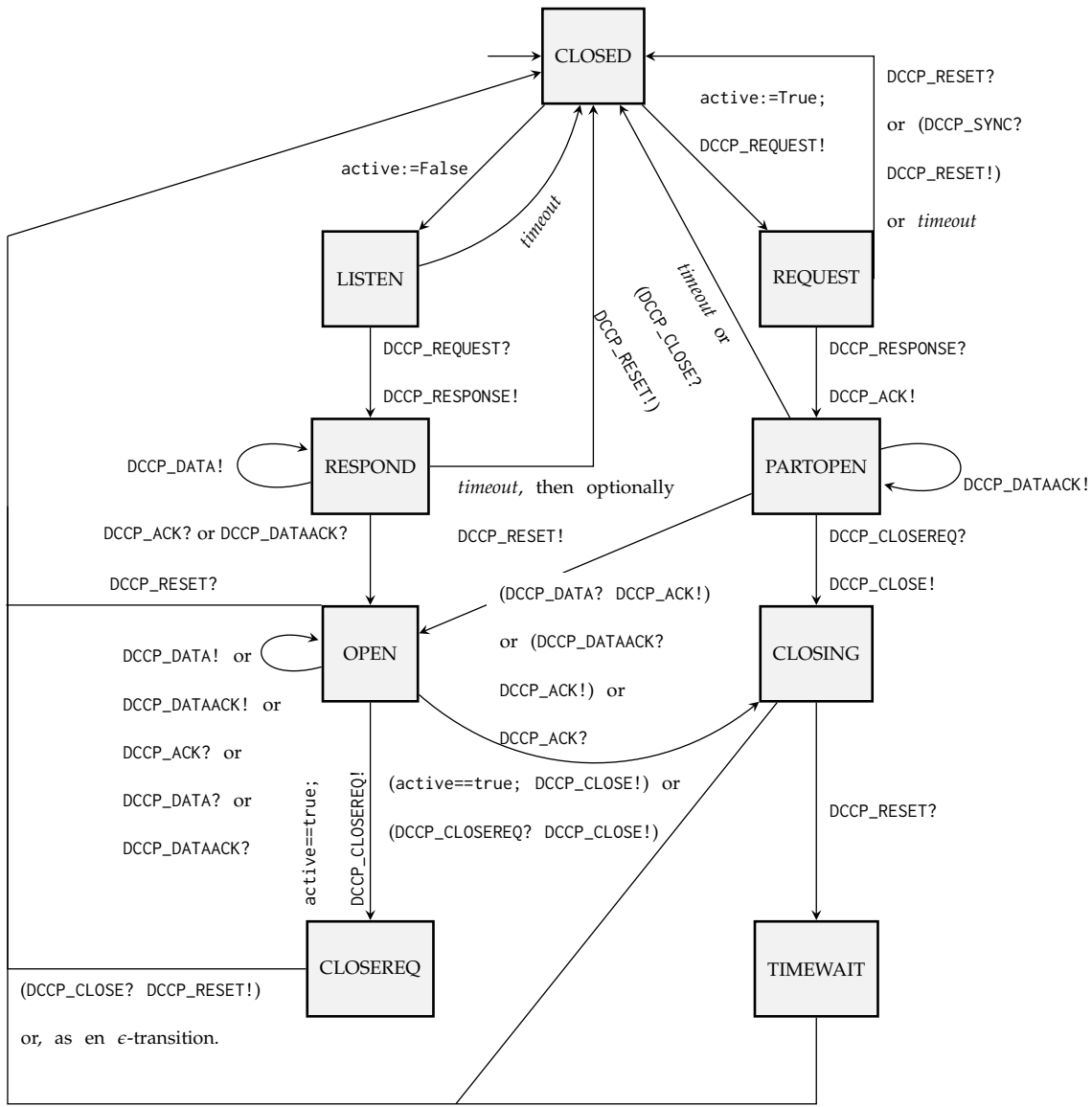
\begin{figure}[H]
\centering
\begin{adjustbox}{width=0.8\textwidth}
\centering
    \begin{tikzpicture}
        \node[state, initial] (CLOSED) {\tiny CLOSED};
        \node[state, below left of=CLOSED] (LISTEN) {\tiny LISTEN};
        \node[state, below right of=CLOSED] (REQUEST) {\tiny REQUEST};
        \node[state, below=1cm and 0.5cm of LISTEN] (RESPOND) {\tiny RESPOND};
        \node[state, below=1cm and 0.5cm of REQUEST] (PARTOPEN) {\tiny PARTOPEN};
        \node[state, below=1cm and 0.5cm of RESPOND] (OPEN) {\tiny OPEN};
        \node[state, below=1cm and 0.5cm of PARTOPEN] (CLOSING) {\tiny CLOSING};
        \node[state, below of=OPEN] (CLOSEREQ) {\tiny CLOSEREQ};
        \node[state, below of=CLOSING] (TIMEWAIT) {\tiny TIMEWAIT};

        \draw[->] (CLOSED) edge[left] node{\tiny \texttt{active:=False}} (LISTEN);
        \draw[->] (CLOSED) edge[above right, text width=1.5cm] 
                       node[xshift=-0.2cm]{\tiny \texttt{active:=True; DCCP\_REQUEST!}} 
              (REQUEST);
        \draw[->] (LISTEN) edge[right,text width=1.5cm] 
                       node{\tiny \texttt{DCCP\_REQUEST? DCCP\_RESPONSE!}} 
              (RESPOND);
        \draw[->] (LISTEN) edge[below, bend right] node[rotate=55]{\tiny \textit{timeout}} 
              (CLOSED);
        \draw[->] (REQUEST) edge[right,text width=1.5cm]
                         node{\tiny \texttt{DCCP\_RESPONSE? DCCP\_ACK!}}
              (PARTOPEN);
        \draw[->] (REQUEST.east) |- node[below right, text width=1.5cm]
                             {\tiny \texttt{DCCP\_RESET?} or
                              (\texttt{DCCP\_SYNC? DCCP\_RESET!}) or
                              \textit{timeout}}
              (CLOSED.east);
        \draw[->] (RESPOND) edge[left=3cm]
                        node[xshift=-0.1cm,yshift=0.3cm]{\tiny \texttt{DCCP\_ACK?} or \texttt{DCCP\_DATAACK?}}
              (OPEN);
        \draw[->] (RESPOND) -| node[below, text width=3cm]
                           {\tiny \textit{timeout}, then optionally
                           \texttt{DCCP\_RESET!}}
              (CLOSED);
        \draw[->] (RESPOND) edge[loop left] node{\tiny \texttt{DCCP\_DATA!}} (RESPOND);
        \draw[->] (PARTOPEN) edge[below,text width=5cm] 
                         node[text width=2.5cm,xshift=0.3cm,yshift=-0.1cm,fill=white]
                         {\tiny (\texttt{DCCP\_DATA? DCCP\_ACK!}) or (\texttt{DCCP\_DATAACK? DCCP\_ACK!}) or \texttt{DCCP\_ACK?}} 
              (OPEN); 
        \draw[->] (PARTOPEN) edge[loop right] node[xshift=-0.1cm,yshift=-0.2cm]
        {\tiny \texttt{DCCP\_DATAACK!}} (PARTOPEN);
        \draw[->] (PARTOPEN) edge[below] node[text width=2cm,rotate=-60,xshift=0.9cm]
                         {\tiny \textit{timeout} or
                          (\texttt{DCCP\_CLOSE? DCCP\_RESET!})} 
              (CLOSED);
        \draw[->] (PARTOPEN) edge[right, text width=1.8cm] 
                         node{\tiny \texttt{DCCP\_CLOSEREQ? DCCP\_CLOSE!}} 
              (CLOSING);

        \draw[->] (OPEN) edge[loop left,text width=1.9cm,below left] 
                     node[xshift=0.5cm,yshift=0.2cm]{\tiny \texttt{DCCP\_DATA!}~or
                                      \texttt{DCCP\_DATAACK!}~or
                                      \texttt{DCCP\_ACK?}~or
                                      \texttt{DCCP\_DATA?}~or
                                      \texttt{DCCP\_DATAACK?}}
              (OPEN);
        \draw[->] (OPEN) edge[left,text width=1.9cm,left]
                     node[rotate=90,xshift=1cm,yshift=0.4cm]{\tiny \texttt{active==true;}
                                      \texttt{DCCP\_CLOSEREQ!}}
              (CLOSEREQ);
        \draw[-] (OPEN.north west) to node[above] {\tiny \texttt{DCCP\_RESET?}} ([xshift=-3.5cm]OPEN.north west);

        \draw[->] (OPEN) edge[bend right=40] node[below,text width=4cm,xshift=-0.2cm]
                                    {\tiny (\texttt{active==true; DCCP\_CLOSE!})~or
                                                 (\texttt{DCCP\_CLOSEREQ?} \texttt{DCCP\_CLOSE!})}
                  (CLOSING);
        \draw (CLOSEREQ.west) -| node[below right,text width=3cm]
                                  {\tiny (\texttt{DCCP\_CLOSE? DCCP\_RESET!})\\ or,
                                       as en $\epsilon$-transition.}
                  (-6.5,-5);
        \draw[->] (CLOSING) to node[right] {\tiny \texttt{DCCP\_RESET?}} (TIMEWAIT);
        \draw (TIMEWAIT) to 
              ([yshift=-1.5cm]TIMEWAIT)
                         -| 
              (-6.5,-2);
        \draw ([xshift=0.2cm]CLOSING.south west) to ([shift=({-4cm,-1.5cm})]TIMEWAIT);
        \draw[->] (-6.5,-2) to (CLOSED);
    \end{tikzpicture}
    \end{adjustbox}
    \caption{\textbf{DCCP Model}.  Note the variable \texttt{active}, used to encode whether the peer is playing the active role or the passive role in the connection.  In practice, this variable doubles the state-space, since the true set of states in the model is the produce of the list of state names and the set $\{ \emph{true}, \emph{false} \}$ of possible assignments of \texttt{active}.}
    \Figl{dccp}
\end{figure}

\section{Properties of the Datagram Congestion Control Protocol Handshake}

We verify all of the following properties of DCCP using \spin.

\begin{description}
\item \textbf{$\theta_1$: The peers don't both loop into being stuck or infinitely looping.}  This is implied by the fact that ``DCCP peers progress through different states during the course of a connection'' ($\mathsection$4.3).  Essentially, a DCCP peer should never transition out of a state and then back into it -- let alone get stuck doing so forever.

\item \textbf{$\theta_2$: The peers are never both in \TIMEWAIT.}  According to $\mathsection$4.3, ``Only one of the endpoints has to enter \TIMEWAIT state (the other can enter \CLOSED state immediately)''.  Moreover, the message sequence charts in the RFC do not show any situation in which both enter \TIMEWAIT at once.  Rather, \TIMEWAIT is described as a mechanism that one peer uses to make sure the other has gracefully closed.

\item \textbf{$\theta_3$: The first peer doesn't loop into being stuck or infinitely looping.}  This is a slightly weaker version of $\theta_1$.

\item \textbf{$\theta_4$: The peers are never both in \CLOSEREQ.}  $\mathsection$4.3 says that a server enters this state from \OPEN.  Since DCCP has no active/active routine, the property logically follows.  That is to say: \CLOSEREQ is only used in active teardown, and only one peer can take the active role during teardown, clearly it cannot be the case that both peers are simultaneously in \CLOSEREQ.

\end{description}

\section{Formal Model of the Stream Control Transmission Protocol Handshake}\Secl{handshakes:sctp}

Our formal model includes timers, out-of-the-blue packet handling, unexpected packet handling, and initiation and verification tags, in addition to the standard handshake state machine logic.

\paragraph{Timers.}
The SCTP connection routines use three timers: Init, Cookie, and Shutdown.
The goal of the Init Timer is to stop the active peer in an establishment routine from getting stuck
waiting forever for the passive peer to respond to its \Init with an \InitAck.
The goal of the Cookie Timer is similar: it stops that same active peer from getting stuck waiting forever for the passive peer to respond to its \CookieEcho.  
The Shutdown Timer plays a similar role but in the teardown routine, stopping the active peer in teardown from getting stuck waiting for a \ShutdownAck.
In our model, each timer is modeled using a boolean variable which is set to \emph{true} iff the timer is enabled.
Whenever a timer is set to \emph{true}, the corresponding nondeterministic ``timeout'' transition is enabled.

\paragraph{Out-of-the-Blue Packet Handling.}
In SCTP a message is considered \emph{out-of-the-blue} (OOTB) if the recipient cannot determine to which association the message belongs, i.e., if it has an incorrect vtag, or is an \Init with a zero-valued itag.
Specifically, an OOTB message will be discarded if:
1) it was not sent from a unicast IP, 
2) it is an \Abort with an incorrect vtag, 
3) it is an \Init with a zero itag or incorrect vtag\footnote{Per RFC 4960, respond with an \Abort having the vtag of the current association.  But per RFC 9260, discard it.},
4) it is a \CookieEcho, \ShutdownComplete, or \CookieError, 
    and is either unexpected in the current state or has an incorrect vtag, or
5) it has a zero itag or incorrect vtag.
We model each of these checks on every \emph{receive} event.
Therefore, in \Figr{sctp:fsm}, the notation $X?$ is shorthand for $X? \land \neg \text{OOTB}(X)$.

\paragraph{Unexpected Packet Handling.}
A message is \emph{unexpected} if it is not OOTB, but nevertheless, the recipient does not expect it.
SCTP handles unexpected packets as described in Algr.~\ref{alg:pktHandler}.

\begin{algorithm}
\caption{Unexpected Packet Handling}
\label{alg:pktHandler}
\begin{algorithmic}
\REQUIRE Unexpected \textit{msg}
\IF{\textit{msg.chunk} = \Init}
    \IF{\textit{state} = \CookieWait \OR \textit{msg} does not indicate new addresses added}
        \STATE Send \InitAck with \textit{vtag} = \textit{msg.itag}
    \ELSE
        \STATE Discard \textit{msg} and send \Abort with \textit{vtag} = \textit{msg.itag}
    \ENDIF
\ELSIF{\textit{msg.chunk} = \CookieEcho}
    \IF{\textit{msg.timestamp} is expired}
        \STATE Send \CookieError
    \ELSIF{\textit{msg} has fresh parameters}
        \STATE Form a new association
    \ELSE 
        \STATE Set \textit{vtag} = \textit{msg.vtag} {\color{gray}{\texttt{// initialization collision}}}
        \STATE \textbf{goto} Established
    \ENDIF
\ELSIF{\textit{msg.chunk} = \ShutdownAck}
    \STATE Send \ShutdownComplete with \textit{vtag} = \textit{msg.vtag}
\ELSE
    \STATE Discard \textit{msg}
\ENDIF
\end{algorithmic}
\end{algorithm}

\paragraph{Packet Verification and Invalid Packet Defenses.}
We model each SCTP message as consisting of a message chunk, 
a vtag, and an itag.
Each of these components are modeled using enums, which in \textsc{Promela} are called \texttt{mtype}s.
The message chunk denotes the meaning of the message, e.g., a message with an \Init chunk is called an \emph{initiate message} and is used to initiate a connection establishment routine. 
The itag and vtag are used to verify the authenticity of the sender of the message.  In our model there are three kinds of tags: expected (\texttt{E}), unexpected (\texttt{U}), or none (\texttt{N}).  A tag is expected if (1) it is a non-zero itag on an \Init or \InitAck chunk, or (2) it is the other peer's vtag in the existing association.  Otherwise, it is unexpected.  The none type is reserved for packets that do not carry the given tag type -- e.g., only \Init and \InitAck chunks carry an itag, so in the other types of messages, the itag is \texttt{N}. The BNF grammar for messages in our model is given below.
\begin{equation}
\begin{aligned}
\textit{msg} & ::= \Init,\texttt{N},\textit{ex} \mid
                 \InitAck,\textit{ex},\textit{ex} \mid
                 \textit{ach},\textit{ex},\texttt{N} \\
\textit{ach} & ::= \Abort \mid \Shutdown \mid \ShutdownComplete \\
             & \mid \CookieEcho \mid \CookieAck \mid \ShutdownAck \\
             & \mid \CookieError \mid \Data \mid \DataAck \\
\textit{ex} & ::= \texttt{E} \mid \texttt{U}
\end{aligned}
\end{equation}
We also support an option where the \textit{msg} can be extended with a TSN.

Upon receiving a message, our model checks that the tags are set as expected,
depending on the message and state.
If a message has an unexpected tag then
the model employs the defenses specified in the RFC, e.g., silently discarding the message or responding with an \Abort. 

\paragraph{State Machine.}
After implementing the timers, OOTB logic, and unexpected packet handling, our SCTP model can be described
by the state machine illustrated in \Figr{sctp:fsm}.
Our SCTP model implements active/passive establishment and teardown, as well as active/active teardown, but not active/active establishment (a.k.a. ``\Init collision''), precisely as described previously and illustrated in \Figr{handshakes:sctp:activePassive}, with the caveat that the itag and vtag are abstracted (as described above).
We also capture the TSN proposal and use throughout an association, although this feature can be
turned off in our model to reduce the state-space
for more efficient verification.
Our model has eight states, described below.
\begin{itemize}
    \item \Closed -- Same as in TCP or DCCP, represents the state where no association exists.
    However, unlike in TCP or DCCP, in SCTP the \Closed state has a self-loop which acknowledges an \Init
    message with an \InitAck.  This allows SCTP to achieve passive establishment in a single transition from \Closed to \Established.
    \item \CookieWait -- Step one in the active establishment routine.
    \item \CookieEchoed -- Halfway through the active establishment routine.
    \item \Established -- An association exists and the peer and communicate data.
    \item \ShutdownReceived -- First step in passive teardown.
    \item \ShutdownPending -- First step in active teardown.
    \item \ShutdownAckSent -- Halfway through passive teardown, or an active role in active/active teardown.
    \item \ShutdownSent -- Halfway through active teardown.
\end{itemize}

\begin{figure}[H]
\begin{adjustbox}{width=0.9\textwidth,center}
\begin{tikzpicture}
\begin{scope}[every node/.style={rectangle,draw,minimum width=2em,minimum height=2em}]
    \node (CLOSED) at (1.75,5.36) {\tiny \Closed}; 
    \node (COOKIEWAIT) at (-1.5,3.3) {\tiny \CookieWait};
    \node (COOKIEECHOED) at (-1.5,1) {\tiny \CookieEchoed};
    \node (ESTABLISHED) at (1.75,-0.4) {\tiny \Established};
    \node (SHUTDOWNRCVD) at (6,1) {\tiny \ShutdownReceived};
    \node (SHUTDOWNPENDING) at (9,1) {\tiny \ShutdownPending};
    \node (SHUTDOWNSENT) at (9,3.3) {\tiny \ShutdownSent};
    \node (SHUTDOWNACKSENT) at (6,3.3) {\tiny \ShutdownAckSent};

\begin{scope}[>={stealth[black]},
              every node/.style={fill=white,rectangle},
              every edge/.style={draw=black,thick}]
    \path [->] (CLOSED) edge [loop above] 
              node[yshift=0.05cm] {\tiny \Init,\texttt{N},\texttt{E}? \InitAck,\texttt{E},\texttt{E}!} 
              (CLOSED);
    \path [->] (-0.5,5.36) edge 
              (CLOSED);
    \path [->] (CLOSED) edge
              node[left,xshift=-0.4cm] {\tiny \UserAssoc? \Init,\texttt{N},\texttt{E}!}
              (COOKIEWAIT);
    \path [->] (CLOSED) edge
              node[text width=1cm,yshift=-0.3cm,xshift=0.7cm,fill=none] 
              {\tiny \CookieEcho,\texttt{E},\texttt{N}?\\\CookieAck,\texttt{E},\texttt{N}!}
              (ESTABLISHED);
    \path [->] (COOKIEWAIT) edge 
              node[below right,text width=2cm,yshift=0.4cm,fill=none] 
              {\tiny \InitAck,\texttt{E},\texttt{E}?\\\CookieEcho,\texttt{E},\texttt{N}!} 
              (COOKIEECHOED);
    \path [->] (COOKIEECHOED) edge 
              node[yshift=0.1cm,xshift=-0.5cm,text width=1.1cm] {\tiny \CookieAck,\texttt{E},\texttt{N}?} 
              (ESTABLISHED);
    \path [->] (COOKIEECHOED) edge [loop left]
               node[below,yshift=-0.35cm,xshift=0.1cm] 
               {\tiny \CookieError,\texttt{E},\texttt{N}?}
               node[below,yshift=-0.7cm,xshift=0.1cm] 
               {\tiny \textit{then optionally,} \Init,\texttt{N},\texttt{E}!}
               (COOKIEECHOED);
    \path [->] (COOKIEECHOED.north west) edge [bend left]
               node[above left,xshift=-0.1cm] {\tiny\CookieError,\texttt{E},\texttt{N}?}
               node[below left,xshift=-0.1cm] {\tiny \Init,\texttt{N},\texttt{E}!}
               (COOKIEWAIT.south west);
    \path [->] (ESTABLISHED) edge [out=east,in=south,looseness=0.5] 
              node[xshift=-1cm,yshift=0.2cm,fill=none] {\tiny \UserShutdown?} 
              (SHUTDOWNPENDING);
    \path [->] (ESTABLISHED) edge 
              node[xshift=1.1cm,fill=none] {\tiny \Shutdown,\texttt{E},\texttt{N}?} 
              (SHUTDOWNRCVD.south west);
    \path [->] (SHUTDOWNPENDING) edge 
              node[yshift=-0.2cm,xshift=1cm] {\tiny \Shutdown,\texttt{E},\texttt{N}!} 
              (SHUTDOWNSENT);
    \path [->] (SHUTDOWNSENT) edge [bend right] 
              node[xshift=2.5cm,text width=3cm,fill=none] 
              {\tiny \ShutdownAck,\texttt{E},\texttt{N}?\\\ShutdownComplete,\texttt{E},\texttt{N}!}
              (CLOSED);
    \path [->] (SHUTDOWNSENT) edge 
              node[below,xshift=0.1cm,yshift=-0.4cm] {\tiny \Shutdown,\texttt{E},\texttt{N}?} 
              node[below,yshift=-0.7cm] {\tiny \ShutdownAck,\texttt{E},\texttt{N}!} 
              (SHUTDOWNACKSENT);
    \path [->] (SHUTDOWNRCVD) edge 
              node[below] {\tiny \ShutdownAck,\texttt{E},\texttt{N}!} 
              (SHUTDOWNACKSENT);
    \path [->] (SHUTDOWNACKSENT) edge [bend right] 
              node[above,xshift=-0.9cm,yshift=-1.5cm,text width=2cm,fill=none] 
              {\tiny \texttt{SHUTDOWN\_}\\\texttt{COMPLETE},\texttt{E},\texttt{N}?\\\textit{or} (\ShutdownAck,\texttt{E},\texttt{N}?\\\ShutdownComplete,\texttt{E},\texttt{N}!)} 
              (CLOSED);
\end{scope}
\end{scope}

\begin{scope}[on background layer]
    \fill[white] (0,0) rectangle (5,5);
  \end{scope}
\end{tikzpicture}
\end{adjustbox}
\caption{SCTP Model.  $x,v,i?$ (or $x,v,i!$) denotes receive (or send) chunk~$x$ with vtag~$v$ and itag~$i$.  Events in multi-event transitions occur in the order they are listed.  Logic for OOTB packets, \Abort messages or \UserAbort commands, unexpected user commands, and data exchange are ommitted but faithfully implemented in the model and described in this chapter.}
\Figl{sctp:fsm}
\end{figure}
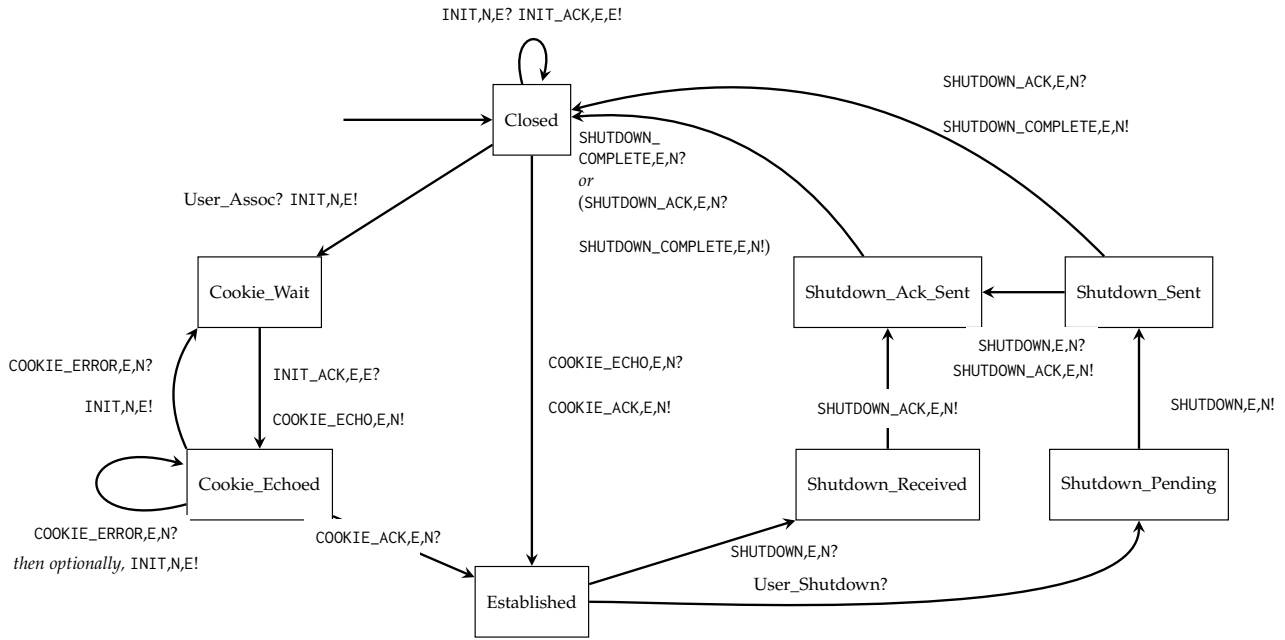

\section{Properties of the Stream Control Transmission Protocol Handshake}
We verify all ten of the following properties for our SCTP model.
\begin{description}

\item \textbf{$\gamma_1$: A peer in \Closed either stays still or transitions to \Established or \CookieWait.}
This is based on the routine described in $\mathsection$5.1, as well as the Association State Diagram in $\mathsection$4.  If a closed peer could transition to any state other than \Established or \CookieWait, it could de-synchronize with the other peer, breaking the four-way handshake and potentially leading to a deadlock, livelock, or other problem.
\item \textbf{$\gamma_2$: One of the following always eventually happens: the peers are both in \Closed, the peers are both in \Established, or one of the peers changes state.}
The property we want to capture here, ``no half-open connections'', is stated in $\mathsection$1.5.1,
was verified in the related work by Saini and Fehnker~\cite{saini2017evaluating},
and was studied for TCP in two prior works~\cite{von2020automated,pacheco2022automated}.
But we have to formalize it subtly, because 
in the case of an in-transit \Abort,
it is possible for one peer to temporarily be in \Established while the other is in \Closed; so we write it as a liveness property, saying half-open states eventually end.
\item \textbf{$\gamma_3$: If a peer transitions out of \ShutdownAckSent then it must transition into \Closed.}
We derived this from the Association State Diagram in $\mathsection$4.  Every transition out of \ShutdownAckSent described in the RFC ends up in either \Closed or \ShutdownAckSent.  If this property fails, it would imply a flaw in the graceful teardown routine, and could cause a deadlock, livelock, or other problem.
\item \textbf{$\gamma_4$: If a peer is in \CookieEchoed then its cookie timer is actively ticking.}
Per $\mathsection$5.1 C), the peer starts the cookie timer upon entering \CookieEchoed.  Per $\mathsection$4 step 3), when the timer expires it is reset, up to a fixed number of times, at which point the peer returns to \Closed.  
If the property fails, then the active peer in an establishment could get stuck in \CookieEchoed forever,
opening a new opportunity for DoS.
\item \textbf{$\gamma_5$: The peers are never both in \ShutdownReceived.}
This property follows from inspection of the Association State Diagram in $\mathsection$4.
From a security perspective, if both peers were in \ShutdownReceived, this would indicate that neither initiated the shutdown (yet both are shutting down); the only logical explanation for which is some kind of DoS.
\item \textbf{$\gamma_6$: If a peer transitions out of \ShutdownReceived then it must transition into either \ShutdownAckSent or \Closed.}
The transition to \ShutdownAckSent is shown in the Association State Diagram in $\mathsection$4.
The transition to \Closed can occur upon receiving either a \UserAbort from the user or an \Abort from the other peer.  No other transitions out of \ShutdownReceived are given in the RFC.  If this property fails, it could de-synchronize the teardown handshake, potentially leading to an unsafe behavior.  For example, if a peer transitioned from \ShutdownReceived to \Established, it would end up in a half-open connection.
\item \textbf{$\gamma_7$: If Peer~A is in \CookieEchoed then B must not be in \ShutdownReceived.}
We derived this from the Association Diagram in $\mathsection$4, which shows A must receive an \InitAck while in the \CookieWait and then send a \CookieEcho in order to transition into \CookieEchoed. B must have been in \Closed to send an \InitAck in the first place, hence B cannot be in \ShutdownReceived.  This property relates to the synchronization between the peers: if one is establishing a connection while the other is tearing down, then they are de-synchronized, and the protocol has failed.
\item \textbf{$\gamma_8$: Suppose that in the last time-step, Peer~A was in \Closed and Peer~B was in \Established.  Suppose neither user issued a \UserAbort, and neither peer had a timer time out.  Then if Peer A changed state, it must have changed to either \Established, or the implicit, intermediary state in \CookieWait in which it received \InitAck but did not yet transmit \CookieEcho.}
The transitions from \Closed to \Established and the described intermediary state are implicit in the Association State Diagram in $\mathsection$4.  The timer caveat is described in $\mathsection$4 step 2, and the aborting caveat is  in $\mathsection$9.1.  If the property fails, the four-way handshake ended, yet was not completed successfully, did not time out, and was not aborted, so somehow, the protocol failed.
\item \textbf{$\gamma_9$: The same as $\gamma_8$ but the roles are reversed.}
The property is: \emph{Suppose that in the last time-step, Peer~B was in \Closed and Peer~A was in \Established ...}

\item \textbf{$\gamma_{10}$: Once connection termination initiates, both peers eventually reach \Closed}. This follows from the description of connection termination in $\mathsection$9. Once connection termination is initiated, there is no way to recover the association.  In other words, termination is final.

\end{description}

\section{Related Work}

\paragraph{TCP} was previously formally studied using a process language called \textsc{SPEX} \cite{schwabe1981formal}, Petri nets \cite{han2005termination}, the \textsc{HOL} proof assistant \cite{bishop2018engineering}, and various other algebras (see Table~2.2 in \cite{Thesis2016}).  Our model is neither the most detailed nor the most comprehensive, but it captures the entire TCP handshake, including every possible establishment or teardown flow.

\paragraph{DCCP} was initially designed in an ad-hoc manner, however,
over the course of its maturation, its designers performed some analysis using
a semi-formal exhaustive state search tool as well as a Colored Petri Net model~\cite{kohler2006designing}.
These analyses revealed some bucks, e.g., a deadlock in connection establishment,
which the authors fixed before publishing RFC 4340,
and the Petri Net analysis resulted in 
multiple publications~\cite{vanit2004initial,vanit2005discovering,gallasch2005sweep,vanit2016validating}.
To the best of our knowledge, ours is the first process model of DCCP amenable to LTL model-checking.

\paragraph{SCTP} is implemented in Linux~\cite{linux} and FreeBSD~\cite{freebsd}.
Both implementations were tested with \textsc{PacketDrill}~\cite{packetDrillSCTP} and fuzz-tests, suggesting they are crash-free and follow the RFCs.
But this does not necessarily imply the \emph{design} 
outlined in the RFCs achieves its intended goals.
Several prior works formally modeled SCTP, however, their models were not as comprehensive and up-to-date as ours.
We summarize the differences between prior models and our own in Table~\ref{tab:comparison} in the Appendix.

Of the prior works that applied formal methods to the security of SCTP, 
only the Uppaal analysis by Saini and Fehnker~\cite{saini2017evaluating} used a technology 
(model-checking) that can verify arbitrary properties.
They reported two properties in their paper; the first is similar to our $\gamma_2$.
The second says an adversary only capable of sending \Init packets cannot cause a victim peer to change state.  This property is trivial for us because we use an FSM model where the peer states are precisely the model states.
And in our model, the only transition out of \Closed that happens upon receiving an \Init is a self-loop that sends an \InitAck and returns to \Closed.  In contrast, in Saini and Fehnker's model the peer state is a variable in memory, while the model states are totally different (e.g., \textsf{LC1}, \textsf{LC2}).  Thus, the property merits verification in their model but not ours.

Another line of inquiry aims to model the performance of SCTP,
e.g., using numerical analyses and simulations~\cite{chukarin2006performance}.
For example, Fu and Atiquzzaman built an analytical model of SCTP congestion control,
including \emph{multihoming}, an SCTP feature not available in TCP.
They compared their model to simulations and found it to be accurate in estimating
steady-state throughput of multihomed paths~\cite{fu2005performance}.
Such models are also used to evaluate new features, e.g., as in~\cite{zou2006throughput}. 
LTL model checking is, generally speaking, a sub-optimal approach for performance evaluation,
thus in our performance analyses (in \Chapr{karn-rto} and \Chapr{gbn}) we rely on interactive provers.

\section{Conclusion}\Secl{handshakes:conclusion}

In \Chapr{karn-rto} and \Chapr{gbn}, we verified properties of infinite-state
systems, using provers (Ivy and \acls).
These provers are extremely powerful, but only semi-automated.
In contrast, model checkers like \spin are fundamentally limited to not just 
finite-state systems and decidable logics, but moreover, to systems and properties
which are ``small'' (i.e., which avoid state-space explosion).
However, they have the advantage of being fully automatic.
This chapter provides a useful case-study in that trade-off.
By making careful modeling decisions (e.g., around how to represent the \emph{itag}
and \emph{vtag} in SCTP), we are able to compress our models enough that they can be 
verified using an LTL model checker in a matter of seconds (\spin).
Thus in contrast to the prior two chapters, in this case, we get our proofs entirely ``for free''.
These proofs include all of the following results:
\begin{itemize}
    \item The TCP handshake avoids half-open connections and deadlocks.  Moreover, its active/passive routine eventually works, and does so in a way which reflects the message sequence chart descriptions in the RFC.
    \item The DCCP handshake avoids infinite looping in any state.  Moreover, it does not support active/active or passive/passive teardown.
    \item The SCTP handshake avoids multiple unsafe states, responds appropriately to messages, uses its timers when needed, and satisfies basic safety and liveness properties reflected in the message sequence charts in the RFC.
\end{itemize}

Our proofs show that the verified handshakes are correct in the sense that they 
satisfy protocol goals outlined in the corresponding RFC documents, which we enumerate.
Also, they set the stage for our study of protocol attacks in the next chapter.
That is: having proven these handshakes work correctly in the absence of an attacker,
once we add an attacker to these systems, if they then behave incorrectly,
we can safely assign blame to the attacker process.

\setcounter{observation}{0}
\setcounter{definition}{0}
\setcounter{problem}{0}
\setcounter{theorem}{0}

\chapter{Automated Attacker Synthesis}\Chapl{korg}
\textbf{Summary.}
Transport protocol handshakes have predefined inputs and outputs, and follow predefined communication patterns 
to synchronize and exchange information.
Such protocols should be robust to both inherent malfunction 
(deadlock or livelock due to unexpected orderings of events) 
and attacks (e.g., message replay).
In the previous chapter, we used LTL model checking to prove that the TCP, DCCP, and SCTP handshakes are correct
in isolation.
Now, we look at their behavior in the context of an attack.
We propose a novel formalism for attacker models, capturing the placement and capabilities of the attacker.
Using this formalism we define two attacker models: one for attackers who sometimes succeed and one for 
those who always succeed.
We argue that the former is more realistic, and derive an automated solution to it, based on LTL model checking.
We prove our solution is sound and complete for a certain class of attacks, and we apply it to TCP, DCCP, and SCTP,
reporting attacks against each.
In the case of SCTP, we find two ambiguities in the RFC, each of which, we show, can enable a novel attack.
We proposed two errata to the RFC, one of which the RFC committee accepted.

\medskip

This chapter includes work originally presented in the following publications:

\medskip

\noindent~Max von Hippel, Cole Vick, Stavros Tripakis, and Cristina Nita-Rotaru. \emph{Automated attacker synthesis for distributed protocols.} Computer Safety, Reliability, and Security, 2020.
\begin{description}
\item \underline{Contribution:} MvH formalized the problem with help from ST, invented the solution, wrote the proofs, wrote most of the code for the implementation and TCP case study, and wrote most of the paper.
\end{description}

\medskip

\noindent~Maria Leonor Pacheco, Max von Hippel, Ben Weintraub, Dan Goldwasser, and Cristina Nita-Rotaru. \emph{Automated attack synthesis by extracting finite state machines from protocol specification documents.} IEEE Symposium on Security and Privacy, 2022.
\begin{description}
\item \underline{Contribution:} MvH wrote the models and properties, as well as the FSM extraction algorithm (not included in this dissertation).
\end{description}

\medskip

\noindent~Jacob Ginesin, Max von Hippel, Evan Defloor, Cristina Nita-Rotaru, and Michael T{\"u}xen. \emph{A Formal Analysis of SCTP: Attack Synthesis and Patch Verification.} USENIX, 2024.
\begin{description}
\item \underline{Contribution:} MvH co-authored the models and properties and wrote more than half of the paper.
\end{description}

\section{Formal Definition of Automated Attacker Synthesis}

We want to synthesize attackers automatically. 
Intuitively, an attacker is a process that, when composed with the system, violates some property. 
To formalize this concept we first introduce a formal notion of \emph{attacker model}, in the context of which
we next introduce a formal definition of an \emph{attacker}.
But first, we need to introduce some mathematical vocabulary which will show up in those definitions.

\subsection{Mathematical Preliminaries}

Let $P = \langle \text{AP}, I, O, S, s_0, T, L \rangle$ be a process.  For each state~$s \in S$,~$L(s)$ is a subset of AP containing the atomic propositions that are true at state~$s$. 
Consider a transition~$(s,x,s')$ starting at state~$s$ and ending at state~$s'$ with label~$x$.
If the label $x$ is an input, then the transition is called an \emph{input transition} and denoted $s \xrightarrow[]{x?} s'$.  Otherwise,~$x$ is an output, and the transition is called an \emph{output transition} and denoted $s \xrightarrow[]{x!} s'$.  A transition $(s, x, s')$ is called \emph{outgoing} from state $s$ and \emph{incoming} to state~$s'$.

A state $s \in S$ is called a \emph{deadlock} iff it has no outgoing transitions.
The state $s$ is called \emph{input-enabled} iff, for all inputs $x \in I$, there exists some state $s' \in S$ such that there exists a transition $(s, x, s') \in T$. 
We call $s$ an \emph{input} state (or \emph{output} sate) if all its outgoing transitions are input transitions (or output transitions, respectively).
States with both outgoing input transitions and outgoing output transitions are neither input nor output states, while states with no outgoing transitions (i.e., deadlocks) are (vacuously) both input and output states.

Various definitions of process determinism exist; ours is a variation on that of \cite{SIGACT17}.
A process $P$ is called \emph{deterministic} iff all of the following hold: (i) its transition relation $T$ can be expressed as a (possibly partial) function $S \times (I \cup O) \to S$; (ii) every non-deadlock state in $S$ is either an input state or an output state, but not both; (iii) input states are input-enabled; and (iv) each output state has only one outgoing transition. 
Determinism guarantees that: each state is a deadlock, an input state, or an output state; when a process outputs, its output is uniquely determined by its state; and when a process inputs, the input and state uniquely determine where the process transitions. 

A \emph{run} of a process $P$ is just a run of its projection, and likewise, a trace of $P$ is just a trace of its projection.

Finally, given two processes \(P_i = \langle \text{AP}_i, I_i, O_i, S_i, s_0^i, T_i, L_i \rangle\) for $i = 1, 2$, we say that $P_1$ is a {\em subprocess} of $P_2$, denoted $P_1 \subseteq P_2$, if $\text{AP}_1 \subseteq \text{AP}_2, I_1 \subseteq I_2, O_1 \subseteq O_2, S_1 \subseteq S_2, T_1 \subseteq T_2,$ and, for all $s \in S_1$, $L_1(s) \subseteq L_2(s)$. 

\subsection{Formal Attacker and Attacker Model Definitions}

An \emph{attacker model} or \emph{threat model} prosaically captures the goals and capabilities of an attacker with respect to some victim and environment. 
Algebraically, it is difficult to capture the attacker goals and capabilities without also capturing the victim and the environment, so our abstract attacker model includes all of the above.
Our attacker model captures: how many attacker components there are; how they communicate with each other and with the rest of the system (what messages they can intercept, transmit, etc.); and the attacker goals.
We formalize the concept of an attacker model next.

\begin{definition}[Input-Output Interface]
    An \emph{input-output interface} is a tuple $(I, O)$ such that $I \cap O = \emptyset$ and $I \cup O \neq \emptyset$.  The \emph{class} of an input-output interface $(I, O)$, denoted $\mathcal{C}(I, O)$, is the set of processes with inputs $I$ and outputs $O$.  Likewise, $\mathcal{C}(P)$ denotes the interface the process $P$ belongs to.
\end{definition}

\begin{sloppypar}
\begin{definition}[Attacker Model]
An \emph{attacker model} is a tuple $(P, (Q_i)_{i = 0}^m, \phi)$ where $P,Q_0,...,Q_m$ are processes,
each process $Q_i$ has no atomic propositions (its set of atomic propositions is empty),
and $\phi$ is an LTL formula such that $P \parallel Q_0 \parallel ... \parallel Q_m \models \phi$.
We also require the system $P \parallel Q_0 \parallel ... \parallel Q_m$ satisfies the formula $\phi$ in a non-trivial manner, that is, that $P \parallel Q_0 \parallel ... \parallel Q_m$ has at least one run.
\end{definition}
\end{sloppypar}

In an attacker model, the process $P$ is called the \emph{invulnerable process},
and the processes~$Q_i$ are called \emph{vulnerable processes}.
The goal of the adversary is to modify the vulnerable processes~$Q_i$ so that composition with the
invulnerable process~$P$ violates the specification~$\phi$.

Having formalized the concept of an attacker model, we next need to say what precisely constitutes an \emph{attacker}.
In most real-world systems, infinite attacks are impossible, implausible, or just uninteresting.
To avoid such attacks, we define an attacker that produces finite-length sequences of adversarial
behavior, after which it behaves like the vulnerable process it replaced (see \Figr{acyclicAttacker}).
In other words, the ``attack'' is merely a malicious piece of code injected as a prefix in an otherwise
reliable system component (or components).

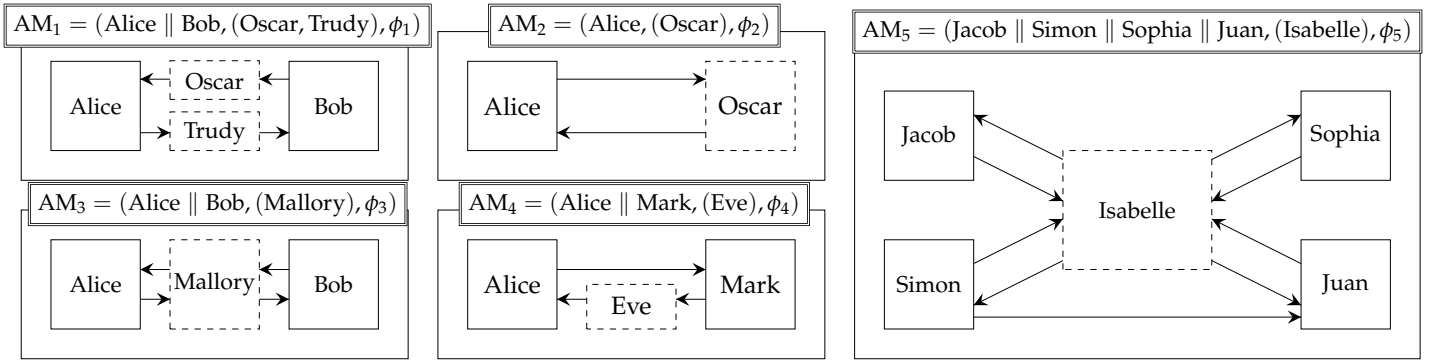
\begin{figure}[h]
\centering
\begin{adjustbox}{max totalsize={.99\textwidth}{.99\textheight},center}
\begin{tikzpicture}
\draw[draw=black] (-0.5,-0.5) rectangle ++(6.5, 2.5);
\draw[draw=black] (0, 0) rectangle ++(1.5,1.5);
\node[] (AliceBL) at (0.75,0.75) {\small Alice};
\draw[draw=black,dashed] (2, 0) rectangle ++(1.5,1.5);
\node[] (MalloryBL) at (2.75,0.75) {\small Mallory};
\draw[draw=black] (4, 0) rectangle ++(1.5,1.5);
\node[] (BobBL) at (4.75,0.75) {\small Bob};
\draw[straight] (1.5,0.5) to (2,0.5);
\draw[straight] (2,1)     to (1.5,1);
\draw[straight] (3.5,0.5) to (4,0.5);
\draw[straight] (4,1)     to (3.5,1);
\node[draw,rectangle,fill=white,double] (BLlabel) at (2.8,2.1) {\small $\text{AM}_3 = (\text{Alice} \parallel \text{Bob}, (\text{Mallory}), \phi_3)$};
\draw[draw=black] (6.5,-0.5) rectangle ++(6.5, 2.5);
\draw[draw=black] (7, 0) rectangle ++(1.5,1.5);
\node[] (AliceBR) at (7.75,0.75) {Alice};
\draw[draw=black,dashed] (9, 0) rectangle ++(1.5,0.75);
\node[] (EveBR) at (9.75,0.4) {Eve};
\draw[draw=black] (11, 0) rectangle ++(1.5,1.5);
\node[] (MarkBR) at (11.75,0.75) {Mark};
\draw[straight] (8.5,1)  to (11,1);
\draw[straight] (11,0.5) to (10.5,0.5);
\draw[straight] (9,0.5) to (8.5,0.5);
\node[draw,rectangle,fill=white,double] (BRlabel) at (9.8,2.1) {\small $\text{AM}_4 = (\text{Alice} \parallel \text{Mark}, (\text{Eve}), \phi_4)$};
\draw[draw=black] (6.5,2.5) rectangle ++(6.5, 2.5);
\draw[draw=black] (7, 3) rectangle ++(1.5,1.5);
\node[] (AliceTR) at (7.75,3.75) {Alice};
\draw[draw=black,dashed] (11, 3) rectangle ++(1.5,1.5);
\node[] (BobTR) at (11.75,3.75) {Oscar};
\draw[straight] (8.5,4.2) to (11,4.2);
\draw[straight] (11,3.3) to (8.5,3.3);  
\node[draw,rectangle,fill=white,double] (TRlabel) at (9.8,5.1) {\small $\text{AM}_2 = (\text{Alice}, (\text{Oscar}), \phi_2)$};
\draw[draw=black] (-0.5,2.5) rectangle ++(6.5, 2.5);
\draw[draw=black] (0, 3) rectangle ++(1.5,1.5);
\node[] (AliceTL) at (0.75,3.75) {\small Alice};
\draw[draw=black,dashed] (2, 3) rectangle ++(1.5,0.65);
\node[] (TrudyTL) at (2.75,3.3) {\small Trudy};
\draw[draw=black,dashed] (2, 3.86) rectangle ++(1.5,0.65);
\node[] (OscarTL) at (2.75,4.16) {\small Oscar};
\draw[draw=black] (4, 3) rectangle ++(1.5,1.5);
\node[] (BobTL) at (4.75,3.75) {\small Bob};
\draw[straight] (1.5,3.3) to (2,3.3);
\draw[straight] (3.5,3.3) to (4,3.3);
\draw[straight] (2,4.2)   to (1.5,4.2);
\draw[straight] (4,4.2)   to (3.5,4.2);
\node[draw,rectangle,fill=white,double] (TLlabel) at (2.8,5.1) {\small $\text{AM}_1 = (\text{Alice} \parallel \text{Bob}, (\text{Oscar}, \text{Trudy}), \phi_1)$};

\draw[draw=black] (-0.5 + 14,-6 + 5.5) rectangle ++(9.5, 5.5);

\draw[draw=black] (0 + 14, -5.5 + 5.5) rectangle ++(1.5,1.5);
\node[] (Simon5) at (0+0.75 + 14, -5.5+0.75 + 5.5) {\small Simon};

\draw[draw=black] (0 + 14, -3 + 5.5) rectangle ++(1.5,1.5);
\node[] (Jacob5) at (0+0.75 + 14,-3+0.75 + 5.5) {\small Jacob};

\draw[draw=black] (11-4 + 14, -5.5 + 5.5) rectangle ++(1.5,1.5);
\node[] (Juan5) at (11+0.75-4 + 14,-5.5+0.75 + 5.5) {\small Juan};

\draw[draw=black] (11-4 + 14, -3 + 5.5) rectangle ++(1.5,1.5);
\node[] (Sophia5) at (11+0.75-4 + 14,-3+0.75 + 5.5) {\small Sophia};

\draw[draw=black,dashed] (5-2 + 14, -4.5 + 5.5) rectangle ++(2.5, 2);
\node[] (Isabelle) at (6.25-2 + 14, -3.5 + 5.5) {\small Isabelle};

\draw[straight] (1.5 + 14, -5.5+0.2 + 5.5) to (11-4 + 14, -5.5+0.2 + 5.5);
\draw[straight] (5-2 + 14, -4.5 + 0.15 + 5.5) to (1.5 + 14, -5.5 + 0.4 + 5.5);
\draw[straight] (1.5 + 14, -5.5 + 1.5 - 0.4 + 5.5) to (5-2 + 14, -4.5 + 0.15 + 0.7 + 5.5);

\draw[straight] (5-2 + 14, -4.5 + 2 - 0.15 + 5.5) to (1.5 + 14, -3 + 1.5 - 0.4 + 5.5);
\draw[straight] (1.5 + 14, -3 + 0.4 + 5.5) to (5-2 + 14, -4.5 + 2 - 0.15 - 0.7 + 5.5);

\draw[straight] (5 + 2.5 - 2 + 14, -4.5 + 0.15 + 5.5) to (11 - 4 + 14, -5.5 + 0.4 + 5.5);
\draw[straight] (11 - 4 + 14, -5.5 + 1.5 - 0.4 + 5.5) to (5 + 2.5 - 2 + 14, -4.5 + 0.15 + 0.7 + 5.5);

\draw[straight] (5 + 2.5 - 2 + 14, -4.5 + 2 - 0.15 + 5.5) to (11 - 4 + 14, -3 + 1.5 - 0.4 + 5.5);
\draw[straight] (11 - 4 + 14, -3 + 0.4 + 5.5) to (5 + 2.5 - 2 + 14, -4.5 + 2 - 0.15 - 0.7 + 5.5);

\node[draw,rectangle,fill=white,double] (TL5Label) at (4 + 14.25,-1 + 6)
    {\small $\text{AM}_5 = (\text{Jacob} \parallel \text{Simon} \parallel \text{Sophia} \parallel \text{Juan}, (\text{Isabelle}), \phi_5)$};
\end{tikzpicture}
\end{adjustbox}
\caption{Example Attacker Models.  The properties $\phi_i$ are not shown.  Solid and dashed boxes are processes; we only assume the adversary can exploit the processes in the dashed boxes.  $\text{AM}_1$ describes a distributed on-path attacker scenario, $\text{AM}_2$ describes an off-path attacker, $\text{AM}_3$ is a classical man-in-the-middle scenario, and $\text{AM}_4$ describes a one-directional man-in-the middle, or, depending on the problem formulation, an eavesdropper.  $\text{AM}_5$ is an attacker model with a distributed victim where the attacker cannot affect or read messages from Simon to Juan.  Note that a directed edge in a network topology from Node 1 to Node 2 is logically equivalent to the statement that a portion of the outputs of Node 1 are also inputs to Node 2.  In cases where the same packet might be sent to multiple recipients, the sender and recipient can be encoded in a message subscript.  Therefore, the entire network topology is {\em implicit} in the interfaces of the processes in the attacker model according to the composition definition.}
\label{exampleThreatModels}
\end{figure}

\begin{definition}[Attacker]
Let $\textsc{AM} = (P, (Q_i)_{i = 0}^m, \phi)$ be an attacker model.  
Suppose that $\vec{A} = (A_i)_{i = 0}^m$ is a list of processes such that, for all $0 \leq i \leq m$,
$A_i$ is a deterministic process in $\mathcal{C}(Q_i)$ consisting of a directed acyclic graph (DAG) with no atomic propositions, ending in the initial state of the vulnerable process $Q_i$, followed by all of the vulnerable process $Q_i$.
Suppose further that $P \parallel A_0 \parallel ... \parallel A_m$ has some run $r$ such that $r \centernot{\models} \phi$ and each $A_i$ eventually reaches $q_0^i$ at some point in $r$.
Then we say that $\vec{A}$ is an \emph{$\textsc{AM}$-attacker}.
\end{definition}
The DAG criteria is illustrated in \Figr{acyclicAttacker}.

\begin{figure}[h]
\centering
\begin{adjustbox}{max totalsize={.9\textwidth}{.9\textheight},center}
\begin{tikzpicture}
\node[] (Ai) at (5,2.9) {\tiny $A_i$};
\draw[draw=gray,dashed] (-1.2,-2.5) rectangle (9.8,2.7);

\node[] (DAG) at (3,2.5) {\tiny DAG};
\draw[draw=gray,dashed] (-1,-2.3) rectangle (7,2.3);

\path[draw,use Hobby shortcut,closed=true,fill=black!25]
(7,-1) .. (7.1,1) .. (8,2) .. (9.5,1.4) .. (8.8,-1) .. (7.3,-2);
\node[] (Qi) at (8.3,0) {\tiny $Q_i$};

\node[] (empty) at (-1,0) {};
\node[draw,circle] (a0) at (0,0) {\tiny $a_0^i$};
\draw[straight] (empty) to (a0);
\node[draw,circle] (a1) at (2,1 ) {\tiny $a_1^i$};
\node[draw,circle] (a2) at (2,0 ) {\tiny $a_2^i$};
\node[draw,circle] (a3) at (2,-1) {\tiny $a_3^i$};
\draw[straight] (a0) to[above] node {\tiny $x_0$} (a1);
\draw[straight] (a0) to[above] node {\tiny $x_1$} (a2);
\draw[straight] (a0) to[above] node {\tiny $x_2$} (a3);
\draw[straight] (a1) to[right] node {\tiny $x_3$} (a2);
\node[] (m0) at (4,2 ) {\tiny ...};
\node[] (m1) at (4,1 ) {\tiny ...};
\node[] (m2) at (4,0 ) {\tiny ...};
\node[] (m3) at (4,-1) {\tiny ...};
\node[] (m4) at (4,-2) {\tiny ...};
\draw[straight] (a3) to[above] node {\tiny $x_4$} (m2);
\draw[straight] (a1) to[above] node {\tiny $x_5$} (m0);
\draw[straight] (a1) to[above] node {\tiny $x_6$} (m1);
\draw[straight] (a2) to[above] node {\tiny $x_7$} (m2);
\draw[straight] (a3) to[above] node {\tiny $x_8$} (m3);
\draw[straight] (a3) to[above] node {\tiny $x_9$} (m4);
\draw[straight] (a0) to[out=south east,in=west,below,looseness=1] node {\tiny $x_{10}$} (m4);
\node[draw,circle,fill=white] (q0) at (7,0) {\tiny $q_0^i$};
\draw[straight] (m0) to[above] node {\tiny $x_k$} (q0);
\draw[straight] (m1) to[above] node {\tiny $x_{k+1}$} (q0);
\draw[straight] (m2) to[above] node {\tiny $x_{k+2}$} (q0);
\draw[straight] (m3) to[above] node {\tiny $x_{k+3}$} (q0);
\draw[straight] (m4) to[left] node {\tiny $x_{k+4}$} (q0);
\end{tikzpicture}
\end{adjustbox}
\caption{Suppose $\vec{A} = (A_i)_{i = 0}^m$ is attacker for $\textsc{AM} = (P, (A_i)_{i = 0}^m, \phi)$.  Further suppose $A_i$ has initial state $a_0^i$, and $Q_i$ has initial state $q_0^i$.  Then $A_i$ should consist of a DAG starting at $a_0^i$ and ending at $q_0^i$, plus all of $Q_i$, indicated by the shaded blob.  Note that if some $Q_i$ is non-deterministic, then there can be no attacker, because $Q_i$ is a subprocess of $A_i$, and all the $A_i$s must be deterministic in order for $\vec{A}$ to be an attacker.}
\Figl{acyclicAttacker}
\end{figure}
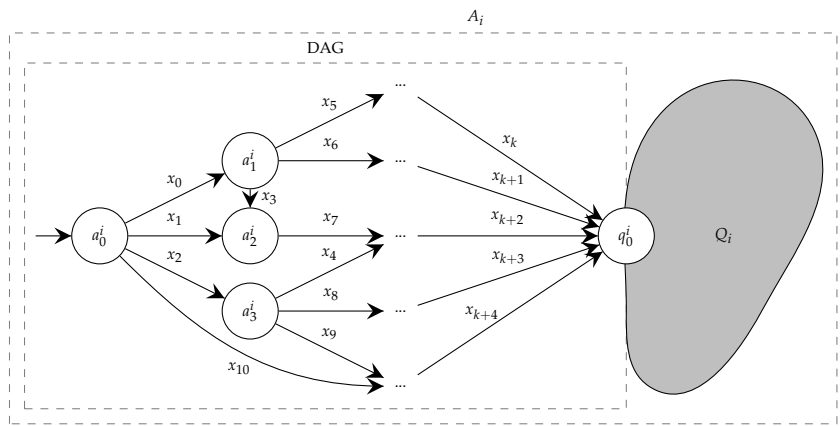

We can naturally characterize attackers depending on how powerful they are, that is to say, 
depending on whether or not they always succeed.

\begin{definition}[$\exists$-Attacker vs $\forall$-Attacker]
Let $\vec{A}$ be a $(P, (Q_i)_{i = 0}^m, \phi)$-attacker.  Then $\vec{A}$ is a \emph{$\forall$-attacker} if $P \parallel A_0 \parallel ... \parallel A_m \models \neg \phi$.  Otherwise, $\vec{A}$ is an \emph{$\exists$-attacker}.
\end{definition}

A $\forall$-attacker $\vec{A}$ always succeeds, because $P \parallel \vec{A}\models \neg \phi$ means that \emph{every} behavior of $P \parallel \vec{A}$ satisfies $\neg \phi$, that is, \emph{every} behavior of $P \parallel \vec{A}$ violates $\phi$.  Since $P \parallel \vec{A} \centernot{\models} \phi$, there must exist a computation $\sigma$ of $P \parallel \vec{A}$ such that $\sigma \models \neg \phi$, so, a $\forall$-attacker cannot succeed by blocking.  An $\exists$-attacker is any attacker that is not a $\forall$-attacker, and every attacker succeeds in at least one computation, so an $\exists$-attacker sometimes succeeds, and sometimes does not.  

\subsection{Automated Attacker Synthesis Problems}

Next, we define the two naturally arising automated attack synthesis problems.
The first problem is to find an attacker which, at least sometimes, induces a malfunction.
The intuition here is that the attacker's success may hinge on decisions in the system which,
in our model, are abstracted nondeterministically.
For example, the attacker might only succeed if the user (modeled nondeterministically) 
issues a specific command, opening the process $P$ up to attack.
Examples of attacks like this -- which sometimes succeed but sometimes do not --
    include Rowhammer~\cite{kim2014flipping},
    Meltdown~\cite{lipp2020meltdown}, or
    the attacks we previously reported against \textsc{GossipSub}~\cite{kumar2024formal}.
In all three cases, the attack succeeds or fails depending on conditions which are not necessarily
observable to the attacker.

\begin{problem}[$\exists$-Attacker Synthesis Problem ($\exists$ASP)]
Given an attacker model AM, find an AM-attacker, if one exists; otherwise state that none exists.
\label{EASP}
\end{problem}

The second problem we define is the dual of the first: it describes attacks that \emph{always} succeed,
no matter what nondeterministic choices the process $P$ makes.
One example is Spectre~\cite{kocher2020spectre}.
Generally speaking most real-world attacks do not satisfy this (very strong) reliability requirement,
so we consider this problem to be more of an academic than a practical one.

\begin{problem}[$\forall$-Attacker Synthesis Problem ($\forall$ASP)]
Given an attacker model AM, find a AM-$\forall$-attacker, if one exists; otherwise state that none exists.
\label{AASP}
\end{problem}

\section{Solution to the $\exists$-Attacker Synthesis Problem}\Secl{korg:solution}

Next, we present a solution to the $\exists$-problem for any number of attackers, and for both safety and liveness properties.  Our solution is sound and complete, and its runtime is 
polynomial in the product of the size of~$P$ and the sizes of the interfaces of the~$Q_i$s, 
and exponential in the size of the property~$\phi$~\cite{IEEECS86}.
The idea is to reduce the problem to model checking by replacing each vulnerable component~$Q_i$
with a process whose language is $(I_i \cup O_i)^* \mathcal{L}(Q_i)$.

We begin by defining \emph{lassos} and \emph{bad prefixes}.
A computation $\sigma$ is a \emph{lasso} if it equals a finite word $\alpha$, then infinite repetition of a finite word $\beta$, i.e., $\sigma = \alpha \cdot \beta^{\omega}$.  A prefix $\alpha$ of a computation $\sigma$ is called a \emph{bad prefix} for $P$ and $\phi$ if $P$ has $\geq 1$ runs inducing computations starting with $\alpha$, and every computation starting with $\alpha$ violates $\phi$.  We naturally elevate the terms \emph{lasso} and \emph{bad prefix} to runs and their prefixes.  We assume a \emph{model checker}: a procedure \textsc{MC}$(P, \phi)$ that takes as input a process $P$ and property $\phi$, and returns $\emptyset$ if $P \models \phi$, or one or more violating lasso runs or bad prefixes of runs for $P$ and $\phi$, otherwise \cite{BaierKatoenBook}.
In practice, the model checker we use is \spin, but in principle our approach should work for any LTL model checker.

Attackers cannot have atomic propositions.  So, the only way for $\vec{A}$ to \emph{attack} $\textsc{AM}$ is by sending and receiving messages, hence the space of attacks is within the space of labeled transition sequences.  The \emph{daisy} nondeterministically exhausts the space of input and output events of a vulnerable process.
We define the daisy as an \emph{abstract process} with two initial states.  We introduce that definition now because this is the only place where we use it; in all other cases we assume processes have just one initial state.

\begin{definition}[Abstract Process]
Let $P = \langle \text{AP}, I, O, S, S_0, T, L \rangle$ such that $S_0 \subseteq S$ is non-empty and, for each $s_0 \in S_0$, $\langle \text{AP}, I, O, S, s_0, T, L \rangle$ is a process.  Then we say $P$ is an \emph{abstract process}.  In other words, an abstract process is a process with more than one possible initial state.
\end{definition}

\begin{definition}[Daisy]
Given a process $Q_i = \langle \emptyset, I, O, S, s_0, T, L \rangle$, the \emph{daisy} of $Q_i$, denoted $\textsc{Daisy}(Q_i)$, is the abstract process $\textsc{Daisy}(Q_i) = \langle \text{AP}, I, O, S', S_0, T', L' \rangle$, with atomic propositions $\text{AP} = \{ \texttt{terminated}_i \}$, states $S' = S \cup \{ d_0 \}$, initial states $S_0 = \{ s_0, d_0 \}$, transitions \(T' = T \cup \{ (d_0, x, w_0) \mid x \in I \cup O, w_0 \in S_0 \}\), and labeling function $L' : S' \to 2^\text{AP}$ that takes $s_0$ to $\{ \texttt{terminated}_i \}$ and other states to $\emptyset$.  (We reserve the symbols $\texttt{terminated}_0, ...$ for use in daisies, so they cannot be sub-formulae of the property in any attacker model.)
\end{definition}

Let $\textsc{AM} = (P, (Q_i)_{i = 0}^m, \phi)$ be an attacker model.  
Our goal is to find an attacker for $\textsc{AM}$, if one exists, or state that none exists, otherwise.
First, we define a new property $\psi$ which says that if all the attacker components eventually terminate 
(in the sense of making it to $Q_i$),
then $\phi$ holds.

\begin{equation}
\psi = \big( \bigwedge_{0 \leq i \leq m} \F\, \texttt{terminated}_i \big) \implies \phi
\label{equationPsi}
\end{equation}

Next, we use the model checker to find runs of the system in which the vulnerable components are replaced with 
corresponding daisies, in which $\psi$ are violated.  Logically, these are traces where all the attacker components
terminate in the sense described above, yet, $\phi$ is violated.

\begin{equation}
R = \textsc{MC}(P \parallel \textsc{Daisy}(Q_0) \parallel ... \parallel \textsc{Daisy}(Q_m), \psi)
\label{equationRrecovery}
\end{equation}

If $R = \emptyset$, or if any $Q_i$ is nondeterministic, then report ``no attack exists''.
Else, choose $r \in R$ arbitrarily and continue as follows.
For each $0 \leq i \leq m$, proceed as follows.
Let $r_i$ be the shortest prefix of $r$ ending in a state $s$ for which $s \models \texttt{terminated}_i$.
Let $\ell_i$ be the sequence of labels on the transitions in $r_i$.
Let \(\ell_i|_{(I_i,O_i)} = l_0, l_1, \ldots, l_c\) be the subsequence of all the labels in $\ell_i$ which are elements of $I_i \cup O_i$.  Define the states $S_i^A = S_i \cup \{ a_0, a_1, \ldots, a_{c - 1} \}$, labeling function $L_i^A = \lambda s . \emptyset$, and transitions as follows.
\begin{equation}
\begin{aligned}
T_i^A & = \{ (a_i, l_i, a_{i+1} \mid i < c - 1) \} \\
      & \cup \{ (a_{c - 1}, l_c, q_0^i \} \\
      & \cup \{ (a_i, x_i, q_0^i) \mid i < c - 1 \land l_i \in I_i \land x_i \in (I_i \cup O_i) \setminus \{ l_i \} \} \\
      & \cup T_i
\end{aligned}
\end{equation}
Finally, define $A_i$ to be the process $\langle \text{AP}_i, I_i, O_i, S_i^A, a_0^i, T_i^A \rangle$.
Once this is done for each $i$, return $\vec{A} = (A_i)_{i=0}^m$.

Next we prove that our solution is sound and complete, provided that the same can be said for the model-checker.
Note that \spin satisfies these conditions in its exhaustive mode, but not when configured with certain state-compressing
optimizations.  Also, being complete does not mean it is \emph{fast} or \emph{efficient}; only that given sufficient time and memory, it will return a result.
\begin{theorem}[Soundness]
Let AM be an attacker model and suppose that given AM, our solution returns $\vec{A}$.  Then $\vec{A}$ is an AM-attacker.
\end{theorem}
\begin{proof}
Determinism of $A_i$ follows from the determinism of $Q_i$ and the construction of $T_i^A$ to include
    \(\{ (a_i, x_i, q_0^i) \mid i < c - 1 \land l_i \in I_i \land x_i \in (I_i \cup O_i) \setminus \{ l_i \} \}\).
The DAG shape requirements and size of $\vec{A}$ both follow from its construction.
Now, consider the run $r$, which we know must exist as otherwise the procedure would have returned ``no attack exists''.
We claim that $P \parallel A_0 \parallel \cdots \parallel A_m$ has some run $r'$ which is trace-equivalent to $r$,
in which each $A_i$s eventually reaches $q_0^i$.
We proceed inductively.
Consider the first transition $\textbf{s} \xrightarrow[]{l} \textbf{s}'$ in $r$.
Let $0 \leq i \leq m$ arbitrarily.
If $l \notin I_i \cup O_i$ then $\textsc{Daisy}(Q_i)$ did not transition in this step,
and neither can $A_i$.
Otherwise, there are three cases.
\begin{enumerate}[(1)]
\item $\textbf{s}[i] = d_0^i = \textbf{s}'[i]$:
Then in the first step of $r'$, $A_i$ can take the matching step $a_0^i \xrightarrow[]{l} a_1^i$.
\item $\textbf{s}[i] = d_0^i$ and $\textbf{s}'[i] = q_0^i$:
Then in the first step of $r'$, $A_i$ can take the matching step $a_0^i \xrightarrow[]{l} q_0^i$.
\item $\textbf{s}[i] = q_0^i$:
Then in the first step of $r'$, $A_i$ can take the same transition that $\textsc{Daisy}(Q_i)$ took in the first step of $r$.
\end{enumerate}
The inductive step is essentially identical, except that $A_i$ might not begin in $a_0^i$, and in the third case, it is possible that $\textbf{s}[i] \in S_i$ but does not equal $q_0^i$, since $\textsc{Daisy}(Q_i)$ may have taken one or more transitions while in its $Q_i$ subprocess.
The last step of the proof is to observe that each $\textsc{Daisy}(Q_i)$ eventually reaches $q_0^i$ in $r$ because of the construction of $\psi$, which by steps 2 and 3 of the argument we just outlined, implies the same for the $A_i$s.
\end{proof}

\begin{theorem}[Completeness]
Let AM be an attacker model and suppose that some AM-attacker exists.
Then our solution does not return ``no attack exists''.
\end{theorem}
\begin{proof}
Suppose $\vec{A}$ is an AM-attacker.
Then there exists some $r' \in \text{runs}(P \parallel A_0 \parallel \cdots \parallel A_m)$ such that 
$r' \centernot{\models} \psi$.
Choose $0 \leq i \leq m$ arbitrarily.
We claim that $\textsc{Daisy}(Q_i)$ can simulate the role of $A_i$ in $r'$.
If $A_i = Q_i$, then the result follows since $Q_i \subseteq \textsc{Daisy}(Q_i)$ and $q_0^i$ is an initial state of the generalized process $\textsc{Daisy}(Q_i)$.
On the other hand, suppose that $Q_i \subsetneq A_i$.
We know that $A_i$ cannot take an infinite number of transitions without entering $q_0^i$ since the part of $A_i$ which is not $Q_i$ is precisely a DAG ending in $q_0^i$.
If $A_i$ takes a finite number of transitions, then this can be emulated by looping on $d_0^i$ (with identical labels) until the last one, at which point $\textsc{Daisy}(Q_i)$ takes a matching-label transition to $q_0^i$.
Else, if $A_i$ takes an infinite number of transitions, then the finite prefix before it first reaches $q_0^i$ can be emulated in the way we just described, and the rest occurs in $Q_i$ and can therefore be repeated verbatim from $q_0^i$.
Since $q_0^i \models \texttt{terminated}_i$ in $\textsc{Daisy}(Q_i)$ it follows that the run which we just described (albeit one $i$ at a time) satisfies $\texttt{terminated}_0 \land \ldots \land \texttt{terminated}_m$.
Moreover, this run~$r$ has the same sequence of labels as $r'$,
meaning that $P$ can take the same sequence of transitions in it as it does in $r'$.
Since the $\texttt{terminated}_i$ propositions do not occur in $\text{AP}$ and the $\textsc{Daisy}(Q_i)$ processes
have no further propositions, and neither do the $A_i$s, it follows that the run $r$ is trace-equivalent to $r'$.
But $r' \centernot{\models} \phi$.  So $r \centernot{\models} \psi$.  Thus $R \neq \emptyset$.
Lastly, by definition the $A_i$s are deterministic; thus so are the $Q_i$s.
The result immediately follows.
\end{proof}

Next, we describe how we implemented our solution.

\section{Implementation in \korg}

We implemented our solution to the $\exists$-attacker synthesis problem in an open-source tool called 
\textsc{Korg}\footnote{Named after the Korg microKORG synthesizer, with its dedicated ``attack" control on Knob 3.  Code and models are freely and openly available at \url{https://github.com/maxvonhippel/AttackerSynthesis}.}.
In this section, we describe the design and features of \korg.
Then in the next three sections, we provide case studies in its use, against TCP, DCCP, and SCTP.

We say an attacker $\vec{A}$ for an \amodel $\textsc{AM} = (P, (Q_i)_{i = 0}^m, \phi)$ is a \emph{centralized attacker} if $m = 0$, or a \emph{distributed attacker}, otherwise.  
In other words, a centralized attacker has only one attacker component $\vec{A} = (A)$, whereas
a distributed attacker has many attacker components $\vec{A} = (A_i)_{i = 0}^m$.
\textsc{Korg} handles the $\exists$-attacker synthesis problem for liveness and safety properties for a centralized attacker.
It is implemented in about 700 lines of \textsc{Python 3} and uses \spin as its backend model checker.

\korg requires three inputs:
    (1) a \textsc{Promela} program $P$ representing the invulnerable part of the system;
    (2) a \textsc{Promela} program $Q$ representing the vulnerable part of the system, as well as its interface (inputs and outputs) in YAML format; and
    (3) a \textsc{Promela} LTL property $\phi$ representing what it means for the system to behave correctly.
Note, \korg can deduce the interface of the program $Q$ automatically by scanning its code.
However, if on paper $Q$ is defined to have an input or output which never appears in any of its transitions,
then said label will likewise not appear in its code, and so will be missed by the interface inference step.
For this reason, users are encouraged to make vulnerable component interfaces explicit.
Given these inputs, which define a centralized attacker model AM=$(P,(Q),\phi)$,
\textsc{Korg} synthesizes attackers using the procedure outlined in \Secr{korg:solution}.
The workflow is illustrated in \Figr{korg:organization}.

\begin{figure}[h]
\begin{adjustbox}{max totalsize={.8\textwidth}{.99\textheight},center}
\begin{tikzpicture}
\node[] (model) at (-1.3,2) {\small \textsc{Promela} program $P$};
\node[text width=200] (placement) at (-1.5,0.8) {\small \textsc{Promela} vulnerable program~$Q$};
\node[text width=200] (phi) at (-2,-0.3) {\small \textsc{Promela} LTL correctness property~$\phi$};

\node[draw,rectangle,fill=white,double] (korg) at (3,1) {\small \textsc{Korg}};

\draw[->] (1,2) -- (korg);
\draw[->] (1,0.8) -- (korg);
\draw[->] (1,-0.3) -- (korg);

\node[draw,rectangle,fill=white,double] (spin) at (7,1) {\small \textsc{Spin}};

\draw[->] (korg) to[above,out=north east,in=north west] node {\small ``$P \parallel \textsc{Daisy}(Q) \models \psi$?"} (spin);

\draw[->] (spin) to[below,out=south west,in=south east] node {\small Counterexamples} (korg);

\node[] (attacks) at (3,-1) {\small Synthesized Attackers};

\draw[->] (korg) -- (attacks);
\end{tikzpicture}
\end{adjustbox}
\caption{\textsc{Korg} workflow.  The property $\psi$ is automatically computed from $\phi$ to ensure the attacker eventually terminates, at which point the original code $Q$ is run.}
\Figl{korg:organization}
\end{figure}
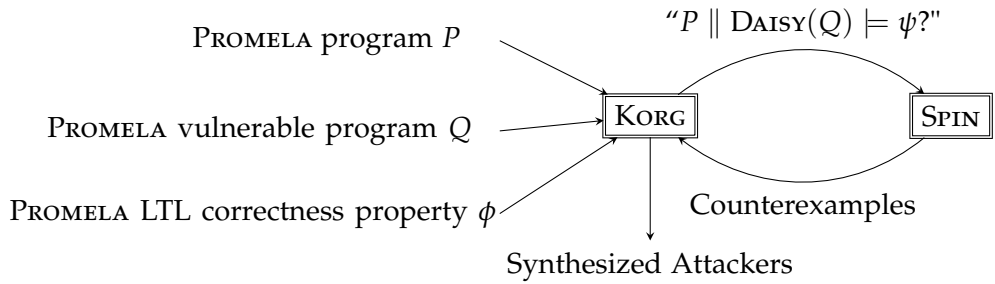

\korg also exposes some additional functionalities beyond those covered in this chapter, including:
\begin{itemize}
    \item partial handshake model extraction from RFC documents, which works in concert with the natural language processing pipeline described in~\cite{pacheco2022automated};
    \item synthesis of so-called ``replay'' attackers with bounded on-board memory (described in~\cite{sctp});
    \item installation via \textsc{pip} or \textsc{Docker}; and
    \item scripts to summarize and categorize attack traces (see \url{https://github.com/rfcnlp}).
\end{itemize}
It comes bundled with our TCP, DCCP, and SCTP models, attacker models, and properties, in addition to some toy models used for tutorials and unit testing.

Next, we describe the representative attacker models we use when applying \korg to TCP, DCCP, and SCTP.

\section{Representative Attacker Models and Experimental Setup}\Secl{amodels}

In this section we describe three representative attacker models which we use
for TCP, DCCP, and SCTP, and how we configure \korg with these \amodels.
These are general purpose and applicable to any transport protocol and correctness property,
and we contribute them to \korg.
We instantiate each \amodel in the context of each protocol model (TCP, DCCP, and SCTP)
and corresponding correctness property.

\paragraph{Off-Path \Amodel.}
In this model, an attacker communicates with one peer in order to disrupt the association
formed by the two peers that want to communicate.
We assume the Off-Path attacker knows the port and IP of the second peer, since otherwise, all its (spoofed) messages will be immediately discarded.
However, it cannot read the communication between the two peers, thus,
in the SCTP model, it cannot deduce the verification tag (vtag) of the association.
Note, since we do not model the sequence number in TCP or DCCP, the Off-Path attacker
can fully spoof the second peer in both of those models.
The Off-Path \amodel is illustrated in \Figr{off-path}.

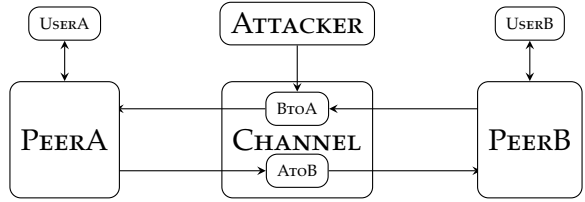
\begin{figure}[H]
\begin{adjustbox}{width=0.4\textwidth,center}
\begin{tikzpicture}
\node[draw, rectangle, rounded corners, minimum width=1.7cm, minimum height=1.5cm, fill=white] (channel) at (3,0) 
    {\small \textsc{Channel}};
\node[draw,rectangle, rounded corners] (AtoB) at (3,-0.4) {\tiny \textsc{AtoB}};
\node[draw,rectangle, rounded corners] (BtoA) at (3,0.4) {\tiny \textsc{BtoA}};
\draw[->] (0,-0.4) to (AtoB);
\draw[->] (AtoB) to (5.37,-0.4);
\draw[->] (6,0.4) to (BtoA);
\draw[->] (BtoA) to (0.65,0.4);
\node[draw, rectangle, rounded corners, minimum height=1.5cm,fill=white] (sender) at (0,0) 
    {\small \textsc{PeerA}};
\node[draw, rectangle, rounded corners, minimum height=1.5cm,fill=white] (receiver) at (6,0) 
    {\small \textsc{PeerB}};
\node[draw, rectangle, rounded corners, fill=white] (attacker) at (3,1.5)
    {\small \textsc{Attacker}};
\draw[->] (attacker) to (BtoA);
\node[draw,rectangle,rounded corners] (userA) at (0,1.5) {\tiny \textsc{UserA}};
\node[draw,rectangle,rounded corners] (userB) at (6,1.5) {\tiny \textsc{UserB}};
\draw[<->] (userA) to (sender);
\draw[<->] (userB) to (receiver);
\end{tikzpicture}
\end{adjustbox}
\caption{Off-Path \Amodel: $P = \textsc{UserA} \parallel \textsc{PeerA} \parallel \textsc{Channel} \parallel \textsc{PeerB} \parallel \textsc{UserB}$, and $Q$ is an empty process with the same inputs and outputs as $\textsc{PeerB}$ (but, in the case of SCTP, the wrong vtag).  The attacker can transmit messages into the BtoA buffer, but cannot receive messages, nor block messages in-transit.}
\Figl{off-path}
\end{figure}

\paragraph{Evil-Server \Amodel.}
In this \amodel, one of the peers behaves maliciously. For example,  the attacker takes the form of a finite sequence of malicious instructions inserted before the code of Peer~B, after which B behaves like normal.  See~\Figr{evil-server}.

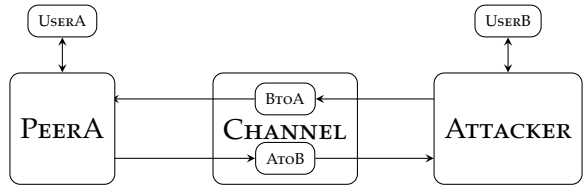
\begin{figure}[H]
\begin{adjustbox}{width=0.4\textwidth,center}
\begin{tikzpicture}
\node[draw, rectangle, rounded corners, minimum width=1.7cm, minimum height=1.5cm, fill=white] (channel) at (3,0) 
    {\small \textsc{Channel}};
\node[draw,rectangle, rounded corners] (AtoB) at (3,-0.4) {\tiny \textsc{AtoB}};
\node[draw,rectangle, rounded corners] (BtoA) at (3,0.4) {\tiny \textsc{BtoA}};
\draw[->] (0,-0.4) to (AtoB);
\draw[->] (AtoB) to (5,-0.4);
\draw[->] (6,0.4) to (BtoA);
\draw[->] (BtoA) to (0.65,0.4);
\node[draw, rectangle, rounded corners, minimum height=1.5cm,fill=white] (sender) at (0,0) 
    {\small \textsc{PeerA}};
\node[draw, rectangle, rounded corners, minimum height=1.5cm,fill=white] (receiver) at (6,0) 
    {\small \textsc{Attacker}};
\node[draw,rectangle,rounded corners] (userA) at (0,1.45) {\tiny \textsc{UserA}};
\node[draw,rectangle,rounded corners] (userB) at (6,1.45) {\tiny \textsc{UserB}};
\draw[<->] (userA) to (sender);
\draw[<->] (userB) to (receiver);
\end{tikzpicture}
\end{adjustbox}
\caption{Evil-Server \Amodel: $P =  \textsc{UserA} \parallel \textsc{PeerA} \parallel \textsc{Channel} \parallel \textsc{UserB}$, $Q = \textsc{PeerB}$.  The attacker can transmit messages into BtoA and receive messages from AtoB.  From the perspective of \textsc{PeerA}, the attacker is indistinguishable from a valid \textsc{PeerB} instance.}
\Figl{evil-server}
\end{figure}

\paragraph{On-Path \Amodel.}
In this \amodel, the attacker controls the channel connecting the two peers, and can drop or insert valid messages at-will.
Note that TCP, DCCP, and SCTP were not designed to withstand such an attacker, so we study it only to understand what could happen in a worst-case scenario.
The \amodel is illustrated in \Figr{on-path}.

\begin{figure}[H]
\begin{adjustbox}{width=0.4\textwidth,center}
\begin{tikzpicture}
\node[draw, rectangle, rounded corners, minimum width=1.7cm, minimum height=1.5cm, fill=white] (channel) at (3,0) 
    {\textsc{Attacker}};
\node[draw,rectangle, rounded corners] (AtoB) at (3,-0.4) {\tiny \textsc{AtoB}};
\node[draw,rectangle, rounded corners] (BtoA) at (3,0.4) {\tiny \textsc{BtoA}};
\draw[->] (0,-0.4) to (AtoB);
\draw[->] (AtoB) to (5.35,-0.4);
\draw[->] (6,0.4) to (BtoA);
\draw[->] (BtoA) to (0.65,0.4);
\node[draw, rectangle, rounded corners,minimum height=1.5cm,fill=white] (sender) at (0,0) 
    {\small \textsc{PeerA}};
\node[draw, rectangle, rounded corners, minimum height=1.5cm,fill=white] (receiver) at (6,0) 
    {\small \textsc{PeerB}};
\node[draw,rectangle,rounded corners] (userA) at (0,1.3) {\tiny \textsc{UserA}};
\node[draw,rectangle,rounded corners] (userB) at (6,1.3) {\tiny \textsc{UserB}};
\draw[<->] (userA) to (sender);
\draw[<->] (userB) to (receiver);
\end{tikzpicture}
\end{adjustbox}
\caption{On-Path \Amodel: $P =  \textsc{UserA} \parallel \textsc{PeerA} \parallel \textsc{PeerB} \parallel \textsc{UserB}$, $Q = \textsc{AtoB} \parallel \textsc{BtoA}$.  The attacker is allowed to perform a finite sequence of send/receive actions, in which it only sends valid messages (but can receive anything).  Once this sequence terminates, it behaves like an honest channel.}
\Figl{on-path}
\end{figure}
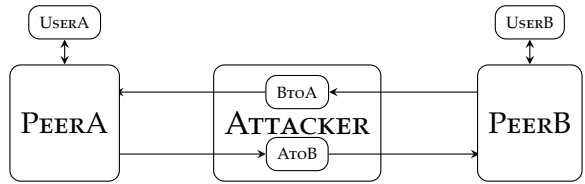

\paragraph{Common Experimental Setup.}
For each handshake model (TCP, DCCP, or SCTP), 
each property thereof,
and each representative attacker model, 
we run \korg using the following common experimental setup.
First, we ask \korg to synthesize $\leq 10$ attacks, because in our 
experience, after the first ten, subsequent attacks tend to be repetitive, differing only by actions that do not impact the attack outcome.  
Second, we configure \korg with a default search depth of 600,000, 
and a maximum depth of 2,400,000.
In our experience, these parameters balance fast-performance on smaller properties
with the ability to also attack more complex ones, without needing to run on a cluster.

\section{Synthesized Attacks Against the Transmission Control Protocol Handshake}

\begin{table}[h]
\centering
\begin{tabular}{l|p{3cm}p{3cm}p{3cm}p{3cm}}
            & $\phi_1$ & $\phi_2$ & $\phi_3$ & $\phi_4$ \\
            & No half-open
            & Passive/active succeeds
            & Peers don't get stuck
            & \SYNREC $\to$ \ESTABLISHED \\\hline
Off-Path    & 7 in 4s & 0 in 1s & 25 in 223m 28.4s & 4 in 2.3s \\
Evil-Server & 1 in 4s & 0 in 1s  & 12 in 72m 57.3s & 24 in 4.7s  \\
On-Path     & 1 in 4s & 9 in 3s & 36 in 218m 24.2s & 17 in 4.2s    
\end{tabular}
\caption{Synthesized attacks against the TCP handshake for each property, and the time required for \korg to compute them (or to determine that none exist) on a 16GB M1 Macbook Air, rounded to the nearest second.}
\end{table}

\korg does not find any attacks in the Off-Path or Evil-Server \amodels against $\phi_2$
because of the placement of the attacker in those \amodels.
In order to violate $\phi_2$, \korg would need to inject a \SYN or \ACK to Peer B,
but in both the Off-Path and Evil-Server \amodels the attacker can only inject packets to Peer A.
With all three \amodels, \korg computes results for $\phi_1$, $\phi_2$, and $\phi_4$ in seconds,
however, it takes a few hours to analyze $\phi_3$.
This is because $\phi_3$ is a considerably larger property than the other three, and \korg reduces to LTL model checking, the runtime of which is exponential in the size of the property~\cite{IEEECS86}.
Next, we describe some example attacks at a high level.

\paragraph{Example Off-Path Attack Against $\phi_1$.} 
Recall that $\phi_1$ forbids half-open connections.  In the first Off-Path attack generated with $\phi_1$, the attacker injects an \ACK and two {\FIN}s to Peer A, in that order.
The attack is illustrated below in \Figr{off-path-tcp-1}.
Note that the second \FIN is injected after the attack has already succeeded.

\begin{figure}[h]
\centering
\begin{adjustbox}{width=0.45\textwidth}
\begin{tikzpicture}
    \node[] (attacker) {\small Attacker};
    \node[right=2cm of attacker] (peerA) {\small \textsc{PeerA}};
    \node[right=2cm of peerA] (peerB) {\small \textsc{PeerB}};
    \node[below=0.18cm of attacker] (a0) {$a_0$};
    \node[below=0.1cm of peerA] (pa0) {\CLOSED};
    \node[below=0.1cm of peerB] (pb0) {\CLOSED};
    \draw[->] (pa0) to node[above] {\SYN} (pb0);
    \node[below=0.18cm of a0] (a1) {$a_1$};
    \node[below=0.1cm of pa0] (pa1) {\SYNSENT};
    \node[below=0.1cm of pb0] (pb1) {\CLOSED};
    \draw[->] (a1) to node[above] {\ACK} (pa1);
    \node[below=0.18cm of a1] (a2) {$a_2$};
    \node[below=0.1cm of pa1] (pa2) {\SYNSENT};
    \node[below=0.1cm of pb1] (pb2) {\CLOSED};
    \draw[->] (pb2) to node[above] {\SYN} (pa2);
    \node[below=0.18cm of a2] (a3) {$a_3$};
    \node[below=0.1cm of pa2] (pa3) {\SYNSENT};
    \node[below=0.15cm of pb2] (pb3) {\SYNSENT};
    \draw[->] (pa3) to node[above] {\ACK} (pb3);
    \node[below=0.18cm of a3] (a4) {$a_4$};
    \node[below=0.1cm of pa3] (pa4) {\ESTABLISHED};
    \node[below=0.1cm of pb3] (pb4) {\SYNREC};
    \draw[->] (a4) to node[above] {\FIN} (pa4);
    \node[below=0.18cm of a4] (a5) {$a_5$};
    \node[below=0.1cm of pa4] (pa5) {\ESTABLISHED};
    \node[below=0.1cm of pb4] (pb5) {\SYNREC};
    \draw[->] (pa5) to node[above] {\ACK} (pb5);
    \node[below=0.18cm of a5] (a6) {$a_6$};
    \node[below=0.1cm of pa5] (pa6) {\CLOSEWAIT};
    \node[below=0.09cm of pb5] (pb6) {\ESTABLISHED};
    \draw[->] (pa6) to node[above] {\FIN} (pb6);
    \node[below=0.18cm of a6] (a7) {$a_7$};
    \node[below=0.1cm of pa6] (pa7) {\LASTACK};
    \node[below=0.1cm of pb6] (pb7) {\ESTABLISHED};
    \draw[->] (pb7) to node[above] {\ACK} (pa7);
    \node[below=0.18cm of a7] (a8) {$a_8$};
    \node[below=0.1cm of pa7] (pa8) {\CLOSED};
    \node[below=0.1cm of pb7, text width=3cm] (pb8) {\ESTABLISHED, about to move to \CLOSEWAIT};
\end{tikzpicture}
\end{adjustbox}
\caption{Attack trace realized by the first Off-Path$(\phi_1)$ attacker synthesized by \korg, illustrated as a message sequence chart ending when the property is violated by a half-open connection.  Subsequent events in the trace are not illustrated since they are irrelevant to the property violation.}
\Figl{off-path-tcp-1}
\end{figure}
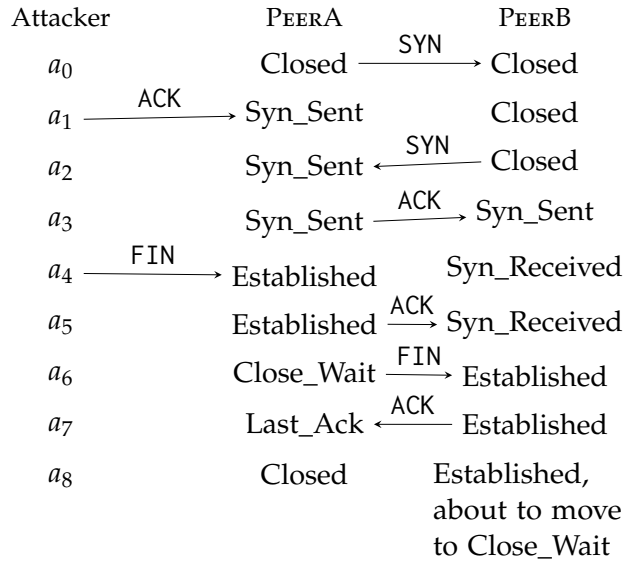

\paragraph{Example On-Path Attack Against $\phi_2$.} 
The attacker spoofs Peer A in order to guide B through a connection routine, resulting in a de-synchronization between A and B which disables them from ever successfully establishing a connection.
Interestingly, despite being On-Path, this particular attack never injects messages to A, nor drops messages from A; it only spoofs A in order to manipulate B.

\paragraph{Example Evil-Server Attack Against $\phi_3$.} 
The attacker communicates with Peer A at length in order to de-synchronize the peers such that, some time after the attack terminates, the peers end up in $(\FINWAITTWO, \CLOSEWAIT)$ with an \ACK in transit to Peer A (who expects a \FIN).  This is a deadlock.

\section{Synthesized Attacks Against the Datagram Congestion Control Protocol Handshake}

\begin{table}[h]
\centering
\begin{tabular}{l|p{3cm}p{3cm}p{3cm}p{3cm}}
            & $\theta_1$ & $\theta_2$ & $\theta_3$ & $\theta_4$ \\
            & Peers don't loop in a state 
            & No passive/passive teardown
            & First peer doesn't loop in a state
            & No active/active teardown \\\hline
Off-Path    & 0 in 5s  & 0 in 3s   & 0 in 5s  & 7 in 10s\\
Evil-Server & 0 in 2s  & 0 in 2s  & 0 in 2s & 0 in 2s  \\
On-Path     & 0 in 3s  & 13 in 12s & 0 in 3s  & 1 in 11s    
\end{tabular}
\caption{Synthesized attacks against the DCCP handshake for each property, and the time required for \korg to compute them (or to determine that none exist) on a 16GB M1 Macbook Air, rounded to the nearest second.}
\end{table}

The most interesting result is that no attacks are found with $\theta_1$ or $\theta_3$.
The type of looping behavior described by these properties is simply impossible in DCCP,
and thus, cannot be triggered by any attacker, regardless of its capabilities.
Next we overview some example attacks.

\paragraph{Example Off-Path Attack Against $\theta_4$.}
The attacker waits until Peer B has reached \CLOSEREQ.
It then injects a \texttt{DCCP\_RESET} to Peer A, guiding it back to \CLOSED
without alerting B.  From there it injects messages to A in order to guide A into \CLOSEREQ.
None of Peer A's response messages are of the type \texttt{DCCP\_CLOSE} and therefore they are all
treated as unexpected packets by Peer B, resulting in eventually both peers simultaneously 
being in \CLOSEREQ, violating $\theta_4$.

\paragraph{Example On-Path Attack Against $\theta_2$.}
The attacker spoofs each peer in order to guide Peer A through 55 establishment routines and Peer B through 40, before eventually leading each into \TIMEWAIT, violating $\theta_2$.  Note, \spin has an option to always return the shortest possible trace, which \korg can be configured to use, however it considerably increases the runtime of both tools.

\paragraph{Example On-Path Attack Against $\theta_4$.}
The attacker spoofs Peer B to guide A through 36 establishment routines and B through 23 before eventually leading each into \CLOSEREQ, violating $\theta_4$.  An attack trace snippet is shown in \Figr{on-path-1-dccp}.

\begin{figure}
\begin{lstlisting}
4570:   proc  2 (DCCP:1) debug.pml:72 (state 20)    [state[i] = 2]
4571:   proc  - (phi4:1) _spin_nvr.tmp:4 (state 4)  [(1)]
Stmnt [AtoN?DCCP_REQUEST] has escape(s): [(timeout)]
4572:   proc  1 (attacker:1) debug.pml:2436 (state 4501)    [AtoN?DCCP_REQUEST]
4573:   proc  - (phi4:1) _spin_nvr.tmp:4 (state 4)  [(1)]
Stmnt [NtoA!DCCP_RESPONSE] has escape(s): [(timeout)]
4574:   proc  1 (attacker:1) debug.pml:2439 Sent DCCP_RESPONSE  -> queue 2 (NtoA)
4574:   proc  1 (attacker:1) debug.pml:2439 (state 4507)    [NtoA!DCCP_RESPONSE]
4575:   proc  2 (DCCP:1) debug.pml:74 (state 21)    [rcv?DCCP_RESPONSE]
4576:   proc  - (phi4:1) _spin_nvr.tmp:4 (state 4)  [(1)]
4577:   proc  2 (DCCP:1) debug.pml:75 Send DCCP_ACK -> queue 1 (snd)
\end{lstlisting}
\caption{Example output from \spin for On-Path$(\theta_4)$, Attack 1.  4,738 trace lines omitted for brevity.  \korg comes with useful built-in tools for parsing verbose \spin output, which can be \texttt{pip}-imported by any \textsc{Python} package.}
\Figl{on-path-1-dccp}
\end{figure}

\section{Synthesized Attacks Against the Stream Control Transmission Protocol Handshake}

SCTP is implemented in Linux~\cite{linux} and FreeBSD~\cite{freebsd}.
These implementations were tested using  \textsc{PacketDrill}~\cite{cardwell2013packetdrill,packetDrillSCTP} and analyzed with \textsc{WireShark}~\cite{rungeler2012sctp}.
However, a recent vulnerability (CVE-2021-3772~\cite{cve}) shows the importance of conducting a much
more comprehensive 
formal analysis. 
Although a patch was proposed in RFC 9260~\cite{rfc9260}, and adapted by FreeBSD, the question remains whether other flaws might persist in the protocol design and whether the patch might have introduced additional vulnerabilities.
To the best of our knowledge, no prior works formally analyzed the 
entire SCTP connection establishment and teardown routines in a security context.
Motivated by this gap in the literature, we chose to conduct a detailed attacker synthesis-based study of SCTP
both with and without the FreeBSD patch.
We attempt to answer two questions.
(1) Does the FreeBSD patch resolve the vulnerability described in CVE-2021-3772?
And (2) do any other vulnerabilities persist in the code, or, were any new vulnerabilities introduced by the patch?

The rest of this section is organized as follows.
We describe the vulnerability disclosed in CVE-2021-3772 and the patches adopted by Linux and FreeBSD in \Secr{sctp-cve-and-patch}.
In \Secr{sctp-attack}, we apply \korg with the same settings we used for TCP and DCCP to the SCTP handshake model outlined in \Secr{handshakes:sctp}, but modified to disable the FreeBSD patch.
(The FreeBSD patch is the canonical patch strategy, in the sense that it is the one given in the latest SCTP RFC.)
Then we repeat the process in \Secr{sctp-patch-analysis} with the default version of our SCTP model, in which the FreeBSD patch is enabled.
Since the vulnerability described in the CVE was enabled by an ambiguity in the RFC, we conclude by manually analyzing the RFC for vulnerabilities, of which we find two.
We describe these ambiguities, and our analysis thereof, in \Secr{sctp-ambiguities}.
Based off our analysis, the IETF published an erratum to the SCTP RFC, which we authored~\cite{erratum7852}.

\subsection{CVE-2021-3772 Attack and Patch.}\Secl{sctp-cve-and-patch}
As reported in CVE-2021-3772~\cite{cve}, the prior version of SCTP specified in RFCs 2960~\cite{rfc2960} and 4960~\cite{rfc4960}
is vulnerable to a denial-of-service attack.
The reported vulnerability worked as follows.  Suppose SCTP peers A and B have established a connection and an off-channel attacker knows the IP addresses and ports of the two peers, but not the vtags of their existing connection.  The attacker spoofs B and sends a packet containing an \Init
to A.  The attacker uses a zero vtag as required for
packets containing an \Init.  The attacker must use
an illegal parameter in the \Init, e.g., a zero itag.  

Peer A, having already established a connection, treats the packet as
out-of-the-blue, per RFC 2960 $\mathsection$8.4 and 5.1, which specify that as an
association was established, A should respond to the \Init
containing illegal parameters with an \Abort and go
to \Closed. But in RFCs 2960 and 4960, it is unspecified which vtag should be
used in the \Abort.  Some implementations
used the expected vtag, which is
where a vulnerability arises.
Since the attacker spoofed the IP and port of Peer B,
Peer A sends the \Abort to Peer B, not the attacker.
When Peer B receives the \Abort, it sees the correct vtag, and 
tears down the connection.  Thus, by injecting a single packet with zero-valued tags, the attacker tears down the 
connection, pulling off a DoS.  
The attack is illustrated in \Figr{sctp-cve-msc}.

RFC 9260 patches CVE-2021-3772 using a strict defensive measure, wherein OOTB \Init packets with empty or zero itags are discarded, without response.  FreeBSD~\cite{freebsd} uses this patch.
Linux, on the other hand, adopts a different patch~\cite{linuxNotes}, wherein
the peer receiving the \Abort with the zero vtag simply ignores
it (rather than close the connection).
We consider the FreeBSD patch canonical because it is the one specified in the latest RFC,
and we enable it by default in our SCTP model (described in \Secr{handshakes:sctp}).

\begin{figure}
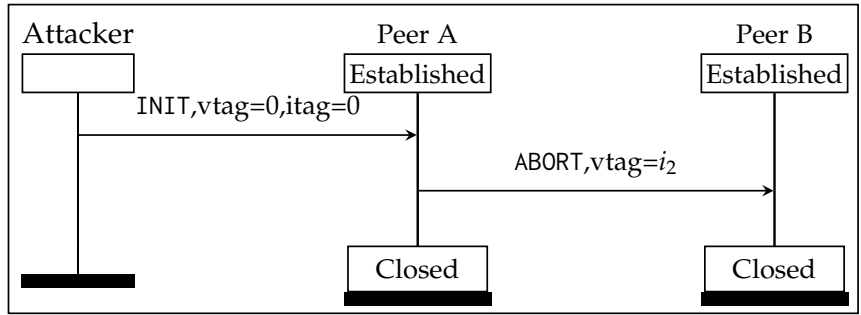

\begin{adjustbox}{width=0.6\textwidth,center}
\centering
\begin{msc}[head top distance=0.7cm,msc keyword=,left environment distance=1cm,right environment distance=1.2cm,foot distance=0.1cm, instance distance=3cm, action width=3cm]{}
    \declinst{attacker}{Attacker}{}
    \declinst{A}{\small Peer A}{\small \Established}
    \declinst{B}{\small Peer B}{\small \Established}

    \mess{\small \Init,vtag=0,itag=0}{attacker}{A}

    \nextlevel

    \mess{\small \Abort,vtag=$i_2$}{A}{B}

    \nextlevel

    \action[action width=2cm]{\small \Closed}{A}
    \action[action width=2cm]{\small \Closed}{B}
\end{msc}
\end{adjustbox}
\caption{Attack disclosed in CVE-2021-3772.  Peers A and B begin having  established an association with vtags~$i_1, i_2$ (resp.).  The Attacker transmits an invalid \Init chunk to  A, spoofing the port and IP of B.  Peer A responds by sending a valid \Abort to B, which closes the association.  By sending a single invalid \Init the Attacker performs a DoS.}
\Figl{sctp-cve-msc}
\end{figure}

\subsection{Synthesized Attacks with the CVE Patch Disabled.}\Secl{sctp-attack}
First, we run \korg with each SCTP \amodel with the CVE patch disabled.
We find at least one attack with each \amodel, all of which we describe below.
The time taken to compute each result, including to report that no attacks exist for combinations where we did not find any attacks, is reported in \Tabr{sctp-times}.

\begin{table}[h]
\centering
\newcolumntype{P}[1]{>{\RaggedRight\hspace{0pt}}p{#1}}
\centering
\small
\begin{tabular}{l|P{2.5cm}P{3cm}P{2.5cm}P{3cm}P{2cm}}
            & $\gamma_1$ & $\gamma_2$ & $\gamma_3$ & $\gamma_4$ & $\gamma_5$ \\
            & {\small Stay closed or establish} 
            & {\small Both closed, both established, \newline or changing state}
            & {\small Active teardown works} 
            & {\small Cookie timer ticks in $\CookieEchoed$} 
            & {\small No passive/ passive teardown} \\\hline
Off-Path    & 0 in 2m 20s & 0 in 8m 43s & 0 in 3m 20s & 0 in 1m 45s & 4 in 2m 57s \\
Evil-Server & 1 in 23s & 0 in 21s & 0 in 20s & 0 in 11s & 0 in 0m 10s \\
On-Path     & 0 in 15s & 0 in 26s & 0 in 25s & 0 in 14s & 0 in 12s \\\hline  
            & $\gamma_6$ & $\gamma_7$ & $\gamma_8$ & $\gamma_9$ & $\gamma_{10}$ \\
            & {\small Passive \newline teardown works}
            & {\small No $(\CookieEchoed,$ $\ShutdownReceived)$} 
            & {\small Correctness \newline of active/ passive teardown}
            & {\small Correctness \newline of passive/ active teardown}
            & {\small Teardown succeeds} \\\hline
Off-Path    & 0 in 3m 19s & 0 in 1m 43s & 0 in 2h 3m 42s & 1 in 1h 26m 10s & 0 in 4s \\
Evil-Server & 1 in 20s & 0 in 11s & 1 in 1m 6s & 1 in 14s & 0 in 12s \\
On-Path     & 0 in 25s & 0 in 13s & 2 in 1m 34s & 2 in 11s & 0 in 4s \\\hline 
\end{tabular}
\caption{Synthesized attacks against the SCTP handshake for each property, and the time required for \korg to compute them (or to determine that none exist) on a 16GB M1 Macbook Air, rounded to the nearest second.  Note, because our SCTP model is so complicated, the preliminary check to confirm $P \parallel Q \models \gamma$ built into \korg proved to be a considerable time-suck.  Hence, we performed this check manually for each property ahead of time, and then disabled it in \korg while running the attacker synthesis pipeline.  Therefore the times shown in this table should not be directly compared to those reported for TCP or DCCP.}
\Tabl{sctp-times}
\end{table}

\paragraph{Example Off-Path Attack Against $\gamma_9$.} A variant of the CVE attack.

\paragraph{Example Evil-Server Attack Against $\gamma_1$.}  The attacker guides A through passive establishment.  Then when A attempts active teardown, if its Shutdown Timer never fires, it deadlocks.  




\paragraph{Example On-Path Attack Against $\gamma_5$.}  The attacker spoofs each peer in order to manipulate the other into \ShutdownReceived.  (All four On-Path attacks against $\gamma_5$ use variations on this strategy.)



\subsection{Verification of the CVE Patch.}\Secl{sctp-patch-analysis}
Next, we re-run our analysis with the CVE patch enabled.
In the Off-Path \amodel, \korg terminates without finding any attacks.
This suffices to prove that the patch resolves the vulnerability.
In the other \amodels, we find the exact same attacks as those reported above, and nothing more, indicating that the patch does not decrease the security of SCTP with respect to the properties we defined
in any way which can be described in our model.

\subsection{Ambiguity Analysis.}\Secl{sctp-ambiguities}
We found two ambiguities in the latest SCTP RFC~\cite{rfc9260}.  First, in $\mathsection$5.2.1, during the description of how a peer should react upon receiving an unexpected \Init chunk:
\begin{quote}
\emph{Upon receipt of an \Init chunk in the \CookieEchoed state, an endpoint
   MUST respond with an \InitAck chunk using the same parameters it sent
   in its original \Init chunk (including its Initiate Tag, unchanged),
   provided that no new address has been added to the forming
   association.}
\end{quote}
Consider two peers (A and B) initially both in \Closed, in addition to some attacker who can spoof the port and IP  of B.  Suppose these machines engage in the sequence of events illustrated on the left-hand side of \Figr{sctp:ambiguity:attack}.  At the end of the sequence, what value should the vtag $V$ take? 

\begin{figure}
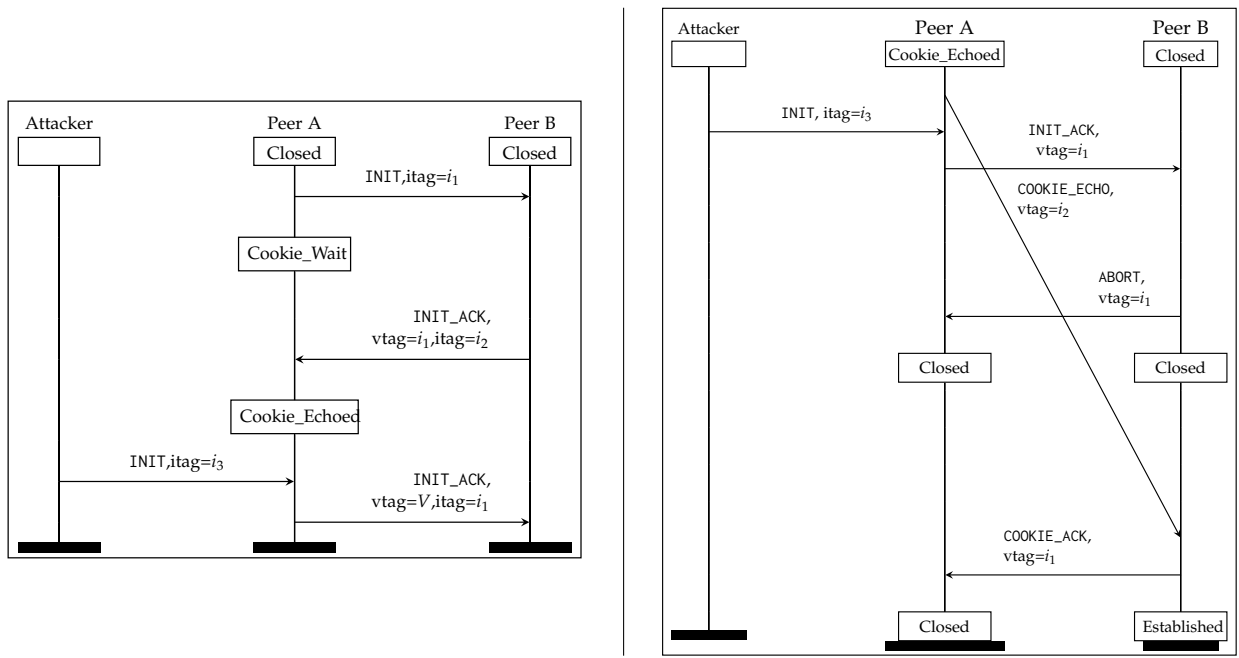

\centering
\begin{minipage}{0.45\textwidth}
\begin{adjustbox}{width=0.9\textwidth,center}
\centering
\begin{msc}[head top distance=0.7cm,msc keyword=,left environment distance=1cm,right environment distance=1cm,foot distance=0.1cm, instance distance=3cm, action width=3cm]{}
    \declinst{attacker}{\footnotesize Attacker}{}
    \declinst{A}{\footnotesize Peer A}{\footnotesize \Closed}
    \declinst{B}{\footnotesize Peer B}{\footnotesize \Closed}

    \mess{\footnotesize \Init,itag=$i_1$}{A}{B}

    \nextlevel

    \action[action width=2.2cm]{\footnotesize \CookieWait}{A}

    \nextlevel
    \nextlevel
    \nextlevel

    \mess[r]{{\parbox[b]{1\instdist}{\flushright \footnotesize \InitAck, vtag=$i_1$,itag=$i_2$}}}{B}{A}

    \nextlevel

    \action[action width=2.5cm]{\footnotesize \CookieEchoed}{A}

    \nextlevel
    \nextlevel

    \mess{\footnotesize \Init,itag=$i_3$}{attacker}{A}

    \nextlevel

    \mess{{\parbox[b]{1\instdist}{\flushright \footnotesize \InitAck, vtag=$V$,itag=$i_1$}}}{A}{B}
\end{msc}
\end{adjustbox}
\end{minipage}
\vline
\begin{minipage}{0.45\textwidth}
\begin{adjustbox}{width=0.9\textwidth,center}
\centering
\begin{msc}[head top distance=0.7cm,msc keyword=,left environment distance=1cm,right environment distance=1.2cm,foot distance=0.1cm, instance distance=3cm, action width=3cm]{}
    \declinst{attacker}{\footnotesize Attacker}{}
    \declinst{A}{Peer A}{\footnotesize \CookieEchoed}
    \declinst{B}{Peer B}{\footnotesize \Closed}

    \mess{\hspace{0.7cm}{\parbox[b]{.2\instdist}{\flushright \footnotesize \CookieEcho, vtag=$i_2$}}}{A}[0.3]{B}[12]

    \nextlevel


    \mess{\footnotesize \Init, itag=$i_3$}{attacker}{A}

    \nextlevel

    \mess{{\parbox[b]{.85\instdist}{\centering \footnotesize \InitAck, vtag=$i_1$}}}{A}{B}

    \nextlevel
    \nextlevel
    \nextlevel
    \nextlevel

    \mess{{\parbox[b]{.5\instdist}{\flushleft \footnotesize \Abort, vtag=$i_1$}}}{B}[0.2]{A}

    \nextlevel

    \action[action width=2cm]{\footnotesize \Closed}{A}
    \action[action width=2cm]{\footnotesize \Closed}{B}

    \nextlevel
    \nextlevel
    \nextlevel
    \nextlevel
    \nextlevel
    \nextlevel

    \mess{{\parbox[b]{.85\instdist}{\flushleft \footnotesize \CookieAck, vtag=$i_1$}}}{B}{A}

    \nextlevel

    \action[action width=2cm]{\footnotesize \Closed}{A}
    \action[action width=2cm]{\footnotesize \Established}{B}
\end{msc}
\end{adjustbox}
\end{minipage}
\caption{Left: ambiguous scenario.  What value should $V$ take?.  Right: Message sequence chart showing the DoS attack enabled by misinterpretation of the ambiguous RFC text.  Note the strict timing requirements necessary for a successful attack.}
\Figl{sctp:ambiguity:attack}
\end{figure}

The ambiguity arises from the use of the words \emph{it} and \emph{its}.
If the \emph{it} in question is interpreted to be the same entity as \emph{an endpoint}, i.e., the responding endpoint (A), and if ``the same parameters'' is interpreted to include the vtag,
then the resulting implementation will be vulnerable to a denial-of-service attack in the form of an induced half-open connection, which we illustrate on the right hand side of \Figr{sctp:ambiguity:attack}.
The fact that this is the wrong interpretation only becomes clear if you fully understand how itags and vtags are used in both directions.

Using a modified version of our SCTP model which implemented the incorrect interpretation of the ambiguous text,
we were able to automatically synthesize variants of the Off-Path 
attack described in \Figr{sctp:ambiguity:attack} using \korg.
We consulted with the lead SCTP RFC author who confirmed that the misinterpretation we describe could enable such attacks.  The attack is not possible if the text is interpreted correctly.
Out of concern that a real implementation might have misinterpreted the RFC document, we 
manually analyzed the source for both the Linux and FreeBSD implementations, 
and tested both implementations with \textsc{PacketDrill}, finding that neither made this mistake.
To make the text unambiguous, we suggest adding the following sentence, which disambiguates the meaning of \emph{it} and \emph{its} in the original quote.
\begin{quote}
\emph{The verification tag used in the packet containing the \InitAck chunk MUST
be the initiate tag of the newly received \Init chunk.}
\end{quote}
This suggestion has not yet resulted in an RFC erratum.

The second ambiguity we found was in $\mathsection$8.5:
\begin{quote}
\emph{When receiving an SCTP packet, the endpoint MUST ensure that the
value in the Verification Tag field of the received SCTP packet
matches its own tag.}
\end{quote}
The problem is that $\mathsection$8.5 does not say \emph{when} the vtag check should happen with respect
to other checks.  In particular, $\mathsection3.3.3$ says that an endpoint in \CookieWait who receives an \InitAck with an invalid itag should respond with an \Abort -- but it is unclear whether this still applies before or after the vtag check in $\mathsection$8.5.
Under the former interpretation, an endpoint in \CookieWait who receives an \InitAck with both an invalid itag and an invalid vtag would respond with an \Abort, whereas under the latter interpretation, the endpoint would silently discard the packet.
To clarify the ambiguity, we proposed the following erratum, which the SCTP RFC committee accepted in Erratum 7852 to RFC 9260~\cite{ourErratum}:
\begin{quote}
\emph{When receiving an SCTP packet, the endpoint MUST first ensure that the
value in the Verification Tag field of the received SCTP packet
matches its own tag before processing any chunks or changing its state.}
\end{quote}
Although it was not obvious to us that misinterpreting the ambiguous text could open the protocol to attack,
when we modeled the second ambiguity and analyzed it with \korg, we were able to find an attack in which an Off-Path
attacker could inject an \InitAck in order to disrupt an association attempt.
The attack is illustrated in \Figr{ambiguity-attack-2}.

\setlength{\instdist}{1.1cm}
\setlength{\levelheight}{0.4cm}
\begin{figure}[H]
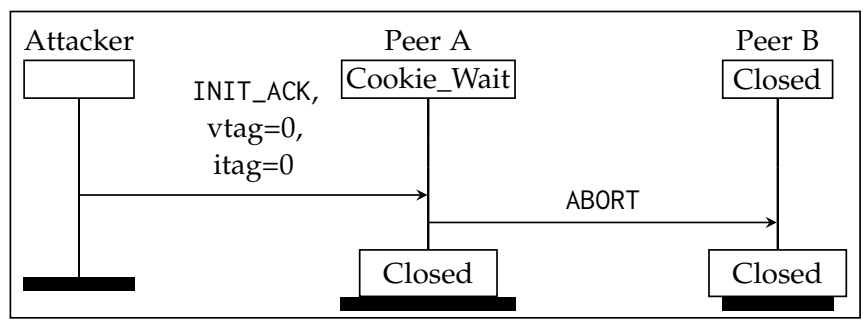

\centering
\begin{adjustbox}{width=0.6\textwidth,center}
\centering
\begin{msc}[head top distance=0.7cm,msc keyword=,left environment distance=1cm,right environment distance=1.2cm,foot distance=0.1cm, instance distance=3cm, action width=3cm]{}
    \declinst{attacker}{Attacker}{}
    \declinst{A}{Peer A}{\CookieWait}
    \declinst{B}{Peer B}{\Closed}

    \nextlevel
    \nextlevel

    \mess{{\parbox[b]{.85\instdist}{\centering \InitAck, vtag=0, itag=0}}}{attacker}{A};

    \nextlevel

    \mess{\Abort}{A}{B}

    \nextlevel

    \action[action width=2cm]{\Closed}{A}
    \action[action width=2cm]{\Closed}{B}
\end{msc}
\end{adjustbox}
\caption{Second ambiguity attack.}
\Figl{ambiguity-attack-2}
\end{figure}
\setlength{\levelheight}{0.6cm}
\setlength{\instdist}{2.5cm}

\section{Related Work}\Secl{korg:related}
Prior works formalized security problems using game theory (e.g., \textsc{FlipIt} \cite{FlipIt}, \cite{klavska2018automatic}), ``weird machines" \cite{USENIX11}, attack trees \cite{ACMCSUR19}, Markov models \cite{valizadeh2019toward}, and other methods.  
Prior notions of attacker quality include $\mathcal{O}$-complexity \cite{LangSec18}, expected information loss \cite{Citeseer13}, or success probability \cite{meira2019synthesis,vasilevskaya2014}, which is similar to our concept of $\forall$ versus $\exists$-attackers.  The formalism of \cite{vasilevskaya2014} also captures attack consequence (cost to a stakeholder).

Nondeterminism abstracts probability, e.g., a $\forall$-attacker 
    is an attacker with $P(\text{success}) = 1$, and, under fairness conditions, 
    an $\exists$-attacker is an attacker with $0 < P(\text{success}) < 1$.  
    Probabilistic approaches are advantageous when the existence of an
    event is less interesting than its likelihood.
    For example, a lucky adversary \emph{could} randomly guess my RSA modulus, 
    but this attack is too unlikely to be interesting.
    We chose to use nondeterminism over probabilities for two reasons: first, 
    because nondeterministic models do not require prior knowledge of event
    probabilities, but probabilistic models do; and second, because the 
    non-deterministic model-checking problem is cheaper than its probabilistic
    cousin~\cite{vardi1999probabilistic}.
    Nevertheless, we believe our approach could be extended to probabilistic
    models in future work.
    Katoen provides a nice survey of probabilistic model checking~\cite{katoen2016probabilistic}.

One work, which built on our own and studied TCP and ABP, suggested reactive controller synthesis (RCS) as an alternative to \korg's approach~\cite{matsui2022synthesis}.  \korg generates attacks that sometimes succeed ($\exists$-attackers), depending on choices made by the peers, whereas the RCS method only outputs attacks that always succeed ($\forall$-attackers); but such attacks do not always exist.
Another approach, which Fiterau-Brostean~et.~al.~\cite{fiterau2022automatabased}
successfully applied to various SSH~\cite{rfc4253} and DTLS~\cite{rfc_dtls} implementations,
describes incorrect behaviors using automata (rather than properties).
This specification style makes sense when generic bug patterns are known ahead of time.

Attacker synthesis work exists in cyber-physical systems~\cite{phan2017synthesis,bang2018online,huang2018algorithmic,lin2019synthesis,meira2019synthesis},
most of which define attacker success using bad states 
(e.g., reactor meltdown, vehicle collision, etc.) 
or information theory 
(e.g., information leakage metrics).
Problems include 
the \emph{actuator attacker synthesis problem} \cite{lin2019synthesisActuator};
the \emph{hardware-aware attacker synthesis problem} \cite{trippel2019security};
and the \emph{fault-attacker synthesis problem} \cite{barthe2014synthesis}.
However, to the best of our knowledge, we are the first to define and propose an approach to attacker synthesis for protocols.

There are many automated 
attack \emph{discovery} tools, which in contrast to attacker synthesis,
are in general sound but incomplete.
Each such tool is crafted to a particular variety of bug or mechanism of attack,
e.g., \textsc{SNAKE}~\cite{IEEE15} 
(which fuzzes network protocols), 
\textsc{TCPwn}~\cite{jero2018automated} 
(which finds throughput attacks against TCP congestion control implementations),
\textsc{MACE}~\cite{cho2011mace}
(which uses concolic execution to find vulnerabilities in protocol implementations),
\textsc{SemFuzz}~\cite{SemFuzz}
(which derives vulnerability proof-of-concept code from written disclosures),
\textsc{Tamarin}~\cite{meier2013tamarin} and \textsc{ProVerif}~\cite{blanchet2016modeling}
(which find attacks against secrecy in cryptographic protocols),
and so on~\cite{hoque2017analyzing,huang2012crax,kayacik2009generating}.
Some of these tools (such as our own, \korg) 
are general purpose, designed to attack any correctness property,
while others (e.g., \textsc{Tamarin} or \textsc{ProVerif}) 
are designed to target specific types of properties such as secrecy and trace-equivalence.
Note that of those, \textsc{SNAKE} was previously applied to TCP and DCCP, and \textsc{TCPwn} was applied to 
multiple TCP implementations.

Saini and Fehnker's work~\cite{saini2017evaluating} is the only one we are aware of that 
studied SCTP in the context of an attacker using formal methods.
But their attacker was only capable of sending \Init messages,
in contrast to our \amodels which are much more sophisticated,
and their attacker could not spoof the port and IP of a peer.
Hence, they could not model (and so did not find) the CVE attack.

The Internet Research Task Force (IRTF) is interested in incorporating formal methods more into the RFC drafting process.  To this end, they created a usable formal methods research group~\cite{fmietf}.  Examples of techniques the group is interested in incorporating include the NLP approach we proposed in a prior work~\cite{pacheco2022automated} (which uses \korg), as well as another such approach proposed by Yen~et.~al., which semi-automatically detects ambiguities in RFC documents~\cite{yen2021semi}.

\section{Conclusion}

In this chapter we showed how, given a protocol handshake model, some LTL specification
it satisfies, and an \amodel indicating the placement and capabilities of an attacker,
one can automatically synthesize a corresponding attack (or determine that none exists).
Although many prior works used formal methods to find attacks against systems or protocols,
and there exists a body of work proposing attacker synthesis techniques for cyber-physical systems 
(which we reviewed in \Secr{korg:related}),
to the best of our knowledge, we are the first to propose a generalizable framework
and approach to the synthesis of attacks against protocols.

Although we focus entirely on transport protocol handshakes, in principal, our approach should
work for other kinds of protocols such as small distributed systems (importantly, systems that are small enough to avoid state-space explosion in a model checking context), concurrent programs with shared resources, etc.
Another interesting line of inquiry is attacker synthesis (e.g.,~\cite{ardeshiricham2019verisketch}).
The tight integration of attack and defense synthesis in a CEGIS-style loop merits future research.

Finally, our SCTP case study highlights how in contrast to heuristic attack discovery tools,
an attacker synthesis approach has the advantage of being able to rule out attacks, which is useful
for confirming that a patch for a given vulnerability indeed accomplishes its stated goal.
However, a disadvantage of attacker synthesis is that the technique is only as good as the model it is fed,
and a very detailed model will lead to state-space explosion, causing \korg to give up without
producing an attack or determining that none exist.
For this reason, in order to gain full assurance that a protocol is totally robust against attacks,
one would need a full LTL specification of what it means for the protocol to be correct, broken up into many 
small (checkable) properties, in addition to some kind of refinement argument indicating that the simplified
model we feed to \korg is an accurate representation of the complete protocol with respect to the correctness
specification.  This could be done using a hybrid approach involving both theorem proving and model checking, as was done in~\cite{bisimulation2003linking}.

\chapter{Conclusion}\Chapl{conclusion}

In this dissertation we studied the functionality, performance, correctness,
and security of transport protocols using a diversity of formal methods,
each of which we explained in \Chapr{intro}.
First, in \Chapr{karn-rto}, we formally defined Karn's Algorithm and proved
what precisely it measures, using inductive invariants in Ivy.
Then we moved to \acls, where we formalized the RTO computation based on 
those RTT measurements.
We showed that when the RTT measurements are 
bounded then so are the internal variables of the RTO, yet nevertheless,
infinitely many timeouts could occur.
Then, in \Chapr{gbn}, we shifted our focus from the timeout mechanism to the
sliding window procedure of Go-Back-$N$.
We defined a realistic network model and formally proved that under ideal conditions
Go-Back-$N$ can achieve perfect efficiency.
Finally, we proved a novel lower bound on the efficiency of Go-Back-$N$ when the sender's
constant transmission rate out-paces the rate at which the bucket refills,
in the absence of reordering.

Between \Chapr{karn-rto} and \Chapr{gbn} we explored two different approaches to the 
analysis of protocol performance: one based on real analysis, and another based on inductive invariants.
However, neither approach involved actually concertizing the real time-line, as was done in~\cite{killian2010finding} or~\cite{sharma2023performal}.  This is especially important for understanding metrics like throughput, which take a duration of time as a denominator.  The natural next step for our research is therefore to extend our models to support a real time-line, so that we can derive concrete performance bounds using actual time durations, whether they be drawn from a distribution, sampled from a real implementation, represented symbolically, etc.  We think that a refinement framework may provide a natural way to connect models with time to more abstract models without, like what was done in~\cite{IEEEACMTON02}.

Next we turned our attention to the actual handshake mechanisms of transport protocols,
which are finite state and amenable to model checking.
In \Chapr{handshakes} we formally modeled the TCP, DCCP, and SCTP handshake procedures 
in \promela and defined LTL correctness properties for each based on a close reading of
the corresponding RFC documents.
We proved that each handshake model satisfied its correctness properties using the model checker \spin.
We did not, however, connect our finite state handshake models to our infinite state models of Karn's Algorithm, the RTO computation, or Go-Back-$N$.
Making this kind of connection and looking at the interplay between the various protocol components
is a natural next step.
This can be done without needing to choose one of either theorem proving or model checking,
by using a hybrid approach involving both~\cite{bisimulation2003linking}.

Having proven our handshake models correct in isolation, we then moved on to the question of whether they are also correct in the presence of an attacker.
In \Chapr{korg}, we proposed a general framework and problem statement for automated attacker synthesis,
and a solution based on LTL model checking, which we proved to be both sound and complete,
and implemented in the open-source tool \korg.
We applied \korg to our TCP, DCCP, and SCTP models using three representative \amodels 
(outlined in \Secr{amodels}), and explained our results.
In general we found that \korg found attacks or determined that none existed in a matter of seconds,
although we also encountered some model/property combinations which took considerably longer,
due to the state-space explosion problem.
Nevertheless, \korg never failed to either find an attack,
or exhaust the search space looking for one, in any of our applications of it
to TCP, DCCP, or SCTP.
Our SCTP analysis was the most detailed and centered on a recent CVE and subsequent patch.
We showed that the attack in question could be found automatically when the patch was disabled,
and moreover, that the patch closed the vulnerability.
The vulnerability in question was enabled by an ambiguity in the RFC text, and inspired by this problem,
we found two more ambiguities, and showed that both could enable a vulnerabilities
if misinterpreted, which the lead SCTP RFC author confirmed.
We proposed two errata to the SCTP RFC, of which so far, one was accepted.

Although the automated attacker synthesis problem we posed was quite general, our solution relies on model checking and therefore does not feasibly scale to large distributed systems, and in fact \korg does not yet support non-centralized \amodels where the attacker consists of multiple coordinated processes.  In order to synthesize attacks against large, distributed systems, we will need other synthesis techniques.  One example can be found in our recent work analyzing \textsc{GossipSub}, a peer-to-peer system used in \textsc{Ethereum} and \textsc{FileCoin}, where we built a custom event \emph{generator} which could steer a system toward a vulnerable state~\cite{kumar2024formal}.  Another interesting direction is the extension of our attacker synthesis approach to other logics beyond LTL, such as Computational Tree Logic, Dynamic Epistemic Logic, Signal Temporal Logic, etc., as well as other kinds of models beyond finite Kripke structures, such as what Oakley~et.~al. did for discrete-time Markov chains~\cite{oakley2022adversarial}.  Building on this work, we would like to investigate applications of probabilistic programming to automated attacker synthesis, where the goal is to steer a system toward low-probability, deleterious outcomes (such as a tied vote in \textsc{RAFT}).
Finally, we note that recent work on the synthesis of distributed protocols may shed light on the corollary problem
for attacks against them; see, e.g.,~\cite{egolf2024efficient}.

Finally, even finite state models such as those outlined in \Chapr{handshakes} are time-intensive to write and validate, making techniques such as model checking and attacker synthesis difficult for practitioners to use as part of their day-to-day engineering workflow.  
More generally, model and specification engineering presents a considerable barrier to the adoption of even
lightweight formal methods in practice~\cite{goldstein2022some,chong2016report}.
This problem can be ameliorated using automated model extraction techniques.  As an example, in a prior work, we used natural language processing to extract protocol handshake models from corresponding RFC documents, which we then attacked using \korg~\cite{pacheco2022automated}.  We found that even ``partial'' models with mistakes or omissions could be used to find real attacks, which succeeded against canonical, hand-written models.  In light of the recent advent of powerful large language models, this research direction deserves further attention.
More speculatively, the converse may also be true -- that is, large language models may benefit from the integration of formal methods.
Regardless, the Internet as a whole stands to benefit from more formal verification,
and the biggest barrier to widespread adoption of these techniques currently is that they are simply 
hard to use.

\bibliographystyle{unsrt}
\bibliography{main}

\chapter{Appendix}

\subsection{Receiver Strategies}\Secl{subsec:delay-ack}
There are numerous possible strategies for when the receiver should send an \ack,
some of which are referred to as \emph{delayed \ack} algorithms (because they involve
a timer).  We summarize a number of these in \Tabr{delACK}.

\begin{table}[h!]
\begin{tabular}{lp{0.7\textwidth}}
Source & Receiver Strategy\\\hline
RFC 1122~\cite{rfc1122} & At least every other packet or every half second.\\
RFC 4681~\cite{rfc4681} & At least every other packet or every second, and within half a second of each previously unACKed packet.\\
RFC 9000~\cite{rfc9000_quic} & Every other packet, every clearly reordered packet, or whenever a timer expires.\\
RFCs 4341~\cite{rfc4341}, 5690~\cite{rfc5690} & Twice per send window, i.e., every $N/2$ packets.\\
RFC 3449~\cite{rfc3449} & Dynamic scheme where ACK frequency scales with traffic.\\
Gomez \& Crowcroft~\cite{ackPull} & Whenever the sender requests one.\\
Fairhurst et.~al.~\cite{quicScaling} & At least once per ten packets.\\
Kuhn et. al.~\cite{quicSatcom} & At least four times per RTT.\\
Chen et. al.~\cite{chen2008tcp} & Adaptive delay based on path length and end-to-end delay.\\
Altman \& Jimen{\'e}z~\cite{altman2003novel} & Adaptive delay based on sequence number.\\
Armaghani et. al.~\cite{armaghani2011performance} & Adaptive delay based on MAC layer collision probability.\\
\end{tabular}
\caption{Receiver strategies.  Adapted and expanded from~\cite{rfcQuicACK}.}\Tabl{delACK}
\end{table}

Empirical evidence in wireless networks suggests that, for fixed-frequency receiver strategies, there generally exists an optimal frequency depending on the network~\cite{chen2008tcp} -- and in some cases, the optimal strategy is to send an ACK after every $N^{\text{th}}$ packet received~\cite{singh2004tcp}.
But in wired networks, where losses are generally caused by the queuing mechanism, 
it is not so simple, with a variety of strategies being adopted by different protocols and implementations.

\subsection{Example LTL Formulae}\Secl{ltl:examples}

Example LTL formulae include: 
\begin{itemize}
    \item Lunch will be ready in a moment: $\X \texttt{lunch-ready}$.
    \item I always eventually sleep: $\G \F \texttt{sleep}$.
    \item I am hungry until I eat: $\texttt{hungry} \U \texttt{eat}$.
    \item $A$ and $B$ are never simultaneously in their \texttt{crit} states: $\G \neg (\texttt{crit}_A \land \texttt{crit}_B)$.
\end{itemize}

\end{document}